\documentclass[11pt,reqno]{amsart}

\usepackage{amsfonts}
\usepackage{amssymb}
\usepackage{amsmath}
\usepackage{latexsym}
\usepackage{graphicx}
\usepackage[table]{xcolor}
\usepackage{colordvi}
\usepackage{cite}
\usepackage[colorlinks=true,urlcolor=black,linkcolor=black,citecolor=black,pagebackref=black,]{hyperref}

\newcount\hh \newcount\mm
\mm=\time
\hh=\time
\divide\hh by 60
\divide\mm by 60
\multiply\mm by 60
\mm=-\mm
\advance\mm by \time
\def\timestamp{\number\hh:\ifnum\mm<10{}0\fi\number\mm}

\numberwithin{equation}{section}

\def\reftab#1{\ref{#1} on page~\pageref{#1}}

\newtheorem{thm}[equation]{Theorem}
\newtheorem{cor}[equation]{Corollary}

\newtheorem{conj}[equation]{Conjecture}

\theoremstyle{definition}
\newtheorem{prop}{Proposition}[section]
\newtheorem{defn}[prop]{Definition}

\newtheorem{example}[prop]{Example}

 \DeclareMathOperator{\tr}{tr} 

\numberwithin{equation}{section}

\renewcommand\a{\alpha}
\renewcommand\b{\beta}
\newcommand\g{\gamma}
\renewcommand\d{\delta}
\newcommand\e{\varepsilon}
\renewcommand\l{\lambda}

\newcommand\G{\Gamma}
\newcommand\f{\frac}
\newcommand\smallf[2]{{\textstyle{\frac{#1}{#2}}}}
\newcommand\smallerf[2]{{\scriptstyle{\frac{#1}{#2}}}}
\newcommand\srel[2]{\begin{smallmatrix} {#1} \\ {#2} \end{smallmatrix}}

\newcommand\calB{\mathcal B}

\newcommand{\Z}{{\mathbb{Z}}}
\newcommand{\R}{{\mathbb{R}}}

\newcommand{\C}{{\mathbb{C}}}

\newcommand{\Q}{{\mathbb{Q}}}

\newcommand{\U}{{\mathbb{H}}}

\renewcommand\Re{\text{Re~}}

\renewcommand\({\left(}
\renewcommand\){\right)}
\newcommand{\ttwo}[4]{
\(\begin{smallmatrix}{#1} & {#2}
\\ {#3} & {#4} \end{smallmatrix}\)}
\newcommand{\tthree}[9]{
\(\begin{smallmatrix}{#1} & {#2} & {#3}
\\ {#4} & {#5} & {#6} \\ {#7} & {#8} & {#9} \end{smallmatrix}\)}

\newcommand{\AV}{\mathrm{AV}}

\newcommand{\gobble}[1]{}
  \newcommand{\rangeref}[2]{%
    \ref{#1}--\afterassignment\gobble\fam 0\ref{#2}%
  }
\makeatletter
\def\imod#1{\allowbreak\mkern10mu({\operator@font mod}\,\,#1)}
\makeatother

\def\hE#1#2#3{ {\hat E}^{#1}_{\textbf{#2};{#3}}}

\def\ZZ{{\mathbb Z}}
\def\IC{{\mathbb C}}

\def\IR{{\mathbb R}}

\def\IZ{{\mathbb Z}}

\def\cE{{\cal E}}

\def\cE{\mathcal{E}}
\def\cF{\mathcal{F}}
\def\cV{\mathcal{V}}
\def\cR{\mathcal{R}}

\newcommand{\be}{\begin{equation}}
\newcommand{\ee}{\end{equation}}
\newcommand{\bea}{\begin{eqnarray}}
\newcommand{\eea}{\end{eqnarray}}

\def\a{\alpha }
\def\g{\gamma }

\def\threeh{{\scriptstyle {3 \over 2}}}
\def\fiveh{{\scriptstyle {5 \over 2}}}

\def\R{\cR}
\def\bE{{\bf E}}

\def\hE{{\hat{E}}}

\def\calE{{\mathcal E}}

\def\rT{\textrm{T}}
\def\calT{{\mathcal T}}
\def\calU{{\mathcal U}}

\def\bsz{\backslash\{0\}}

\def\calV{{\mathcal V}}

\def\nn{\nonumber}
\def\half{{\scriptstyle {1 \over 2}}}
\def\third{{\scriptstyle {1 \over 3}}}
\def\quart{{\scriptstyle {1 \over 4}}}

\title[Small representations, string instantons, and Fourier modes]{Small representations, string instantons, and Fourier modes of  Eisenstein series  \\ \hspace{1cm} \\ \scriptsize{Michael B. Green, Stephen D. Miller, and Pierre Vanhove}  \\ \hspace{1cm} \\ \scriptsize{with appendix ``Special unipotent representations'' by Dan Ciubotaru and Peter E.~Trapa}
}

\author[M.B. Green]{}
\address{Michael B. Green\\
Department of Applied Mathematics and
Theoretical Physics\\
 Wilberforce Road, Cambridge CB3 0WA, UK}
\email{ M.B.Green@damtp.cam.ac.uk}
\author[S.D. Miller]{}
\address{Stephen D Miller\\
Department of Mathematics\\
 Rutgers University, Piscataway, NJ 08854-8019, USA}
\email{miller@math.rutgers.edu}
\author[P. Vanhove]{}
 \address{Pierre Vanhove\\
Institut des Hautes Etudes Scientifiques\\
 Le Bois-Marie, 35 route de Chartres\\
 F-91440 Bures-sur-Yvette, France\hfill\break
Institut de Physique Th{\'e}orique,\\
CEA, IPhT, F-91191 Gif-sur-Yvette, France\\
CNRS, URA 2306, F-91191 Gif-sur-Yvette, France}
\email{pierre.vanhove@cea.fr}

\thanks{DAMTP-2011-102, IPHT-t11/188, IHES/P/11/25}
\date{}

\begin{document}

 \begin{abstract}

This paper concerns some novel features of    maximal parabolic
Eisenstein series at certain special values of their analytic
parameter, $s$.  These series arise as coefficients in the $\R^4$ and
$ \partial^4\, \R^4$ interactions in the low energy expansion of the
scattering amplitudes in  maximally supersymmetric string theory
reduced to $D=10-d$ dimensions on a torus, $\rT^{d}$ ($0\le d \le 7$).
For each $d$ these amplitudes are automorphic functions on  the  rank
$d+1$ symmetry
group $E_{d+1}$.

Of particular significance is the orbit content of the Fourier modes of these series when expanded
in three different parabolic subgroups, corresponding to certain
limits of string theory. This  is of interest in the classification of
a variety of instantons that correspond to minimal or ``next-to-minimal''  BPS orbits.  In the limit of decompactification from $D$ to $D+1$ dimensions many such instantons are related to charged $\smallf 12$-BPS or $\smallf 14$-BPS black holes with euclidean world-lines wrapped around the large dimension. In a different limit the instantons give nonperturbative corrections to string perturbation theory, while  in a third limit they describe nonperturbative contributions in eleven-dimensional supergravity.

A proof is given that these three distinct Fourier expansions have certain vanishing coefficients that are expected from string theory.  In particular, the Eisenstein series for these special values of $s$  have markedly fewer Fourier coefficients
than typical maximal parabolic Eisenstein series.  The corresponding mathematics involves showing that the wavefront sets of the Eisenstein series in question are supported on only a limited number of coadjoint nilpotent orbits
 -- just the minimal and trivial orbits in the $\smallf12$-BPS case, and  just the next-to-minimal, minimal and trivial orbits in the $\smallf14$-BPS case.  Thus as a byproduct we demonstrate
 that the next-to-minimal representations occur automorphically for $E_6$, $E_7$, and $E_8$, and hence the first two nontrivial low energy coefficients in scattering amplitudes can be thought of as exotic $\theta$-functions for these groups.
The proof includes an appendix by Dan Ciubotaru and Peter E. Trapa which calculates wavefront sets for these and other special unipotent representations.

{\bf keywords:~}automorphic forms, scattering amplitudes, string theory, small representations, Eisenstein series, Fourier expansions,  unipotent representations, charge lattice, BPS states, coadjoint nilpotent orbits.
\end{abstract}

\maketitle
\newpage\tableofcontents

\section{Introduction}\label{sec:introduction}

String theory is expected to be invariant under a very large set of discrete symmetries (``dualities''),
 associated with arithmetic subgroups of a variety of reductive Lie
 groups.    For example,  maximally supersymmetric string theory (type
 II superstring theory),   compactified on a $d$-torus to $D = 10-d$
 space-time dimensions, is strongly conjectured to be invariant under
 $E_{d+1}(\Z)$, the integral  points of the  rank $d+1$ split real
 form\footnote{The split real forms are conventionally denoted
   $E_{n(n)}$, but in this paper we will truncate this to $E_n$ except when other forms of $E_n$ are needed.} of one of the groups in the
 sequence $E_8$, $E_{7}$, $E_{6}$, $Spin(5,5)$, $SL(5)$, $SL(3)\times
 SL(2)$, $SL(2)\times \IR^+$, $SL(2)$ listed in table~\ref{tab:Udual}.\footnote{\label{Gzdeffootnote}
Unfortunately the literature contains  some disagreement over precisely which groups  $E_{d+1}(\IR)$ occur here, an ambiguity amongst the split real groups having the same Lie algebra.  For example some authors have $SO(5,5,\IR)$ instead of its double cover $Spin(5,5,\IR)$; in general possible groups are related by taking quotients by a subgroup $G_0$ of the center of the larger group.  The choices listed here, which  represent the current consensus, are  each the real points of an (algebraically) simply connected Chevalley group.  (The real groups $E_{d+1}(\IR)$ and $K_{d+1}$ are not topologically simply connected, except in the trivial $D=10A$ case.)

Although we will try to be precise in our definitions, this discrepancy does not affect the results in this paper.  We note, in particular, that $E_{d+1}(\Z)$ is mathematically defined as the stabilizer of the Chevalley lattice in the Lie algebra ${\mathfrak e}_{d+1}$ under the adjoint action.  Since the  center acts trivially under the adjoint action, the integral points of the larger group factors as the direct product of $G_0$ with the integral points of the smaller group.  In particular the Eisenstein series for the two groups are the same (see for example (\ref{SpinddandSoddseries})).}
 \begin{figure}[ht]
 \centering\includegraphics[width=8cm]{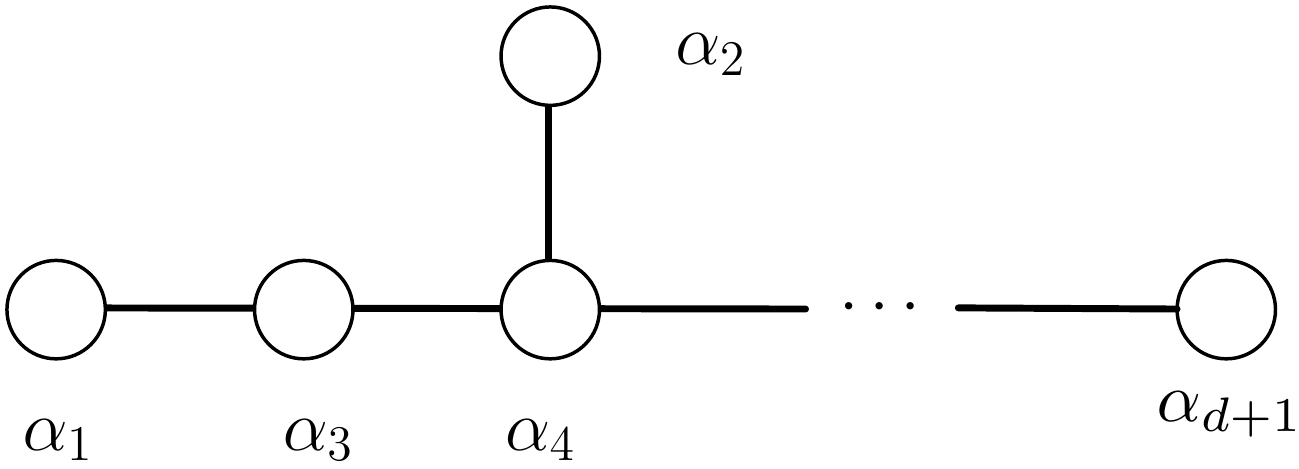}
 \caption{\label{fig:dynkin}  The  Dynkin  diagram for  the rank  $d+1$ Lie group $E_{d+1}$, which defines the symmetry group for $D=10-d$.}
 \end{figure}

These symmetries  severely constrain the dependence of string scattering amplitudes on the symmetric space coordinates (or ``moduli''),  $\phi_{d+1}$,  which parameterise the coset
$E_{d+1}/K_{d+1}$, where the stabiliser $K_{d+1}$ is the maximal
compact subgroup of $E_{d+1}$.  The list of these
symmetry\footnote{The continuous groups,
  $E_{d+1}(\IR)$, will be referred to as {\sl symmetry groups} while the  discrete arithmetic  subgroups, $E_{d+1}(\ZZ)$, will be referred to as  {\sl duality groups}.} groups and stabilisers is given in   table~\ref{tab:Udual}.  These moduli are scalar fields that are interpreted as coupling constants in string theory.  A general consequence of the dualities is that scattering amplitudes are functions of $\phi_{d+1}$ that must transform as automorphic  functions  under the appropriate duality group $E_{d+1}(\Z)$.  It is   difficult to determine the precise restrictions these dualities impose on general amplitudes, but certain exact properties  have been obtained in the case of the   four-graviton interactions, where a  considerable amount of information has been obtained  for  the first three terms in the low energy (or ``derivative'')  expansion of the four graviton scattering amplitude  in \cite{Green:2010kv} (and references cited therein).  These are described by terms in the effective action of the form
\be
\calE_{(0,0)}^{(D)} (\phi_{d+1})\, \R^4\,,\quad  \calE_{(1,0)}^{(D)}(\phi_{d+1})\, \partial^4\, \R^4\,,  \ \ \  \text{and}  \ \ \  \calE_{(0,1)}^{(D)}(\phi_{d+1})\,\partial^6\, \R^4\,,
\label{coeffs}
\ee
 where the symbol $\R^4$ indicates a contraction of four powers of the Riemann tensor with a standard rank  16 tensor.
 The coefficient functions,  $\calE_{(p,q)}^{(D)}(\phi_{d+1})$,  are automorphic functions that are the main focus of our interests (the notation is taken from \cite{Green:2010wi, Green:2010kv} and will be reviewed later  in (\ref{amp})).  More precisely we will focus on the three terms shown in \eqref{coeffs}  that are protected by supersymmetry, which accounts for the relatively simple form of their coefficients.
 \begin{table}[h]
   \centering
   \begin{tabular}{||c|c|c|c||}
   \hline
 $D$ & $E_{d+1}(\IR)$ & $K_{d+1}$&$E_{d+1}(\ZZ)$\\
 \hline
 10A&$ \IR^+$&1&1\\
 10B&$SL(2,\IR)$&$SO(2)$& $SL(2,\ZZ)$\\
 9&$SL(2,\IR)\times \IR^+$& $SO(2)$& $SL(2,\ZZ)$\\
 8&$SL(3,\IR)\times SL(2,\IR)$& $SO(3)\times SO(2)$& $SL(3,\ZZ)\times SL(2,\ZZ)$ \\
 7&$SL(5,\IR)$ & $SO(5)$& $SL(5,\ZZ)$\\
 6 &  $Spin(5,5,\IR)$ & $(Spin(5)\times Spin(5))/\ZZ_2$& $Spin(5,5,\ZZ)$\\
 5&$ E_{6}(\IR)$& $USp(8)/\ZZ_2$& $ E_{6}(\ZZ)$\\
 4 &$ E_{7}(\IR)$& $SU(8)/\ZZ_2$& $E_{7}(\ZZ)$\\
 3 &$ E_{8}(\IR)$& $Spin(16)/\ZZ_2$& $E_{8}(\ZZ)$\\
 \hline
   \end{tabular}
   \label{tab:Udual}
      \caption{
      The   symmetry   groups  of   maximal  supergravity   in  $D=10-d\le 10$
     dimensions.  The group $E_{d+1}(\IR)$ is a split real form of
       rank  $d+1$,
    and $K_{d+1}$ is its maximal compact subgroup.
     In  string  theory  these  groups  are  broken  to  the  discrete
     subgroups, $E_{d+1}(\ZZ)$, as indicated in the last column (see \cite{Hull:1994ys} and its updated version in \cite{Hull:2007}).
     The split real form $E_{d+1}(\IR)$   is determined among possible
     covers or quotients  by its maximal compact subgroup $K_{d+1}$,
     which shares the same fundamental group. The terminology $10A$ and $10B$ in the first column  refers to the two possible superstring theories (types IIA and IIB) in $D=10$ dimensions. }
 \end{table}

The coefficients of the first two terms satisfy Laplace eigenvalue equations (\ref{laplaceeigenone}-\ref{laplaceeigentwo}) and are subject to specific boundary conditions that  are required for consistency with string perturbation theory and  M-theory.   The solutions to these equations are particular maximal parabolic Eisenstein series that were studied in \cite{Green:2010wi} (for cases with rank $\le 5$) and \cite{Green:2010kv} (for the $E_6$, $E_7$ and $E_8$
cases), and will be reviewed in the next section.    The required
boundary conditions in each limit  amount to conditions on the
constant terms  in the expansion of these series in three limits
associated with particular  maximal\footnote{The
   $D= 8$ case is degenerate and also involves  non-maximal
   parabolics (see table~\ref{tab:Udual}).} parabolic subgroups of relevance
to the string theory analysis.   Such subgroups have the form
$P_\alpha = L_\alpha\, U_\alpha$, where $\alpha$ labels a simple root,
$U_\alpha$ is the unipotent radical and $L_\alpha=  GL(1) \times
M_\alpha$ is the Levi factor.\footnote{For clarity, we emphasize that its usage here indicates that every element of $L_\a$ can be written as an element of $GL(1)$ times an element of $M_\a$ (and not that $L_\a$ is the direct product of the two factors, which is a stronger statement). }  The  three subgroups of relevance here
have Levi factors  $L_{\alpha_1}= GL(1)\times Spin(d,d)$,
$L_{\alpha_2}= GL(1)\times SL(d+1)$, and $L_{\alpha_{d+1}}=
GL(1)\times E_d$,  respectively.  In each case the $GL(1)$ parameter,
$r$, can be thought of as measuring the distance to the
cusp\footnote{Each of the groups we are considering has a single cusp.
  The various limits correspond to different ways of approaching this
  cusp.},  as will be discussed in the next section.
 A key feature of the boundary conditions is that they require these constant terms to have very few components with distinct  powers of the parameter  $r$.    These conditions pick out the unique
solutions  to the Laplace equations, which
are,\footnote{In~\cite{Green:2010wi,Green:2010sp,Green:2010kv}
    the series were indexed by the  label $[1\, 0\cdots0]$ of
    the root $\alpha_1$. In the present paper, we will index the series
    according the labeling of the simple root in
    figure~\ref{fig:dynkin}. We have as well changed the
    normalisations of the Eisenstein series, since our series there was instead
    $\bE_{[10\cdots0];s}^{E_{d+1}}=2\zeta(2s) E^{E_{d+1}}_{\alpha_1;s}$.
  }
\begin{equation}
\calE_{(0,0)}^{(10-d)} \ \ = \ \ 2\,\zeta(3)\, E^{E_{d+1}}_{\alpha_1;\, \threeh}\,,
\label{rfourcoeff}
\end{equation}
for the groups $E_1$, $E_4$, $E_5$, $E_6$, $E_7$, and $E_8$
\cite{Green:2010wi,Green:2010kv} and
\begin{equation}
\calE_{(1,0)}^{(10-d)} \ \  = \ \  \zeta(5)\,E^{E_{d+1}}_{\alpha_1;\, \fiveh}\,,
\label{dfourrfourcoeff}
\end{equation}
for the groups $E_1$, $E_6$, $E_7$, and $E_8$ \cite{Green:2010kv}. Here
$E_{\beta;s}^{G}$ is the maximal parabolic Eisenstein series for a parabolic subgroup  $P_\b\subset G$  that is specified by the node $\beta$ of the Dynkin diagram (see (\ref{maxparabeisdef}) for a precise definition).
 This generalizes results for the $SL(2,\Z)$ case (relevant to the ten-dimensional type IIB string theory). The functions  $\calE_{(0,0)}^{(10-d)}$ and $\calE_{(1,0)}^{(10-d)}$ in the intermediate rank cases involve linear combinations of Eisenstein series  \cite{Green:2010wi}, which will be discussed later in section~\ref{lowrank}.
The third coefficient function, $\calE_{(0,1)}^{(10-d)}$  satisfies an
interesting inhomogeneous Laplace equation and is not an Eisenstein
series~\cite{Green:2005ba,Green:2010kv}.  Its constant terms in the
three limits under consideration were also analysed in the earlier
references but it  will not be considered in this paper, which is
entirely concerned with Eisenstein series.

In other words, our previous work showed that  the particular
Eisenstein series in \eqref{rfourcoeff} and \eqref{dfourrfourcoeff}
have strikingly sparse constant terms as required to correctly
describe the coefficients of the $\smallf 12$-BPS and $\smallf 14$-BPS
interactions.  But the string theory boundary conditions also
determine the support of the non-zero Fourier coefficients in each of
the three limits under consideration.
In string theory, the non-zero Fourier modes describe instanton
contributions to the amplitude. These are classified in
BPS orbits obtained by acting on a representative instanton
configuration with the appropriate Levi subgroup.
A given instanton configuration generally depends on only a subset of
the parameters of the Levi group, $L_\alpha = GL(1) \times M_\alpha$,
so that a given orbit depends on the  subset of the moduli that live
in a coset space of the form $M_\alpha/H^{(i)}$,  where
$H^{(i)}\subset M_\a$ denotes the stabiliser of the $i$-th orbit.
The dimension of the  $i$-th orbit   is the dimension of this coset
space.

In particular, the coefficients in the $s=3/2$ cases covered by \eqref{rfourcoeff} must  be
  localized within  the smallest possible non-trivial orbits
  (``minimal orbits'') of the Levi actions, as required by the
  $\smallf 12$-BPS condition.  Furthermore,  in the $s=5/2$ cases covered by \eqref{dfourrfourcoeff} the
  coefficients  are shown to be localized within the ``next-to-minimal'' (NTM)
  orbits (see section~\ref{sec:mathematicsbackground}).  The role of next-to-minmal orbits was also considered in~\cite{Pioline:2010kb}.  However, the specific suggestion there was based on the next-to-minimal representations of Gross and Wallach~\cite{GrossWI,GrossWII}, who did not consider the split groups of relevance to the duality symmetries of type IIB string theory, which have very distinctive properties (as we shall see).
\vspace{.2cm}

This provides motivation  from string theory for the following
\vspace{.2cm}

\noindent
 \centerline {\bf String motivated vanishing of Fourier modes of Eisenstein series:}
{\it
 \begin{itemize}\label{stringmotivatedprediction}
\item[(i)] The non-zero Fourier coefficients of  $E^{E_{d+1}}_{\alpha_1;\,
    \smallerf32}$ ($d=5,6,7$) in any of the three parabolic subgroups
  of relevance are localized within the smallest possible non-trivial orbits (``minimal orbits'') of the action of the Levi subgroup associated with that parabolic, as required by the $\smallf12$-BPS condition.

\item[(ii)] The non-zero Fourier coefficients of $E^{E_{d+1}}_{\alpha_1;\, \smallerf52}$ ($d=5,6,7$) are localized within ``next-to-minimal'' (NTM)   orbits, as required by the $\smallf14$-BPS condition.
\end{itemize}}

\vspace{.1cm}

While the special properties of the Fourier coefficients of the
$s=3/2$ series is implied by the results in \cite{grs}, the corresponding properties for the NTM orbits at $s=5/2$ is novel.     One of the main mathematical contributions of this paper is to give a
rigorous proof of these statements using techniques from representation
theory, by connecting these automorphic forms to small representations
of the  split real groups $E_{d+1}$.   The Fourier coefficients in the intermediate rank cases not covered by \eqref{rfourcoeff} and \eqref{dfourrfourcoeff} satisfy analogous properties as we will determine by explicit calculation later in this paper.


  \section{Overview of scattering amplitudes and Eisenstein series}

Since this paper covers topics of interest in both string theory and mathematics, this section will  present a brief description of the background to these topics from both points of view followed by a detailed outline of the rest of the paper.

\subsection{String theory Background}\label{sec:stringtheorybackground}

We are concerned with exact (i.e., non-perturbative) properties of the low energy expansion of the four-graviton scattering amplitude in dimension $D=10-d$, which   is a function of the moduli, $\phi_{d+1}$,  as well as of the  particle momenta $k_r$ ($r=1,\dots,4$) that are null Lorentz $D$-vectors ($k_r^2=k_r\cdot k_r=0$) which are conserved ($\sum_{r=1}^4 k_r=0$).
They arise in the invariant combinations (Mandelstam invariants), $s=-(k_1+k_2)^2$, $t=-(k_1+k_4)^2$ and $u=-(k_1+k_3)^2$ that satisfy $s+t+u=0$.  At low orders in the low-energy expansion the amplitude can usefully be separated into analytic and nonanalytic parts
  \be
  A_D(s,t,u) \ \ = \ \  A_D^{analytic}(s,t,u)  \ + \  A_D^{nonanalytic}(s,t,u)
  \label{ampcomp}
  \ee
  (where the dependence on $\phi_{d+1}$ has been suppressed).
 The analytic part of the amplitude has the form
  \be
  A_D^{analytic} (s,t,u) \  \ = \  \  T_D(s,t,u)\, \ell_D^6\,\R^4\,,
  \label{scalardef}
  \ee
where $\ell_D$ denotes the $D$-dimensional Planck length scale and    the factor $\R^4$ represents the particular contraction  of four Riemann curvature tensors,  $\tr(\R^4) - (\tr \R^2)^2/4$, that is fixed by maximal supersymmetry in a standard fashion~\cite{Green:1987mn}.  The scalar function  $T_D$ has the expansion
(in the Einstein frame\footnote{The Einstein frame is the frame in which lengths are measured in Planck units rather than string units, and is useful for discussing dualities.})
  \bea\label{amp}
 T_D(s,t,u) &=& \cE_{(0,-1)} \, \sigma_3^{-1}+ \sum_{p,q\geq0}   \calE^{(D)}_{(p,q)}\,
\sigma_2^p\, \sigma_3^q\\
&=&3\, \sigma_3^{-1}
  +  \,  \calE^{(D)}_{(0,0)} \,   +  \,    \calE^{(D)}_{(1,0)}\, \sigma_2 \,   +  \,   \calE^{(D)}_{(0,1)}\, \sigma_3 \,   +  \, \cdots \,.\nn
\eea
  Symmetry under interchange of the four gravitons implies that the Mandelstam invariants only appear in the combinations $\sigma_2$ and $\sigma_3$ with  $\sigma_n =(s^n + t^n + u^n)\, (\ell_D^{2}/4)^n$. Since $s,t,u$ are quadratic in momenta the successive terms in the expansion are of order $n = 2p +3q$ in powers of  $($momenta$)^2$.
The degeneracy, $d_n=\lfloor (n+2)/2\rfloor-\lfloor (n+2)/3\rfloor$,  of terms with power $n$  is given by the generating function\footnote{This is the same as the well-known dimension formula for the space of weight $2n$ holomorphic modular forms for $SL(2,\Z)$,  which are expressed  as polynomials in the (holomorphic)  Eisenstein series $G_4$ and $G_6$.},
\be
{1\over (1-x^2)(1-x^3)} \ \ = \  \  \sum_{n=0}^\infty d_n \,x^n\,,
\label{degen}
\ee
so $d_0=1$, $d_1=0$ and $d_n=1$ for  $2\leq n\le 5$.

The coefficient functions in \eqref{amp},
$\calE^{(D)}_{(p,q)}(\phi_{d+1})$,  are automorphic  functions of the
moduli $\phi_{d+1}$ appropriate to compactification  on $\rT^d$. The
first term on the right-hand side of (\ref{amp})
 is  identified with the tree-level contribution of classical
 supergravity and has a constant coefficient given by $\calE^{(D)}_{(0,-1)}(\phi_{d+1})=3$.  The terms of higher order in $s$, $t$, $u$ represent stringy modifications of supergravity, which depend on the moduli in a manner consistent with duality invariance.
This expansion is presented in the Einstein frame so  the curvature, $\cR$, is invariant under $E_{d+1}(\Z)$ transformations, whereas it transforms nontrivially in the string frame   since   it is nonconstant in
 $\phi_{d+1}\in E_{d+1}(\IR)/K_{d+1}$.

Apart from the   first
  term, the power series expansion  in \eqref{amp} translates into a
  sum of local interactions in the effective action.  The first two of
  these have the form
\be
\ell_{D}^{8-D}\int d^{D}x\, \sqrt{-G^{(D)}}  \,  \calE^{(D)}_{(0,0)}\, \R^4\,, \quad
\ell_{D}^{12-D}\int d^{D}x\, \sqrt{-G^{(D)}}\, \,
\calE^{(D)}_{(1,0)} \partial^4\R^4\,.
\label{effacts}
\ee
The three interactions with coefficient functions  $\calE^{(D)}_{(0,0)}$,  $\calE^{(D)}_{(1,0)}$  and  $\calE^{(D)}_{(0,1)}$  displayed in the second equality in \eqref{amp} are specially  simple since they are protected by supersymmetry from renormalisation beyond a given order in perturbation theory.
 In particular, the $\R^4$ interaction breaks $16$ of the $32$
 supersymmetries of the type II theories and is thus $\smallf 12$-BPS,
 while the $\partial^4 \R^4$ interaction breaks $24$ supersymmetries
 and is $\smallf 14$-BPS;  likewise,  the $\partial^6 R^4$ interaction
 breaks 28 supersymmetries and is $\smallf 18$-BPS.
 The next  interaction is the $p=2, q=0$ term in \eqref{amp},
 $\calE^{(D)}_{(2,0)}\,\partial^8\R^4$.
 Naively this interaction breaks all supersymmetries, in which case it
 is expected to be much more complicated, but it would be of interest
 to discover if  supersymmetry does constrain this interaction.\footnote{A discussion of the properties of
     $\cE^{(9)}_{(2,0)}$ in nine dimensions can be found
     in~\cite[section~4.1.1]{Green:2008bf}.}

It was argued in   \cite{Green:2010wi}, based on consistency under various dualities,  that  the coefficients $\calE^{(D)}_{(0,0)}$,  $\calE^{(D)}_{(1,0)}$  and  $\calE^{(D)}_{(0,1)}$  satisfy the equations
  \begin{eqnarray}
 \!\!\!\!\!\! \left( \Delta^{(D)}                  -{3(11-D)  (D-8)\over
    D-2}\right)\,\cE^{(D)}_{(0,0)}&=& \!  6\,\pi\,\delta_{D,8} \,,
\label{laplaceeigenone}\\
\!\!\!\!\!\! \left( \Delta^{(D)}  -{5(12-D) (D-7)\over D-2}\right)\,\cE^{(D)}_{(1,0)}&=&\! 40\,\zeta(2)\, \delta_{D,7} \,,
\label{laplaceeigentwo}\\
\! \!\! \!\!\! \left( \Delta^{(D)} -{6(14-D) (D-6)\over D-2} \right)\,\cE^{(D)}_{(0,1)}&=& \! -\left(\cE_{(0,0)}^{(D)}\right)^2\!\! + 120\zeta(3)  \delta_{D,6}\,,
\label{laplaceeigenthree}\end{eqnarray}
where $\Delta^{(D)}$  is the Laplace  operator on the  symmetric space  $E_{11-D}/K_{11-D}$.  The discrete  Kronecker $\delta$ contributions on the
right-hand-side of these equations arise  from anomalous behaviour and can be related to the logarithmic ultraviolet divergences of  loop amplitudes in maximally supersymmetric supergravity \cite{Green:2010sp}.

Recall that automorphic forms for $SL(2,\Z)$ have Fourier expansions (i.e., $q$-expansions)  in their cusp.  For higher rank groups,
automorphic forms have Fourier expansions coming from any one of
several maximal  parabolic subgroups $P_{\alpha_r}$, where the simple
root $\alpha_r$ corresponds to  node $r$  in the Dynkin
diagram for $E_{d+1}$  in figure~\ref{fig:dynkin}.
  We are particularly interested in this Fourier expansion for $r=1$, 2, or $d+1$, because
each of these expansions has a distinct string theory interpretation in terms of the contributions of instantons in the limit in which a special combination of moduli degenerate.
These three limits are:
\begin{itemize}
\item[(i)]  {\it The decompactification limit} in which one circular
  dimension, $r_d$, becomes large.  In this case the amplitude reduces to the
  $D+1$-dimensional case  with $D=10-d$.  The BPS  instantons of the
  $D=(10-d)$-dimensional theory  are classified by orbits of  the Levi
  subgroup  $GL(1)\times E_d $.  Apart from one exception,  these instantons  can be described in terms of  the wrapping of the world-lines of black hole states in the decompactified $D+1$-dimensional theory around the large circular dimension (the exception will be described later).
  This limit is  associated with the parabolic subgroup $P_{\alpha_{d+1}}$.

\item[(ii)] {\it The string perturbation theory limit} of small string coupling constant, in which  the string coupling constant, $\sqrt{y_D}$,  is small, and  string perturbation theory amplitudes are reproduced.
    The instantons are exponentially suppressed contributions  that
    are classified by orbits of the Levi subgroup
    $GL(1)\times Spin(d,d)$.
    This  limit is associated with the parabolic subgroup $P_{\alpha_1}$.

\item[(iii)] {\it The $M$-theory limit} in which the $M$-theory torus has large volume $\calV_{d+1},$ and the semi-classical approximation to eleven-dimensional supergravity is valid.  This involves
      the compactification of M-theory from $11$ dimensions  on the
      $(d+1)$-dimensional $M$-theory torus, where  the instantons are
      classified by orbits of the Levi subgroup $GL(1)\times SL(d+1)$.
 This  limit is associated with the parabolic subgroup $P_{\alpha_2}$.

\end{itemize}

\vspace{.1cm}

 The special features of the constant terms that lead to consistency of all perturbative properties in these three limits appear to be highly nontrivial, and indicate particularly special mathematical properties of the Eisenstein series that define the coefficients of the $\R^4$ and $\partial^4 \R^4$ interactions.  The
 solutions to equations
 (\ref{laplaceeigenone}-\ref{laplaceeigenthree}) satisfying requisite
 boundary conditions on the constant terms (zero modes) in the Fourier
 expansions in the limits (i), (ii), and (iii)   were obtained for
 $7\le D \le 10$ in   \cite{Green:2010wi}, and for $3\le D \le 6$ in
 \cite{Green:2010kv}.   In particular, \eqref{rfourcoeff} and
 \eqref{dfourrfourcoeff} were found to be solutions for the cases with
 duality groups $E_6$, $E_7$ and $E_8$.
Whereas the coefficient functions $\calE_{(0,0)}^{(D)}$ and $\calE_{(1,0)}^{(D)}$ are given in terms of Eisenstein series that satisfy Laplace eigenvalue equations on the moduli space, the  coefficient $\calE_{(0,1)}^{(D)}$, of the $\smallf 18$-BPS interaction $\partial^6 \R^4$,  is an automorphic function that  satisfies an inhomogeneous Laplace equation.  Various properties of its constant terms in these three
limits were also determined in \cite{Green:2010wi,Green:2010kv}.

Whereas the earlier work concerned the zero Fourier modes of the coefficient functions, in this paper we are concerned with the non-zero modes in the Fourier expansion in any of the three limits listed above.  These Fourier coefficients should have the exponentially suppressed form that is characteristic of instanton contributions.
In more precise terms, the angular variables involved in the  Fourier expansion with respect to a maximal parabolic subgroup $P_\a$ come from the  abelianization\footnote{See (\ref{fourierexp2}).} $U_\a/[U_\a,U_\a]$ of the unipotent radical $U_\a$ of $P_\a$, and are conjugate to integers that define the instanton ``charge lattice".
Asymptotically close to a cusp a given Fourier coefficient is expected to have  an exponential factor of $\exp{ (-S^{(p)})}$, where $S^{(p)}$ is the action for an instanton of a given charge, as will be defined in section~\ref{sec:BPSinst}.  In the case of fractional BPS instantons the leading asymptotic behaviour in the cusp is  the real part of $S^{(p)}$, and  is related to the  charge \eqref{tencharge}, which enters the phase of the mode.

In each limit the $\smallf 12$-BPS orbits are minimal orbits
(i.e., smallest nontrivial orbits) while the $\smallf 14$-BPS orbits are ``next-to-minimal'' (NTM)  orbits  (i.e., smallest nonminimal or nontrivial orbits).   The next largest are $\smallf 18$-BPS orbits, which only arise for groups of sufficiently high rank; in the $E_8$ case there is a further $\smallf 18$-BPS orbit beyond that.
These come up again as ``character variety orbits'', a major
consideration in sections~\ref{sec:NTMdetails}  and~\ref{sec:shrunkFourierCoeff}. They are closely related to  -- but not to be confused with  -- the minimal and next-to-minimal  coadjoint nilpotent orbits that are attached to the Eisenstein series that arise in the solutions for the coefficients, $\calE_{(0,0)}^{(D)}$ and $\calE_{(1,0)}^{(D)}$ in \eqref{rfourcoeff}  and \eqref{dfourrfourcoeff}, respectively.


{\bf Note on conventions.}  Following
\cite[Section 2.4]{Green:2010kv},
  the parameter associated with the $GL(1)$ factor that parameterises the approach to any cusp will be called $r$ and is normalised in a mathematically convenient manner.   It translates into distinct  physical parameters  in each  of  the three limits described above, that correspond to parabolic subgroups defined at nodes  $d+1$,  $1$ and $2$, respectively,
of the Dynkin diagram  in figure~\ref{fig:dynkin}.  These are summarised as follows:
\bea
&&{\rm Limit~(i)}\, \quad \  r^2  \ = \  {r_d/ \ell_{11-d}}\,,  \ \ \,  r_d \ = \  {\rm radius\ of\ decompactifying\ circle}\,,\nn\\
&&{\rm Limit~(ii)} \, \quad \  r^{-2}  \  = \ \sqrt{ y_D} \ = \  {\rm string\ coupling\ constant} \,,\nn\\
&&{\rm Limit~(iii)} \quad \,  r^{2(1+d)\over 3}  \ = \ { \calV_{d+1}/ \ell_{11}^{d+1}} \,, \ \ {\calV_{d+1}}  = {\rm vol.\, of\, M-theory\ torus}\,.\nn\\
\label{notation}
\eea
The  $D$-dimensional string  coupling constant  is defined  by  $ y_D =
g_s^2\, \ell_s^d/V_d$,  where $D=10-d$ and  $g_s$  is either the $D=10$  IIA string coupling constant, $g_A$, or the  IIB string coupling constant, $g_B$,  and $V_d$ is the volume of $\rT^d$ in string units.\footnote{We will use the symbol $\rT^d$ to denote the string theory $d$-torus while using the symbol $\calT^{d+1}$ for the corresponding M-theory $(d+1)$-torus expressed in eleven-dimensional units.}
The Planck length scales in different dimensions are related to each other and to the string scale, $\ell_s$, by
\bea
 (\ell^A_{10})^8 &=& \ell_s^8\, g_A^2\,, \quad   (\ell^B_{10})^8\ =\  \ell_s^8 \, g_B^2\,,\quad \ell_{11}\ =\ g_A^{\third}\, \ell_s\,,\nn\\(\ell_D)^{D-2} &=&  \ell_s^{D-2}\, y_D\ =\ (\ell_{D+1})^{D-1}\, {1\over r_d}\,,\  \ {\rm for}\ D\le 8 \ (d \ge 2)\nn\\
\ell_9^7
 &=&
  \ell_s^{7}\, y_9= (\ell^A_{10})^{8}\, {1\over r_A}\ =\  (\ell^B_{10})^{8}\, {1\over r_B}\,.
\label{plancks}
\eea
(note the two distinct Planck lengths in the ten-dimensional case and the distinction between $r_1=r_A$ and $r_1=r_B$  in the two type II theories).

\subsection{Mathematics background}\label{sec:mathematicsbackground}

Let us begin by recalling some  notions from the theory of automorphic
forms that are relevant to the expansion (\ref{amp}), specifically
from \cite[Section 2]{Green:2010kv}.   Let $G$ denote the split real Lie group $E_n$, $n\le 8$, defined in table~\ref{tab:Udual}.
For convenience we fix (as we may) a Chevalley basis of the Lie algebra $\mathfrak g$ of $G$, and a choice of positive roots $\Phi_+$ for its root system $\Phi$.  Letting $\Sigma\subset \Phi_+$ denote the positive simple roots, the Lie algebra $\mathfrak g$ has the triangular decomposition
\begin{equation}\label{triangulardecomp}
    \mathfrak g \ \ = \ \ \mathfrak n \,\oplus\,\mathfrak a\,\oplus\,\mathfrak n_{-}\,,
\end{equation}
where $\mathfrak n$ (respectively, $\mathfrak n_{-}$) is spanned by the
Chevalley basis root vectors $X_\a$ for   $\a\in \Phi_+$
(respectively, $\a\in \Phi_{-}$), and $\mathfrak a$ is spanned by their
commutators $H_\a=[X_\a,X_{-\a}]$.  Let $N\subset G$ be the
exponential of $\mathfrak n$; it is a maximal unipotent subgroup. Likewise
$A=\exp(\mathfrak a)$ is a maximal torus, and is isomorphic to
$\operatorname{rank}(G)$ copies of $\IR^+$.  The group $G$ has an Iwasawa decomposition $G=NAK$, where $K=K_n$ is the maximal compact subgroup of $G$ listed in table~\ref{tab:Udual}.  There thus
  exists a logarithm map $H:A\rightarrow {\mathfrak a}$ which is inverse
  to the exponential, and which extends to all $g\in G$ via its
  value on the $A$-factor of the Iwasawa decomposition of $g$.   The integral points $G(\Z)$ are defined as all elements $\g \in G$ such that the adjoint action $Ad(\g)$ on $\mathfrak g$ preserves the integral span of the Chevalley basis.

The standard maximal parabolic subgroups of $G$ are in bijective correspondence with the positive simple roots of $G$.  Given such a root $\beta$ and a standard maximal parabolic $P_\b$,   the {\sl maximal parabolic Eisenstein series} induced from the constant function on $P_\b$ is defined by the sum
\begin{equation}\label{maxparabeisdef}
    E_{\beta;s}^G \ \ := \ \ \sum_{\g\,\in\,(P_\b\cap G(\Z))\backslash G(\Z)} e^{2\,s\,  \omega_\b(H(\g g))}\ , \ \ \ \Re s \, \gg \, 0\,,
\end{equation}
where
 $\omega_\b$, the fundamental weight associated to $\beta$, is defined by the condition $\langle \omega_\b,\a\rangle=\d_{\a,\b}$.
These series  generalize the classical nonholomorphic Eisenstein series (the case of $G=SL(2)$), and more generally  the Epstein Zeta functions (the case of $G=SL(n)$ and $\b$  either the first or last node of the $A_{n-1}$ Dynkin diagram).  Because of this special case, we     often   refer to the $\b=\a_1$ series (in the numbering of figure~\ref{fig:dynkin}) as the {\sl Epstein} series for a particular group, even if it is not $SL(n)$.
 These series are the main mathematical objects of   this paper.

  As we remarked in footnote~\ref{Gzdeffootnote} changing $G$ to another Chevalley group with Lie algebra $\mathfrak g$ changes $G(\Z)$ by a central subgroup, and so Eisenstein series for the cover descend to the corresponding Eisenstein series on the quotient.  For example,
\begin{equation}\label{SpinddandSoddseries}
    E^{Spin(d,d)}_{\beta;s}(g) \ \ = \ \ E^{SO(d,d)}_{\beta;s}(\pi(g))\,,
\end{equation}
where $\pi:Spin(d,d,\IR)\rightarrow SO(d,d,\IR)$ is the covering map. We shall sometimes refer to either as $E^{D_d}_{\beta;s}$ when we wish to emphasize that a particular statement applies to both $E^{Spin(d,d)}_{\beta;s}$ and $E^{SO(d,d)}_{\beta;s}$.

As shorthand, we often denote a root  by its ``root label'', that is, stringing together its coefficients when written as a linear combination of the positive simple roots $\Sigma$.   Thus $\a_2+\a_3+2\a_4+\a_5$ could be denoted $0112100\cdots$ or $[0112100\cdots]$, with brackets sometimes added for clarity.
  Note that
  Eisenstein series of the type (\ref{maxparabeisdef}) are parameterized by a single complex variable, $s$,  whereas the more general minimal parabolic series in (\ref{minparabseries}) has $\operatorname{rank}(G)$ complex parameters.

The series (\ref{maxparabeisdef}) is initially absolutely convergent
for $\Re{s}$ large, and has a meromorphic continuation to the entire
complex plane as part of a more general analytic continuation of
Eisenstein series due to Langlands.  Its special value at $s=0$ is the
constant function identically equal to one.  This corresponds to the
trivial representation of $G(\IR)$, and clearly has no nontrivial
Fourier coefficients.  The main mathematical content of this paper
extends this phenomenon to other special values of $s$ which are
connected to small representations of real groups (see
sections~\ref{sec:autreps} and \ref{sec:NTMdetails}), and which have very few nontrivial Fourier
coefficients.  This will be demonstrated to be in complete agreement
with a number of  string theoretic predictions, in particular the one stated
at the end of  section~\ref{sec:introduction}.

The main results of \cite{Green:2010kv} were the identifications (\ref{rfourcoeff}) and (\ref{dfourrfourcoeff}) of $\cE^{(D)}_{(0,0)}$ and $\cE^{(D)}_{(1,0)}$, respectively, in terms of special values of the Epstein series, for $3 \le D=10-d \le 5$. The more general automorphic function $\cE^{(D)}_{(0,1)}$ which  satisfies  \eqref{laplaceeigenthree}  was also analysed in  \cite{Green:2010kv}, but will not be relevant in this paper.       The case of $Spin(5,5)$ was also covered in \cite{Green:2010kv}, but is somewhat more intricate; it will be explained separately.   We will show in a precise sense that these Epstein series at the special values at $s=0$, $3/2$, and $5/2$ correspond, respectively, to the three smallest types of representations of $G$ (see theorem~\ref{mainthm}) below.

\subsubsection{Coadjoint nilpotent orbits}

Let $\mathfrak g$ be the Lie algebra of a matrix Lie group $G$, whether over $\IR$ or $\C$.  An element of $\mathfrak g$ is {\sl nilpotent} if it is nilpotent as a matrix, i.e., some power of it is zero.  The group $G$ acts on its Lie algebra $\mathfrak g$ by the adjoint action $Ad(g)X=gXg^{-1}$, and hence dually on linear functionals $\l:{\mathfrak g}\rightarrow\C$ through the coadjoint action given by $(Coad(g)\l)(X)=\l(Ad(g^{-1})X)=\l(g^{-1}Xg)$.  Actually $\mathfrak g$ is isomorphic to its space of linear functionals via the Killing form, and so the coadjoint action is equivalent to the adjoint action.  Following common usage, we thus refer to the orbits of the adjoint action of $G$ on $\mathfrak g$ as {\sl coadjoint nilpotent orbits} (even though they are, technically speaking, adjoint orbits).

The book \cite{Collingwood} is a standard reference for the general theory of  coadjoint nilpotent orbits.  When $G$ is a real or complex semisimple Lie group there are a finite number of orbits, each of which is even dimensional.  The smallest of these is the trivial orbit, $\{0\}$.  On the other hand, there is always an open, dense  orbit, usually referred to as the {\sl principal} or {\sl regular orbit}.   Another orbit which will be important for us is the {\sl minimal} orbit, the  smallest orbit aside from the trivial orbit.  Since our groups $G$ are all simply laced, it can be described as the orbit of any root vector $X_\a$, for any  root $\a$.

Table \ref{tab:minimalnilpotent3}
gives a list of some orbits that are important to us, along with their
basepoints.

\begin{table}[center]
  \centering
  \begin{tabular}[ht]{||c|c|c||}
    \hline
Group & Orbit Dimension & Basepoint\\
\hline
$SL(2)$&0&0\\
&2&$X_{1}$\\
\hline
&  0 & 0 \\
$SL(3)\times SL(2)$ & 2& an $SL(2)$ root  \\
& 4 & an $SL(3)$ root \\
\hline
 & 0 & 0 \\
$SL(5)$ & 8 & $X_{1111}$\\
 & 12 & $X_{1110}+X_{0111}$\\
\hline
 & 0 & 0 \\
$Spin(5,5)$ & 14 & $X_{12211}$\\
& 16 & $X_{11110}+X_{11101}$ \\
 & 20 & $X_{01111}+X_{11211}$\\
\hline
 & 0 & 0 \\
$E_6$&  22 & $X_{122321}$ \\
 & 32 &$ X_{111221}+X_{112211}$ \\
 & 40 & $X_{011221}+X_{111210}+X_{112211}$ \\
\hline
 & 0 & 0 \\
$E_7$ & 34 & $X_{2234321}$\\
 & 52 & $X_{1123321}+X_{1223221}$\\
 & 54 &$X_{0112210}+X_{1112221}+X_{1122110}$ \\
\hline
 & 0 & 0 \\
 & 58 & $X_{23465432}$\\
$E_8$ & 92 & $X_{23354321}+X_{22454321}$\\
 & 112 & $X_{22343221}+X_{12343321}+X_{12244321}$ \\
 & 114 & $X_{11232221}+X_{12233211 }$ \\
\hline
  \end{tabular}
  \caption{Basepoints  of  the smallest  coadjoint nilpotent   orbits
    for the complexified $E_n$ groups.  The notation $X_{\a}$ denotes the
    Chevalley basis root vector for the simple root $\a$, which is written
     here in terms of the root labels described in the text.   The basepoints are given as a description of the orbit
     but are not otherwised used.
    The $SL(3)\times SL(2)$ case comes from the $E_3$ Dynkin diagram,
    which is the $E_8$ Dynkin diagram from figure~\ref{fig:dynkin}
    after the removal of  nodes 4, 5, 6, 7, and 8.  Its Lie algebra is a product of
    two simple Lie algebras  and has a different orbit structure than
    the others;  its smallest orbits come from the respective
    factors. }
  \label{tab:minimalnilpotent3}
\end{table}

\subsubsection{Automorphic representations}\label{sec:autreps}

The right translates of an automorphic function by the group $G$ span a vector space on which $G$ acts.  For a suitable basis of square-integrable automorphic forms and  most Eisenstein series, this action furnishes an irreducible representation.
As we discussed in \cite[Section 2]{Green:2010kv}, the Eisenstein
series are specializations of the larger  ``minimal parabolic
Eisenstein series'' defined in~(\ref{minparabseries}).  The automorphic representations connected to the  latter are generically principal series representations, an identification which can be made by comparing the infinitesimal characters (that is, the action of all $G$-invariant differential operators).  However, at special points the principal series reduces, and the Eisenstein series is part of a smaller  representation.

An irreducible representation is related to coadjoint nilpotent orbits
through its {\sl wavefront set}, also known as the ``associated
variety'' of its ``annihilator ideal''.  It is a theorem of Joseph
\cite{joseph} and Borho-Brylinski \cite{borhobryl}    that this set is
always the closure of a unique coadjoint nilpotent orbit.  Thus a
coadjoint nilpotent orbit is attached to every irreducible
representation of $G$.

 \begin{figure}[ht]
\qquad\qquad\qquad\qquad\centering\includegraphics{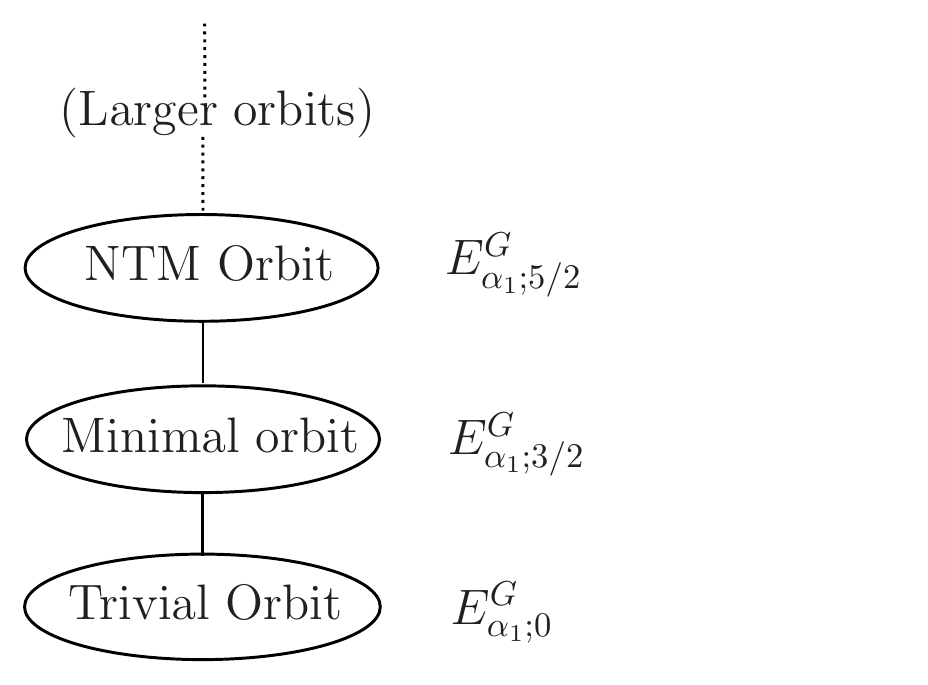}
 \caption{\label{fig:smallorbit}  Schematic of small representations and Eisenstein special values.}
 \end{figure}

Part (iii) of the following theorem is the main mathematical result of this paper, in particular the cases of $E_7$ and $E_8$.
Part (i) is trivial, while part (ii) is contained in results of Ginzburg-Rallis-Soudry~\cite{grs},
following earlier work of Kazhdan-Savin~\cite{KazhdanSavin}.

\begin{thm}\label{mainthm}  Let $G$ be  one of the groups $E_6$, $E_7$, or $E_8$ from table~\ref{tab:Udual}.  Then
\begin{itemize}
  \item[(i)] The wavefront set of the automorphic representation attached to the $s=0$ Epstein series is the trivial orbit.
  \item[(ii)] The wavefront set of the automorphic representation attached to the $s=3/2$ Epstein series is the closure of the minimal orbit.
  \item[(iii)] The wavefront set of the automorphic representation attached to the $s=5/2$ Epstein series is the closure of the next-to-minimal  (NTM) orbit.
\end{itemize}
\end{thm}
\noindent The closure of the minimal orbit is simply the union of the minimal orbit and the trivial orbit, while the closure of the next-to-minimal orbit is the union of itself, the minimal orbit, and the trivial orbit.  Theorem~\ref{mainthm} will be used in   proving theorem \ref{thm:vanishingofcoefficients}, which is the mathematical proof of the  statement concerning vanishing Fourier modes at the end of section \ref{sec:introduction}  that was motivated by string considerations.

\subsection{Outline of paper.}

This paper combines information deduced from string theory with results in number theory involving properties of Eisenstein series, which we hope will  be  of interest to both string theorists and number theorists. In particular, each subject is used to make nontrivial statements about the other. Sections  \ref{bpsinstantons}--\ref{lowrank} and appendices \ref{susyinst}--\ref{euclidDbrane}   are framed in  string theory language and provide information concerning the structure expected of the non-zero Fourier modes  based on instanton contributions in superstring theory and supergravity.   The subsequent sections provide the mathematical foundations of these observations and generalize them significantly.

Section \ref{bpsinstantons}   presents the classification of the expected orbits of fractional  BPS  instantons in the
three limits  (i), (ii), and (iii) considered in section~\ref{sec:stringtheorybackground},
  from the point of view of  string theory.  The BPS constraints imply
  that these instantons  span particular small orbits generated by the
  action of the Levi subgroup acting on the unipotent radical
  associated with the parabolic  subgroup appropriate to a given
  limit.    These orbits can be thus thought of as character
    variety orbits, which are discussed at the beginning of  section
    \ref{lowrank}.

In the rest of section~\ref{lowrank} and appendix~ \ref{modesdetails}  we
will consider explicit low-rank examples (with rank $d+1 \le 5$) of the Fourier expansions of the functions $\calE_{(0,0)}^{(10-d)}$ and $\calE_{(1,0)}^{(10-d)}$ in the  parabolic subgroups corresponding to each limit.
 In the  cases with $d+1\le 4$ ($D\ge 7$), the definition \eqref{maxparabeisdef} implies that the coefficient functions are combinations of $SL(n)$ Eisenstein series that can easily be expressed in terms of elementary lattice sums.  In these cases it is straightforward to use standard  Poisson summation techniques to exhibit the precise form of their Fourier modes.
 In particular, the non-zero Fourier modes of $\calE_{(0,0)}^{(10-d)}$ will be
 determined in the three limits under consideration for the rank $d+1\le
 4$ cases. These modes are localized within the minimal
   character variety orbits that contain
 precisely the $\smallf 12$-BPS instantons that are anticipated in section~\ref{bpsinstantons}.
  We
 will see, in particular,  that  in the decompactification limit (i)
 the precise form for each of these coefficients matches in detail
 with the expression determined directly from a quantum mechanical
 treatment of $D$-particle world-lines wrapped around an  $S^1 \subset \rT^d$.\footnote{The term {\sl $D$-particle}  refers to any point-like BPS particle state obtained by completely wrapping the spatial directions of $Dp$-brane states.}

  Explicit examples of Fourier expansions of
 the coefficients of the $\smallf 14$-BPS interactions, $\calE_{(1,0)}^{(D)}$,
 will also be presented in section~\ref{lowrank} and
 appendix~\ref{modesdetails}.
In the $D=10B$ case (with symmetry group $SL(2)$)  this function is simply equal to
 $\zeta(5)\,E_{\alpha_1;5/2}^{SL(2)}$,  and  the extension to $D=9$ and $D=8$ is also straightforward.
 But in the $D=7$ case (with symmetry group $E_4=SL(5)$)  the coefficient function $\calE^{(7)}_{(1,0)}$ is a sum
 of the regularized Epstein series $\hat E_{\alpha_1;5/2}^{E_4}$ and
 the non-Epstein Eisenstein series $\hat E_{\alpha_4;5/2}^{E_4}$ (coming from the third node of the Dynkin diagram).
 The analysis of the Fourier modes of $E_{\alpha_4;s}^{E_4}$
 involves the use of several lattice summation identities that are
 proved in  appendices~\ref{sec:latticesums},
 \ref{sec:latticeidentity} and \ref{sec:unfolding}.  In particular we
 will derive an expression for the non-Epstein Eisenstein
 series coming from either the second or second-to-last nodes as a Mellin transform of a certain lattice
 sum that is closely related to the $Spin(d,d)$ Epstein Eisenstein series,
 $E_{\alpha_1;s-d/2}^{Spin(d,d)}$.  In appendix~\ref{sec:unfolding}  we
 will derive a theta lift between $SL(d)$ and $Spin(d,d)$ Eisenstein series.  This relation was presented in a less rigorous form in~\cite{Green:2010wi}.
The resulting Fourier  expansions contain instanton contributions localized within the minimal
 ($\smallf 12$-BPS) character variety orbit and the next-to-minimal
 ($\smallf 14$-BPS)  character variety orbit,  comprising precisely the instantons anticipated in
 section~\ref{bpsinstantons}.

The coefficient   $\calE^{(6)}_{(0,0)}$ is proportional to the series   $E_{\alpha_1;3/2}^{Spin(5,5)}$,  which we will analyse by using the integral representation  proved in proposition~\ref{prop:Ddintegralrepn}. As expected, its non-zero Fourier modes are supported within the minimal ($\smallf 12$-BPS) character variety orbits in any of the three limits.  On the other hand the $\smallf 14$-BPS coefficient,
      $\calE_{(1,0)}^{(6)}$, involves the sum of the regularized values of
      $\hE_{\alpha_1;5/2}^{Spin(5,5)}$  and  $\hE_{\alpha_5;3}^{Spin(5,5)}$.  Although we have not computed the Fourier expansion of the second series, it is still possible to show that the non-zero Fourier coefficients of this sum are supported within the minimal and next-to-minimal (i.e., $\smallf 12$- and $\smallf 14$-BPS) character variety orbits in each of the three limits.
This will be discussed at the end of  section~\ref{lowrank}.

Sections~\ref{sec:NTMdetails},~\ref{sec:shrunkFourierCoeff}, and~\ref{sec:L2} are primarily concerned with the exceptional group cases, which correspond to $d\ge 5$ and  $D\le 5$.
Since classical lattice summation techniques are difficult to apply in this context, we instead use results from representation theory to show a large number of the Fourier coefficients vanish.  Indeed, avoiding explicit computations here is one of the main novelties of the paper.
 Section~\ref{sec:NTMdetails} discusses aspects connected to representation theory  and contains a proof of theorem (2.13), which makes important use of appendix~\ref{sec:trapendix} by Ciubotaru and Trapa on  special unipotent representations.

 Section~\ref{sec:shrunkFourierCoeff} then applies these results to Fourier expansions, using a detailed analysis of character variety orbits.  We will see that the spectrum of instantons that are expected to vanish on the
 basis of string theory is  precisely  reproduced by the Eisenstein
 series in \eqref{rfourcoeff} and \eqref{dfourrfourcoeff}.
 For the $s=3/2$ case (the $\smallf 12$-BPS case) we will reproduce
 the statements  in~\cite{grs,Kazhdan:2001nx,KazhdanPolishchuk:2002} that only the
 minimal orbit and  the trivial orbit contribute to the Fourier expansions of the
 Eisenstein series.  The relevance of this work to $\smallf 12$-BPS states was
 suggested by~\cite{Pioline:2010kb,Neitzke}. In addition,  we will find that this generalizes for $s=5/2$ (the $\smallf 14$-BPS case) to the statement that no orbits larger than the next-to-minimal (NTM) orbit can contribute.
The analysis  in \cite{Green:2010kv} showed the striking fact that the
particular Eisenstein series in \eqref{rfourcoeff} and
\eqref{dfourrfourcoeff}  have constant terms  with very few powers of
$r$ (defined in \ref{notation}) in their expansion around any of the
three limits under consideration.  The analysis in this paper
demonstrates analogous special features of the orbit structure of the
non-zero modes.
Theorem~\ref{thm:vanishingofcoefficients} gives a precise statement
about which Fourier modes automatically vanish because of
representation theoretic reasons.    This set of vanishing
coefficients is exactly those that  are argued to vanish for string
theory reasons   in section~\ref{bpsinstantons}.

 It is important to point out that our methods show the vanishing of a
 precise set of Fourier coefficients, but typically do not show the
 {\sl nonvanishing} of the remaining Fourier coefficients.  However,
 this is accomplished in a number of low rank cases by explicit
 calculations in section~\ref{lowrank} and
 appendix~\ref{modesdetails}, and we hope to treat some of the
 higher rank cases in the future.
   Section~\ref{sec:L2} discusses square-integrability of the coefficients and  conditions under which
    $\cE^{(D)}_{(0,0)}$ and $\cE^{(D)}_{(1,0)}$ are square-integrable for higher rank groups.

\section{Orbits of supersymmetric instantons}
\label{bpsinstantons}
From the string theory point of view our main interest is in the systematics of orbits of BPS instantons that enter the Fourier expansions of the coefficients of the low order terms in the low energy expansion of the four graviton amplitude.  Before describing these orbits in sections \ref{sec:representation} -- \ref{sec:BPSMorbits} we begin with a short overview of the special features of such instantons that follow from supersymmetry.  A short summary of the M-theory supersymmetry algebra and BPS particle states is given in appendix~\ref{susyinst} (although this barely skims the surface of a huge subject), where  the structure of the eleven-dimensional superalgebra is seen to imply the presence of an extended two-brane (the $M2$-brane) and five-brane (the $M5$-brane) in eleven dimensions.   Compactification on a torus also leads to Kaluza--Klein ($KK$) point-like states and Kaluza--Klein monopoles ($KKM$), one of which is interpreted in string theory as a $D6$-brane.    All the particle states in lower dimensions can be obtained by wrapping the spatial directions of these objects around cycles of the torus.
\subsection{BPS instantons}\label{sec:BPSinst}

 One class of BPS instantons can be described from the eleven-dimensional semi-classical M-theory point of view by wrapping euclidean
world-volumes of $M2$- and $M5$- branes around compact directions so that the
brane actions are finite. These branes couple to the three-form M-theory potential and its dual, and the BPS conditions constrain  their charges, $Q^{(p)}$,  to be proportional to their tensions,  $T^{(p)}$, where $p=2$ or $5$ (as briefly reviewed in appendix~\ref{susyinst}).
Wrapping  the world-volume of a euclidean $M2$-brane  around a 3-torus, $\calT^{3} \subset \calT^{d+1}$,  or a euclidean $M5$-brane around a 6-torus, $\calT^6\subset \calT^{d+1}$,
gives a  $\smallf 12$-BPS instanton, which has a euclidean action of
the form $S^{(p)}  =2\pi \, (T^{(p)} + i Q^{(p)} )$.   This gives a factor in amplitude of the form $e^{-  S^{(p)}}$ that has a characteristic phase determined by the charge of the brane.

In addition, the ``$KK$ instanton''  is identified with the euclidean world-line of a $KK$ charge winding around a circular dimension.  The magnetic version of this is the ``$KKM$ instanton'', one manifestation of which appears in string theory as a wrapped euclidean $D6$-brane.
Recall that a $KK$ monopole  in eleven dimensional (super)gravity with
one compactified direction labelled $x^\#$ has a metric of the form~\cite{Townsend:1995kk}
\begin{equation}
ds^2=  V^{-1}\, (dx^\#+ {\bf A}\cdot {\bf dy})^2     +V\, {\bf dy}\cdot {\bf dy}   -dt^2 + dx_6^2 \,,\qquad V=1 +\frac{R}{2|{\bf y}|}\,,
\label{kkmonopole}
\end{equation}
where $ds_7^2=-dt^2+dx_6^2$ is the seven-dimensional Minkowski metric
and the other four dimensions, $x^\#$, ${\bf y}=(y_1,y_2,y_3)$,
define a Taub--NUT space, and $|{\bf y}|^2=\sum_{i=1}^3 y_i^2$.  The coordinate $x^\#$ is periodic with
period $2\pi R$ and the potential, ${\bf A}$, satisfies the equation
${\bf \nabla} \times {\bf A} = - {\bf \nabla} V= {\bf B}$.
Poincar\'e duality in the ten dimensions $(t, x_6, {\bf y})$ relates
the 1-form potential, ${\bf A}$, to a 7-form,  i.e, $*d{\bf A}= dC^{(7)}$.  If  $x^\#$ is identified with the M-theory circle, $C^{(7)}$ couples to a $D6$-brane in the string theory limit.  This gives an instanton when its world-volume is wrapped around a 7-torus.
 More generally,  $x^\#$ can be identified with other circular
 dimensions of the torus $\calT^{d+1}$, giving a further $d$ distinct
 $KKM$'s, each one of which appears as a finite action instanton when
 wrapped on an M-theory 8-torus, $\calT^8$ (i.e., when $d=7$).  When
 describing these in the string theory parameterisation (on the string
 torus $\rT^7$) these will be referred to as ``stringy $KKM$
 instantons''.
   Furthermore, it is well understood how to combine wrapped
   branes to make $\smallf 12-$,  $\smallf 14-$ and $\smallf 18$-BPS instantons
   \cite{Becker:1995kb,Harvey:1999as}\footnote{We are concerned with
     compactification on tori, but more generally the BPS condition
     requires branes to be  wrapped on  special lagrangian
     submanifolds (SLAGs) or on holomorphic cycles
     \cite{Becker:1995kb}.}    in a manner analogous to combining
   $p$-branes to make states preserving a fraction of the symmetry.

 This  description of instantons is directly relevant to  the discussion of  the semi-classical M-theory limit  (case (iii))  associated with the Fourier expansion in the parabolic subgroup $P_{\alpha_2}$ in section~\ref{sec:BPSMorbits}. This is the large-volume limit in which eleven-dimensional supergravity is a valid approximation.  Similarly, the instanton contributions in limits (i) and (ii)  can be described  by translating from the M-theory description to the string theory description of the wrapped branes. These wrapped string theory objects comprise: the fundamental string and
 the Neveu--Schwarz five-brane (NS5-brane) that couple to  $B_{\rm
   NS}$; $Dp$-branes that couple to the
Ramond--Ramond $(p+1)$-form potentials $C^{(p+1)}$ (with $-1\le p\le
9$);  and $KK$ charges and  $KK$ monopoles that couple to
modes of the metric associated with toroidal compactification on
$\rT^d$.

  Knowledge of this instanton spectrum is a valuable ingredient in understanding the systematics of the Fourier modes of the Eisenstein series that enter into the definitions of the coefficients of the low order interactions in the expansion of the scattering amplitude.   In particular, it connects closely with the study of the Fourier expansions of specific Eisenstein series
that enter into $\calE_{(0,0)}^{(D)}$ and $\calE_{(1,0)}^{(D)}$ (that
will be discussed later in this paper),  as well as with the Fourier
expansion of the more general automorphic function
$\calE_{(p,q)}^{(D)}$ (that will not be discussed in this paper).

\subsection{Fourier modes and orbits of BPS charges.}
\label{fourierorbits}

The Fourier expansion  associated with any parabolic subgroup,
$P_\alpha= L_{\alpha}\, U_\alpha$, of $E_{d+1}$ is a sum over integer
charges that  are conjugate to the angular variables that enter in  its unipotent radical
$U_\alpha$.  These determine the phases of the modes.
The Levi factor
is a  reductive group that has the form
$L_\alpha =GL(1)\times M_{\alpha}$, where $M_\a$ is its
  semisimple component.

  The conjugation action on $U_\alpha$ of    $L_{ \alpha}$  -- or more specifically,
  its intersection with the discrete duality group  $L_\a\cap E_{d+1}(\Z)$   --   relates
  these charges  by Fourier duality.  Thus this action carves out orbits within the charge
  lattice, with each given orbit  only covering a subset of the total charge space.  This viewpoint is
   expanded upon in more detail in section~\ref{sec:nreFourierCoeff}.  In this subsection we classify
   these orbits in cruder form, by considering the action of the continuous group  $L_{\a}$  on the charge
    lattice.  Indeed,
     since we are mainly interested in the algebraic nature of the group action, we sometimes look at the
     less refined action of the complexification of  $L_\a$, e.g.,  in order to avoid subtle issues about
     square roots.    Though this loses information by grouping charges into broader families,
    those families still retain some important common features.

As will be explained in section~\ref{sec:nreFourierCoeff}, the action of  $L_\a$  on the charge lattice is related to the adjoint representation on the Lie algebra of $U_\a$.  This representation is irreducible if and only if $U_\a$ is abelian.   That is the case for the unipotent radicals we consider of every symmetry group $E_{d+1}(\IR)$ of rank $d+1 <6$.  Otherwise,  the Fourier expansion is only well-defined after averaging over the commutator subgroup (see (\ref{fourierexp2})), and hence does not capture the full content of the function.  We devote the rest of this section to relating these orbits to BPS instantons in the three limits we consider.
In each particular case we will explain the origin of the non-abelian
nature of the unipotent radicals,
 which have  charges that do not commute with the other brane charges.  A
 discussion of such effects within string theory can be
 found, for example, in~\cite{Freed:2006yc}.

We now describe the adjoint action
$V_{\hat\alpha}$ on the unipotent radical, where $\hat \alpha$ labels the node immediately
adjacent to $\alpha$ in the Dynkin diagram
(fig.~\ref{fig:dynkin}).
For the three parabolic subgroups of interest  to us the representations of the unipotent radical are as follows:

\begin{itemize}
\item[(i)]{The maximal parabolic $P_{\alpha_{d+1}}$.}

In this case ${\hat \alpha}=\alpha_d$ and  $L_{\alpha_{d+1}}=GL(1)\times E_d$.
The following lists the representations $V_{\alpha_d}$ for each value of $2\le d\le 7$.
\begin{center}
   \begin{tabular}[c]{|c|c|c|}
\hline $ E_{d+1}$ & $M_{\alpha_{d+1}}$  &$ V_{\alpha_d}$ \\
\hline
  $  E_8$ & $E_7$ & $q^i : {\bf 56}$, $q : {\bf 1}$\\
  $  E_7$ &$ E_6$ & $q^i : {\bf 27}$\\
   $ E_6$& $Spin(5,5)$& $S_\alpha : {\bf 16}$\\
   $ Spin(5,5)$ &$ SL(5)$& $v_{[ij]} : {\bf 10}$\\
  $ SL(5) $& $SL(3)\times SL(2) $& $v_{ia} : {\bf 3}\times {\bf 2}$\\
  $ SL(3)\times SL(2)$& $SL(2)\times \IR^+$& $v v_a : {\bf 2}$\\
\hline
  \end{tabular}
  \end{center}
The notation in the last column indicates the irreducible
  representations are indexed by their dimensions.  Both the fundamental representation and the
  trivial representation of $E_7$ occur, because the unipotent radical $U_{\a_8}$ is a Heisenberg group. The lower dimensional representations are:~
the fundamental representation for $E_6$;  a spinor
representation for $Spin(5,5)$; the rank  2 antisymmetric tensor
representation for $SL(5)$; a bivector representation for $SL(3)\times
SL(2)$; and a scalar-vector representation for $SL(2)\times\IR^+$.

\item[(ii)] {The maximal parabolic $P_{\alpha_1}$.

In this case ${\hat \alpha} = \alpha_3$, which is a spinor node
(following the numbering of figure~\ref{fig:dynkin})  and
$L_{\alpha_1} = GL(1)\times Spin(d,d)$. The representation
$V_{\hat{\alpha}}$ always includes a spinor representation of
$Spin(d,d)$.  It is irreducible except in  the cases of $d=6,7$.    The
case of   $Spin(6,6)\subset E_7$ also includes a copy of the trivial
representation, because the unipotent radical is again a Heisenberg
group; the case of $Spin(7,7)\subset E_8$ also includes a copy of the
standard 14-dimensional ``vector'' representation.}

\item[(iii)] {The maximal parabolic $P_{\alpha_2}$.

In this case ${\hat \alpha} = \alpha_4$ and   $L_{\alpha_2} = GL(1)\times SL(d+1)$.   The representation $V_{\hat{\a}}$ always includes   a rank 3 antisymmetric tensor of $SL(d+1)$, $v_{ijk}$, of dimension $\f{1}{3!}(d+1)d(d-1)$.  It is irreducible when the rank is less than 6 (see table~\ref{tab:dimUnipotent} for the dimensions in the higher rank cases.)}
\end{itemize}

\begin{table}[center]
  \centering
  \begin{tabular}[t]{||c|c|c||c|c||c|c||}
    \hline
Group & \multicolumn{2}{|c||}{first node}& \multicolumn{2}{|c||}{second node}&\multicolumn{2}{|c||}{last node}\\
\hline
$SL(3)\times SL(2)$&2&0& 1 &0&3&0\\
$SL(5)$ & 4& 0 &4 &0 &6 &0\\
$Spin(5,5)$ & 8&0&10&0&10&0\\
$E_6$ & 16&0&20&1&16&0\\
$E_7$ &32&1&35&7&27&0        \\
$E_8$ & 64  & 14 & 56 & 28 + 8 & 56 &1      \\
\hline
  \end{tabular}
  \caption{Dimensions of  the unipotent  radical $U_{\a_i}$ for  the standard
    maximal parabolic  subgroup $P_{\alpha_i}$ where  $i=1$, $i=2$ and
    $i=d+1$. For each  node the first column gives  the dimension of the
    character variety $\mathfrak u_{-1}$  (see section~\ref{sec:nreFourierCoeff}), and the second column gives the dimension of
the derived subgroup $[U,U]$.  The sum of the two is the dimension of $U$.  The unipotent radical $U$ is abelian when the   dimension in the second column  is zero; it is a Heisenberg group when this dimension equals 1 and even more non-abelian when it is $>1$.}
  \label{tab:dimUnipotent}
\end{table}

 In each case, the charges form a lattice within the first listed
 piece of $V_{\hat{\a}}$, that is, the irreducible subrepresentation
 coming from the ``abelian part'' of $U_\a$.  More precisely, these
 are the nontrivial representations in part (i), the spinor
 representations in part (ii), and the rank 3 antisymmetric tensors
 $v_{ijk}$ in part (iii).  The space $V_{\hat{\a}}$ is identical with the
 ``character variety orbit'' $\mathfrak u_{-1}$ introduced in
 section~\ref{sec:nreFourierCoeff}.

Before proceeding with the explicit list of orbits based on the
counting of states and instantons in the next three subsections, we will recall
basic properties of the space of  nontrivial charges.   Apart from the most trivial
case (with duality group $SL(2,\Z)$), the $\smallf 12$-BPS orbits only fill a
small fraction of the whole space.
For the $E_{d+1}$ groups with $1\le d \le 5$ the complementary space to
the $\smallf 12$-BPS space is filled out by $\smallf 14$-BPS orbits.
 For  $E_7$ and $E_8$  the full space is spanned by the union of
 $\smallf 12$-, $\smallf 14$- and $\smallf 18$-BPS orbits.
The Fourier coefficients of the BPS protected operators will have
nonvanishing Fourier coefficients only associated to these nilpotent orbits.
The classification of possible charge orbits only depends on the
semi-classical nature of the associated BPS configurations, but does
not provide any detailed information about strong quantum corrections. Such
information should be encoded in the precise form  of the instanton contributions to the Fourier modes.

The instanton spectrum will now be considered in each of these limits
in turn.  In each case we will list the single BPS instantons  that form   basepoints  of
the charge orbits.  The dimension of the full spaces of charges spanned by the orbits in each case of interest is shown in table~\ref{tab:dimUnipotent}.  For each of the three limits (i), (ii), (iii), the two columns in the table show the dimensions of the abelian and nonabelian charge spaces, respectively.
Since we will be only interested in BPS (supersymmetric)
  orbits we will not discuss all the possible nilpotent orbits of $E_7$
  and $E_8$. A  complete discussion of the orbit structure is
  given in section~\ref{sec:OrbitCharacter}.

\subsection{BPS instantons in the decompactification limit:  $P_{\alpha_{d+1}}$}\hfill\break
 \label{sec:representation}
 The parabolic subgroup of relevance to the expansion of the amplitude in $D=10-d$ dimensions when  the radius $r_d$ defined in~(\ref{notation})  of one circle of the torus
$\rT^d$ becomes large is  $P_{\alpha_{d+1}}$, which has Levi factor $L_{\alpha_{d+1}} = GL(1)\times E_d$.   In this limit there is a close correspondence between the spectrum of instantons in $D=10-d$ dimensions and the spectrum of black hole states in $D+1=11-d$ dimensions.  This follows from the identification of the euclidean world-line of a charged black hole of mass $M$  wrapping around a circular dimension of radius $r$ with an instanton with action $2\pi M r$ that gives rise to an exponential factor of $e^{-2\pi Mr}$ in the amplitude.  In addition to instantons of this type, there can be instantons that do not decompactify to particle states in the higher dimension because their actions are singular in the large-$r$ limit.
 In any dimension there are also instantons with actions independent of $r$ that are inherited from the higher dimension in a trivial manner.

The spectrum of BPS black hole states in compactified string theory
has been studied extensively.  We will here follow the analysis
in~\cite{Ferrara:1997ci,Ferrara:1997uz},  which considered the
spectrum of  branes wrapped on $\rT^d$.  This generates charged
$\smallf 12$- and $\smallf 14$- BPS black hole states that correspond to singular solutions in supergravity since they have zero horizon size and hence zero entropy.  In addition,  for $E_6$, $E_7$ and $E_8$ there are $\smallf 18$-BPS states that correspond to black holes that have non-zero entropy  (as well as states with zero entropy), the prototypes being the analysis of black holes in $D=5$ dimensions (with $E_6$ duality group) in \cite{Strominger:1996sh,Callan:1996dv}.  The discussion of the associated  nilpotent  orbits  was given  in~\cite{Lu:1997bg}.   Our main interest is to extend the analysis in order to account for BPS instantons.

  We shall, for convenience, use the M-theory description starting from eleven dimensional supergravity compactified on a $(d+1)$-torus that will be denoted $\calT^{d+1}$.   The BPS  particle states  in any dimension are
obtained  by   wrapping all the spatial dimensions of the various
extended objects in supergravity around the torus.    These include
the $M2$-brane  and the $M5$-brane, together with the Kaluza--Klein
modes of the metric and the magnetic dual Kaluza--Klein monopoles.
The BPS instantons  can be listed by completely wrapping the euclidean
world-volumes of these objects on  these tori.

Despite their similarities, there is a fundamental mathematical
difference between the orbits of BPS states and the orbits of BPS
instantons.  The former are orbits under the semisimple part
$M_{\a_{d+1}}$ of the Levi component $L_{\a_{d+1}}=GL(1)\times
M_{\a_{d+1}}$, while the latter are orbits under the larger group
$L_{\a_{d+1}}$ itself.  Often these orbits coincide, but not
always:~the 27-dimensional orbit of $E_6$ and 56-dimensional orbit of
$E_7$ are actually unions of infinitely many $M_{\a_{d+1}}$-orbits
which are related by the $GL(1)$ action. This $GL(1)$ action is reminicent of the so-called trombone symmetry of supergravity~\cite{trombone}.  Similar examples occur in
other limits as well.  The $GL(1)$ parameter, $r$, described in
(\ref{notation}) is always normalized to act by the scalar factor of
$r^2$ on the BPS instantons, and so never acts trivially.  This action
is typically compensated by a different $GL(1)$ factor in the
stabilizer of a BPS instanton. When this happens we will shorten the
orbit notation by canceling these two factors, even though they are
mathematically different.  We use a horizontal line to denote a
quotient $\f{G}{H}$ of a group $G$ by a stabilizer $H$, in order to
match orbit descriptions with those commonly found in the physics
literature.  We have also made an attempt to correct mathematical
imprecisions in some existing descriptions. Since we do not use the
explicit descriptions of these orbits this should cause no confusion.

\subsubsection{Features of $P_{\alpha_{d+1}}$ orbits }
\label{generalfeatures}
The details of the enumeration of BPS states and instantons in the decompactification limit are reviewed in appendix~\ref{orbit1appendix}, the results of which are summarised in this subsection.  These states are labelled by a set of charges that couple to components of the various tensor potentials in the theory and span a space whose dimension is given in the second-to-last column of table~\ref{tab:BPScounting} for each Levi group, $M_{\alpha_{d+1}}$, with $0\le d \le 7$.  Correspondingly, the dimension of the space of instanton charges is given in the last column.   Table~\ref{tab:bpsorbits1} lists the  BPS orbits for each Levi group in the range $0\le d \le 7$.


\begin{table}[center]
  \centering
  \begin{tabular}[t]{||c|c|c|c||}
    \hline
  $D=$& $M_{\alpha_{d+1}}=E_d$ &dim point charges& dim instanton charges\\
$10-d$ && = \ $\dim\,U_{\alpha_{d+1}}$  &  =\  \# $+$ve roots of $E_d$\\
\hline
10A&1& 1&0\\
10B&$SL(2)$&0&1\\
9&$SL(2)\times \IR^+$& 3& 1\\
8&$SL(3)\times SL(2)$&6&4\\
7&$SL(5)$&10&10\\
6&$Spin(5,5)$&16&20\\
5&$E_{6}$&27&36\\
4&$E_{7}$&56\, (57)&63\\
3&$E_{8}$&120&120\\
\hline
  \end{tabular}\vskip 0.2cm
  \caption{The dimensions of the spaces spanned by the BPS point-like charges  and BPS  instantons of maximal  supergravity for the Levi subgroups in $P_{\alpha_{d+1}}$.
  The parentheses for $M_{\alpha_8}= E_{7}$  indicate  that the number of BPS states is one less than the dimension of the unipotent radical, $U_{\alpha_8}$, of the parabolic subgroup $P_{\alpha_8}$ of $E_{8}$.
  }
 \label{tab:BPScounting}
\end{table}

Table~\ref{tab:BPScounting} shows that, with one exception, the number of BPS
instantons in dimension $D$ equals the sum of the number of  BPS
particle states and the BPS instantons in dimension $D+1$, as
anticipated above.   The exceptional case is the parabolic subgroup
with $M_{\alpha_8}= E_7$,  where the number of instantons, 120,
is one greater than  the  number of BPS states, 56,  plus instantons, 63 in $D=4$.   The
 string theory interpretation of this extra state is discussed at the end of section~\ref{spinorbit}.

\begin{table}[!ht]
  \centering
    \begin{tabular}[t]{||c|c|c|c|c||}
   \hline
 $M_{\alpha_{d+1}}=E_d$   &    BPS &      BPS condition &   Orbit & Dim.\\
\hline
&&&&\\[-1em]
\hline
$GL(1)$&$\smallf 12$& - & GL(1)& 1\\
\hline
$SL(2)\times \IR^+$& $\smallf 12$& $v\, v_a=0$ &   Union of 2 orbits  &  1 and 2  \\
        & $\smallf 14$& $v\, v_a\neq 0$ &   $  \frac{GL(1)\times SL(2)}{\IR}   $  &3 \\
\hline
& $\smallf 12$  & $\epsilon^{ab}   \,  v_{i\,a}  v_{j\, b}=0$   &
                            $\frac{SL(3)\times
                             SL(2)}{ GL(2) \ltimes \IR^3} $ & 4
                           \\
$SL(3)\times SL(2)$ &&&&\\[-1.2ex]
                         & $\smallf 14$  & $\epsilon^{ab}   \,  v_{i\,a}  v_{j\, b}\neq 0$   &
                          $\frac{SL(3)\times
                             SL(2)}{SL(2)\ltimes \IR^2}$  &6 \\
\hline
 &  $\smallf 12$  &  $\epsilon^{ijklm}\, v_{ij}\, v_{kl}=0$ &
$\frac{SL(5)}{(SL(3)\times SL(2))\ltimes \IR^6}$  &7\\
$SL(5)$ &&&&\\[-1.2ex]
                       &  $\smallf 14$  &  $\epsilon^{ijklm}\, v_{ij}\, v_{kl} \neq 0$ &  $\frac{SL(5)}{Spin(2,3)\ltimes \IR^4}$  &10\\
\hline
       &  $\smallf 12$    &   ${\scriptstyle (S\Gamma^mS)=0}$
 &  $\frac {Spin(5,5)}{SL(5)\ltimes \IR^{10}}$   &11\\
$Spin(5,5)$ &&&&\\[-1.2ex]
                                               &  $\smallf 14$    &   ${\scriptstyle (S\Gamma^mS)\neq 0}$        & $\frac{Spin(5,5)}{Spin(3,4)\ltimes \IR^8}$   &16 \\
\hline
     &  $\smallf 12$     &    $\srel{I_3 = {\partial I_3\over \partial q^i}=0,}{ \text{and~} {\partial^2 I_3\over \partial q^i\partial q^j}\neq0.} $    &  $\frac{E_{6}}{Spin(5,5)\ltimes
  \IR^{16}} $  &17\\
 &&&&\\[-1.2ex]
  $E_6$                 &  $\smallf 14$     &   $ \srel{{\scriptstyle
      I_3=0,\  {\partial I_3\over \partial q^i}\neq0}}{}$
  & $\frac{E_{6}}{Spin(4,5)\ltimes \IR^{16}}$   & 26\\
 &&&&\\[-1.2ex]
                         & $\smallf 18$    &   ${\scriptstyle I_3\neq 0} $           &  $\frac{ GL(1)\times  E_{6}}{F_{4(4)}}$  &27\\
\hline
      &  $\smallf 12$      &  $\srel{ I_4  = {\partial I_4\over \partial q^i}=\left.{\partial^2   I_4\over   \partial
    q^i\partial q^j}\right|_{Adj_{E_7}}=0\,,}{\srel{\text{and~} {\partial^3   I_4\over   \partial
    q^i\partial q^j\partial q^k}\neq0.}{\srel{}{}}}$   &$\frac{E_{7}}{E_{6(6)}\ltimes
  \IR^{27}}$      &28 \\
                $E_{7}$                   &  $\smallf 14$      &  $\srel{ I_4= {\partial I_4\over \partial q^i}=0,}{\srel{\text{and~}
\left.{\partial^2 I_4\over \partial q^i\partial q^j}\right|_{Adj_{E_7}}\neq0.}{}}$        &$\frac{E_{7}}{Spin(5,6) \ltimes(
\IR^{32}\ltimes\IR)}$  &   45 \\
                                  & $\smallf 18$     & $ {\scriptstyle I_4=0, \  {\partial I_4\over \partial q^i}\neq0}$       &   $\frac{E_{7 } } {F_{4(4)}\ltimes
\IR^{26}} $  & 55\\
         &&&&\\[-1.2ex]
                            & $\smallf 18$     &  ${\scriptstyle I_4>0}$  &$\frac{  \IR^+\times E_{7}}{E_{6(2)}}$  &56 \\
\hline
  \end{tabular}
\vskip 0.2cm
    \caption{The orbits of instantons associated with the parabolic
      subgroup $P_{\alpha_{d+1}}$.  With one exception these are
      orbits of charged black hole states satisfying fractional BPS
      conditions that  are generated by the action of the Levi
      subgroup, $GL(1)\times E_d$, on a representative BPS state.
      The notation is explained in the text.  The degenerate case
      with $d=0$ is omitted here but will be discussed  in
      section~\ref{sl2examples}.
 The information in the  third and fourth columns is
  taken from~\cite{Ferrara:1997ci} and \cite{Lu:1997bg}, respectively. Details are provided in appendix~\ref{orbit1appendix}.
  Note the presence of the nonabelian 33-dimensional unipotent radical $\IR^{32}\ltimes \IR$ in the $\smallf 14$-BPS entry for $E_7$.
 }
  \label{tab:bpsorbits1}
\end{table}

The BPS orbits for each value of $d=10-D$ with Levi  factor $L_{\a_{d+1}}=GL(1)\times E_{d}$ are shown in table~\ref{tab:bpsorbits1}.
The tensors $v,\, v_a,\, v_{ia},\, v_{ij}$ and the spinor $S$ were introduced  in section~\ref{fourierorbits}.  $I_3$ and $I_4$ are cubic and
quartic invariants of $E_6$ and $E_7$, respectively, which are defined
in terms of the fundamental representation, $q^i$, of $E_6$ and $E_7$, as reviewed in appendices~\ref{sec:D5} and \ref{sec:D4}.
A general feature that is valid  for $d>1$  is that the $\smallf 12$-BPS states fill out orbits
 of the form
   \begin{equation}
    \label{e:HalfOrbit}
\mathcal{O}_{\frac12-BPS}\,=\,\frac{E_{d+1}}{E_{d}\ltimes \IR^{n_{d+1}}}\,,\qquad (n_3,\dots,n_7)\,=\,(3,6,10,16,27)\,.
  \end{equation}
The integers $n_{d+1}$  are the dimensions of the unipotent radicals,
$U_{\alpha_{d+1}}$, listed in table~\ref{tab:dimUnipotent}; they are  also
the dimensions of the spaces of BPS point charges for the symmetry  groups $E_{d+1}$ listed
in table~\ref{tab:BPScounting}, apart from the case of $d=7$
where $U_{\alpha_8}$ is   a non-abelian Heisenberg group.  As mentioned earlier, $U_{\alpha_8}$ has dimension $57$ while the $E_7$ point-like states (charged black holes) are labelled by only 56 charges.
The missing charge arises from the fact that among the 120 instantons in $D=3$ dimensions (see table \ref{tab:BPScounting}) there is one that is a wrapped $KKM$ with $x^\#$ (the fibre coordinate in \eqref{kkmonopole}) wrapped around the direction that is identified with (euclidean) time.
Since particle states in $D=4$ dimensions are obtained by identifying the decompactified direction with time, the exceptional instanton is one for which $x^\#$ grows in the cusp and its action becomes singular. By contrast,  56 of the $D=3$ instantons have action proportional to $r_7$ and are seen as point-like states in four dimensions, and the other 63 have no $r_7$ dependence and decompactify to instantons in four dimensions.

It is interesting to speculate about an additional line to table~\ref{tab:bpsorbits1} which we did not list, namely one for $M_{\a_9}=E_8$ inside the  affine Kac-Moody group $E_9$. While this latter group is infinite dimensional, one can still make sense of the orbits in terms of the finite dimensional vector space $\mathfrak{u}_{-1}$ in (\ref{gradedstructureforuminus}).  Indeed, $\mathfrak{u}_{-1}$ here is 248-dimensional and the action of $E_8$ is equivalent to the adjoint action
on its Lie algebra.  Thus
the orbits there coincide with the coadjoint nilpotent orbits for $E_8$.

\subsection{The string perturbation theory limit: $P_{\alpha_1}$.}
\label{sec:BPSTorbits}
In this limit BPS  instantons give non-perturbative corrections to
string perturbation theory. This involves an expansion in the
parabolic subgroup $P_{\alpha_1}$, with Levi factor $L_{\alpha_1}=
GL(1)\times Spin(d,d)$.
This limit is analogous to the limit considered in the previous subsection with the role of the decompactifying circle radius, $r_d$, replaced by the inverse string coupling in $D=10-d$ dimensions, which is denoted $1/\sqrt{y_D}$.
 In this case the orbits of BPS charges do not correspond to black hole charge orbits.

The BPS instantons that enter in this limit are easiest to analyse in
terms of the wrapping of euclidean world-volumes of $Dp$-branes, the
NS5-brane and stringy $KKM$ instantons.  The $Dp$-branes enter for all
values of $d\ge 0$ and their contribution alone leads to an abelian
unipotent radical, $U_{\alpha_1}$.  The NS5-branes contribute on tori
of dimension $d\ge 6$ and the $KKM$ instantons contribute for
$d=7$. Both these kinds of instantons render the unipotent radical nonabelian.   In section~\ref{spinorbit}  and
appendix~\ref{euclidDbrane} we review  the classification of
$Dp$-brane instantons in terms of the classification of $Spin(d,d)$
chiral spinor orbits, which leads to the following
features:
\begin{itemize}
\item{
For $d\le 3$ there is only one non-trivial orbit, which is $\smallf 12$-BPS. }
\item{
$\smallf 14$-BPS orbits arise when $d \geq 4$.  For $d=4$ and $5$ there is one orbit, namely the full spinor space of dimension $2^{d-1}$.   For $d=6$ and $d=7$ there is again a single  $\smallf 14$-BPS orbit  given by constrained spinors, which has dimensions 25 and 35, respectively.
}
\item {For $d=4$
the $\smallf 12$-BPS orbit is parameterised by a  spinor satisfying the
$Spin(4,4)$ pure spinor constraint, $S\cdot S=0$, while the full
eight-component spinor space (with $S\cdot S\ne 0$)  parameterises the $\smallf 14$-BPS orbit.}
\item{For $d=5$
the $\smallf 12$-BPS orbit is parameterised by an $Spin(5,5)$ spinor satisfying the pure
spinor constraint,\footnote{The Dirac matrices $\Gamma^i$ ($i=1,\dots
  2d$) form a $2^{\frac{d}{2}-1}\times 2^{\frac{d}{2}-1}$
 representation of the Clifford algebra $C\ell(d,d)$. We will   denote
the antisymmetric product of $r$ Dirac $\Gamma$ matrices  by $\Gamma^{i_1\cdots i_r} ={1\over r!} \,\sum_{\sigma\in\mathfrak
  S_r}\,(-)^{\sigma}\Gamma^{i_{\sigma(1)}}\cdots\Gamma^{i_{\sigma(r)}}$,
where $(-)^{\sigma}$ is the signature of the permutation $\sigma$.}
$S\Gamma^i S=0$, and once again the unconstrained spinor parameterises
the $\smallf 14$-BPS orbit. }
\item{For  $d=6$ the $\smallf 12$-BPS orbit is defined by an $Spin(6,6)$ spinor satisfying the  pure spinor constraint,
\be
\label{sixpure}
F_2 \ \  := \ \  \smallf12\,\sum_{i,j=1}^{12} S\Gamma^{ij} S \, dx^i\wedge dx^j  \ \ = \ \ 0\,,
\ee
where
the $\smallf 14$-BPS orbit is parameterised by a spinor satisfying the weaker constraints
\be
F_2 \ \  \ne \ \  0\ \,, \qquad F_2\wedge F_2 \ \ = \ \ 0\,.
\label{sixpuretwo}
\ee
In addition there is a $\smallf 18$-BPS orbit which is identified with the space of a spinor satisfying
\be
F_2\wedge F_2 \ \ \ne  \ \ 0\ \,, \qquad *F_2\wedge F_2 \ \ = \ \ 0\,,
\label{sixpurethree}
\ee
where $*$ is the Hodge star operator, and  a second $\smallf 18$-BPS orbit identified with the space spanned
by an unconstrained 32-component spinor.}
\item{For $d=7$ there are nine nontrivial orbits (in addition to the trivial orbit) that were determined by Popov \cite{Popov}.
The $\smallf 12$-BPS case is the smallest non-trivial orbit, which is the space spanned by a spinor satisfying
\be
F_3 \ \ := \ \ {1\over3!}\,\sum_{i,j,k=1}^{14} S\Gamma^{ijk} S \, dx^i \wedge dx^j \wedge dx^k  \ \ = \ \ 0\,,
\label{sevenpure}
\ee
where $S$ is a $Spin(7,7)$ spinor and $\Gamma^i$ ($i=1,\dots,14$) are
corresponding Dirac matrices.  However, the description of the
remaining orbits in terms of covariant constraints involving $F_3$
analogous to those of \eqref{sixpuretwo} and \eqref{sixpurethree} is apparently unknown.
}
\end{itemize}

We now turn to a detailed description of these orbits, which draws from the information in section~\ref{sec:OrbitCharacter}.
\subsubsection{Classification of spinor orbits}
\label{spinorbit}

 A review of the method for classifying spinor  orbits of  $G=
 Spin(d,d)$, when viewed as the subgroup of even and invertible elements of the
 Clifford algebra $C\ell(d,d)$ associated with $SO(d,d)$, can be found in~\cite{Trautman} (based on the original work in \cite{Igusa}   for   $d\leq   6$,    and\cite{Popov} for   $d=7$).

   The following tables will summarise some facts about these orbits, which are typically  cosets of the form $\mathcal O=Spin(d,d)/H$, $H$ being the stabilizer of a point in the orbit;  in three particular cases the quotients are actually $(GL(1)\times G)/H$ for reasons similar to those explained just above section~\ref{generalfeatures}.
  Since we do not require any specific features of these orbits we shall simplify their description by writing the real points the corresponding complex algebraic variety.
     For each value of $d$ we will
give a representative spinor of each orbit  (labelled $S^0$ in column
1 and defined in appendix~\ref{euclidDbrane}),
together with its stabiliser ($H$ in column $2$), its dimension
(dim$(G/H)$ in column 3)   and the fraction of supersymmetry it
preserves -- i.e., its BPS degree $N/2^{d-1}$, which  is
  determined by  the number of linearly independent spinors $N$ of
  the orbit representative $S^0$. In the
  following we will only list the BPS orbits appearing into the
  Fourier coefficients of the coefficients we are interested in. A
  more complete discussion is given  in section~\ref{sec:OrbitCharacter}.

The tables that follow have the following general properties:
\begin{itemize}
\item The bottom row is the trivial orbit and the top row is the
  dense orbit of a full spinor.
\item
The   second to bottom row  is the smallest  non-trivial orbit, which is the $\smallf 12$-BPS configuration with
orbit parametrized by the coset\footnote{Although the orbits listed in this
section are over  $\IR$  or $\C$, the structures are largely independent of the ground field.  For example, this particular orbit has the same form over
any  field $k$ with characteristic different from 2, but with the $\IR$ factor replaced by $k^{d(d-1)\over2}$.}
\begin{equation}
  \label{eq:halfBPS-orbit}
\mathcal O_{\tfrac12-BPS}= {Spin(d,d)\over SL(d)\ltimes \IR^{d(d-1)\over2}}
\end{equation}
of dimension $1+d(d-1)/2$. This is the  orbit of a spinor satisfying the pure spinor constraint and can be obtained by acting on the ground state of the Fock space representation of the spinor  with $SO(d,d)$ rotations.

\item The third to bottom row is the second smallest  non-trivial orbit (the NTM, or $\smallf 14$-BPS,  orbit), which  arises for $d\ge 4$ and is the coset
  \begin{equation}
    \label{eq:quarterBPS-orbit}
 \mathcal O_{\tfrac14-BPS} \ \ = \ \    {Spin(d,d)\over (Spin(7)\times SL(d-4))\ltimes U_{(d-4)(d+11)\over2}}\ ,
  \end{equation}
 where $U_s$ is a
 unipotent group of dimension $s$ (which is nonabelian for $d\ge 6$).
\end{itemize}

In more detail, the specific orbits for each $Spin(d,d)$ group are as follows:

\smallskip

\noindent$\blacktriangleright$ $Spin(1,1)$ is trivial.  For $Spin(2,2)$ and $Spin(3,3)$  the action of the spin group
is transitive  and there are only two orbits:~the trivial one of dimension $0$, and the
Weyl spinor orbit.  This is in accord with the discussion in the previous subsection.
$$
  \begin{tabular}[h]{||c|c|c|c||}
\hline
\multicolumn{4}{||c||}{$G=Spin(2,2)$}\\
\hline
$S^0$& stabilizer $H$ & dim$(G/H)$ & BPS\\
\hline
    1 & $ SL(2)\ltimes \IR$& $2$ & $\smallf 12$\\
0  & $Spin(2,2)$ & 0  & $--$\\
\hline
  \end{tabular}
  $$
  $$
 \begin{tabular}[h]{||c|c|c|c||}
\hline
\multicolumn{4}{||c||}{$G=Spin(3,3)$}\\
\hline
$S^0$& stabilizer $H$ & dim$(G/H)$ & BPS\\
\hline
    1 & $ SL(3)\ltimes\IR^3 $& $4$ & $\smallf 12$\\
0  & $Spin(3,3)$ & 0  & $--$\\
\hline
  \end{tabular}
$$

\noindent$\blacktriangleright$ For $d\geq4$ the action of the spin group is
not transitive and there are several non-trivial orbits represented by
constrained spinors.\footnote{The symbols  $e_{i_1\cdots i_r}$ and
 $e^*_{i_1\cdots   i_r}$ labelling the spinor $S^0$ are defined in appendix~\ref{euclidDbrane}.}  The first orbit listed  in the $Spin(4,4)$ table, the full spinor orbit of dimension 8, is actually the quotient $( GL(1) \times Spin(4,4))/Spin(7)$.  A similar $GL(1)$ factor occurs  for the largest orbit  of the groups   $Spin(6,6)$ and $Spin(7,7)$, but not for $Spin(5,5)$ (see below).
$$  \begin{tabular}[h]{||c|c|c|c||}
\hline
\hline
\multicolumn{4}{||c||}{$G=Spin(4,4)$}\\
\hline
$S^0$& stabilizer $H$ & dim$(G/H)$ & BPS \\
\hline
$1+e_{1234}$ & $Spin(7)$ & 8& $\smallf 14$\\[1ex]
    1 & $ SL(4)\ltimes \IR^{6}$& 7 & $\smallf 12$\\
0  & $Spin(4,4)$ & 0  & $--$\\
\hline
  \end{tabular}
$$
$$  \begin{tabular}[h]{||c|c|c|c||}
\hline
\hline
\multicolumn{4}{||c||}{$G=Spin(5,5)$}\\
\hline
$S^0$& stabilizer $H$ & dim$(G/H)$ & BPS\\
\hline
$1+e_{1234}$&  $ Spin(7)\ltimes \IR^{8}$& 16&$\smallf 14$\\[1ex]
    1 & $ SL(5)\ltimes \IR^{10}$& 11 & $\smallf 12$\\
0  & $Spin(5,5)$ & 0  & $--$\\
\hline
  \end{tabular}
$$

\noindent$\blacktriangleright$ The $Spin(6,6)$ case involves  some noncommutative unipotent subgroups $U_{s}$ of dimension $s$. The full spinor orbit of dimension 32 is $(GL(1)\times
Spin(6,6))/SL(6)$.
$$  \begin{tabular}[h]{||c|c|c|c||}
\hline
\hline
\multicolumn{4}{||c||}{$G=Spin(6,6)$}\\
\hline
$S^0$& stabilizer $H$ & dim$(G/H)$ & BPS\\
\hline
$1+e^*_{14}+e^*_{25}+ e^*_{36} $&$SL(6) $&32&0\\
$1+e^*_{14}+e^*_{25} $&$  Sp(6)\ltimes\IR^{14} $&31&$\smallf 18$\\[1ex]
$1+e^*_{14}$ &$( SL(2)\times Spin(7))\ltimes U_{17}$& 25&$\smallf  14$\\[1ex]
    1 & $ SL(6)\ltimes \IR^{15}$& 16 & $\smallf 12$\\
0  & $Spin(6,6)$ & 0  & $--$\\
\hline
  \end{tabular}
$$

\noindent$\blacktriangleright$ For $Spin(7,7)$ the full spinor orbit of dimension 64 is
$( GL(1)\times
Spin(7,7))/(G_2\times_{\ZZ_2} G_2)$, where $G_2$ is the exceptional
group of rank 2 and where $H_1\times_{\ZZ_2} H_2$
denotes the almost direct product of two groups intersecting on
$\ZZ_2$.  Of the total of 10 orbits obtained in~\cite{Popov},
  we  quote only  the ones relevant for the analysis of the Fourier modes
  discussed in this paper:
$$
  \begin{tabular}[h]{||c|c|c|c||}
\hline
\hline
\multicolumn{4}{||c||}{$G=Spin(7,7)$}\\
\hline
$S^0$& stabilizer $H$ & dim$(G/H)$ & BPS\\
\hline
$1+e^*_{7} $&$ SL(6)\ltimes \IR^{12} $&44&$\smallf 18$\\[1ex]
$1+e^*_{147}+e^*_{257} $&$(  Sp(6)\times_{\ZZ_2}GL(1))\ltimes\IR^{26}
$&43&$\smallf 18$\\[1ex]
$1+e_{1234}$ &$( SL(3)\times Spin(7))\ltimes U_{27}$& 35&$\smallf 14$\\[1ex]
    1 & $ SL(7)\ltimes \IR^{21}$& 22 & $\smallf  12$\\
0  & $Spin(7,7)$ & 0  & $--$\\
\hline
  \end{tabular}
$$
\subsubsection{Neveu--Schwarz five-brane and stringy $KKM$ instantons}\hfill\break
\label{sec:nsinstantons}
The wrapped world-volume of the NS5-brane produces a new kind of
instanton when $d\ge 6$, which is a source of $B_{\rm NS}$ flux.  Whereas
the $Dp$-brane instantons have actions of the form $C/g_s$ with
  $C$ independent of $g_s$, the
wrapped  NS5-brane   has an action of the form $C/g_s^2$.  This means
that such NS5-instantons are suppressed by $e^{-C/g_s^2}$, and so, in
the string  perturbation theory regime they are suppressed relative to
the $Dp$-brane instantons.
The presence of the charge carried by this
wrapped NS5-brane instanton leads to a non-commutativity of the unipotent
radical, $U_{\alpha_1}$, which is  a Heisenberg group  (this is
analogous to the fact that the $KKM$ instanton in $D=3$ led to
non-commutativity of the  unipotent radical $U_{\alpha_8}$ in the
$P_{\alpha_8}$ parabolic subgroup of $E_8$).
 The non-commutativity arises because the presence of a
    NS5-brane  charge generates a non-trivial $B_{\rm
      NS}$ background.  This affects the  definition of the $D$-brane charges
    due to the dependence on    $B_{\rm
      NS}$ of their field-strengths,  $F^{(4)}:=dC^{(3)}+dB_{\rm NS}\wedge C^{(1)}$ and
$*F^{(4)}=dC^{(5)}+C^{(3)}\wedge dB_{\rm NS}- dC^{(3)}\wedge B_{\rm
  NS}$.
 Since there is only one euclidean NS5-brane configuration on a
6-torus (the $D=4$ case) the non-commutative part of
$U_{\alpha_1}$ is one-dimensional, so the unipotent radical forms a Heisenberg group.

Upon further compactification on $\rT^7$ to $D=3$ there are $7$
distinct wrapped NS5-brane world-volume instantons, one for each
six-cycle.  In addition, there are 8 M-theory $KKM$ instantons that are distinguished from each other in the M-theory
description by identifying the coordinate $x^\#$ with any one of the
1-cycles, as explained earlier.  In string language, one of these is
the wrapped euclidean $D6$-brane  that has been counted as one of the
64 components of the $SO(7,7)$ spinor space and contributes to the
abelian part of the unipotent radical $U_{\alpha_1}$.  The other $7$
are $KKM$ instantons with $x^\#$ identified with a circle in one of
the $7$ other directions.  These are T-dual to the 7 wrapped
NS5-branes. The presence of the $D6$-brane and KKM instantons  leads to a higher degree of non-commutativity of the unipotent radical, due for example,  to
the non-linear dependence of the D6-brane field strength on  $B_{\rm NS}$ through
$*dC^{(1)}=dC^{(7)}+\frac12\, B_{\rm NS}\wedge dC^{(5)}- \frac12
dB_{\rm NS}\wedge C^{(5)}-\frac13 B_{\rm NS}\wedge B_{\rm NS} \wedge
dC^{(3)}+\frac13 B_{\rm NS}\wedge dB_{\rm NS}\wedge dC^{(5)}$.

 This counting coincides with that expected from
a group theoretic analysis of the dimension of the abelian and
non-abelian (i.e., derived subgroup) parts  of the  unipotent radical
summarised in the columns labelled ``first node'' of table
\reftab{tab:dimUnipotent}.

\subsection{ BPS instantons in the semi-classical M-theory limit: $P_{\alpha_2}$.}
\label{sec:BPSMorbits}

 This is the  limit in which  the volume, $\calV_{d+1}$, of  the
 M-theory torus  $\calT^{d+1}$ becomes large and semi-classical
 eleven-dimensional supergravity is a good approximation.  The Fourier
 modes of interest are those associated with the maximal parabolic
 subgroup  $P_{\alpha_2}$, which has Levi subgroup
 $L_{\alpha_2}=GL(1)\times SL(d+1)$.
The constant terms in the Fourier expansion  were considered in \cite{Green:2010kv} and shown to match expectations based on perturbative eleven-dimensional supergravity.

The instanton charge space can be described as follows. The wrapped $KK$
world-lines do not give instantons in this limit since their action is
independent of the volume, $\calV_{d+1}$.  Wrapped euclidean
$M2$-branes appear in $D\le 8$ dimensions (corresponding to symmetry
groups with rank $\ge 3$), while the wrapped euclidean $M5$-brane
arises for $D\le 5$ dimensions (corresponding to symmetry groups with
rank   $ \ge 6$) and the wrapped world-volume associated with the $KKM$
enters first in $D=3$ dimensions (i.e., for symmetry group $E_8$).
These instanton  actions have the  exponentially suppressed form  $\exp(-C\calV_{d+1}^a)$,   where $C$ is independent of $\calV_{d+1}$ in the
limit $\calV_{d+1}\to \infty$, and $a=3/(d+1)$ for the wrapped $M2$-brane,
$a=6/(d+1)$ for the wrapped $M5$-brane and $a=7/(d+1)$ for the wrapped
$KKM$.

  The space spanned by the 3-form, $v_{[ijk]}$ that couples  to
  $M2$-brane world-sheets wrapping 3-cycles inside  the $M$-theory torus $\calT^{d+1}$ has dimension
\bea
  D^{d+1}_{M2}  \ \ = \ \  {(d+1)!\over 3!\,  (d-2)!}\,,
    \label{e:dM2}
  \eea which equals $1$, $4$, $10$, $20$, $35$, and $56$, respectively, for tori of dimensions $d+1=3$, $4$, $5$, $6$, $7$, and $8$ (corresponding to the duality groups $E_3, \dots, E_8$).
Similarly,  the space of euclidean five-branes wrapping
6-cycles inside $\calT^{d+1}$ has dimension
\begin{equation}
  \label{e:dM5}
  D^{d+1}_{M5}= {(d+1)!\over 6! (d-5)!}\,,
\end{equation} which equals $1$, $7$, and $28$, respectively, for $d+1=6$, $7$, and $8$ (corresponding to duality groups $E_6$, $E_7$, and $E_8$).
Finally,  a finite action $KKM$ instanton only exists if there are 8 circular dimensions, so it only contributes for the $E_8$ case.  As argued earlier, there are 8 distinct objects of this kind since $x^\#$ is distinguished from the other circular coordinates.

Again these dimensions can be compared with those  listed in section~\ref{sec:OrbitCharacter}  and summarised in table \reftab{tab:dimUnipotent} under the heading ``second node''. The wrapped euclidean $M2$-branes contribute the dimensions of abelian part of the
unipotent  radical for this  maximal parabolic  subgroup.  In fact the numbers in the left-hand column of the second node heading are equal to $D^{d+1}_{M2}$ for all $0\le d\le 7$.
 The $M5$-brane charge space of dimension
$D^{d+1}_{M5}$,   equals  the dimension of the non-commutative   part  (i.e., derived subgroup)  of  the  unipotent  radical for $E_6$ and $E_7$, while
for $E_8$ there is also a contribution of 8 from the  $KKM$
instantons.
In this case the  non-abelian component of the unipotent radical
     arises from the $KKM$ instanton dependence on the  3-form
     $A^{(3)}$ configurations (analogous to the way the $B_{\rm
       NS}$ configurations induced the non-commutativity in the
     previous section).

Although we have given a list of dimensions of the space spanned by the orbits, in this case we have not analysed the BPS conditions to discover how the complete space decomposes into orbits with fractional supersymmetry.   However, the latter part of this paper analyses the complete orbit structure for the subgroup $P_{\alpha_2}$ and the
list of orbits is given in table~\reftab{tab:minimalnilpotent22}.
From this we can identify, for each value of $d$, the minimal
($\smallf 12$-BPS) and NTM ($\smallf 14$-BPS) orbits, as well as many others that arise when $d\ge 5$ (i.e. for $E_6$, $E_7$ and $E_8$).

\section{Explicit examples of Fourier modes for rank $\le 5$.}
\label{lowrank}

\subsection{Fourier expansions for higher rank groups  }\label{sec:nreFourierCoeff}

Suppose that $\phi \in C^\infty(\G\backslash G)$ is an automorphic function, and that $A\subset G$ is an abelian subgroup which is isomorphic to $\IR^m$ for some $m>0$.  If $\G\cap A$ corresponds to a lattice in $\IR^m$ under this identification, then $\phi$'s restriction to $A$, $\phi(a)$, has a Fourier expansion.  The same is true for any right translate $\phi(ag)$, for $g$ fixed.
  A prime example of this is  $A$ equal to  the unipotent radical $U$ of a maximal parabolic subgroup $P=LU$ of $G$, when $U$ is abelian and $\G$ is arithmetically defined:
\begin{equation}\label{fourierexp1a}
    \phi(ug) \ \ = \ \  \sum_{\chi}\chi(u)\phi_\chi(g)\ , \  \ \ \phi_\chi(g) \ \ = \ \ \int_{\G\cap U\backslash U} \phi(ug)\,\chi(u)^{-1}\, du\,,
\end{equation}
where the sum is taken over all characters $\chi$ of $U$ which are trivial on $\G\cap U$.
In particular the special case $u=e$,
\begin{equation}\label{fourierexp1b}
    \phi(g) \ \ = \ \ \sum_{\chi}\,\phi_\chi(g)\,,
\end{equation}
 reconstructs $\phi$ as a sum of its Fourier coefficients $\phi_\chi$.  These Fourier coefficients
 are in general distinct from Whittaker functions, which are Fourier coefficients for the minimal parabolic.
 When $U$ fails to be abelian the coefficients $\phi_{\chi}$ defined by (\ref{fourierexp1a}) still make sense, though $\phi$ is no longer a sum of them.  Instead, it is the integral of $\phi$ over the commutator subgroup\footnote{The commutator subgroup $[U,U]$ is the smallest normal subgroup of $U$ which contains all elements of the form $[u_1,u_2]$, for $u_1,u_2\in U$.} $[U,U]$ of $U$ which has an expansion
\begin{equation}\label{fourierexp2}
    \int_{\G\cap [U,U]\backslash [U,U]}\phi(ug)\,du \ \ = \ \ \sum_{\chi}\,\phi_\chi(g)\,;
\end{equation}
in other words, the Fourier expansion only captures a small part of $\phi$'s restriction to $U$ -- the part which transforms trivially under $[U,U]$.

A character on $U$ can be viewed as a linear functional on its Lie algebra $\mathfrak u$ via the differential.  In our case, in which $U$ is the unipotent radical of a maximal parabolic subgroup $P=P_{\a_j}$ for some simple root $\a_j$, $\mathfrak u$ has a graded structure
\begin{equation}\label{gradedstructureforu}
    \mathfrak u \ \ = \ \ {\mathfrak u}_1 \,\oplus\,{\mathfrak u}_2 \,\oplus\,\cdots
\end{equation}
where $\mathfrak u_k$ is the span of root vectors for roots of the form $\a=\sum c_k \a_k$, with $c_j=k$.  The Killing form $B(\cdot,\cdot)$ exhibits the dual ${\mathfrak u}^*$ of $\mathfrak u$ as the complexification of the Lie algebra
\begin{equation}\label{gradedstructureforuminus}
    \mathfrak u_{-} \ \ = \ \ {\mathfrak u}_{-1} \,\oplus\,{\mathfrak u}_{-2} \,\oplus\,\cdots\,.
\end{equation}
The commutator subgroup $[U,U]$ has Lie algebra ${\mathfrak u}_2\oplus{\mathfrak u}_3\oplus\cdots$, so the differential of a character is sensitive only to $\mathfrak u_1$.  Again through the bilinear pairing of the Killing form, its dual space $\mathfrak u_1^*$ is isomorphic to the complexification ${\mathfrak u}_{-1}\otimes\C$ of ${\mathfrak u}_{-1}$.  The exponential of any such a linear functional is a character of $U$, and hence ${\mathfrak u}_{-1}\otimes\C$ is known as the {\sl character variety} of $U$.

Now let $\chi$ be a character of $U$ which is invariant under the discrete subgroup $\G\cap U$.  The above correspondence guarantees the existence of a unique
\begin{equation}\label{Ychar}
Y \  \in  \  {\mathfrak u}_{-1}\otimes \C  \ \ \  \text{such that} \ \ \ \chi(e^X)\  = \  e^{iB(Y,X)}\,.
\end{equation}
 The set of all such $Y$ produced from characters $\chi$ of $(\G\cap U)\backslash U$ forms the {\em charge lattice} in $  {\mathfrak u}_{-1}$.
 Decompose $P=LU$, where $L$ is the Levi component.  Then   formula (\ref{fourierexp1a}) and the automorphy  of $\phi$ under any  $\g\in \G\cap L$
imply that
\begin{equation}\label{fourierexp3}
\aligned
    \phi_\chi(\g g) \ \  & = \ \ \int_{\G\cap U\backslash U} \phi(\g^{-1} u\g g)\,\chi(u)^{-1}\, du \\
& = \ \ \int_{\G\cap U\backslash U} \phi(u g)\,\chi(\g u\g^{-1})^{-1}\, du \,.
\endaligned
\end{equation}
Here we have changed variables $u \mapsto \g u \g^{-1}$, which preserves the measure $du$
 because $\g$ lies in the arithmetic subgroup $\G\cap L$.
  In terms of (\ref{Ychar})
\begin{equation}\label{fourierexp4}
    \chi(\g e^{X} \g^{-1}) \ \  = \ \  \chi(e^{\g X \g^{-1}}) \ \ = \ \ e^{iB(Y,\g X \g^{-1})}
    \ \ = \ \ e^{iB(\g^{-1} Y\g ,X )}\,,
\end{equation}
because of the invariance of the Killing form under the adjoint action; the character in the second line of (\ref{fourierexp3})   is hence equal to the character for the Lie algebra element $\g^{-1}Y\g \in {\mathfrak u}_{-1}\otimes \C$.

Consequently, the Fourier coefficients (\ref{fourierexp1a}) are related for characters $\chi$ which lie in the same $\G \cap L$-orbit  under the adjoint action on ${\mathfrak u}_{-1}\otimes \C$.
It should be remarked that $\mathfrak u_{-1}$ -- like each space $\mathfrak u_j$ -- is
invariant under the adjoint action of  $L$, and in fact furnishes an irreducible representation of $L$ (a fact
which can be verified in each example using the Weyl character formula -- see the tables in \cite[\S5]{Miller-Sahi}, for example, for a complete list).
The complexification  $L_\C$ of $L$ likewise acts on ${\mathfrak u}_{-1}\otimes \C$ according to an irreducible representation, and  carves it up into finitely many complex character variety orbits.

  Similarly, the adjoint action of $\G\cap L$  on the set of characters of $U$ which are trivial on $\G\cap U$ refines these complex orbits into myriad further ``integral'' orbits.
   Those characters naturally form a lattice inside of $i {\mathfrak u}_{-1} \subset {\mathfrak u_{-1}}\otimes \C$, and this last action is that of a discrete subgroup of $L$ on a lattice, e.g., the action of $GL(n,\Z)$ on $\Z^n$ in a particular special case.  The integral orbits are more subtle to describe because of number-theoretic reasons; indeed, even describing $\G\cap L$ for a large exceptional group is quite complicated.

Each of these complex character variety orbits (and hence each of the $\G\cap L$-orbits on the set of characters that are trivial on $\G\cap U$) is thus contained in a single (complex) coadjoint
nilpotent orbit.  It therefore makes sense to categorize the
complex  character variety orbits by giving their basepoints and
dimensions.
 Some of this information  was provided in section~\ref{bpsinstantons}, based on the analysis of BPS states in string
  theory. This analysis  focused on the  supersymmetric orbits and did
  not cover all possible orbits. A  systematic and detailed analysis of the remaining orbits for the maximal parabolic subgroups we study will be given in~\ref{sec:OrbitCharacter}.  These have long been known for the classical groups by the study of ``classical rank theory''; the paper \cite{Miller-Sahi} contains a listing for all maximal parabolic subgroups of exceptional groups.  In addition, the integral orbits are also known in some important cases:~Bhargava  \cite[Section 4]{bhargava} and Krutelevich~\cite{krutelevich} treat certain cases, with additional cases to appear in forthcoming work of Bhargava.

Note that the calculation (\ref{fourierexp3}) shows that each coefficient $\phi_\chi$ -- which is determined by its values on $L$ -- is automorphic under any $\g$ that lies in both $\G$ and $\operatorname{Stab}_{L}(\chi)$, the stabilizer of $\chi$ within $L$. In terms of the differential, these are the elements of $\G\cap L$
for which
the adjoint action fixes the element $Y\in {\mathfrak u}_{-1}\otimes \C$  from  (\ref{Ychar}).  One can therefore use (\ref{fourierexp3}) to write the sum of $\phi_\chi(g)$ ranging over $\chi$ in one of the integral orbits, as the sum of left $\g$-translates of a fixed $\phi_\chi$, where $\g$ now ranges over cosets of $\G\cap L$ modulo the stabilizer of this fixed character.   This shows that not all of the Fourier coefficients need to be  computed individually; knowing them for orbit representatives of characters is tantamount to knowing them for all characters.    Furthermore, the vanishing of any Fourier coefficient $\phi_\chi$ as a function of $L$ is equivalent to that of all Fourier coefficients in its orbit.


The following subsections (together with details that are presented in
appendix~\ref{modesdetails})  concern some specific, explicit
examples of the   Fourier modes of  the coefficient functions
$\calE_{(0,0)}^{(D)}$ and $\calE_{(1,0)}^{(D)}$   for the low rank
duality groups with $d \le 4$ (i.e. $D\ge 6$).   In these cases
standard, classical techniques can be used to obtain exact
expressions, including the arithmetical divisor sums that appear.
These techniques have the virtue of being relatively simple in these
special low rank cases; the higher rank cases of $E_6$, $E_7$ and
$E_8$ will be discussed in the later sections, although without
precise calculations -- our chief contribution is to use
representation theory to show that   many of them vanish.  The divisor sums could also be calculated using Hecke operators, though we do not do so here.

 In
 each particular case we will explicitly identify the character $\chi$,
 which lies in the lattice of characters of $U$ that are trivial on
 $\G\cap U$, with a tuple of integral parameters $m_i$, and use the
 notation
\begin{equation}
 \cF^{(D)\alpha}_{(p,q)}(\ell;m_i) \ \ : = \ \ \(\calE_{(p,q)}^{(D)}\)_\chi\!\!(\ell)   \ \ \ \ \text{and} \ \  \ \ \
  F^{G\,\alpha}_{\beta; s}(\ell;m_i)\ \ : = \ \ \(E_{\beta;s}^{G}\)_\chi\!\!(\ell)
\label{invfouriermodes}
\end{equation}
to refer to the Fourier modes of $\calE_{(p,q)}^{(D)}$ and $E_{\beta;s}^{G} $, respectively.   For brevity we shall sometimes drop the dependence on $\ell \in L$ from the notation.

The precise details of these Fourier coefficients could, in principle, be independently checked against an explicit evaluation of instanton contributions to the graviton scattering amplitude, but in practice such detailed verification is very difficult.  However,  most details of the contribution of $\smallf 12$-BPS instantons  to these coefficients in limit (i),  the decompactification limit in which  $r_d\gg1$, can be motivated directly from string theory.  This is the limit in which, for these low rank cases, the instantons are identified with wrapped world-lines of small black holes of the $(D+1)$-dimensional theory.
The
asymptotic behaviour can be understood by studying the
fluctuations around $\smallf 12$-BPS $D$-particle
configurations in a manner that generalises the arguments of~\cite{Green:1997tn}, leading to an expression
for the modes  in  $D=10-d \le 9$ dimensions of the form
\begin{equation}\label{e:AsympBPS}
\cF^{(D)\alpha_{d+1}}_{(0,0)}(k) \ \  = \ \
(\smallf{r_d}{\ell_{D+1}})^{n_D}\,  \sigma_{7-D}(|k|) \,
{e^{-S_D(k)}\over
  S_D(k)^{8-D\over 2}}\, \left(c_D+O(\smallf{\ell_{D+1}}{r_d})\right).
\end{equation}
Here   $c_D$ is a positive constant and  $S_D(k)=2\pi |k|
r_d m_{\frac12}$ is the action for the world-line of the $D$-particle
wound around the circle of radius $r_d$ and $m_{\frac12}$, which is a
function of the moduli, is the mass of a ``minimal'' $\smallf 12$-BPS
point-like particle state in $D+1$ dimensions -- that is, a state that
is related by duality to the lightest mass single-charge $D$-particle.
Such states can form threshold bound $D$-particles of mass $p\,
m_{\frac12}$. The divisor sum, $\sigma_n(k) = k^n\, \sigma_{-n}(k)=
\sum_{q|k} q^n$, sums over the winding number $q$ of the world-lines
of such  $D$-particles (where $k = p\times q$) and can be identified
with a matrix model partition function.  The factor of
$S_D(k)^{(D-8)/2}$  comes from
integration over the  bosonic  and fermionic zero modes and $n_D$ is a constant that depends on the dimension  $D$.  Because of the high degree of supersymmetry
preserved by the $\smallf  12$-BPS configuration  it turns out that
this approximation  is exact in several cases.  In $D=6$ our
  results are in
  agreement with~\cite{Pioline:2010kb}.    We have not completed
an independent quantum calculation of the $\smallf 14$-BPS instanton
contributions, which are more subtle.   We do not know a general pattern for the
exponent  $n_D$, though it is easily computable in each of the examples below.

The Fourier coefficients for different characters satisfy a number of relations between them due to (\ref{fourierexp3}).  This phenomenon is particularly striking on the symmetry groups with $D\ge 7$, which are products of $SL(n)$'s.  For example, the formula in (\ref{e:431}) depends on $p_1$ and $p_2$ only through the combination $|p_2+p_1\Omega|$, which is actually an instance of the principle in (\ref{fourierexp3}) (see \cite{bumpgl3} for more details).  Thus these coefficient functions (aside from a substitution in their argument) are determined by the ones  having $(p_1,p_2)=(1,0)$.  In general, a theorem of Piatetski-Shapiro \cite{psmult} and Shalika \cite{shalika} computes the Fourier expansion of an automorphic form on $SL(n)$ in terms of similarly simple Fourier coefficients.  In particular, they demonstrate that  the ``abelian Fourier coefficients'' that appear in (\ref{fourierexp2}) determine ones absent there that come from nonabelian charges.

 However, this theorem is not true for groups other than $SL(n)$.  Certain features still persist for the ``small'' automorphic representations which are the focus of this paper; see \cite[section 4]{Miller-Sahi} for the analogous result for minimal automorphic representations of $E_6$ and $E_7$.
 An expression that includes contributions from non-commutative charges (which are not addressed in this paper) is presented in
  \cite{Pioline:2010kb} in the case of $D=3$.  See also \cite{Pioline:2009qt,Bao:2009fg,Persson:2011xi} for a discussion of noncommutative contributions in a
  different context.

\subsection{$D=10B$: $SL(2,\Z)$}
\label{sl2examples}
The simplest nontrivial (but very degenerate) example arises in the case of the IIB theory with $D=10$,
where  the  discrete duality  group  is  $SL(2,\Z)$.\footnote{The type IIA theory has no instantons, which means that only the $0$-dimensional trivial orbit contributes.}
In this case  the   $\smallf 12$-    and   $\smallf 14$-BPS   interactions,    $\cE^{(10)}_{(0,0)}$   and
$\cE^{(10)}_{(1,0)}$, are given by Eisenstein series \cite{Green:1997tv,Green:1998by}
\begin{eqnarray}
  \label{e:Couplings10D}
  \cE^{(10)}_{(0,0)}= 2\zeta(3)\,E^{SL(2)}_{3\over2}(\Omega)\,, \qquad  \cE^{(10)}_{(1,0)}= \zeta(5)\,E^{SL(2)}_{5\over2}(\Omega)\,,
\end{eqnarray}
where   $E^{SL(2)}_s(\Omega)$ is a non-holomorphic Eisenstein
series and   $\Omega:= \Omega_1 +
i\Omega_2=C^{(0)} + i/\sqrt{y_{10}}$.

It is useful to parametrize the coset  space  $SL(2,\IR)/SO(2)$ (i.e., the upper half
plane)  associated with the continuous symmetry group, $SL(2,\IR)$,
by   matrix representatives   of the form
 \be
 e_2 \ \  = \   \begin{pmatrix}
1& \Omega_1\\
 0 &1
 \end{pmatrix}
 \begin{pmatrix}
\sqrt{\Omega_2}&0\\
 0 & {1\over\sqrt{\Omega_2}}
 \end{pmatrix}.
 \label{sl2param}
 \ee
 This matrix lies in the maximal parabolic subgroup of upper triangular matrices in $SL(2,\IR)$;
its first factor is in the unipotent radical and the second factor lies in its standard Levi component.
 The
$SL(2)$  Eisenstein series can be expressed as
\begin{equation}
\label{e:EisSl2}
 2\zeta(2s)\,  E_s^{SL(2)} (\Omega) :=     \sum_{M_2\in\ZZ^2\bsz}
  (m^2_{SL(2)})^{-s}  =\sum_{(m,n)\in\ZZ^2\bsz} \frac{\Omega_2^s} {|n\Omega+m|^{2s}} \,,
  \end{equation}
where the $SL(2,\IZ)$-invariant  $($mass$)^2$ is defined   by
\begin{equation}
  m_{SL(2)}^2:=M_2G_2 M_2^t = {|n\Omega+m |^2\over \Omega_2}\,,
  \label{sl2mass}
  \end{equation}
 with  $G_2= e_2e_2^t$  and $M_2=(n\ m)\in \ZZ^2\bsz$.

It   is  straightforward   to  determine   the  Fourier
coefficients   using the  standard
expansion of such series in terms of Bessel functions,
\begin{equation}
  \label{e:ESl2fourier}
 E^{SL(2)}_s(\Omega)=\sum_{n\in\ZZ} F^{SL(2)}_s(n) \,e^{2i\pi n \Omega_1}\, .
\end{equation}
The zero Fourier mode is
\be
  \label{e:ESl2fourierCoef1}
   F_s^{SL(2)}(0)= \Omega_2^s+
  {\xi(2s-1)\over\xi(2s)}\,\Omega_2^{1-s}\,,
  \ee
  where $\xi(s)=\pi^{-s/2}\Gamma(s/2)\zeta(s)$.   The non-zero mode with phase $e^{2i\pi n \Omega_1}$ is
\begin{equation}
F^{SL(2)}_s(n) ={2\,        \Omega_2^{\frac12}\over\xi(2s)}
{\sigma_{2s-1}(|n|)\over |n|^{s-\frac12}}\,
K_{s-\frac12}(2\pi|n|\Omega_2)\, ,
  \label{e:ESl2fourierCoef}
\end{equation}
where  $\sigma_\alpha(n)=\sum_{0<d|n} d^{\alpha}$ is a divisor  sum.  Thus  the non-zero mode with frequency $n$ is  proportional to $K_{s-\half}$, which  is a modified Bessel function of the second kind.

In this degenerate case the only  limit to consider is
$\Omega_2\to \infty$, which is the limit of string perturbation theory
organized as a power series in $\Omega_2^{-2}$ corresponding to
  the genus expansion of a closed Riemann surface.  In
this limit the expansion of the coefficient functions  is dominated by
the two power behaved constant terms in the zero mode  $F_s^{SL(2)}(0)$ in
\eqref{e:ESl2fourierCoef1}, while the non-zero modes have asymptotic
behaviour at large $\Omega_2$,
\begin{equation}
F^{SL(2)}_s(n)
 =
{\sigma_{2s-1}(|n|)\over \xi(2s)|n|^{s}}\, e^{-2\pi|n|\Omega_2}\left(1 + O(\Omega_2^{-1})    \right) ,
\label{asymptotic}
\end{equation}
where the asymptotic expansion of the Bessel function
\be
    K_\nu (x) = \sqrt{\frac{\pi}{2x}} e^{-x} \left(1 +  O(x^{-1}) \right) \,,  \ \ \  x\gg 1\,,
\label{besseldef}
\ee
has been used.

The two power behaved terms   in (\ref{e:ESl2fourierCoef1})  have the interpretation of terms in
string perturbation theory, which is an expansion  in  $y_{10}$,
the square of the string coupling constant.  Furthermore, the Eisenstein series with $s=3/2$ and with
$s=5/2$ have the correct power-behaved terms to account precisely  for
the known behaviour of the $\R^4$ and $\partial^4\R^4$ terms in the
low energy expansion of the four graviton amplitude
  in 10 dimensions.  In \cite{Green:2010kv} it was
shown that this is in agreement with string perturbation theory  and  extends to the higher rank cases where the pattern of constant terms is more elaborate.
 Furthermore, the exponential terms in the expansion in \eqref{asymptotic} correspond to the expected $D$-instantons that
 arise in the $D=10$ type IIB theory.  This  illustrates the fact,
common to all BPS instanton processes, that the exponential decay of a
Fourier mode is proportional to the charge $n$ that determines the
phase of the mode.  The correction term of order $\Omega_2^{-1}$ in
\eqref{asymptotic}  indicates  perturbative corrections to the instanton contribution  given by an expansion in powers of  the string coupling constant that corresponds to the addition of boundaries in the
Riemann surface.

In this case the only instantons are $\smallf 12$-BPS $D$-instantons
-- there are no $\smallf 14$-BPS instantons in the ten-dimensional
type IIB theory.  However, it is known from string theory arguments
that the Eisenstein series at $s=3/2$ is associated with the $\smallf
12$-BPS $\R^4$ term while the series at $s=5/2$ is associated with the
$\smallf 14$-BPS $\partial^4 \R^4$  contribution  (\ref{e:Couplings10D}).   This leaves
unresolved  the question as to what features of these series at
special values of $s$ encode the fraction of supersymmetry that these
terms preserve?  This must be encoded in the measure.    Indeed in the
$s=3/2$ case it was argued  in \cite{Green:1997tn,Green:1998yf} that
the measure factor  $\sigma_{2}(|n|)$   arises from the $\smallf
12$-BPS $D$-instanton matrix model, which was verified in
\cite{Moore:1998et}.  Presumably, the $s=5/2$ measure
should arise  in a similar manner.

In most of the higher-rank examples that follow there is a less subtle
distinction between the $\smallf 12$-BPS and $\smallf 14$-BPS cases
since in typical cases there are $\smallf 14$-BPS instanton
configurations that break $\smallf 34$ of the supersymmetry. As will
be shown in the following,  these generally enter into  non-zero
Fourier modes of the coefficient $\calE_{(1,0)}^{(D)}$ for $3\le D
<10$ (although, as will also be seen later, only the $\smallf 12$-BPS
orbit contributes in  the $P_{\alpha_1}$ parabolic with $D=7,8,9$).
The subtleties of the measure factor are not required in order to
identify the fraction of supersymmetry preserved in such cases.
However, there are no $\smallf 18$-BPS configurations for $D>5$.
Therefore, for  $D>5$ the distinction between the coefficient
$\calE_{(0,1)}^{(D)}$ and the ones which preserve more supersymmetry
is again not determined by the spectrum of instantons that contribute
in the various limits under consideration. This indicates that the
$\smallf 18$-BPS nature of $\calE_{(0,1)}^{(D)}$ must be encoded in
the form of the measure factor.

\subsection{$D=9$: $SL(2,\ZZ)$}\label{section:4.3}\hfill\break
The coefficients of the $\R^4$ and $\partial^4\R^4$ interactions in this case are \cite{Green:1997di, Basu:2007ru,Green:2010wi}
\bea
  \label{e:e9D1}
  \cE^{(9)}_{(0,0)}=2\zeta(3)\,\nu_1^{-\frac37}\,E^{SL(2)}_{\frac32}+4\zeta(2)\nu_1^{\frac47},
  \eea
  \be
    \label{e:e9D2}
  \cE^{(9)}_{(1,0)}=\zeta(5)\nu_1^{-\frac57}\,E^{SL(2)}_{\frac52}+{4\zeta(2)\zeta(3)\over15}\nu_1^{\frac97}E^{SL(2)}_{\frac32}+{4\zeta(2)\zeta(3)\over15}\nu_1^{-\frac{12}7}\,,
\ee
where $\nu_1=(\ell^B_{10}/r_B)^2=g_A^{\scriptstyle{\frac {7}{8}}}\, (r_A/\ell^A_{10})^\threeh $ with $r_B$ the radius of the compact dimension in the IIB theory and $r_A = \ell_s^2/r_B$  the radius in the IIA theory.  The IIA string coupling, $g_A$, is related to that of the IIB theory by  $g_A = g_B \, \ell_s/r_B$.  Furthermore, the $D=9$ theory can be viewed as the compactification of M-theory from 11 dimensions on a 2-torus, $\calT^2$, with volume $\calV_2 = \nu_1^{2/3}\,\ell_{11}^2$.

The limit  $\nu_1\to 0$ is the limit in which the $\IR^+$ parameter of the continuous symmetry,  $SL(2,\IR) \times \IR^+$,  becomes infinite, which is the  decompactification limit  to the $D=10$ IIB theory ($r_B\to\infty$), while the limit  $\nu_1\to \infty$  is the semi-classical M-theory limit  in which   $\calV_2$,  the volume of $\calT^2$,  becomes infinite.
 Equations \eqref{e:e9D1} and \eqref{e:e9D2} show that there are no
 non-zero modes in either of these limits.  Since   $\Omega_2=
 g_A^{-1} r_A/ \ell_s$, the  perturbative IIB limit, $\Omega_2\to
 \infty$, is also the $D=10$ type IIA limit, $r_A \to \infty$.  This
 is the limit in the parabolic subgroup $GL(1) \times U$
 of the $SL(2)$ factor (given in \eqref{sl2param})   in which the parameter  in the $GL(1)$ Levi factor in the $SL(2)$
 becomes infinite.
 The non-zero Fourier modes of the expression for $\cE^{(9)}_{(0,0)}$
in \eqref{e:e9D1}  that contribute  in this limit are obtained by using the mode expansion
of $E_{3/2}$  given in the previous section in~\eqref{e:e9D1},  giving
\begin{eqnarray}
\nn \cF^{(9)}_{(0,0)}(k)&:=& \int_{[0,1]} d\Omega_1\, \cE^{(9)}_{(0,0)}\,
 e^{-2i\pi k\Omega_1}\\
 &=&8\pi   \Omega_2^{\frac12}\,   \nu_1^{-\frac37}\,
 {\sigma_2(|k|)\over   |k|}\,    K_1(2\pi|k|\Omega_2)\,.
\end{eqnarray}
The limit $\Omega_2\to \infty$ in the Bessel function in the
second line gives the $D$-instanton contribution to the coefficient of the $\R^4$ interaction in the  type IIB  perturbative string theory limit, which has the form, after reinstating the power of $\ell_9$ in the effective action, \eqref{effacts},
\begin{equation}
 \frac{1}{\ell_9}\,\cF^{(9)}_{(0,0)}(k)= {r_B\over\ell_s^2}\, \sqrt{8\pi}\,
 \sigma_{-2}(|k|)\, { e^{-2\pi   |k|
  \Omega_2}\over (2\pi |k|\Omega_2)^{-\frac12}}\, (1+O(\Omega_2^{-1}))\,,
\end{equation}
where the factor of $r_B/\ell_{s}$ shows that this term survives the limit $r_B\to \infty$.
Here we have used the relations $\nu_1=(\ell_{10}/r_B)^2$, $\ell_9^7=\ell_{10}^8/r_B$, and $\ell_{10}=\ell_s\Omega_2^{-1/4}$.

On the other hand, taking the large radius $r_A/\ell_{10}\to\infty$
limit in the IIA case gives
\begin{equation}
\frac{1}{\ell_9}\, \cF^{(9)}_{(0,0)}(k)
=\frac{1}{\ r_A} \,\sqrt{8\pi}\,
 \sigma_{-2}(|k|)\,  {e^{-2\pi   |k|
   r_A m_{\frac 12}} \over (2\pi|k| r_A m_{\frac12})^{-\frac12}}\,
\left(1+O(\ell_{10}/r_A)\right)\,,
\label{tenawrap}
\end{equation}
 where  $m_{\frac12}=1/(\ell_s g_A)$.    Here we have used the relations $\Omega_2=\f{r_A}{\ell_s g_A}$, $\nu_1=g_A^{1/2}r_A^{3/2}\ell_s^{-3/2}$, and $\ell_9=g_A^{2/7}\ell_s^{8/7}r_A^{-1/7}$. This  expression  reproduces  the
 asymptotic behaviour for the $\smallf 12$-BPS contribution given
 in~\eqref{e:AsympBPS} with $D=9$,
$n_D=-8/7$ and $S_9(k)=2\pi|k| r_A m_{\frac12}$.
The exponent has the interpretation of the action of the euclidean world-line of a type IIA $D0$-brane of charge $p$ wrapped $q$ times around the circle of radius $r_1=r_A$, where $k=p\times q$ (and the sum over $q$ is in $ \sigma_{-2}(|k|)$).

A similar expansion of the two Eisenstein series in~\eqref{e:e9D2}
gives the mode expansion of the coefficient $\calE^{(9)}_{(1,0)}$ as
the sum of two terms.  The occurrence of both the $s=3/2$ and $s=5/2$
series demonstrates that the $\partial^4\, \R^4$ interaction contains
a piece that is $\smallf 14$-BPS as well as a piece that is $\smallf
12$-BPS.
Repeating the above analysis for the $\smallf 14$-BPS part of $\cE_{(1,0)}^{(9)}$ (the $E_{5/2}$ term in  \eqref{e:e9D2}) and  making use of  \eqref{asymptotic} with $s=5/2$ gives (after multiplying by $\ell^3_9$ to reproduce the $\partial^4\, \R^4$ interaction in \eqref{effacts})
\bea
\left. \ell_9^3\, \cF^{(9)}_{(1,0)}(k)\right |_{\quart-BPS}
&\sim &   \f{\sqrt{2\pi}}{3} \,  (\ell^A_{10})^3   \, g_A^{-\frac 12} \,
\left(\frac{\ell_{10}^A}{r_A} \right)^{3} \,  \sigma_{-4}(|k|)\,\,
{e^{-S_9(k)}\over (S_9(k))^{-\threeh}} \,.  \nn\\
\label{quarter9}
\eea
 As with the $D=10$ examples, the distinction between the
$s=3/2$ and $s=5/2$ Eisenstein series is not seen in the instanton orbits (both series
contain the same 1-dimensional orbit) but must be encoded in the
different measure factors, such as  the divisor sum, which takes
the form $\sigma_{-4}(|k|)$ when $s=5/2$.  In contrast to the $\smallf
12$-BPS case we have not derived \eqref{quarter9}, or the analogous
expressions for $D <9$ obtained below,  by explicitly evaluating the
$\smallf 14$-BPS instanton contributions.

\subsection{$D=8$:  $SL(3,\ZZ)\times SL(2,\ZZ)$}\hfill\break
\label{exampleeight}

The coefficient function $ \cE^{(8)}_{(0,0)}$ is  given in terms of Eisenstein series by \cite{Green:1997di,Kiritsis:1997em,Basu:2007ru,Green:2010wi}
\begin{equation}
  \label{e:eEsi}
  \cE^{(8)}_{(0,0)}:= \lim_{\epsilon\to 0}
  \,\left(2\zeta(3+2\epsilon)\,E^{SL(3)}_{\alpha_1;\threeh+\epsilon}+4\zeta(2-2\epsilon)
  \,E^{SL(2)}_{1-\epsilon}(\calU)\right).
  \end{equation}
It was shown  in~\cite{Green:2010wi} that the  poles in $\epsilon$ of the  individual series in
  parentheses cancel and the expression is analytic at $\epsilon=0$. The coefficient function $  \cE^{(8)}_{(1,0)}$ is given by
    \begin{equation}
\label{e:eEti}
  \cE^{(8)}_{(1,0)}=\zeta(5)\,E^{SL(3)}_{\alpha_1;\frac52}+{4\zeta(4)\over3}\,E^{SL(3)}_{\alpha_1;-\frac12}\,E^{SL(2)}_2(\calU)\,.
\end{equation}
We have suppressed the dependence of the $SL(3)$ series on the $5$ parameters of the $SL(3)/SO(3)$ coset, but have indicated that the $SL(2)$ series depends on $\calU$, the complex structure of the 2-torus, $\rT^2$ (see appendix~\ref{appendixcoeff8} for details).

\medskip
 (i)  {\bf The nonmaximal parabolic\footnote{ In this
      somewhat degenerate case, the decompactification limit is
      associated with a nonmaximal parabolic so that its Levi matches the $D=9$ duality group.} $P_{\alpha_3}=GL(1)\times
  SL(2)\times \IR^+\times U_{\alpha_3}$}

This is relevant for the decompactification limit $r_2/\ell_9\to
\infty$. The Fourier modes, which are  integrals with respect to the
$U_{\alpha_3}$ factor   \eqref{e:UnipotentN3E3},  get contributions
from the sum of the modes of the $SL(3)$ and $SL(2)$ Eisenstein
series.  The modes of  $\calE_{(0,0)}^{(8)}$    are defined by
\begin{equation}
  \cF^{(8)\alpha_3}_{(0,0)} (kp_1,kp_2,k') :=\int_{[0,1]^3}\!\!\!\!dC^{(2)}dB_{\rm NS}d\calU_1\, e^{-2i\pi k(p_1C^{(2)}+p_2B_{\rm
      NS})-2i\pi k' \calU_1}\, \cE_{(0,0)}^{(8)}\,,
 \label{fouriereight}
 \end{equation}
where $\gcd(p_1,p_2)=1$ and $C^{(2)}$, $B_{\rm NS}$ and $\calU_1$ are the components of
the unipotent radical in~(\ref{e:UnipotentN3E3}).
Using the definition in~(\ref{e:eEsi}) the Fourier modes of  $ \cE_{(0,0)}^{(8)}$
are given by the sum of the Fourier modes of the $SL(3)$ and $SL(2)$
series defined in~\eqref{e:F2def}
and~\eqref{e:F22def}:
\begin{equation}
  \cF^{(8)\alpha_3}_{(0,0)} (kp_1,kp_2,k') = 2\zeta(3)\,F^{SL(3)\, \beta_2}_{\beta_1;\threeh} (kp_1,kp_2) + 4\zeta(2) \,F^{SL(2)}_1 (k')\,.
 \label{fouriereight2}
 \end{equation}
We have used the notation $\b_1$ and $\b_2$ on the righthand side to indicate the nodes of the $SL(3)$ Dynkin diagram that correspond to $\a_1$ and $\a_3$  (see figure~\ref{fig:E3lab}).

\begin{figure}[h]
  \centering
  \includegraphics[]{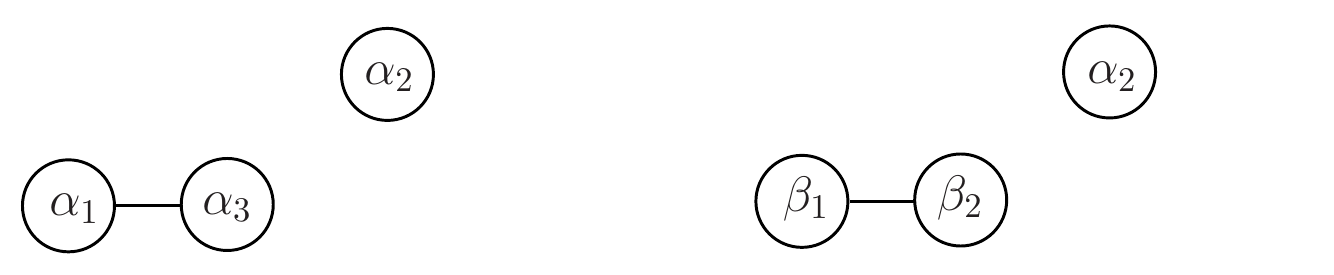}
  \caption{Correspondence between the labelling of the $SL(3)$ nodes
    in the $E_3$ Dynkin diagram according to figure~\ref{fig:dynkin} (in terms of
    $\alpha_1$ and $\alpha_3$) and the conventional labelling of the $SL(3)$
    Dynkin diagram (in terms of $\beta_1$ and $\beta_2$).}
  \label{fig:E3lab}
\end{figure}

Note that both contributions are nonsingular at $\e=0$ despite the simple poles in (\ref{e:eEsi}).  The reason that these Fourier coefficients do not have poles is that the residues of the series in (\ref{e:eEsi}) are constant.
Using the expression
  \eqref{e:F22}  for the $SL(2)$ Fourier modes and  setting  $\calU_2=r_2/r_1=r_2/r_B$ we obtain\footnote{Here, and in the following we will use the type IIB description, in which $r_1=r_B$.}
 \begin{equation}
  4\zeta(2)\,  F^{SL(2)}_1 (k')=4\pi\,
    \sigma_{-1}(|k'|)\,  e^{-2\pi \,  |k'|\,  r_2\times
      {1 \over r_1}}\, .
    \label{e:N2R4bis} \end{equation}
The exponent  can be identified with minus the action of the
  world-line of a $\smallf 12$-BPS charge $p$ KK state wrapped $q$ times
  around a circle of radius $r_2$, with  $p\times q = k'$. The divisor
  sum $\sigma_{-1}(|k'|)$ weights the different values of $p$ with a
  factor of $1/p$.   The expression \eqref{e:N2R4bis} agrees with the
  general asymptotic formula \eqref{e:AsympBPS},
 but it is notable that in this case there are no perturbative corrections.

The $SL(3)$ part is obtained from~\eqref{e:F2} with $s=3/2$,
 \begin{eqnarray}\label{e:431}
 2\zeta(3)\,  F^{ SL(3)\,\beta_2}_{\beta_1;\threeh} (kp_1,kp_2)
=4\pi
    \,\sigma_{-1}(|k|)\, e^{-2\pi |k|
  {|p_2+p_1\Omega|\over\sqrt\Omega_2}{1\over \sqrt\nu_2}} \,,
 \label{e:N3R41}
 \end{eqnarray}
where $\gcd (p_1,p_2) =1$.
This expression reproduces the asymptotic behaviour (which is again
exact) for the $\smallf 12$-BPS contribution given
  in~\eqref{e:AsympBPS} with $D=8$.
The exponent can be written as
\bea
  \label{e:DFbound}
- 2\pi |k|
  {|p_2+p_1\Omega|\over\sqrt\Omega_2}{1\over     \sqrt\nu_2} &=&  -2\pi |k| r_2 \, m_{p_1,p_2}\,,
  \eea
 where the $k=1$ contribution is  minus the action for the world-line of a state of mass
  \be
  m_{p_1,p_2} \, \ell_s =  |p_2+p_1\Omega|\,{r_1\over \ell_s}\,,
  \label{stringuni}
  \ee
  wound around the circle of radius $r_2$.  This is  the mass of a (non-threshold) bound state of $p_2$ fundamental strings and $p_1$ $D$-strings wound around the dimension of radius $r_1$.
 In the limit $r_2/\ell_9 \to \infty$ the Fourier coefficients with
 different $p_1$'s and $p_2$'s fill out an orbit under the action of
 the discrete subgroup $SL(2,\Z)$ of the Levi factor, which is the
 nine-dimensional duality group. This is
 made manifest by expressing $  m_{p_1,p_2}$  in nine-dimensional
 Planck units,
  \be
    m_{p_1,p_2} \, \ell_9 =  {|p_2+p_1\Omega|\over  \sqrt{\Omega_2}}\, \nu_1^{-3/7}\,,
 \ee
 where $SL(2,\Z)$ acts  with the usual linear fractional
 transformation on $\Omega$ and leaves $\nu_1$  invariant. When $k>1$ (\ref{e:DFbound}) is minus the world-line action of a threshold bound state
    of mass $p \times m_{p_1,p_2}$  wound $q$ times around the circle of radius $r_2$, where $k=p\times q$ and the divisor sum weights the contributions with a factor of $1/|q|$.

Thus,  in the  decompactification limit these instantons correspond to the expected contributions from the
point-like $\smallf 12$-BPS black hole  states in nine dimensions listed in appendix~\ref{sec:D9}.
The  Kaluza--Klein  $\smallf 12$-BPS   states  in \eqref{e:N2R4bis} are  in  the  singlet   $v$  and  the
$(p,q)$-string bound state in \eqref{e:N3R41} in the doublet $v_a$ of $SL(2)$.
 These contributions come from separate configurations
 ($v=0$, $v_a\neq0$)  and  ($v\neq0$,
$v_a=0$)    so that the condition $v v_a=0$ is satisfied.

The Fourier modes of the coefficient $\cE_{(1,0)}^{(8)}$ in the
$P_{\alpha_3} $ parabolic are defined as
\begin{equation}
 \label{e:N3D4R4}
\cF_{(1,0)}^{(8) \alpha_3}(kp_1,kp_2,k'):= \int_{[0,1]^3} \!\!\!\!\!dC^{(2)}
dB_{\rm NS}d\calU_1\, e^{-2i\pi \, k\, (p_1C^{(2)}+p_2B_{\rm
        NS})-2i\pi k'\,\calU_1}\, \cE^{(8)}_{(1,0)}\, ,
\end{equation}
where we have chosen to extract the greatest common divisor  $k$ of the coefficients of $C^{(2)}$ and $B_{\rm NS}$ so that $\gcd(p_1,p_2)=1$.
 Note that, unlike in the case of $\calE^{(8)}_{(0,0)}$, the integral does not split into the sum of two terms even though $U_{\alpha_3}$ is block diagonal since $\calE^{(8)}_{(1,0)}$ contains the product of two Eisenstein series.
Substituting the expression \eqref{e:eEti} for $\calE_{(1,0)}^{(8)}$
(which includes a term quadratic in Eisenstein series), it is
straightforward to perform the Fourier integration with the result
\begin{eqnarray}\label{e:d4R48a3}
 \cF^{(8)\alpha_3}_{(1,0)}(kp_1,kp_2,k')&=&
  \zeta(5)\,F^{SL(3)\beta_2}_{\beta_1;\frac52}(kp_1,kp_2)\\
\nn&+&
  {2\pi^4\over135}\, F^{SL(3)\beta_2}_{\beta_1;-\frac12}(kp_1,kp_2)\,F^{SL(2)}_{2}(k').
\end{eqnarray}
The $k=0$ or $k'=0$ terms are determined by $\smallf 12$-BPS instantons
arising from the  winding  of the nine-dimensional $\smallf 12$-BPS states,
listed in appendix~\ref{sec:D9}, around the  decompactifying circle.

The $\smallf 14$-BPS part is contained in the $k\ne 0$,  $k'\ne 0$
modes of the second contribution
  in~\eqref{e:d4R48a3}.  For the physical interpretation we
 extract the greatest common divisor $\ell=\gcd(k,k')$, and set
 $k = \ell q_1, k' = \ell q_2$ with $\gcd(q_1,q_2)=1$.  Applying   \eqref{e:F2} with $s=-1/2$ and \eqref{e:F22} with $s=2$,  it  can be written as
 \begin{eqnarray}
 &&{2\over
  \pi}\,{\Omega_2^{4\over3}\over T_2^{1\over3}}\,       \sigma_{-3}(|\ell  q_1 |) \,    \sigma_{-3}(|\ell q_2|)\, {1+2\pi
  |\ell q_1 | |p_2+p_1\Omega| T_2\over
  |p_2+p_1\Omega|^3}\, {1+2\pi
  |\ell  q_2 |  \calU_2\over
  \calU_2}\nn\\
&&\qquad\qquad \times\, \exp(-2\pi |\ell q_1 |
  |p_2+p_1\Omega|T_2-2\pi  |\ell  q_2|  \calU_2)\,.
  \label{d4r4d8}
  \end{eqnarray}
   Taking the limit $r_2/\ell_9\to\infty$ and recalling that  $T_2=
 \nu_1^{-{3\over7}}\,\Omega_2^{-\frac12}\,r_2/\ell_9$ and $\calU_2 =
 r_2/r_1= \nu_1^{{4\over7}}\,r_2/\ell_9$, the leading behaviour of this expression is
\begin{equation}\label{exponfac}
 8\pi\,{\ell_9^4\over\ell_8^4}\,\sigma_{3}(|\ell  q_1 |) \,    \sigma_{3}(|\ell q_2|)\,{\exp(-2\pi\,\ell\, r_2 m_{\quart})\over(|\ell q_1|\,
  {|p_2+p_1\Omega|\over\sqrt{\Omega_2}}\, \nu_1^{-{3\over7}})^2\times(|\ell q_2|\, \nu_1^{{4\over7}})^2}\,,
\end{equation}
where  $r_2/\ell_9^7=1/\ell_8^6$ and   the  instanton action is described by the world-lines of the constituents  (in this case bound states of $F$ and $D$ strings and the KK charge)  of
$\smallf 14$-BPS bound states wound $\ell$ times around the circle
$S^1$ of radius $r_2$. The $\smallf 14$-BPS mass is given by
\begin{equation}
\label{massquart} m_{\quart}\, \ell_9= |q_1|
  {|p_2+p_1\Omega|\over\sqrt{\Omega_2}}\, \nu_1^{-{3\over7}}+ |q_2|\,
  \nu_1^{{4\over7}}\,,
\end{equation}
or in string units
\begin{equation}
m_{\quart}\, \ell_s=   |q_1|    |p_2+p_1\Omega|{r_1\over\ell_s}+|q_2| {\ell_s\over r_1}\,.
\end{equation}
Much as before,  the divisor sums in~\eqref{exponfac} encode the combinations of winding numbers and charges carried by these world-lines (although the combinatorics are here more complicated than in the $\smallf 12$-BPS case and deserve further study).

\medskip
 (ii) {\bf The maximal parabolic\footnote{Note that $Spin(2,2)$ is isomorphic to $SL(2)\times SL(2)$.} $P_{\alpha_1}=GL(1)\times
  Spin(2,2)\times U_{\alpha_1}$.}

This is relevant to the string perturbation theory limit, in which the string coupling constant, $y_8$  gets small.  The  unipotent  factor $U_{\alpha_1}$ in  \eqref{e:UnipotentN1E3} is  parametrized  by
  $(C^{(2)},\Omega_1)$.  In  this case the non-zero Fourier modes of $\calE^{(8)}_{(0,0)}$ are obtained from \eqref{e:F1} with $s=3/2$,
\bea \label{e:R4Dstring}
\cF^{(8)\alpha_1}_{(0,0)}(kp_1,kp_2)&:=& \int_{[0,1]^2}d\Omega_1dC^{(2)} e^{-2i\pi k\,(p_1 C^{(2)}+p_2\Omega_1)} \cE^{(8)}_{(0,0)}\\
&=& {8\pi\over  \sqrt{y_8}}\,  {\sigma_2(|k|)\over     |k|}\,    {\sqrt{T_2}\over    |p_2+p_1T|}K_1\left(2\pi
    |k|{|p_2+p_1T|\over\sqrt{T_2 y_8}}\right),\nn
\eea
where again $\gcd(p_1,p_2)=1$.
 Note that the second term in \eqref{e:eEsi} does not contribute since it is constant in $(C^{(2)},\Omega_1)$.
Its asymptotic form for $y_8\to0$ is given  by
\begin{equation}
\lim_{y_8\to0}  \cF^{(8)\alpha_1}_{(0,0)}(kp_1,kp_2)  \sim  {4\pi\over y_8}\,
  \sigma_{2}(|k|)\,    \left(\sqrt{T_2 y_8}\over
   |k|\, |p_2+p_1T|\right)^{3\over2} e^{-2\pi |k|\frac{|p_2+p_1T |
    }{\sqrt{T_2y_8}}} \,,
\end{equation}
where $\gcd(p_1,p_2)=1$ and the asymptotic form of the Bessel function  \eqref{besseldef} has been used in the last line in order to extract the leading instanton contribution in the perturbative limit, $y_8\to0$
 with  $T_2$   fixed~\cite{Green:2010wi} (recall  $y_8=(\Omega_2^2T_2)^{-1}$ is the square of the string coupling).  In this limit  these  non-perturbative
 effects behave as $e^{-C/\sqrt {y_8}}$, as expected of $D$-brane instantons. The $p_1=0$ and $p_2\neq0$ terms
are $D$-instanton contributions and those with $p_1\neq0$ are the wrapped $D$-string
contributions of charge $(p_1,p_2)$ that are related by
 the $SL(2,\ZZ)$ action on the $T$ modulus, which is  part  of  the
 perturbative T-duality symmetry.

The Fourier modes of $\calE_{(1,0)}^{(8)}$ are given by
\begin{eqnarray}\label{e:D4R4Dstring}
     \cF_{(1,0)}^{(8) \alpha_1} &:=&\int_{[0,1]^2}d\Omega_1dC^{(2)}\,e^{-2i\pi    k(p_1 C^{(2)}+p_2\Omega_1)} \cE^{(8)}_{(1,0)}\\
\nn &=& \f{  16\zeta(2)}{y_8^{\frac23}}\,
  {\sigma_4(|k|)\over |k|^2}\, {T_2\over |p_2+p_1T|^2}\, K_2\left(2\pi |k|{|p_2+p_1T|\over \sqrt{T_2y_8}}\right)\\
\nn&+&{16\zeta(4)E_2^{SL(2)}(\calU)\over  \pi\, y_8^{\frac16}}\,
{\sigma_2(|k|)\over |k|}\,{{|p_2+p_1T|\over \sqrt{T_2}}}\, K_1\left(2\pi|k|
    {|p_2+p_1T|\over \sqrt{T_2y_8}}\right)\,,
\end{eqnarray}
with $\gcd(p_1,p_2)=1$.   In the limit of small string coupling,
$y_8\to 0$  and recalling that  $\ell_8=\ell_s\,y_8^{1/6}$,
 the first line on the right-hand side behaves as
\begin{equation}
{\ell_s^4\over\ell_8^4}\, {8\zeta(2)\over y_8}\,
 \sigma_{4}(|k|)\,\left( \sqrt{ y_8 T_2}\over |k|\, |p_2+p_1T|\right)^{{5\over2}}\, \exp\left(-2\pi |k|{|p_2+p_1T|\over \sqrt{T_2y_8}}\right),
\label{asyme8}
\end{equation}
which is characteristic of the $\smallf 12$-BPS configuration due to  a
euclidean world-sheet of a $(p_1,p_2)$ $D$-string wrapped $k$ times around $\rT^2$.

 The second line  behaves in the small string coupling limit
   $y_8\to0$ as
\begin{equation}
{\ell_s^4\over\ell_8^4}\, \f{8\zeta(4)}{\pi} y_8\,E_2^{SL(2)}(\calU)\,
 \sigma_{-2}(|k|)\,\left( \sqrt{ y_8 T_2}\over |k|\, |p_2+p_1T|\right)^{-{1\over2}}\, \exp\left(-2\pi |k|{|p_2+p_1T|\over \sqrt{T_2y_8}}\right),
\label{asyme82}
\end{equation}
which  is suppressed  relative to~\eqref{asyme8} by  $y_8^2$ (which is  itself four powers of the string coupling).
As in the $D=9$ and $D=10$ cases, the distinction between the $\smallf
12$-BPS and $\smallf 14$-BPS cases is not seen in the argument of the
Bessel function, which determines the exponential suppression at small
$y_8$.  In other words, there are no $\smallf 14$-BPS instantons so
the second term on the right-hand side  of \eqref{e:D4R4Dstring} has
the same exponential suppression in the $y_8\to 0$ limit as the first
term. The distinction between the $\smallf 12$- and $\smallf 14$-BPS contributions in \eqref{e:D4R4Dstring} again lies in the properties of the measure rather than in the spectrum of instantons.

\medskip
 (iii) {\bf The maximal parabolic $P_{\alpha_2}=GL(1)\times
  SL(3)\times U_{\alpha_2}$}

This corresponds to the limit in which the volume of the M-theory 3-torus, $\calV_3$, gets large. The unipotent  factor  $U_{\alpha_2}$ \eqref{e:UnipotentN2E3} depends only on $\calU_1$ and the Fourier modes in this case only involve the modes of the $SL(2,\Z)$ Eisenstein series,
  \begin{equation}
    \label{e:N2R4}
 \cF^{(8)\alpha_2}_{(0,0)} :=    \int_{[0,1]}   d\calU_1\,  e^{-2i\pi\,  k   \calU_1}\,  \cE^{(8)}_{(0,0)}=4\pi
    {\sigma_{-1}(|k|)}\, e^{-2\pi |k|\calU_2}\, .
  \end{equation}
Recalling~\cite{Green:2010wi} that  $\calU_2=\cV_3/\ell_P^3$  is  the
volume of the M-theory 3-torus, we see that
these coefficients  are exponentially suppressed in $\cV_3$,  and correspond to the expected contributions from  euclidean $M2$-branes wrapped $k$ times on the 3-torus.

Furthermore, the  divisor   sum  reproduces the  one derived  from a direct partition
function   calculation in~\cite{Sugino:2001iq}.
The form of this measure factor can also be seen from a simple duality
argument using the fact that the wrapped $M2$-brane instanton is
related to the Kaluza--Klein world-line instanton by the $SL(2,\ZZ)$
part of the duality group.  This duality interchanges $T$ and $\calU$
and, hence, the factor $\exp(-2\pi |k|/\sqrt{\Omega_2\nu_2})  =
\exp(-2\pi |k|\,T_2)$ in~\eqref{e:N3R41} for $p_1=0$ and $p_2=1$  is related to $\exp(-2\pi
|k|\,\calU_2)$ in \eqref{e:N2R4}.  This explains the fact that the
measure factor, $\sigma_{-1}(|k|)$,  is the same in both these
equations.

\subsection{$D=7$: $SL(5,\ZZ)$}\hfill\break
\label{exampleseven}

{\bf Convention on $SL(d)$ labelling}:
In the following we will consider the maximal parabolic series
  $E^{SL(d)}_{\beta_i;s}$ associated with a node $\beta_1,\ldots,\b_{d-1}$ of
  the $SL(d)$ Dynkin diagram using its usual labeling.  For example, in the particular case of $SL(5)$ this labeling is shown on the righthand side of  figure~\ref{fig:A4vsE4},
  whereas the previous labeling (coming from the $E_4$ labeling in figure~\ref{fig:dynkin})
  is shown on the lefthand side.
  The
  correspondence between the two labelings is given by $\beta_1=\alpha_1$,
  $\beta_4=\alpha_2$, $\beta_2=\alpha_3$, $\beta_3=\alpha_4$.
 \begin{figure}[ht]
 \centering\includegraphics[width=12cm]{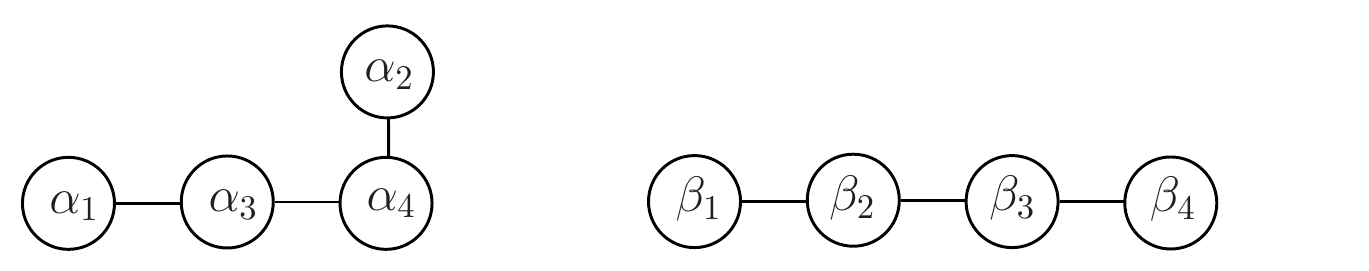}
 \caption{\label{fig:A4vsE4}  Different labelings of the $A_4$ Dynkin diagram.}
 \end{figure}

In the case  $D=7$ the coefficient functions are given in terms of Eisenstein series by\footnote{In this
  work this non-Epstein series is related to the one
  in~\cite{Green:2010wi} by
  $\bE^{SL(5)}_{\beta_3;s}=2\zeta(2s-1)\zeta(2s)\,
  E^{SL(5)}_{[0010];s}$. The $SL(d)$ nodes are labelled according the
  natural order as indicated in figure~\ref{fig:A4vsE4}.} \cite{Green:2010wi}
\begin{eqnarray}
  \cE^{(7)}_{(0,0)}=&2\zeta(3)\,E^{SL(5)}_{\beta_1;\frac32}\,,\label{esl510a}\\
\cE^{(7)}_{(1,0)} =&\!\!\lim_{\epsilon\to0}\left(\! \zeta(5+2\epsilon)E^{SL(5)}_{\beta_1;\frac52+\epsilon}+{6\zeta(4-2\epsilon)\zeta(5-2\epsilon)\over\pi^3}E^{SL(5)}_{\beta_3;\frac52-\epsilon}\right)\label{esl5501a}
\end{eqnarray}
It was shown in~\cite{Green:2010wi} that the poles of the
  individual series in the parenthesis cancel in the limit
  $\epsilon\to0$ and the resulting expression is analytic at $\epsilon=0$.
 The detailed analysis of properties of the Fourier modes of Epstein series $E_{\beta_1;s}^{SL(5)}$ in the three limits of
 interest is determined in appendix~\ref{sec:sl5modes}.  The modes of
 the non-Epstein series, $E^{SL(5)}_{\beta_{3};\frac52+\epsilon}$ in
 these three limits are obtained in appendix~\ref{nonepsteinmodes},
 making use of the representation of $E^{SL(d)}_{\beta_3;s}$ as a
 Mellin transform of the automorphic lift of a certain lattice sum (see proposition~\ref{prop:nonepsteinintegralrep} below).

\medskip
 (i) {\bf The maximal parabolic $P_{\alpha_4}=GL(1)\times
  SL(2)\times SL(3)\times U_{\alpha_4}$}

This is the decompactification limit in which $r_3/\ell_8= r^2 \to \infty$ (where $r$ is the $GL(1)$ parameter that
parameterises the approach to the cusp).  Recalling the relation  between the  volume of the 3-torus $\nu_3$  and the volume of
the 2-torus $\nu_2$~\cite{Green:2010wi},  the limit under consideration is one in which
$\nu_3 = \nu_2^{\frac{5}{6}}\, (r_3/\ell_8)^{-2} \to 0$.
The unipotent radical is abelian  and has the form
\begin{equation}
  U_{\alpha_4} \ \ = \ \ \left\{
  \begin{pmatrix}
    I_2&Q_4 \\
    0&I_3\\
  \end{pmatrix}\right\},
  \label{unialpha4}
  \end{equation}
  where $I_n$ is the rank $n$ identity matrix and $Q_4$ is the $2\times 3$ matrix defined in \eqref{Q4defa}. In the discussion of this limit in this subsection we will write the Levi component as
\begin{equation}\label{levialpha4}
    \ttwo{r^{6/5}e_2}{0}{0}{r^{-4/5}e_3}\,,
\end{equation}
where $e_2\in SL(2,\IR)$ and $e_3\in SL(3,\IR)$.

Specialising the Fourier modes of $E^{E_4}_{\alpha_1;s}=E^{SL(5)}_{\beta_1; s}$  that are given
in~\eqref{e:EpPalpha4FourierModes} to the case $s=3/2$ and using the relation  between the  $GL(1)$ parameter
and the radius of compactification,  $r^2= r_3/\ell_8$, gives the
Fourier modes of $\calE_{(0,0)}^{(7)}$ in~\eqref{esl510a}
\begin{eqnarray}\label{e:E4P4R4}
\cF^{(7)\alpha_4}_{(0,0)}(k,\tilde N_4) &:=& \int_{[0,1]^6} d^6Q_4 \,e^{-2i\pi k \tr(\tilde N_4\cdot Q_4)}\,
  \cE^{(7)}_{(0,0)}\cr
 & =& \left(r_3\over \ell_8\right)^{\frac65}\, 8\pi \,
  \sigma_0(|k|) \, K_0(2\pi|k|\, r_3\, m_{\frac12}),
\end{eqnarray}
where $\gcd (\tilde N_4)=1$ and the support of the non-vanishing Fourier coefficients is  equal to  the rank 1  integer-valued matrices  $k\tilde N_4$   in
$M_{3,2}(\ZZ)$; these have the form  $kn^tm$ with  $n=(n_i) \in
\ZZ^3$  and $m=(m_a)\in\ZZ^2$  row vectors satisfying $\gcd(n_1,n_2,n_3)=\gcd(m_1,m_2)=1$.  This factorization is unique up to signs of the three factors.     The matrix $\tilde N_4=  n^tm$  satisfies the relation
\begin{equation}
  \sum_{a,b=1}^2\,  \epsilon_{ab}
(\tilde N_4)_i{}^a (\tilde N_4)_j{}^b=0,\qquad \forall i,j=1,2,3\,,
\end{equation}
with $\epsilon_{12}=-\epsilon_{21}=1$ and $\epsilon_{11}=\epsilon_{22}=0$, which is precisely
$\smallf 12$-BPS condition discussed in appendix~\ref{sec:D8}.
The argument  of the Bessel  function in~\eqref{e:E4P4R4} is
proportional to the mass of $\smallf 12$-BPS states, where
\begin{equation}\label{e:mhalf8}
  m_{\frac12}\,\ell_8  \ \ :=  \ \  \|me_2\|\times\|n(e_3^t)^{-1}\|\,.
\end{equation}
This expression does not depend on the factorization $N_4=k\tilde N_4=kn^tm$,
and transforms  covariantly
  under the $SL(2)$ and $SL(3)$ factors of the Levi component. This
 is the mass of a  $\smallf 12$-BPS bound state of  fundamental
strings and $D$-strings with
Kaluza--Klein momentum.  This
expression is covariant under the action of the
symmetry group
$SL(2)\times SL(3)$ of the Levi factor.
In the limit $r_3/\ell_8\to\infty$ the expression for the
  Fourier modes $\cF^{(7)\alpha_4}_{(0,0)}$ takes the form
\begin{equation}\label{e:E4P4R4Asymp}
\cF^{(7)\alpha_4}_{(0,0)}(k,\tilde N_4)= \left(r_3\over \ell_8\right)^{6\over5}\, 4\pi \,
  \sigma_0(|k|) \,{e^{-2\pi |k|\, r_3\, m_{\frac12}}
   \over \sqrt{|k| r_3\, m_{\frac12}}}\, (1+O(\ell_8/r_3)) \, ,
\end{equation}
where $\ell_8/r_3$  is the inverse square of
  the $GL(1)$ parameter  (see (\ref{e:EpPalpha4FourierModes})).  The exponent is proportional to $r_3m_{1/2}$ with
  $r_3\to\infty$ and $m_{1/2}$ fixed,
which is in accord with the behaviour described in~\eqref{e:AsympBPS} with $D=7$.

The Fourier modes of $\calE_{(1,0)}^{(7)}$ in~\eqref{esl5501a} in this
parabolic subgroup are defined as
\begin{equation}
  \label{e:F7D4R4}
  \cF^{(7)\alpha_4}_{(1,0)}(k,\tilde N_4) \ \ := \ \ \int_{[0,1]^6} d^6Q_4 \,e^{-2i\pi k \tr(\tilde N_4\cdot Q_4)}\,\cE^{(7)}_{(1,0)}\,,
\end{equation}
with $\gcd(\tilde N_4)=1$.  An expression for these Fourier modes is obtained by adding
\eqref{e:EpPalpha4FourierModes} for the  Epstein series
$E_{\b_1;s}^{SL(5)}$ to the  modes of the non-Epstein  series $E_{\b_3;s}^{SL(5)}$  with the  correct proportionality constants
and setting $s=5/2$.  Since each has a constant residue at $s=5/2$ we can directly use the formulas for the  nonzero Fourier modes derived in appendix~\ref{modesdetails}.

 The Fourier modes of
   $E^{SL(5)}_{\beta_3;s}$  are computed via
   appendix~\ref{nonepsteinmodes}  using the following proposition,
   which represents this series as
   the  Mellin transform of the lattice sum
   \begin{equation}
    {\mathcal G}(\tau,X) \ \ := \ \ \sum_{ [\srel{m}{n}]\, \in \, {\mathcal M}_{2,d}^{(2)}(\Z)}
    e^{-\pi \tau_2^{-1} (m+n\tau)X(m+n\bar\tau)^t}\,.
\end{equation}
Here as  in the  usual physics notation $\tau=\tau_1+ i\tau_2\in\U$ and   $X=G+B$, with $G$  a positive definite symmetric $d\times d$ matrix and  $B$  an antisymmetric $d\times d$ matrix;
${\mathcal M}_{2,d}^{(i)}$ represents $2\times d$ matrices of rank
$i$.
This contribution is the rank 2 part of the lattice sum
  $\Gamma_{(d,d)}$ for even self-dual Lorentzian lattices.
 The properties of this sum are studied in appendix~\ref{sec:cG}, and the proof of the proposition given at the end of appendix~\ref{sec:E3}.
\begin{prop}\label{prop:nonepsteinintegralrep}
For $\Re{s}$ large (and consequently for all $s\in \C$ by meromorphic continuation)
\begin{equation}\label{I04}
\aligned
     \int_0^\infty I(0,uG) \, u^{2s-1}\,du \ \ & = \ \ \frac12\,\xi(2s)\xi(2s-1)
E^{SL(d)}_{\beta_2;s}(e)\\ &
= \ \ \frac12\,\xi(d-2s)\xi(d-2s-1) E^{SL(d)}_{\b_{d-2};\f{d}{2}-s}(e)  \, ,
\endaligned
\end{equation}
where the function $I(s,X)$ is defined as
\begin{equation}\label{Isdeff}
    I(s,X) \ \ := \ \ \int_{SL(2,\Z)\backslash \U}E^{SL(2)}_s(\tau)\,{\mathcal G}(\tau,X) \,\f{d^2\tau}{\tau_2^2}\,.
\end{equation}
\end{prop}
The equality of the two formulas on the righthand side of (\ref{I04}) represents a well-known functional equation of Eisenstein series.  There is an additional functional equation between these two Eisenstein series coming from the diagram automorphism:
\begin{equation}\label{diagramautfe}
    E_{\b_2;s}^{SL(d)}(e) \ \ = \ \   E_{\b_{d-2};s}^{SL(d)}(w_d(e^t)^{-1}w_d)\,,
\end{equation}
where $w_d$ is formed from the $d\times d$ identity matrix by reversing its columns.  Unlike the functional equation in (\ref{I04}), the functional equation (\ref{diagramautfe}) alters the group variable $e\in SL(d,\IR)$, and consequently relates Fourier coefficients of these series in different parabolics.

The formulas for Fourier coefficients of $E_{\b_2;s}^{SL(5)}$ in appendix~\ref{nonepsteinmodes} can be adapted to $E_{\b_3;s}^{SL(5)}$ using either functional equation, resulting in different (yet of course equivalent) formulas.
Using (\ref{diagramautfe}) and (\ref{levialpha4}) gives the identity
\begin{equation}\label{useofcontragredient2}
    F_{\b_3;s}^{SL(5)\,\b_3}(r^{6/5}e_2,r^{-4/5}e_3;N_4)    \ \ = \ \ F_{\b_2;s}^{SL(5)\,\b_2}(r^{4/5}\tilde{e}_3,r^{-6/5}\tilde{e}_2;-w_2N_4^tw_3)\, ,
\end{equation}
where $N_4\in M_{3,2}(\Z)$.
Here we have used   the ``contragredient'' notation $\widetilde{e}$ to represent $w_d(e^t)^{-1}w_d$ (see
 (\ref{Soddpf2})), and the relation
  \begin{equation}\label{contragredientofe2}
    \widetilde{e} \ \ = \ \ \ttwo{I_3}{-w_3Q^tw_2}{}{I_2}\ttwo{r^{4/5}\tilde{e}_3}{}{}{r^{-6/5}\tilde{e}_2}
\end{equation}
for   $e=\ttwo{I_2}{Q}{}{I_3}\ttwo{r^{6/5}e_2}{}{}{r^{-4/5}e_3}$.

 Applying    (\ref{gl5node2h22punchline})
we arrive at the formula
\begin{multline}\label{gl5node2h22punchlineconverted}
    F_{\b_3;s}^{SL(5)\,\b_3}(r^{6/5}e_2 ,r^{-4/5}e_3 ;N_4) \ \ = \\
    \f{8\,r^{4+4s/5}}{\xi(2s)\xi(2s-1)} \int_{\IR}
 \sum_{[\srel pq ] \, \in \, SL(2,\Z)\backslash {\mathcal M}_{2,3}^{(2)}(\Z)}
 \sum_{\srel{\hat{m},\hat{n}\,\in\,\Z^2}{\hat{m}p-\hat{n}q=-w_2N_4^tw_3}}
    \left(\|(p+q\tau_1) \tilde e_3\|\over \|\tilde e_2^{-1}\hat{m}\|  \right)^{1/2-s} \ \times \\
 \left( \|q \tilde e_3\|\over \|\tilde e_2^{-1}(\hat{n}+\hat{m}\tau_1)\|\right)^{3/2-s}
K_{s-1/2}(2\pi r^{2} \|(p+q\tau_1) \tilde e_3\|\|\tilde e_2^{-1}\hat{m}\|) \ \times \\ K_{s-3/2}(2\pi r^{2} \|q \tilde e_3\|\|\tilde e_2^{-1}(\hat{n}+\hat{m}\tau_1)\|) \, d\tau_1 \\
 + \ \
  \f{2\, \G\left(s-\frac12\right)}{\xi(2s)\xi(2s-1)} r^{1+14s/5} \sum_{\srel{p\,\neq\,0}{\srel{n\,\neq
        \,0}{\srel{\hat{m}\,\perp\,n}{\hat{m}p= -w_2N_4^tw_3}}}}
 \left(\| \tilde e_2^{-1}\hat{m}  \| \over
    \pi \|n\tilde e_2 \|^2 \|  p\tilde e_3 \|\right)^{s-1/2} \ \times \\
     K_{s-1/2}(2\pi r^{2}\|\tilde e_2^{-1}\hat{m} \| \| p\tilde e_3\| )
\end{multline}
(here $\hat{m}\in\Z^2$ is thought of as a column vector and $p\in\Z^3$ as a row vector).

Returning to (\ref{esl5501a}), we factor $N_4=k\tilde N_4$, where $k=\gcd(N_4)$,
and furthermore factor  $\tilde N_4^t$  as $\tilde N_4^t=\hat m' p'$, where $\gcd(\hat m')=\gcd(p')=1$.  This factorization is unique up to multiplication by $\pm 1$.    Fixing such a factorization, the solutions to the equation $\hat{m}p=-k\tilde N_4^t$ have the form $\hat{m}=\pm d \hat m'$ and $p=\mp \smallf kd p'$ for positive divisors $d$ of $k$.  We now group   the coefficient of the $\partial^4\R^4$ interaction as  the sum of  two  contributions
   \begin{equation}
     \label{eq:FFF}
       \cF^{(7)\alpha_4}_{(1,0)}(k,\tilde N_4) \ \ = \ \
       \cF^{(7)\alpha_4}_{(1,0)\,I}(k,\tilde N_4) \ + \   \cF^{(7)\alpha_4}_{(1,0)\,II}(k,\tilde N_4)\,,
   \end{equation}
where $ \cF^{(7)\alpha_4}_{(1,0)\,I}(k,\tilde N_4)$  comes from applying (\ref{e:EpPalpha4FourierModes}) to the first term in (\ref{esl5501a}),  and from the last line of~\eqref{gl5node2h22punchlineconverted}; it is supported on
rank one integer valued matrices $\tilde N_4$  (i.e., it contains  the $\smallf 12$-BPS configurations). The second contribution
   $ \cF^{(7)\alpha_4}_{(1,0)\,II}(k,\tilde N_4)$
   comes from the first term of  \eqref{gl5node2h22punchlineconverted} and  contains the $\smallf
   14$-BPS contributions.
    Using (\ref{e:mhalf8}) (with the current notation where $\tilde N_4^t=\hat m' p'$) explicit formulas for these are given as
\begin{multline}\label{new12bpsforsl5limiti}
 \cF^{(7)\alpha_4}_{(1,0)\,I}(k,\tilde N_4) \ \  = \ \ 8\pi^2\,  r_3 \,
  {\sigma_2(|k|)\over 3\,|k|}\,{m_{\frac12}\over  \|m'e_2\|^2}\,K_1(2\pi|k|\, r_3\,
  m_{\frac12}) \\
 +  \  \f{32}{\pi}\,\f{\sigma_4(|k|)}{k^2}\, \f{(r_3/\ell_8)^2\,(r_3 m_{\f12})^2}{ \|p'(e_3^t)^{-1}\|^4}K_2(2\pi|k|r_3m_{\frac12})\sum_{\srel{n\,\neq\,0}{n\tilde N_4^t=0}}\|n(e_2^t)^{-1}\|^{-4}
  \,.
\end{multline}

The remaining contribution to (\ref{eq:FFF}) is given by the formula
\begin{multline}\label{e:E4P4D4R4b}
\cF^{(7)\alpha_4}_{(1,0)\,II}(k,  \tilde N_4) \ \ =  \\   64\pi^4 r^6
   \int_{\IR}
 \sum_{\srel{[\srel pq ] \, \in \, SL(2,\Z)\backslash {\mathcal M}_{2,3}^{(2)}(\Z)}{\srel{\hat{m},\hat{n}\,\in\,\Z^2}{\hat{m}p-\hat{n}q=-kw_2\tilde N_4^tw_3}}} \left(\|\tilde e_2^{-1}\hat{m}\| \over \|(p+q\tau_1) \tilde e_3\|  \right)^{2} \ \times \\
{  \|\tilde e_2^{-1}(\hat{n}+\hat{m}\tau_1)\| \over  \|q \tilde e_3\| }
K_{2}(2\pi r^{2} \|(p+q\tau_1) \tilde e_3\|\|\tilde e_2^{-1}\hat{m}\|) \ \times \\ K_{1}(2\pi r^{2} \|q \tilde e_3\|\|\tilde e_2^{-1}(\hat{n}+\hat{m}\tau_1)\|) \, d\tau_1\,.
\end{multline}
 We have not succeeded in simplifying the $\tau_1$ integral in this expression and therefore the interpretation of the asymptotic behaviour  as $r_3/\ell_8\rightarrow \infty$   remains obscure.

\medskip
(ii) {\bf The maximal parabolic $P_{\alpha_1}=GL(1)\times
  SL(4)\times U_{\alpha_1}$}

The  instanton contributions to $\calE_{(0,0)}^{(7)}$  in  the  perturbative string limit associated with
$L_{\alpha_1}=GL(1)\times  SL(4)$ are  given  by~\eqref{e:EpPalplha1FourierModes} upon setting
$s=3/2$.  The  relation between the  $GL(1)$ parameter  and the  string coupling
constant   in  7  dimensions is  $r^{-2}=y_7^{\half}$  and   the  relation
between  the  7 dimension  Planck  length  and  the string  length is
$\ell_7=\ell_s \, y_7^{1/5}$~(cf.~(\ref{plancks})).
In this case the unipotent radical is abelian and  has the form
\begin{equation}
  U_{\alpha_1} \ \ = \ \
  \begin{pmatrix}
    I_4&Q_1 \\
    0&1\\
  \end{pmatrix},
  \label{Q1def}
  \end{equation}
  where $Q_1$ is a $SL(4)$ spinor  defined in \eqref{Q1defa}.

 This leads to the expression for the Fourier modes
\begin{eqnarray}
\nn \cF^{(7)\,\alpha_1}_{(0,0)}(k,\tilde N_1)&:=&  \int_{[0,1]^4} d^4Q_1   \,   e^{-2\pi i k\,  \tilde N_1    Q_1}
\, \cE^{(7)}_{(0,0)}\\
&=& {8\pi\over  y_7^{7\over10}}\, {\sigma_2(|k|)\over
    |k|}\,  {K_1\left(2\pi|k|\, \lVert \tilde N_1e_4\rVert\over  \sqrt{y_7}\right)\over \lVert \tilde N_1e_4\rVert}\,,
    \label{sevenzero}
\end{eqnarray}
where
$\tilde N_1
\neq 0 $ is a row vector in $\ZZ^4$ such that $\gcd(\tilde N_1)=1$.
In  the limit $y_7\to 0$ the right hand side of~\eqref{sevenzero} has the exponential
 suppression characteristic of an instanton contribution and contributes
  \begin{equation}
\ell_7\, \cF^{(7)\,\alpha_1}_{(0,0)}(k,\tilde N_1) \sim \ell_s\, {4\pi\over y_7} \, \sigma_2(|k|)\, \left(\sqrt{y_7}\over |k|\,\lVert \tilde N_1e_4\rVert
\right)^{{3\over2}}\, \exp\left(-2\pi|k|{\lVert \tilde N_1e_4\rVert\over \sqrt{y_7}}\right)\,
  \end{equation}
 to the effective $\R^4$ action with $D=7$ in ~\eqref{effacts}.

Terms with  $\tilde N_1=[1\,0\,0\,0]$ are $D$-instanton contributions.   Terms with  $\tilde N_1\neq
[1\,0\,0\,0]$ are $\smallf 12$-BPS contributions  due to wrapped Euclidean
bound states of fundamental and  $D$-strings.  The rank  4 integer vector
$k\tilde N_1$ is unrestricted, other than being nonzero.

The Fourier modes of $\cE^{(7)}_{(1,0)}$ can be computed in terms of the individual Eisenstein series it is expressed from in  (\ref{esl5501a}).  The modes of $E_{\beta_1;5/2}^{SL(5)}$  are given in~\eqref{e:EpPalplha1FourierModes}, while the modes of $E_{\beta_3;5/2}^{SL(5)}$ can be determined from those of $E_{\beta_2;5/2}^{SL(5)}$  in (\ref{nonepsinP4d}) using the contragredient mechanism described in (\ref{useofcontragredient2}-\ref{contragredientofe2}).
This results in the expression
\begin{eqnarray}
\nn \cF^{(7)\alpha_1}_{(1,0)}(k,\tilde N_1)&:=& \int_{[0,1]^4} d^4Q_1   \,   e^{-2\pi i k\,  \tilde N_1^t\cdot   Q_1}
\, \cE^{(7)}_{(1,0)}\\
&=& {8\pi^2\over 3\, y_7}\, {\sigma_4(|k|)\over
    |k|^2}\,  {1\over \lVert \tilde N_1e_4\rVert^2}\,K_2\left(2\pi|k|\,\lVert \tilde N_1e_4\rVert\over  \sqrt{ y_7}\right)
   \label{onezeroseven} \\
   &+& \f{16}{\pi\sqrt{y_7} }  \, \times \sum_{\srel{\srel{p\,>\,0}{n\,\neq\,0}}{\srel{p\hat{m}\,=\,-kw_4\tilde N_1^t }{n\,\perp\,w_4\tilde N_1^t}}}
  \f{  \|\tilde e_4^{-1}\hat{m}\| }{  p   \, \|n\tilde e_4\|^4} \,   K_1\left(2\pi|k|\,\lVert \tilde N_1e_4\rVert\over  \sqrt{ y_7}\right)\,, \nn
\end{eqnarray}
where again $\tilde N_1\in\ZZ^4\bsz$ such that $\gcd(\tilde N_1)=1$.
   Since all factorizations $p\hat m = -kw_4\tilde N_1^t$ with $p>0$ have the form  $\hat m = -\smallf kp w_4\tilde N_1^t$ for some divisor $p$ of $k$, the second term on the righthand side can be rewritten as
 \begin{equation}\label{onezerosevenpart2}
    \f{16}{\pi\sqrt{y_7} }  \,
   \| \tilde N_1 e_4  \|   \,  |k|\,\sigma_{-2}(|k|) \, K_1\left(2\pi|k|\,\lVert \tilde N_1e_4\rVert\over  \sqrt{ y_7}\right)\,\sum_{\srel{n\,\neq\,0}{n\,\perp\,w_4\tilde N_1^t}}  \|n\tilde e_4\|^{-4}\,.
 \end{equation}

The two contributions to the Fourier modes have the same
  support (i.e., in both cases the charges are labelled by the matrix $\tilde N_1$) because there  are no $\smallf 14$-BPS instantons in the expansion at node $\alpha_1$ (see section~\ref{spinorbit}). The different BPS nature of
  each contribution must be encoded in the  factor multiplying the Bessel functions.

  \medskip
 (iii) {\bf The maximal parabolic $P_{\alpha_2}=GL(1)\times
  SL(4)\times U_{\alpha_2}$}

  Although we do not work out the details here, explicit expressions for ${\mathcal F}_{(0,0)}^{(7)\,\a_2}$ and  ${\mathcal F}_{(1,0)}^{(7)\,\a_2}$ can be calculated using the expressions for the Fourier coefficients of $E_{\b_1;s}^{SL(5)}$ and $E_{\b_3;s}^{SL(5)}$ given in appendix~\ref{nonepsteinmodes}.

\subsection{$D=6$: $Spin(5,5, \ZZ)$}\hfill\break
\label{so55case}

The coefficient functions in this case are given by combinations of Eisenstein series \cite{Green:2010kv},
\be
\cE_{(0,0)}^{(6)}=2\zeta(3)\, E^{Spin(5,5)}_{\alpha_1;\frac32}\,,
\label{eisen006}
\ee
and
\be
 \cE_{(1,0)}^{(6)}=\lim_{\epsilon\to0}\,\left(\zeta(5+2\epsilon)\, E^{Spin(5,5)}_{\alpha_1;\frac52+\epsilon}+\frac{8\zeta(6 -2\e)}{45}\,E^{Spin(5,5)}_{\alpha_5;3-\epsilon}\right)\,.
\label{eisen016}
\ee
It was shown in~\cite{Green:2010kv} that the pole of the individual series in the parentheses cancel in the limit $\epsilon\to0$ and the resulting expression is analytic at $\epsilon=0$.
Whereas the previous cases involved $SL(n)$ Eisenstein series, which
could be expressed as lattice sums that were easy to manipulate, there
is much less understanding of the $Spin(5,5)$ series in terms of such
explicit lattice sums.  Various properties of  $E_{\alpha_1;s}^{Spin(5,5)}$  were considered in
\cite{Green:2010wi} (where the series was denoted $(2\zeta(2s))^{-1}\, \bE^{Spin(5,5)}_{[10000];s}$),
based on the integral representation contained in the following proposition.  We give a rigorous proof of it through proposition~\ref{prop:Ddintegralrepn} (from which it immediately follows via proposition~\ref{prop:nonepsteinintegralrep}).
\begin{prop}\label{prop:SldtoSOd}
For $\Re{s}$ large (and consequently for all $s\in \C$ by meromorphic continuation)
\begin{multline}\label{469}
\frac14 \xi(2s)\xi(2s-1) E^{SL(d)}_{\beta_2;s}(e)\\
= \int_0^\infty
\Big(   u^{-d/2}
  \xi(d-2)\,E^{Spin(d,d)}_{\a_1;d/2-1}\ttwo{u^{1/2}e}{}{}{u^{-1/2}\tilde{e}}
+u^{-1}\, \xi(d-2)\,E^{SL(d)}_{\beta_{d-1};d/2-1}(e) \\
 +   \xi(2)
\Big) u^{2s-1}\, du\,.
\end{multline}
\end{prop}
The convergence of this integral is not {\it a priori} obvious and is explained in appendix~\ref{sec:unfolding} (cf.~its concluding remark).
 Proposition~\ref{prop:SldtoSOd} relates $E_{\b_2;s}^{SL(d)}(e)$ to a Mellin transform of
   $E_{\alpha_1;s}^{Spin(d,d)}$;  note that the last two terms in (\ref{469}) are not present in
  \cite{Green:2010wi,Obers:1999um,Angelantonj:2011br}.
This integral representation will be used  in  appendix~\ref{sec:dfive} to obtain the Fourier modes of
$E_{\alpha_1;s}^{Spin(d,d)}$.   This is
sufficient to discuss the Fourier modes of the coefficient $\calE_{(0,0)}^{(6)}$, but
$\calE_{(1,0)}^{(6)}$ also involves the series
$E_{\alpha_5;s}^{Spin(5,5)}$.  The evaluation of its Fourier modes
appears to be much more complicated and  will not be performed in this
paper.  However, we will be able to determine its orbit content as
will be discussed later.

(i) {\bf The maximal parabolic $P_{\alpha_5}=GL(1)\times
  SL(5)\times U_{\alpha_5}$ }

This parabolic subgroup has Levi factor $L_{\alpha_5} = GL(1)\times SL(5)$ (recalling from figure~\ref{fig:dynkin} that in our conventions  $\alpha_5$ is a
spinor node of $E_5= Spin(5,5)$).    Here we will   evaluate the Fourier modes using     methods similar to those  used in
  computing the  constant term of the
  series $E_{\alpha_1;s}^{Spin(d,d)}$  in~\cite[appendix~C]{Green:2010wi}.
The Fourier modes are defined as
 \begin{equation}
F^{Spin(5,5)\alpha_5}_{\alpha_1;s}(N_2) \ \ := \ \ \int_{[0,1]^{10}} dQ_2\, e^{-\pi i
  \tr(N_2  Q_2)}\, E^{Spin(5,5)}_{\alpha_1;s}
\end{equation}
where $Q_2$ is a $5\times 5$  antisymmetric matrix parameterizing
the  abelian unipotent radical $U_{\alpha_5}$, and $N_2$ is
an  antisymmetric $5\times 5$ matrix with integer entries.


We find that the Fourier modes of the series
$E^{Spin(5,5)}_{\alpha_1;s}$  are localized on the rank 1  contributions
where $N_2$ satisfies the constraints
\begin{equation}
 \sum_{i,j,k,l=1}^5 \epsilon^{ijklm} (N_2)_{ij}(N_2)_{kl}=0,\qquad \forall
 1\leq m\leq 5  \,,
\end{equation}
where $\epsilon^{ijklm}$ is the totally antisymmetric symbol with
$\epsilon^{12345}=1$.  This constraint is the $\smallf 12$-BPS condition
discussed in appendix~\ref{sec:D7}.
 This condition can be solved as
\begin{equation}\label{472}
  N_2=n^t m - m^t n; \qquad m, n\in \ZZ^5 -\{[0\,0\,0\,0\,0]\}\,.
\end{equation}
 In this case $e^{-i\pi\tr(N_2  Q_2)}=e^{-2\pi i mQ_2n^t}$.

The Fourier modes of $\cF_{\alpha_1;s}^{Spin(5,5)\,
  \alpha_5}$ are computed in (\ref{D5Epsinspinb}) using   the method of orbits for the $SL(2)$
action on $\tau$.
That formula simplifies for the special value of $s=3/2$ to
\begin{equation}\label{new473}
   \cF_{(0,0)}^{(6) \alpha_5}(N_2) \ \ = \ \ \f{1}{2\,\xi(3)}  \,\(\f{r_4}{\ell_7}\)^{5/2}\,  \sum_{\srel{ [\srel{m'}{n'}]\, \in \, GL(2,\Z)\backslash {\mathcal M}_{2,5}^{(2)}(\Z)'}{N_2\,=\,k((n')^t m'-(m')^t n')}}  \sigma_1(k)  \,\f{
   e^{-2\pi  \,k\, r_4\,m_{\frac12}}}{k\, r_4\,m_{\frac12}}\,,
\end{equation}
where
\begin{equation}\label{mhalfell7}
  m_{\frac12}^2\,\ell_7^2 \ \ := \ \ \det([\srel{m'}{n'}]G_5 [\srel{m'}{n'}]^t) \ \ = \ \ \|m'e_5\|^2\,\|n'e_5\|^2 \ - \ (m'e_5\cdot n'e_5)^2   \,,
\end{equation}
$k=\gcd(N_2)$, $G_5=e_5e_5^t$, and ${\mathcal M}_{2,5}^{(2)}(\Z)'$ represents all possible bottom two rows of matrices in $SL(5,\Z)$
(see (\ref{D5Epsinspine})).
The expression in~\eqref{new473} reproduces the
asymptotic (actually exact  in this case) behaviour for $\smallf 12$-BPS
contribution in~\eqref{e:AsympBPS} with $D=6$.

The Eisenstein series $E^{Spin(5,5)}_{\alpha_1;s}$
has a
single pole at $s=5/2$  with  residue proportional to  the $s=3/2$ series
$E^{Spin(5,5)}_{\alpha_1;3/2}$ discussed above. The finite part of the  $E^{Spin(5,5)}_{\alpha_1;s}$ series at $s=5/2$ only
receives $\smallf 12$-BPS contributions (see the comment following (\ref{D5Epsinspina})).
The complete coefficient $\cE^{(6)}_{(1,0)}$, defined
in~\eqref{eisen016}, also gets a $\smallf 14$-BPS contribution from
$E^{Spin(5,5)}_{\alpha_5;s}$, which has a pole  at $s=3$ such that the
resulting combination in~\eqref{eisen016} is analytic as shown in~\cite{Green:2010kv}.

\medskip
(ii) {\bf The maximal parabolic  $P_{\alpha_1}=GL(1)\times
Spin(4,4)\times U_{\alpha_1}$ }

In this parabolic subgroup the Levi factor is $L_{\alpha_1} =GL(1)\times Spin(4,4)$.
  The elements of the unipotent  radical are parametrized by the   $4\times 2$ matrix
  \begin{equation}
    \label{e:Q1D5}
    Q_1=
    \begin{pmatrix}
      Q_{1I}& Q^I_{2}
    \end{pmatrix},
\qquad \forall 1\leq I\leq 4\,,
  \end{equation}
  where    $Q_1=(u_1, u_2, u_3, u_4)$ and $Q_2=(u_8, u_7,
u_6, u_5)$  using the variables parameterizing the unipotent radical
in~\eqref{Ua1param} in appendix~\ref{sec:dfive}.
In the type~IIA string theory description this matrix is  parametrized by the four euclidean $D0$-branes wrapped on 1-cycles  and four euclidean $D2$-branes wrapped on 3-cycles of $\rT^4$.

The Fourier modes of  $E_{\alpha_1;s}^{Spin(5,5)}$ are defined as
\begin{equation}\label{e:477}
F^{Spin(5,5)\alpha_1}_{\alpha_1;s}(N_1):= \int_{[0,1]^8}\, d^8Q_1
e^{-2i\pi\tr(N_1 Q_1)}\, E_{\alpha_1;s}^{Spin(5,5)}\,.
\end{equation}
We will write the $2\times 4$ matrix  $N_1$ as
\begin{equation}\label{e:N11}
N_1 \ \ := \ \ [\srel MN]\,,
\end{equation}
where the row vectors have components $M=[m^1 \, m^2\, m^3 \, m^4]$ and $N=[n_1\, n_2\, n_3\, n_4]$.
The $ m^I$ ($I=1,2,3,4$) integers associated with the  windings of
the one-dimensional euclidean world-volume of a $D0$-brane on the four cycles of the 4-torus,
and $n_I$ are associated with the four distinct windings of the three dimensional euclidean
world-volume of a $D2$-brane on a  4-torus\footnote{As in the earlier cases each integer should be interpreted as a product of a $D$-particle charge and its  world volume winding number.}. This means, for example, that
on a square 4-torus with radii $R_I$ the action of a euclidean $D0$-brane
is $\sum_{I=1}^4  m^I R_I/(\ell_sg_s)$ while the action of a
$D2$-brane is $V_4\,\sum_{I=1}^4  n_I \ell_s/(R_I g_s)$, where
$V_4= R_1R_2R_3R_4/\ell_s^4$.  Because of space considerations we will omit the analysis of the case when $N=[0\,0\,0\,0]$, and instead indicate how the calculations can be performed in appendix~\ref{sec:dfive}.
More generally, the various configurations of $(D0,D2)$ states can be
classified by introducing the vector $(p_L,p_R)$ in the even self-dual
Lorentzian lattice $\Gamma_{(4,4)}$,
\begin{equation}\label{PLRdefgeneral}
    \aligned
 \sqrt2\, p_{L} \ \ & = \ \ (M+N(B - G_4))(e_4^t)^{-1}\\
\sqrt2\,   p_{R} \ \ & = \ \ (M+N(B + G_4))(e_4^t)^{-1}\,,
    \endaligned
\end{equation} where
$G_4=e_4e_4^t$ is the metric on the torus and $B$ is an antisymmetric $4\times 4$ matrix.
 Introducing $y_6=g_s^2/V_4$ the $GL(1)$ parameter is $r^2=y_6^{-
   \half}$ according to~\eqref{notation}. We remark that    the lattice is even because
$p_L^2-p_R^2=2 \sum_{I=1}^4 m_I n^I\in2\ZZ$. In terms  of the modes matrix $N_1$
in~\eqref{e:N11} this is expressed as $p_L^2-p_R^2=\tr(N_1 J N_1^t)$
where $J=
\begin{pmatrix}
  0&1\\1&0
\end{pmatrix}$.
By triality the $SO(4,4)$ vector
$(p_L,p_R)$ is equivalent to a $SO(4,4)$ chiral spinor used for the
orbit classification in section~\ref{spinorbit}.

The Fourier modes are derived in \eqref{D5epsineps6} using
  the $\theta$-lift representation of the $Spin(5,5)$ Eisenstein series,  yielding
\begin{eqnarray}
 F_{\alpha_1;s}^{Spin(5,5)\,
    \alpha_1}(N_1)&=&{1\over
   \xi(2s)y_6^{ 1/2}}
    \sum_{p\,|\,\gcd(N_1)}\!\!\int_0^\infty d\tau_2\, e^{-\pi
    {p^2\over\tau_2y_6}-\pi \frac{\tau_2}{p^2}\,(p_L^2+p_R^2)}\cr
&\times&\int_{-\frac12}^{\frac12}d\tau_1\,
  E^{SL(2)}_{s-\frac32}(\tau) \, e^{i\pi\tau_1 \, \frac{(p_L^2-p_R^2)}{p^2}}
\,,
  \label{so55modes}
\end{eqnarray}
 where we used that $\gcd(N_1)=\gcd(m^1,\ldots, m^4, n_1,\dots, n_4)$.
It is significant that setting $s=3/2$ and using $E^{SL(2)}_0(\tau)=1$,  the  integration over
$\tau_1$ projects onto the  condition $p_L^2-p_R^2=0$ which is the
pure spinor condition for $SO(4,4)$.  Using  the triality relation
between vector and spinor representation of $SO(4,4)$ this condition
is the  $\smallf 12$-BPS (pure spinor)
condition $S\cdot S=0$ discussed in section~\ref{spinorbit}.
It is then straightforward to compute the integrals in \eqref{so55modes} to evaluate the Fourier modes of the
coefficient function  $\cE_{(0,0)}^{(6)}$, giving
\begin{equation}\label{e:F00alpha16d}
\cF_{(0,0)}^{(6)\, \alpha_1}( N_1) \ \ = \ \  \f{4\sqrt{2}\,\pi\,\sigma_2(\gcd(N_1))}{y_6\sqrt{p_L^2}}\,K_1(2\pi y_6^{-1/2}\sqrt{2p_L^2})
 \, \delta_{p_L^2=p_R^2}\,,
\end{equation}
where the Kronecker $\d$-function localizes the contributions to $\smallf 12$-BPS pure
spinor  locus $p_L^2=p_R^2$ (specified by the condition
$\tr(N_1JN_1^t)=0$ on the mode matrix $N_1$).
As expected, the argument of the Bessel function is proportional to
$1/\sqrt{y_6}$, the inverse of the string coupling with
$D=6$, so its asymptotic expansion is that expected from the
contribution of $\smallf 12$-BPS states from wrapped
  D-brane on the 4-torus $\rT^4$.
The asymptotic form for $y_6\to\infty$ in the weak coupling
  regime is given by
\begin{equation}
\ell_6^2\,
\cF_{(0,0)}^{(6)\, \alpha_1}( N_1) \ \ \sim \ \
 {4\pi\ell_s^2\over y_6}\,
 \sigma_2(\gcd(N_1))\,{e^{-2\pi\,
  {\sqrt{2p_L^2}\over\sqrt{y_6}}}\over
   (\sqrt{2p_L^2}\, y_6^{-\frac12})^{3\over2}} \, \delta_{p_L^2=p_R^2}\,,
\end{equation}
where we made use of the relation between the Planck length in six
dimensions and the string scale $\ell_6=\ell_s\, y_6^{-\frac14}$.

When $s\neq 3/2$ the  $\tau_1$ integral in \eqref{so55modes} does not impose the restriction
$p_L^2-p_R^2=0$ and so the solution fills a generic $Spin(4,4)$ orbit
and  is $\smallf 14$-BPS.
Although the function $\cE_{(1,0)}^{(6)}$ in  \eqref{eisen016} is a linear combination of the vector Eisenstein series,  $E^{Spin(5,5)}_{\alpha_1;5/2}$, and the spinor series, $E^{Spin(5,5)}_{\alpha_5;3}$, at present we know little about the explicit structure of the latter, so we will only discuss the former here.    However, in this parabolic  the $\smallf 14$-BPS content of  $\cE_{(1,0)}^{(6)}$ is entirely contained in $E^{Spin(5,5)}_{\alpha_1;5/2}$.\footnote{The fact that the spinor series $E^{Spin(5,5)}_{\alpha_5;3}$ contains only the $\smallf 12$-BPS orbit follows from the theorem of Matumoto \cite{Matumoto} that will be used in the context of the higher-rank groups in  section~\ref{Matumoto}.}

 Therefore  we can obtain the complete $\smallf 14$-BPS content  of  (\ref{eisen016}) by analysing the Fourier modes of the Epstein series  $E^{Spin(5,5)}_{\alpha_1;5/2}$ when  $p_L$ and $p_R$ are assumed to satisfy  $\smallf 14$-BPS condition   $p_L^2-p_R^2\neq 0$.  We shall therefore assume that $p_L^2- p_R^2\neq 0$ for the rest of this section.  Hence
the $\smallf14$-BPS  Fourier modes of the first term  are obtained from the $s=5/2$ limit of (\ref{so55modes})\,,
\begin{eqnarray}
\label{485}
\cF_{(1,0)}^{(6)\, \alpha_1}( N_1)
&=&{\pi^{5/2}\over \G(\smallf 52)\,y_6^{1/2}}\,\sum_{p\,|\,\gcd(N_1)} \int_0^\infty d\tau_2\, e^{-\pi
    {p^2\over\tau_2y_6}-\pi {\tau_2\over p^2}\, (p_L^2+p_R^2)}\times\cr
&\times&\int_{-\frac12}^{\frac12}d\tau_1\,
  \hE_1^{SL(2)}(\tau) \, e^{i\pi\tau_1 \,{p_L^2-p_R^2\over p^2}}
\,,
\end{eqnarray}
where
\begin{equation}\label{E1hatSL2expansion}
    \hE_1^{SL(2)}(\tau) \ \ = \ \
    \tau_2-\smallf{3}{\pi}\log(\tau_2e^{-\hat{c}}) \ + \  \sum_{n\,\neq\,0} \f{\sigma_1(|n|)}{\xi(2)\,|n|}\,e^{-2\pi|n|\tau_2}\,e^{2\pi i n\tau_1}\,,
\end{equation}
where $\hat{c}=0.9080589548722\cdots$
(see (\ref{e:ESl2fourierCoef1}-\ref{e:ESl2fourierCoef})).  Note that since the residue of $E_s^{SL(2)}$ at $s=1$ is constant, the nonzero Fourier modes of $\hE_1^{SL(2)}$ are indeed the limits of the corresponding modes of $E_s^{SL(2)}$ as $s\rightarrow 1$; these are the only coefficients relevant to the $\tau_1$-integral in (\ref{485}) because of the assumption $p_L^2-p_R^2\neq 0$.
 Evaluation of (\ref{485}) gives the result
 \begin{eqnarray}
\nn   F_{\alpha_1;\frac52}^{Spin(5,5)\, \alpha_1}( N_1)&=&{16\pi
   \over y_6}\!
\sum_{p\,|\,\gcd(N_1)}  p^2\,
{\sigma_{-1}\left(|p_L^2-p_R^2|\over 2p^2\right) } \times \\
&\times&{ K_1(2\pi y_6^{-\frac12} \,\sqrt{p_L^2+p_R^2+|p_L^2-p_R^2|}) \over\sqrt{p_L^2+p_R^2+|p_L^2-p_R^2|}}
\,,
\end{eqnarray}
 where the lattice momenta are such that   $(p_L^2-p_R^2)/k^2\in
2\ZZ$. Using $SO(4,4)$ triality this corresponds to the full
spinor orbit $S$ characterizing the $\frac14$-BPS orbits as described in
section~\ref{spinorbit}.
In the weak coupling regime $y_6\to\infty$ these Fourier modes take
the form
 \begin{eqnarray}
\nn F_{\alpha_1;\frac52}^{Spin(5,5)\, \alpha_1}( N_1)&\sim&{8   \pi
   \over  y_6^{\frac34}}\!
  \,
\sum_{p\,|\,\gcd(N_1)}  p^2 \, {\sigma_{-1}\left(|p_L^2-p_R^2|\over 2p^2\right) } \times \\
&\times&{e^{-2\pi y_6^{-\frac12} \,\sqrt{p_L^2+p_R^2+|p_L^2-p_R^2|}}
\over(p_L^2+p_R^2+|p_L^2-p_R^2|)^{\frac34}}
\,.
\end{eqnarray}

In summary, the non-zero Fourier modes of $\calE_{(0,0)}^{(6)}$ have
support on the $\smallf 12$-BPS orbit in limits (i), (ii)  and (iii).
One of the contributions to $\calE_{(1,0)}^{(6)}$ is the regularised
series  $ E^{Spin(5,5)}_{\alpha_1;s}$ at $s=5/2$.  This has  non-zero
Fourier modes with support on the $\smallf 12$-BPS orbit in limits (i)
and (iii), but on both the $\smallf 12$-BPS and $\smallf 14$-BPS
orbits in limit (ii).  Although we have not computed the modes for the
other contribution to $\calE_{(1,0)}^{(6)}$ -- the  spinor series  -- we do know  its orbit content by use of
techniques similar to those in section~\ref{Matumoto}.  The result is
that the non-zero Fourier modes of this series have support on the
$\smallf 12$-BPS and $\smallf 14$-BPS orbits  in limits (i) and (iii),
but only on the $\smallf 12$-BPS  orbit in limit (ii).   In other
words the complete coefficient $\calE_{(1,0)}^{(6)}$ has the expected
content of both the  $\smallf 12$-BPS and $\smallf 14$-BPS in its
non-zero Fourier modes in all three limits.

\section{The next to minimal (NTM) representation}\label{sec:NTMdetails}

This section contains the proof of  theorem~\ref{mainthm}, drawing on
some results in representation theory that can be found in
appendix~\ref{sec:trapendix} by Ciubotaru and Trapa.  As we remarked just before its statement,  cases (i) and (ii) are by now well known, and so we restrict our attention to  case (iii):~the $s=5/2$ series.  To set some terminology, let $G=NAK$ be the Iwasawa decomposition of the split real Lie group $G$, $B$ the minimal parabolic subgroup of $G$ containing $NA$,  and  ${\mathfrak a}_\C = {\mathfrak a}\otimes_\IR \C$ be the complexification of the Lie algebra of $A$.  Without any loss of generality we may assume it is the complex span of the Chevalley basis vectors $H_\a$, where $\a$ ranges over the positive simple roots.  For any $\l\in {\mathfrak a}_\C^*$, the dual space of complex valued linear functionals on ${\mathfrak a}_\C$, define the vector space of functions on $G$
\begin{equation}\label{Vlambdadef}
    V_\l \ \ : = \ \ \left\{ \, f:G\rightarrow\C  \,|\,  f(nag) \,=\,e^{(\l+\rho)(H(a))}f(g), \,\forall \ n\in N, a\in A,g\in G \, \right\}.
\end{equation}
The transformation law and Iwasawa decomposition show that all functions in $V_\l$ are determined by their restriction to $K$.
Then $G$ acts on $V_\l$ by the right translation operator
\begin{equation}\label{pilambdadef}
    \(\pi_\l(h)f\)(g) \ \ : = \ \ f(gh)\,,
\end{equation}
making $(\pi_\l,V_\l)$ into a representation of $G$ commonly called a {\sl (nonunitary) principal series} representation.  It is irreducible for $\l$ in an open dense subset of ${\mathfrak a}_\C^*$, but reduces at special points with certain integrality properties -- such as the ones of interest to us. The representation $V_\l$ has a unique $K$-fixed vector up to scaling, namely any function whose restriction to $K$ is constant.  These are also known as the {\sl spherical} vectors of the representation, and any representation which contains them is also called ``spherical''.  When $V_\l$ is reducible, it clearly can have at most one irreducible spherical subrepresentation.

The minimal parabolic Eisenstein series is defined as
   \begin{equation}\label{minparabseries}
    E^G(\l,g) \ \ = \ \ \sum_{\g\in B(\Z)\backslash G(\Z)}e^{( \l+\rho)(H(\g g))}\, ,
 \end{equation}
 initially for $\l$ in Godement's range $\{\l | \langle \l,\a\rangle > 1$ for all $\a\in \Sigma\}$, and then by meromorphic continuation to  ${\mathfrak a}_\C^*$.   When $\l$ has the form $\l=2s\omega_\b-\rho$, it specializes to the maximal parabolic Eisenstein series (\ref{maxparabeisdef}).   For generic $\l$ in the range of convergence, the right translates of $E^G(\l,g)$ span a subspace of functions on $G(\Z)\backslash G(\IR)$ which furnish a representation of $G$ that is equivalent to  $V_\l$; the group action here is also given by the right translation operator (\ref{pilambdadef}).  The spherical vectors in this representation are the scalar multiplies of $E^G(\l,g)$, because the function $H(g)$ --  the logarithm of the Iwasawa $A$-component -- is necessarily right invariant under $K$.
  For general $\l$ at which $E^G(\l,g)$ is holomorphic, its right translates span a spherical subrepresentation of $V_\l$, again with the group action given by the right translation operator (\ref{pilambdadef}).

As mentioned above, the principal series $V_\l$ reduces for special values of $\l$.  This reducibility reflects special behavior of the Eisenstein series $E^G(\l,g)$.  This is most apparent at the point $\l=-\rho$, where the transformation law (\ref{Vlambdadef}) indicates that the constant functions on $K$ extend to constants on $G$, and hence that the
 trivial representation is a subrepresentation of $V_{-\rho}$.  Likewise,  the specialization of the minimal parabolic Eisenstein series at $\l=-\rho$  is the constant function identically equal to 1, a compatible fact.

   The proof of theorem~\ref{mainthm} rests upon special properties of the spherical constituent (i.e., Jordan-H\"older composition factor) of  $V_\l$ at the values of $\l$ relevant to the $s=5/2$ Epstein series. We recall that for this maximal parabolic Eisenstein series, $\l$ has the form
$\l=2s\omega_{\a_1}-\rho$; it is characterized by having
inner product $2s-1$ with $\a_1$, and inner product $-1$ with each
 $\a_j$, $j\ge 2$.  Write $\l_{\text{dom}}$ for  the  dominant weight in the Weyl orbit of $\l$, i.e., one whose inner product with all positive roots is nonnegative.  Table~\ref{tab:varioussandlambda} gives dominant weights for the groups  in Theorem~\ref{mainthm}  as well as its three values of $s\in \{0,3/2,5/2\}$, although of course only the last value is of immediate relevance in this section.

\begin{table}
\begin{center}\begin{tabular}{|c|c|c|c|}
\hline
  &  $G=E_6$ &  $G=E_7$ & $G=   E_8$ \\
\hline
 $s=0$ &   &   &  \\
\hline
 $\l_{\text{dom}}$ &  [1,1,1,1,1,1] &  [1,1,1,1,1,1,1] & [1,1,1,1,1,1,1,1] \\
 $s_{\text{GRS}}$ &   &   &  \\
 $z_{\text{KS}}$ &   &   &  \\
\hline
 $s=3/2$ &   &   &  \\
\hline
 $\l_{\text{dom}}$ &  [1,1,1,0,1,1] &  [1,1,1,0,1,1,1] & [1,1,1,0,1,1,1,1] \\
 $s_{\text{GRS}}$ &  1/4 &  5/18 & 19/58 \\
 $z_{\text{KS}}$ &  7/22 &  11/34 & 19/58 \\
\hline
 $s=5/2$ &   &   &  \\
\hline
 $\l_{\text{dom}}$ &  [0,1,1,0,1,1] &  [1,1,1,0,1,0,1] & [1,1,1,0,1,0,1,1] \\
 $s_{\text{GRS}}$ &  -1/2 &  1/18 & 11/58 \\
 $z_{\text{KS}}$ &  none &  33/34 & 11/58 \\
\hline
\end{tabular}\end{center}\caption{\label{tab:varioussandlambda} The values of $\l$ for the three values of $s$ and three groups in theorem~\ref{mainthm}.  Weights $\l\in {\mathfrak a}_\C^*$ are listed here in terms of their inner products with the positive simple roots as  $[\langle \l,\a_1 \rangle,\langle \l,\a_2 \rangle,\ldots]$.
For comparison with \cite{grs,KazhdanSavin}, we have listed the
 parameters $s_{\text{GRS}}$  (the quantity $s$  on  \cite[p.71]{grs}) and $z_{\text{KS}}$ (the quantify $z(G)$ from \cite[p.86]{grs}) for $s=3/2$, as well as  their corresponding generalizations for $s=5/2$.  These parameters coincide for the group $E_8$.  The parameter $z_{\text{KS}}$ is not defined in the $s=5/2$ case for $E_6$ because the relevant Weyl orbits do not intersect (cf.~\cite[Section 3.1]{Green:2010kv}).}
\end{table}

The case of $G=E_6$ is slightly easier than the others because of a
low-dimensional coincidence, which in fact is mostly independent of the actual
value of $s$ in that the same statement holds for generic $s$.
Namely, the representation $V_\l$ we consider is part of a family of
degenerate principal series representations, induced from the trivial representation
on the semisimple $Spin(5,5)$ factor of the Levi component $GL(1)\times
Spin(5,5)$ of the maximal parabolic subgroup $P_{\a_1}$.  These
representations are indexed by the one dimensional family
$\l=2s\omega_{\a_1}-\rho$, $s\in \C$, which is related to  the $GL(1)$
factor.  Though they may reduce at particular points, their
Gelfand-Kirillov dimension\footnote{The Gelfand-Kirillov dimension is
  a numerical index of how ``large'' a representation is; it is half
  the dimension of the associated coadjoint nilpotent orbit (i.e., the
  orbit whose closure is the wavefront set of the representation).
  For example, finite dimensional representations have
  Gelfand-Kirillov dimension equal to zero.} is equal to the dimension
of the unipotent radical of that parabolic, 16; likewise, any
subrepresentation of it cannot have larger dimension.  Since  the
dimension of the wavefront set of a representation is  twice the
Gelfand-Kirillov dimension, it is bounded by 32.  For $E_6$, the
orbits in figure~\ref{fig:smallorbit} have dimensions 0, 22, and 32;
all other orbits have larger Gelfand-Kirillov dimension.  Hence the
orbit attached to the $s=5/2$ Eisenstein series for $E_6$ is either
the trivial orbit, the minimal orbit, or the next-to-minimal orbit.
It cannot be the trivial orbit, because only the trivial
representation is attached to it.  Likewise,
Kazhdan-Savin~\cite{KazhdanSavin}  proved a uniqueness statement for the minimal orbit,
that (up to Weyl equivalence) only the $s=3/2$ series is related to
the minimal representation.  We thus conclude it is attached to the
next-to-minimal orbit.

To explain the $s=5/2$ cases for $E_7$ and $E_8$ we need to rely on
some recent results from representation theory, and some notions from
there concerning {\sl unipotent} and {\sl special unipotent}
representations (see appendix~\ref{sec:trapendix}).  A striking feature from table~\ref{tab:varioussandlambda} is that  $\langle
\l_{\text{dom}},\a_j \rangle$ has all 1's except for a single zero for
the $s=3/2$ case, and two zeroes for the $s=5/2$ case.  This
phenomenon, which came up here because of physical arguments, also
arose in work on special unipotent representations.
 These
$\l_{\text{dom}}$ each have the property that there exists an element $H$  of the Cartan subalgebra of
 $\mathfrak g$ such that $[H,X_\a]=\langle \l_{\text{dom}},\a_j\rangle X_{\a_j}$ for each positive simple root $\a_j$.
 Furthermore, there exists a  homomorphism
from ${\mathfrak{sl}}_2$ to $\mathfrak g$ carrying
$\ttwo 1{0}{0}{-1}$ to $H$, and  $\ttwo 0100$ to a nilpotent element $X$.  Thus a ``dual'' coadjoint
nilpotent orbit, namely
the one containing $X$, is associated to $\l_{\text{dom}}$.
 \begin{figure}[ht]
 \centering\includegraphics[width=12cm]{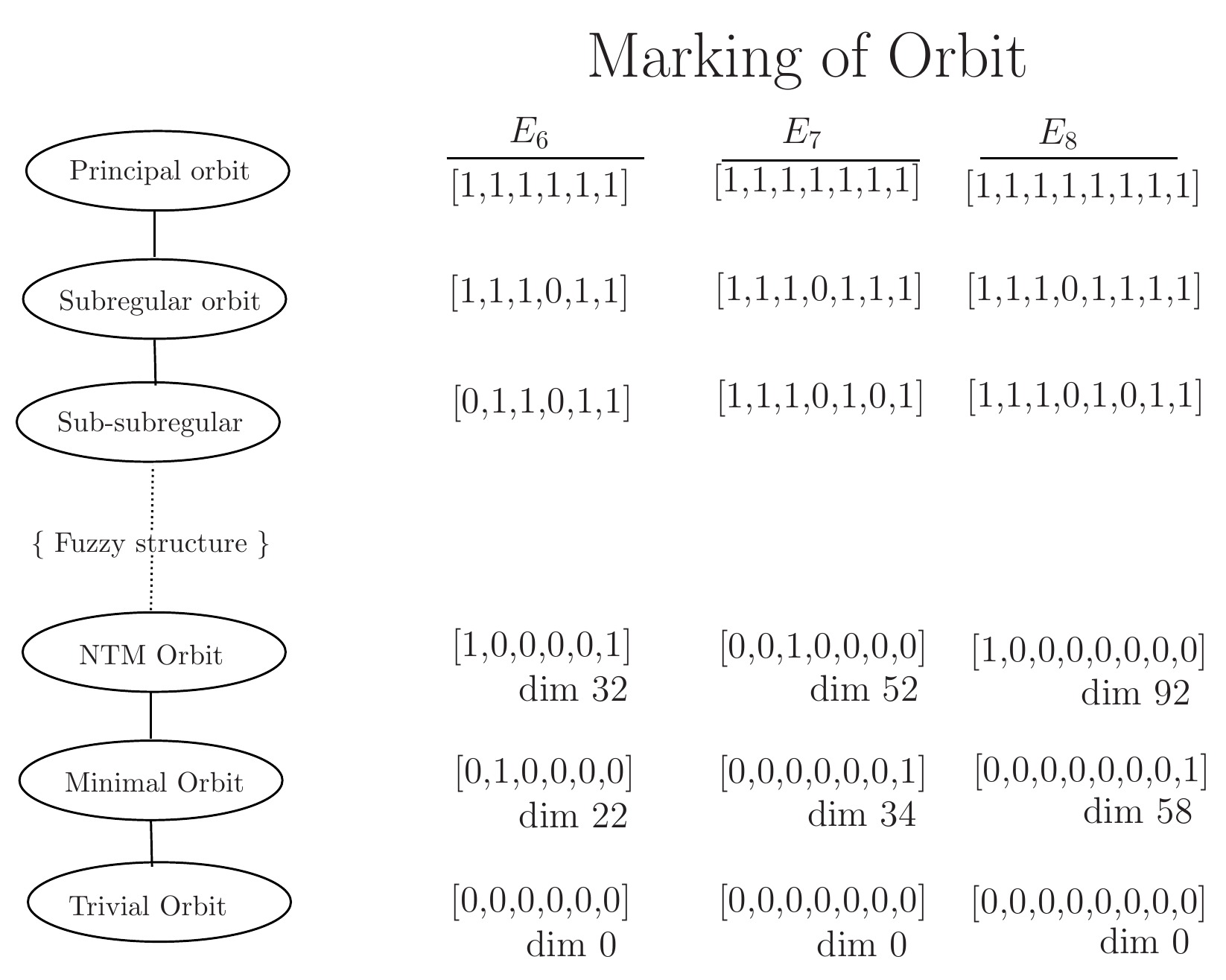}
 \caption{ The largest and smallest orbits, with markings (also known as ``weightings'') listed in terms of the inner products $\langle
\l_{\text{dom}},\a_j \rangle$  described in the text.\label{fig:topandbottomorbit} }
 \end{figure}
 In terms of
figure~\ref{fig:topandbottomorbit}, in our three examples these   related dual orbits are the top
three listed, though in the reverse order.    Appendix~\ref{sec:trapendix} describes a related construction for more general types of orbits beyond the ones considered in this paper.

As part of the more general result given in appendix~\ref{sec:trapendix}, corollary~\ref{c:su} then asserts that the spherical constituent of each of the three principal series
$V_{\l_{\text{dom}}}$  has wavefront set equal to     the closure of the dually related orbit listed in figure~\ref{fig:topandbottomorbit}.  This
proves theorem~\ref{mainthm} for $E_7$ and $E_8$.

\section{Fourier coefficients and their vanishing}\label{sec:shrunkFourierCoeff}

\subsection{Dimensions of orbits in the character variety}\label{sec:OrbitCharacter}

In  sections~\ref{sec:representation}-\ref{sec:BPSMorbits} we listed a number of explicit features of the orbits of instantons for the parabolic subgroups $P_{\a_1}$, $P_{\a_2}$, and $P_{\a_{d+1}}$ (in the numbering of figure~\ref{fig:dynkin}).  These are the character variety orbits discussed at the beginning of section~\ref{sec:nreFourierCoeff}.  In this section we give more details, in particular basepoints and dimensions for each of the finite number of orbits under the  complexification $L_\C$ of the Levi factor of the parabolic.  As shorthand, we will refer to these as the ``complex orbits of the Levi''.  We shall also use the notation $Y_{\a}$ to refer to the root vector $X_{-\a}$, in order to keep the listing of basepoints more readable.

  This information is quoted from the paper \cite{Miller-Sahi},
  which lists the corresponding information for any maximal parabolic subgroup of  any Chevalley group,
   whether classical or exceptional (see \cite[\S5]{Miller-Sahi} for more examples and details
    of how these are computed).
    We also describe the group action of the Levi in some of the  cases, the rest being described in \cite{Miller-Sahi}.
  Recall that the dimensions of the character varieties were given earlier in
  table~\reftab{tab:dimUnipotent}.  In the following subsections, we  expand upon this for the groups $E_5=Spin(5,5)$, $E_6$, $E_7$, and $E_8$.  For ease of reference, tables~\ref{tab:minimalnilpotent21}, \ref{tab:minimalnilpotent22}, and \ref{tab:minimalnilpotent2}  give  the orbit dimensions  for the parabolic subgroups  $P_{\a_1}$,  $P_{\a_2}$,  and $P_{\a_{d+1}}$ of each of these  groups, respectively.

\begin{table}[center]
  \centering
  \begin{tabular}[t]{||c|c|c|c|c|c|c|c|c|c|c||}
    \hline
Group & \multicolumn{10}{|c||}{dimensions}\\
\hline
$SL(2)$&0&1&-&-&- &-&-&-&-&-\\
$SL(3)\times SL(2)$&0& 2 &-&-&-&-&-&-&-&-\\
$SL(5)$ & 0 & 4 &-&-&-&-&-&-&-&-\\
$Spin(5,5)$ & 0 & 7 & 8&- &-&-&-&-&-&-\\
$E_6$ & 0 & 11 & 16&- &- &-&-&-&-&-\\
$E_7$ & 0  & 16 & 25 & 31 &   32&-&-&-&-&-       \\
$E_8$ & 0  &22& 35& 43& 44 & 50 &54 &59 & 63 & 64       \\
\hline
  \end{tabular}
  \caption{Dimensions of  character variety  orbits  for the Levi  component of
    the  parabolic formed by  deleting the  first node  of $E_4=SL(5)$,
    $E_5=Spin(5,5)$, $E_6$, $E_7$, and $E_8$. A dash, $-$,  signifies that
    there is no  orbit.
The character variety orbits in this parabolic subgroup are the $Spin(d,d)$ spinor orbits listed in section~\ref{spinorbit}.}
  \label{tab:minimalnilpotent21}
\end{table}

\begin{table}[center]
  \centering
  \begin{tabular}[t]{||c|c|c|c|c|c|c|c|c|c|c||}
    \hline
Group & \multicolumn{10}{|c||}{dimensions}\\
\hline
$SL(2)$&0&-&-&-&-&-&-&-&-&-\\
$SL(3)\times SL(2)$&0&1&-&-&-&-&-&-&-&-\\
$SL(5)$ & 0 & 4& -&-&-&-&-&-&-&-\\
$Spin(5,5)$ & 0 & 7 & 10&-&-&-&-&-&-&-\\
$E_6$ & 0 & 10 & 15&19&20&-&-&-&-&- \\
$E_7$ & 0  & 13 & 20 & 21 & 25& 26 & 28&31&34&35   \\
$E_8$ & 0  & 16& 25& 28&31& 32&35&38&40&$\cdots$\\
\hline
  \end{tabular}
  \caption{Dimensions of   character variety orbits  of the Levi  component for
    the  parabolic formed by  deleting the  second node  of $E_4=SL(5)$,
    $E_5=Spin(5,5)$, $E_6$, $E_7$, and  $E_8$.  A dash, $-$,  signifies that
    there is no orbit.  Not all $E_8$ orbits are listed
    (there are 23 in total).}
  \label{tab:minimalnilpotent22}
\end{table}

\begin{table}[center]
  \centering
  \begin{tabular}[t]{||c|c|c|c|c|c||}
    \hline
Group & \multicolumn{5}{|c||}{dimensions}\\
\hline
$SL(2)$&0&-&-&-&-\\
$SL(3)\times SL(2)$&0&1 &3&-&-\\
$SL(5)$ & 0 & 5 & 6&-&-\\
$Spin(5,5)$ & 0 & 7 & 10&-&-\\
$E_6$ & 0 & 11 & 16&-&- \\
$E_7$ & 0  & 17 & 26 & 27 &  -         \\
$E_8$ & 0  & 28 & 45 & 55 & 56       \\
\hline
  \end{tabular}
  \caption{Dimensions of   character variety orbits  of the Levi  component for
    the  parabolic formed by  deleting the  last node  of $E_4=SL(5)$,
    $E_5=Spin(5,5)$,  $E_6$, $E_7$,  and  $E_8$. A dash, $-$,  signifies that
    there is no orbit. The character variety
      orbits in this parabolic  subgroup were also listed in table~\ref{tab:bpsorbits1} based on enumeration of instanton orbits.}
  \label{tab:minimalnilpotent2}
\end{table}

\subsubsection{$Spin(5,5)$}

Recall that we label our $E_5=Spin(5,5)$ Dynkin diagram according to the numbering in
 figure~\ref{fig:dynkin}.  This does not match the customary numbering of the $Spin(5,5)$ Dynkin diagram,
  but has the advantage of allowing for a uniform discussion of all of our cases of interest.

Node 1 is the so-called ``vector'' node, because $P_{\a_1}$ has Levi
component isomorphic to $GL(1)\times Spin(4,4)$, which acts on the
$8$-dimensional, abelian unipotent radical by the  8-dimensional spin representation of $Spin(4,4)$.  This action breaks into 3 complex
orbits:~the trivial orbit; a 7-dimensional orbit with basepoint
$Y_{\a_1}$; and the open, dense 8-dimensional orbit with basepoint
$Y_{11110}+Y_{10111}$ (see table~\ref{tab:minimalnilpotent21}).

Nodes 2 and 5 are the ``spinor nodes'', and have identical orbit structure (up to relabeling the nodes).  Here the Levi component of
$P_{\a_2}$ or $P_{\a_5}$ is now isomorphic to $GL(1)\times SL(5)$, and
acts on the 10-dimensional abelian unipotent radical by the second
fundamental representation, also known as the exterior square
representation.  In other words, the action of the $SL(5)$ piece is
equivalent to that on antisymmetric 2-tensors $x\wedge y = - y \wedge
x$, where $x$ and $y$ are 5-dimensional vectors.  This action also
has 3 complex orbits (which can be seen as part of a general description for abelian
unipotent radicals of maximal parabolic subgroups given in
\cite{richardson-rohrle-steinberg}):~the trivial orbit; a
7-dimensional orbit with basepoint $Y_{\a_2}$ in the case of node 2,
and $Y_{\a_5}$ in the case of node 5; and the open, dense
10-dimensional orbit with basepoint $Y_{01121}+Y_{11111}$ (see
table~\ref{tab:minimalnilpotent22} or
table~\ref{tab:minimalnilpotent2}).  This last basepoint is in the
open dense orbit for either $P_{\a_2}$ or $P_{\a_5}$.

\subsubsection{$E_6$}

Nodes 1 and 6 are related by an automorphism of the Dynkin diagram, and
have identical orbit structure (up to relabeling the nodes).  Here the Levi component is isomorphic
to $GL(1)\times Spin(5,5)$, which acts on the 16-dimensional, abelian
unipotent radical by the spin representation of $Spin(5,5)$.  There are
three complex orbits:~the trivial orbit; an 11-dimensional orbit with
basepoint $Y_{\a_1}$ in the case of node 1, and $Y_{\a_6}$ in the case
of node 6; and the open, dense 16-dimension orbit with basepoint
$Y_{111221}+Y_{112211}$ for either nodes 1 or 6 (see
table~\ref{tab:minimalnilpotent21} or
table~\ref{tab:minimalnilpotent2}).

Node 2 is the first case we encounter with a non-abelian unipotent radical.  It is instead a 21-dimensional Heisenberg group, and its
character variety has 5 complex orbits (another general fact for
Heisenberg unipotent radicals of maximal parabolic subgroups
\cite{Rohrle}):~the trivial orbit; a 10-dimensional orbit with
basepoint $\a_2$; a 15-dimensional orbit with basepoint
$Y_{111221}+Y_{112211}$; a 19-dimensional orbit with basepoint
$Y_{011221}+Y_{111211}+Y_{112210}$; and the open, dense 20-dimensional
orbit with basepoint $Y_{010111}+Y_{112210}$
(see table~\ref{tab:minimalnilpotent22}).

\subsubsection{$E_7$}

This is the first group for which the three nodes have mathematically
different structures.  Node 1 has a 33-dimensional unipotent radical
which is a Heisenberg group, and Levi component isomorphic to
$GL(1)\times Spin(6,6)$.    The action on the  32-dimensional  character variety again has
5 complex orbits:~the trivial orbit; a 16-dimensional orbit with basepoint $Y_{\a_1}$; a 25-dimensional orbit with basepoint $Y_{1123321}+Y_{1223221}$; a 31-dimensional orbit with basepoint $Y_{1122221}+Y_{1123211}+Y_{1223210}$; and the open, dense 32-dimensional  orbit with basepoint $Y_{1011111}+Y_{1223210}$ (see Table~\ref{tab:minimalnilpotent21}).

Node 2 has a 42-dimensional unipotent radical, and a 35-dimensional
character variety.  The Levi component $GL(1)\times SL(7)$ acts with 10 complex orbits:~the trivial orbit; a 13-dimensional orbit with basepoint $Y_{\a_2}$; a 20-dimensional orbit with basepoint $Y_{1122221}+Y_{1123211}$; a 21-dimensional orbit with basepoint $Y_{0112221}+Y_{1112211}+Y_{1122111}$; a 25-dimensional orbit with basepoint $Y_{1112221}+Y_{1122211}+Y_{1123210}$; a 26-dimensional orbit with basepoint $Y_{1111111}+Y_{1123210}$; a 28-dimensional orbit with basepoint $Y_{0112221}+Y_{1112211}+Y_{1122111}+Y_{1123210}$; a 31-dimensional orbit with basepoint $Y_{0112221}+Y_{1111111}+Y_{1123210}$; a 34-dimensional orbit with basepoint $Y_{0112211}+Y_{1112111}+Y_{1112210}+Y_{1122110}$; and the open, dense 35-dimensional orbit with basepoint $Y_{0112111}+Y_{0112210}+Y_{1111111}+Y_{1112110}+Y_{1122100}$ (see table~\ref{tab:minimalnilpotent22}).

Node 7 has a 27-dimensional abelian unipotent radical, and Levi
component isomorphic to $GL(1)\times E_{6,6}$.  The latter acts with 4 complex orbits:~the trivial orbit, a 17-dimensional orbit with basepoint $Y_{\a_7}$, a 26-dimensional orbit with basepoint $Y_{1123321}+Y_{1223221}$, and the open, dense 27-dimensional orbit with basepoint $Y_{0112221}+Y_{1112211}+Y_{1122111}$ (see Table~\ref{tab:minimalnilpotent2}).

\subsubsection{$E_8$}

This is the largest of our groups, and the unipotent radicals of its maximal parabolics  are never abelian.

Node 1 has a 78-dimensional unipotent radical, and a 64-dimensional
character variety.  The Levi component is isomorphic to  $GL(1)\times Spin(7,7)$ and acts on the character variety according to the spin representation of $Spin(7,7)$, with 10 complex orbits:~the trivial orbit; a 22-dimensional orbit with basepoint $Y_{\a_1}$; a 35-dimensional orbit with basepoint $Y_{12244321}+Y_{12343321}$; a 43-dimensional orbit with basepoint $Y_{12233321}+Y_{12243221}+Y_{12343211}$; a 44-dimensional orbit with basepoint $Y_{11122221}+Y_{12343211}$; a 50-dimensional orbit with basepoint $Y_{11233321}+Y_{12233221}+Y_{12243211}+Y_{12343210}$; a 54-dimensional orbit with basepoint $Y_{11222221}+Y_{12243211}+Y_{12343210}$; a 59-dimensional orbit with basepoint $Y_{11122221}+Y_{11233211}+Y_{12232211}+Y_{12343210}$; a 63-dimensional orbit with basepoint $Y_{11222221}+Y_{11232211}+Y_{11233210}+Y_{12232111}+Y_{12232210}$; and the open, dense 64-dimensional orbit with basepoint $Y_{11122111}+Y_{11221111}+Y_{11233210}+Y_{12232210}$
(see table~\ref{tab:minimalnilpotent21}).

Node 2 has a 92-dimensional unipotent radical, and a 56-dimensional
character variety.  The Levi component is isomorphic to $GL(1)\times
SL(8)$ and acts according to the third fundamental representation of
$SL(8)$, also known as the exterior cube representation.  It acts with
23 complex orbits, the four smallest of which are:~the trivial orbit;
a 16-dimensional orbit with basepoint $Y_{\a_2}$; a 25-dimensional
orbit with basepoint $Y_{11232221}+Y_{11233211}$; and a 28-dimensional
orbit with basepoint $Y_{11122221}+Y_{11222211}+Y_{11232111}$ (see
table~\ref{tab:minimalnilpotent22}).

Node 8 has a 57 dimensional unipotent radical which is a Heisenberg
group.  The Levi factor is isomorphic to $GL(1)\times E_{7,7}$ and acts
with 5 complex orbits on the 56-dimensional character variety:~the
trivial orbit; a 28-dimensional orbit with basepoint $Y_{\a_8}$; a
45-dimensional orbit with basepoint $Y_{22454321}+Y_{23354321}$; a
55-dimensional orbit with basepoint
$Y_{12244321}+Y_{12343321}+Y_{22343221}$; and the open, dense
56-dimensional orbit with basepoint $Y_{01122221}+Y_{22343211}$ (see
table~\ref{tab:minimalnilpotent2}).

\subsection{Applications of Matumoto's theorem}
\label{Matumoto}
In Section~\ref{sec:autreps} we mentioned
 that representations of real groups have an
invariant attached to them, the wavefront set, that in a sense
measures how big the representation is.   Theorem~\ref{t:su} indeed computes this wavefront set in many cases, including ours.
There is  a theorem due
to Matumoto \cite{Matumoto}
 that asserts, in a precise sense, that automorphic forms in  small representations cannot have large Fourier coefficients.  Namely, he proves that if an element $Y\in {\mathfrak u}_{-1}$ associated to the character $\chi$ from (\ref{Ychar}) does not lie in the wavefront set, then the Fourier coefficient $\phi_\chi$ from (\ref{fourierexp1a}) must vanish identically.   We will use real group methods here in deference to the importance of the underlying symmetry groups $E_{d+1}(\IR)$, but it is notable that we could obtain the same results using $p$-adic methods via a vanishing result of M\oe glin-Waldspurger \cite{moewal}.  Related information is given at the end of appendix~\ref{sec:trapendix}.

For example, the trivial representation has wavefront set $\{0\}$, and likewise the constant function does not have any nontrivial Fourier coefficients.
In~\cite{Miller-Sahi} a detailed analysis is given of the different
character variety orbits for each maximal parabolic subgroup of an exceptional group, and which
coadjoint nilpotent orbits they are contained in. It is then a simple matter to apply Matumoto's theorem and determine a set of Fourier coefficients which automatically vanishes because their containing coadjoint nilpotent orbits lie outside the wavefront set.
In particular, it is shown in \cite{Miller-Sahi}  that the closure of the minimal coadjoint nilpotent orbit contains the two smallest character variety orbits in each of
the examples of $P_{\a_1}$, $P_{\a_2}$, and $P_{\a_{d+1}}$ for the groups $E_{d+1}$, $5\le d \le 7$, but no others
(this was known to experts, at least in special cases
-- see for example \cite{grs}).  Likewise, it is also verified
there that the closure of the  next-to-minimal coadjoint nilpotent orbit contains the three smallest
character variety orbits in each of these nine configurations of maximal parabolics and $E_{d+1}$ groups, but no others.

 Combining this with the
characterization in Theorem~\ref{mainthm} of the wavefront sets for
the Epstein series at $s=0$, $3/2$, and $5/2$, we get the following
statement about the vanishing of Fourier coefficients.
 This gives a rigorous proof of the vanishing statements on page~\pageref{stringmotivatedprediction}.

\begin{thm}\label{thm:vanishingofcoefficients}Let $5 \le d \le 7$ and
  $G=E_{d+1}$ as defined in table~\reftab{tab:Udual}.  Then:

\begin{enumerate}
  \item[(i)] All Fourier coefficients of the $s=0$ Epstein series vanish in any of the parabolics $P_{\a_1}$, $P_{\a_2}$, or $P_{\a_{d+1}}$,
      with the  exception of the constant terms (which were calculated in \cite{Green:2010kv}).
  \item[(ii)] All Fourier coefficients of the $s=3/2$ Epstein series $E^G_{\a_1;3/2}$ vanish in any of the parabolics $P_{\a_1}$, $P_{\a_2}$, or $P_{\a_{d+1}}$, with the  exceptions of the constant term and the smallest dimensional character variety orbit.  This  orbit has:~dimension 11 for $E_6$ and either $P_{\a_1}$ or $P_{\a_6}$, and  dimension 10 for  $P_{\a_2}$;  dimensions 16, 13, and 17 for $E_7$ and   $P_{\a_1}$,  $P_{\a_2}$, and  $P_{\a_7}$, respectively; and dimensions 22, 16, and 28 for $E_8$ and  $P_{\a_1}$,  $P_{\a_2}$, and  $P_{\a_8}$, respectively.
  \item[(iii)] All Fourier coefficients of the $s=5/2$ Epstein series $E^G_{\a_1;5/2}$  vanish in any of the parabolics $P_{\a_1}$, $P_{\a_2}$, or $P_{\a_{d+1}}$, with the  exceptions of the constant term and the next two smallest dimensional character variety orbits.  This additional character variety orbit is:~the 16, 15, and 16-dimensional orbit  for $E_6$ and  $P_{\a_1}$,  $P_{\a_2}$, and  $P_{\a_6}$, respectively; the 25, 20, and 26-dimensional orbit for $E_7$ and  $P_{\a_1}$,  $P_{\a_2}$, and  $P_{\a_7}$, respectively; and the 35, 25, and 45-dimensional orbit for $E_8$ and $P_{\a_1}$,  $P_{\a_2}$, and  $P_{\a_8}$, respectively.
\end{enumerate}
\end{thm}

\section{Square integrability of special values of Eisenstein series}\label{sec:L2}

In this section we remark that some of the coefficient functions
$\mathcal{E}^{(D)}_{(0,0)}$ and  $\mathcal{E}^{(D)}_{(1,0)}$ from the
expansion (\ref{amp}) provide examples of  square-integrable
automorphic forms on higher rank groups.  In particular, we
  will prove this is the case for $\mathcal{E}^{(D)}_{(1,0)}$ on  $E_7$ and $E_8$.  In light of (\ref{dfourrfourcoeff}), this proves the associated automorphic representation is unitary, since it can be realized in the Hilbert space $L^2(E_{d+1}(\Z)\backslash E_{d+1}(\IR))$.  This unitary can also be demonstrated by purely representation theoretic methods.   It is an instance of broader conjectures of James Arthur on unitary automorphic representations, which are studied in more detail in \cite{Mill:2012}.   This fact about the residual $L^2$ spectrum is
at present   more of a curiosity as far as our investigations here are concerned, since we are not aware of any particular importance for our applications.    The analysis in the proof also determines the exact asymptotics of these coefficients in various limits, generalizing those studied in \cite{Green:2010kv}.

\begin{thm}\label{thm:L2}
Let $G$ denote the group $E_{d+1}$ defined in table~\reftab{tab:Udual}.
\begin{enumerate}
  \item[(i)] The Epstein series $E^{G}_{\a_1;0}$ is constant, and hence always square-integrable.
  \item[(ii)] The Epstein series  $E^{G}_{\a_1;3/2}$ and hence $\mathcal{E}^{(10-d)}_{(0,0)}$ is  square-integrable  if  $4\le d \le 7$.
  \item[(iii)] The Epstein series $E^{G}_{\a_1;5/2}$ and hence  $\mathcal{E}^{(10-d)}_{(1,0)}$ is square-integrable  if   $6\le d \le 7$.
\end{enumerate}
\end{thm}

\noindent
Case (i) is obvious since the quotient $E_{d+1}(\Z)\backslash E_{d+1}(\IR)$ has finite volume, while case (ii)  was proven earlier  by \cite{grs}.  We have included them here in the statement for convenience and comparison.  It should be stressed, though, that  $E^{G}_{\a_1;s}$ is certainly not  square integrable for {\sl general} $s$.  The same method treats the lower rank groups as well, though since the statements are not needed here we refer to papers \cite{grs} and\cite{GinzSayag} for $Spin(5,5)$.
\begin{proof}
Recall that the series  $E^G_{\a_1;s}$ is a  specialization of the {\sl minimal parabolic} Eisenstein series $E^G(\l,g)$
 from (\ref{minparabseries}) at $\l=2s\omega_1-\rho$. This is explained in our context in  \cite[Section 2]{Green:2010kv}, where
Langlands' constant term formula is also given in Theorem~2.18.  The latter shows that the constant term of $E^G(\l,g)$  along any maximal parabolic subgroup $P$ is a sum of other minimal parabolic Eisenstein series on its Levi component.  By induction, this is also true if $P$ is  not maximal.  In particular, since these Eisenstein series on smaller groups are orthogonal to all cusp forms on those groups, the constant terms are therefore orthogonal to all cusp forms on the Levi components -- a meaningful statement only, of course, when the parabolic $P$ is not the Borel subgroup $B$ (so that  the Levi is nontrivial).  This means $E^G(\l,g)$ has ``zero cuspidal component along any such $P$'' in the sense of  \cite[Section 3]{langlandsSLN}, or equivalently that it is ``concentrated'' on the Borel subgroup $B$.

The constant term along $B$ is explicitly given in terms of a sum over the Weyl group:
 \begin{equation}\label{constterminBorel}
 \int_{N(\Z)\backslash N(\IR)}E^G(\l,ng)\,dn \ \ = \ \ \sum_{w\,\in\,\Omega}e^{(w\l+\rho)(H(g))} M(w,\l)\,,
\end{equation}
where $M(w,\l)$ is given by the explicit product over roots whose sign is flipped by $w$,
\begin{equation}\label{Mwl}
    M(w,\l) \ \ = \ \ \prod_{\srel{\a\,>\,0}{w\a\,<\,0}} c(\langle \l,\a \rangle)\,,
\end{equation}
with
\begin{equation}\label{csdef}
    c(s) \ \ := \ \ \f{\xi(s)}{\xi(s+1)} \  \  \   \ \text{and} \,  \ \ \ \ \xi(s) \ \ := \ \ \pi^{-{s\over2}}\,\G(\smallf{s}{2})\,\zeta(s)
\end{equation}
(see, for example,  \cite[(2.16)-(2.21)]{Green:2010kv}).
 This formula is valid for generic $\l$, and develops logarithmic terms at special points via meromorphic continuation.  Moreover, certain coefficients $M(w,\l)$ may vanish, for example when $\langle \l,\a\rangle=-1$ and the respective factor in (\ref{csdef}) has a zero owing to the pole of $\xi(s+1)$ at $s=-1$ (unless it  is canceled by a pole from another factor).   Indeed, $c(s)$ has a simple zero at $s=-1$, a simple pole at $s=1$, and is holomorphic at all other integers. Because $E^G(\l,g)$ is ``concentrated on $B$'', Langlands' criteria in \cite[Section 5]{langlandsSLN} asserts that it is square-integrable if and only if the surviving exponents $w\l$ have negative inner product with each fundamental weight:
  \begin{equation}\label{Langlandscritera}
    \langle w\l , \omega_\a \rangle \ \ < \ \ 0 \ \ \ \text{for each~~}\a\,>\,0\,.
  \end{equation}
The rest of the proof involves an explicit calculation to check that for each possible value of $w\lambda$, either the sum of $e^{(w'\l+\rho)(H(g))} M(w',\l)$ over all $w'\in \Omega$ with $w'\lambda=w\lambda$ vanishes, or instead that (\ref{Langlandscritera}) holds.  Actually, despite the enormous size of the Weyl groups involved, $M(w,\l)$ vanishes for all but very few $w$ (because of the special nature of $\l$).

Though the individual terms in (\ref{constterminBorel}) are frequently singular at the values of $\l$ in question,  the overall sum can be calculated explicitly by taking limits.  We now present the result of this calculation.  To make the condition (\ref{Langlandscritera}) more transparent, we take $g=a$ to be an element of the maximal torus $A$ (as we of course may, given that $H(g)$ depends only on the $A$-component of $g$'s Iwasawa decomposition).
We then furthermore parameterize $a$ by real numbers $r_1,r_2,\ldots$  via the condition that the simple roots on $a$ take the  values
\begin{equation}\label{rparameterization}
    a^{\a_1} \ = \ e^{r_1}\,, \  a^{\a_2} \ = \ e^{r_2}\,, \ \ldots.
\end{equation}
For example, for $G=E_6$
 the limiting value of (\ref{constterminBorel}) as $\l$ approaches $3\omega_1-\rho$ can be calculated explicitly as  $e^{2 r_1+3 r_2+4 r_3+6 r_4+4
  r_5+2 r_6}$ times
\begin{equation}\label{3halvesinborelforE6}
 \frac{3 \zeta (3) \left(e^{2 r_1+r_3}+e^{r_5+2 r_6}\right)+\pi ^2 ( e^{r_2}+   e^{r_3}+ e^{r_5})+6 \pi ( r_4+ \gamma   -  \log (4\pi )) }{3 \zeta (3)}\,.
\end{equation}
The exponentials are all  dominated by $e^{\rho(H(g))}=e^{8r_1+11 r_2+15 r_3+21
  r_4+15r_5+8 r_6}$ for $r_i>0$, that is, (\ref{Langlandscritera}) holds  and
hence $E^G_{\a_1;3/2}$ is square-integrable -- verifying a fact  proven
in~\cite{grs}.

We now turn to the two new cases, those of the $s=5/2$ series for $E_7$ and $E_8$.  We recall the computational method of \cite[Section 2.4]{Green:2010kv} to find the minimal parabolic constant terms, namely to precompute the set
\begin{equation}\label{precomputedset}
    {\mathcal S} \ \ : = \ \ \{\ w \,\in\,\Omega \ \mid \ w\a_i\,>\,0\ \ \ \text{for all} \ \, i \, \neq \,1 \ \}\,.
\end{equation}
 For $w\notin {\mathcal S}$, $M(w,\l)$  will include the factor $c(\langle \l,\a_i\rangle)=c(\langle 2s\omega_1-\rho,\a_i\rangle)=c(-\langle \rho,\a_i\rangle)=c(-1)=0$ for some $i>1$.  At the same time, at least for $\Re{s}< \half$, all inner products $\langle \l,\a\rangle$ will be negative, and hence none of the other factors in (\ref{Mwl}) can have a pole (after all, $c(s)$ is holomorphic for $\Re{s}<0$). Thus the term for $w$ in (\ref{constterminBorel}) vanishes identically in $s$ by analytic continuation, and the sum in (\ref{constterminBorel}) reduces to one over $w\in {\mathcal S}$.

For $E_7$ there are only 126 elements in ${\mathcal S}$ out of the 2,903,040 elements of the full Weyl group $\Omega$.  It can be calculated that all but three $w$ of these 126 satisfy Langlands' condition (\ref{Langlandscritera}), and the three that do not have the following expressions for $M(w,\l)$ for $s=5/2+\e$:
\begin{equation}\label{threexceptionsforE7}
    \aligned
\text{Exception 1 : } \,   &    c(2 (\epsilon -5)) c(2 \epsilon )^2 c(2 \epsilon -9) c(2 \epsilon
-8)^2 c(2 \epsilon -7)^2 c(2 \epsilon -6)^3 \  \times \\ & \  \times \  c(2 \epsilon -5)^3 c(2
\epsilon -4)^3 c(2 \epsilon -3)^3 c(2 \epsilon -2)^3 c(2 \epsilon
-1)^3 \  \times \\ &  \  \times \   c(2 \epsilon +1)^2 c(2 \epsilon +2) c(2 \epsilon +3) c(2
\epsilon +4) c(4 \epsilon -7)\,,
   \\
\text{Exception 2 : } \,   &    c(2 \epsilon )^2 c(2 \epsilon -9) c(2 \epsilon -8)^2 c(2 \epsilon
-7)^2 c(2 \epsilon -6)^3 c(2 \epsilon -5)^3  \  \times \\ &  \  \times \  c(2 \epsilon -4)^3 c(2
\epsilon -3)^3 c(2 \epsilon -2)^3 c(2 \epsilon -1)^3 c(2 \epsilon
+1)^2  \  \times \\ &  \  \times \  c(2 \epsilon +2)  c(2 \epsilon +3) c(2 \epsilon +4) c(4 \epsilon
-7) \,,
   \\
  \text{Exception 3 : } \,   &    c(2 (\epsilon -5)) c(2 \epsilon )^2 c(2 \epsilon -11) c(2 \epsilon
-9) c(2 \epsilon -8)^2 c(2 \epsilon -7)^2  \  \times \\ &  \  \times \,   c(2 \epsilon -6)^3 c(2
\epsilon -5)^3 c(2 \epsilon -4)^3  c(2 \epsilon -3)^3 c(2 \epsilon
-2)^3  \  \times \\ & \   \times \,  c(2 \epsilon -1)^3 c(2 \epsilon +1)^2 c(2 \epsilon +2) c(2
\epsilon +3) c(2 \epsilon +4) c(4 \epsilon -7)  \,.
    \endaligned
\end{equation}
Each of these terms is in fact zero by dint of the triple zero from the term $c(2\e-1)^3$ counterbalancing the double pole from the term $c(2\e+1)^2$ at $\epsilon=0$.
(Incidentally, the overall series $E^{G}_{\a_1;5/2}$ was shown to be non-zero in \cite{Green:2010kv} for both $G=E_7$ and $G=E_8$).

For $E_8$ there are 2160 elements in ${\mathcal S}$ out of the 696,729,600 elements of the full Weyl group $\Omega$.  Likewise,  all but 258 of these 2160 $w$ satisfy (\ref{Langlandscritera}).  Again, all 258 of these terms vanish at $s=5/2$ because their products have a triple zero (coming from three $c(s)$ factors evaluated at near $s=-1$) that counterbalance two poles (coming from two $c(s)$ factors evaluated near $s=1$).

\end{proof}

\section{Discussion and future problems}

In this paper we have studied the Fourier modes of the Eisenstein  series that define
the coefficients of the  first two nontrivial interactions in the low
energy expansion of the four-graviton amplitude in maximally
supersymmetric string theory compactified on $\rT^d$, and verified they  have certain expected features.  In
particular, we have shown that their non-zero Fourier coefficients
contain the expected  minimal and next-to-minimal  ($\smallf 12$-BPS
and $\smallf 14$-BPS) instanton orbits for any  of the
symmetry groups, $E_{d+1}$ ($0\le d \le 7$).  This extends the
analysis of these functions in \cite{Green:2010kv}, where the constant
terms of these functions were shown to reproduce all the expected
features of string perturbation theory and semi-classical M-theory.
Furthermore, in low rank cases we were able to present the explicit
Fourier coefficients of these functions and show that they have the
form expected of BPS-instanton contributions.  Indeed, the form of
the $\smallf 12$-BPS contributions match those deduced from string
theory calculations as summarised by \eqref{e:AsympBPS}.

For high rank cases  this involved  a detailed analysis of the automorphic
representations connected to these coefficients.  Namely, we explained
that they are automorphic realizations of the smallest two types of
nontrivial representations of their ambient Lie groups, and why this
property automatically implies the vanishing of a slew of Fourier
coefficients -- precisely the Fourier coefficients that  the BPS
condition ought to force to vanish.  We furthermore showed the most
interesting cases -- those of the next-to-minimal representation for
$E_7$ and $E_8$  -- occur in $L^2(E_{d+1}(\Z)\backslash
E_{d+1}(\IR))$.

This raises some obviously interesting questions, both from the string theory perspective and from the mathematical perspective.

An immediately interesting mathematical direction would be the explicit computation of the non-zero Fourier modes of $\calE^{(D)}_{(0,0)}$  and $\calE^{(D)}_{(1,0)}$ for the high rank cases with groups $E_6$, $E_7$ and $E_8$, in particular to get finer information  using the work of Bhargava and Krutelevich on the integral structure of the character variety orbits.
In a different direction, as mentioned in section~\ref{generalfeatures}  it would be of interest to extend the considerations of this paper to affine $E_9$ and behind that to hyperbolic extensions.\footnote{After this paper was first posted on the arXiv  the paper \cite{Fleig:2012xa} by Fleig and Kleinschmidt appeared, which makes important steps in this direction.
 }

Another question that is natural to ask in the context of string theory is  to what extent does our analysis generalise to higher
order interactions in the low energy expansion, which preserve a
smaller fraction of supersymmetry?  Could there be a role for Eisenstein
series with other special values of the index $s$ in the description of such terms?
However, the evidence is that such higher order terms involve
automorphic functions that are not Eisenstein series.
For example,  $\calE_{(0,1)}^{(D)}$  (the coefficient of the $\smallf 18$-BPS $\partial^6\, \R^4$ interaction) is expected to satisfy a  particular inhomogeneous Laplace eigenvalue equation
\cite{Green:2005ba}.  Although its constant term has, to a large
extent,  been analysed  for the relevant values of $D$
\cite{Green:2010kv}, it would be most interesting to analyse the
non-zero Fourier modes of $\calE_{(0,1)}^{(D)}$, which should describe
the couplings of $\smallf 18$-BPS instantons in the four-graviton amplitude
for low enough dimensions, $D$.  This should reveal a rich structure.
For example, the instantons that contribute in the limit of
decompactification from $D$ to $D+1$ include the $\smallf 18$-BPS  black
holes of $D+1$ dimensions, which  can have non-zero horizon size and
exponential degeneracy.  It is not apparent at first sight
whether  this degeneracy should   be encoded in the solutions of the
inhomogeneous equation satisfied by $\calE^{(D)}_{(0,1)}$.   Indeed, we have seen in the $\smallf 14$-BPS cases that the Fourier expansion of the coefficient function $\calE_{(1,0)}^{(D)}$ in the decompactification limit does not determine the Hagedorn-like degeneracy of   $\smallf 14$-BPS  small black holes in $D+1$ dimensions.  Rather, the divisor sums weight particular combinations of charges and windings of the wrapped world-lines of such objects.

  These issues involve mathematical challenges.  For example, the study of inhomogeneous Laplace equations for the group $SL(2,\IR)$ heavily relies on explicit formulas for automorphic Green functions, which do not generalize in an obvious manner to higher rank groups because they involve automorphic Laplace eigenfunctions   which do not have moderate growth in the cusps (at present the existence of such functions is itself an open problem).

Another issue is to what extent this  analysis can be extended to
discuss the automorphic properties of yet higher order terms in the
expansion of the four-graviton amplitude.  Further afield are issues concerning the extension of
these ideas to multi-particle amplitudes, to amplitudes that transform
as modular forms of non-zero weight, and extensions to processes
with less supersymmetry.

\section*{Acknowledgements}
\label{sec:acknowledgements}

We are grateful to Jeffrey Adams, Ling Bao, Iosif Bena, Manjul Bhargava, Dan Ciubotaru, Nick Dorey, Howard Garland, David Kazhdan,   Axel Kleinschmidt, Laurent
Lafforgue, Peter Littelmann, Dragan Milicic, Andrew Neitzke, Daniel Persson,
 Boris Pioline,  Gerhard
R\"ohrle, Siddhartha Sahi, Simon Salamon, Gordan
Savin, Wilfried Schmid,  Freydoon Shahidi, Sheer El-Showk, and Peter Trapa for
enlightening conversations.

MBG is grateful for the support of European Research Council  Advanced Grant No. 247252, and  to the Aspen Center for Physics for support under NSF grant \#1066293. SDM is grateful for the support of NSF grant \#DMS-0901594.   DC is partially supported by NSF-DMS 0968065 and NSA-AMS 081022. PT is partially supported by NSF-DMS 0968275.

\break
\appendix

\section{Special unipotent representations, \newline by Dan Ciubotaru and Peter E.~Trapa}\label{sec:trapendix}

\begin{center}

Department of Mathematics

University of Utah

 Salt Lake City, UT 84112-0090

{\tt ciubo@math.utah.edu},
{\tt ptrapa@math.utah.edu}

\end{center}

The representations considered in Theorem \ref{mainthm} are examples of a
wider class of representations which have attracted intense attention in the
mathematical literature.  The purpose of this appendix is to recall certain
results (from a purely local point of view) which are especially
relevant for the discussion of Section \ref{sec:NTMdetails}.

To begin, let $G$ denote a real reductive group arising as the real
points of a connected complex algebraic group $G_\mathbb{C}$.  In
\cite{arthur.old} and \cite{arthur.conj}, Arthur set forth a conjectural description of
irreducible (unitary) representations contributing to the
automorphic spectrum of $G$.  In many cases, these conjectures could
be reduced to a fundamental set of representations attached to
(integral) ``special unipotent'' parameters.  In the real case, Arthur's
conjectures --- and, in particular, the definition of the
corresponding special unipotent representations --- are made precise
and refined in the work of Barbasch-Vogan \cite{barbasch.vogan} and,
more completely, in the work of Adams-Barbasch-Vogan
\cite{adams.barbasch.vogan}.  The perspective of these references is
entirely local.  (Of course an extensive literature approaching  Arthur's
conjectures by global methods exists and, for classical groups, is
summarized in \cite{arthur.classical}.)  As we now explain, the
representations appearing in Theorem \ref{mainthm} are indeed special
unipotent in the sense of Adams-Barbasch-Vogan.

Write $\mathfrak{g}_\mathbb{C}$ for the Lie algebra of $G_\mathbb{C}$
and fix a Cartan subalgebra $\mathfrak{h}_\mathbb{C}$ arising as the
Lie algebra of a maximal torus in $G_\mathbb{C}$.  Write $\Omega$ for the
Weyl group of $\mathfrak{h}_\mathbb{C}$ in $\mathfrak{g}_\mathbb{C}$.
The classification of connected reductive algebraic groups naturally
leads from $G_\mathbb{C}$ to the Langlands dual $G_\mathbb{C}^\vee$, a
connected reductive complex algebraic group, e.g.~\cite{springer}. Let
$\mathfrak{g}_\mathbb{C}^\vee$ denote the Lie algebra of
$G_\mathbb{C}$.  The construction of $G_\mathbb{C}^\vee$ includes the
definition of a Cartan subalgebra $\mathfrak{h}^\vee_\mathbb{C}$ which
canonically identifies with the linear dual of
$\mathfrak{h}_\mathbb{C}$,
\begin{equation}
\label{e:hstar}
\mathfrak{h}^\vee_\mathbb{C} \simeq (\mathfrak{h}_\mathbb{C})^*.
\end{equation}
Let $\mathcal{N}$ denote the cone of nilpotent elements in
$\mathfrak{g}_\mathbb{C}$, and
likewise let $\mathcal{N^\vee}$ denote the cone of nilpotent elements in
$\mathfrak{g}_\mathbb{C}^\vee$. Write $G_\mathbb{C} \backslash \mathcal{N}$ and
$G_\mathbb{C}^\vee \backslash \mathcal{N}^\vee$ for the corresponding sets
of adjoint orbits.  These sets are partially ordered by the inclusion of
closures.  Spaltenstein defined an order-reversing
map
\[
d \; : \; G_\mathbb{C}^\vee \backslash \mathcal{N}^\vee
\longrightarrow G_\mathbb{C} \backslash \mathcal{N}
\]
with many remarkable properties which were refined in \cite[Appendix]{barbasch.vogan}; see Theorem \ref{t:BVav} below.

\begin{example}
\label{e:spaltenstein}
Suppose the Dynkin diagram corresponding to
$\mathfrak{g}_\mathbb{C}$ is simply
laced (as is the case for the groups $E_{d+1}$ from figure~\ref{fig:dynkin} and table~\ref{tab:Udual}).  Then $\mathfrak{g}_\mathbb{C} \simeq \mathfrak{g}^\vee_\mathbb{C}$
and $G_\mathbb{C}^\vee$ and $G_\mathbb{C}$ are isogenous.  Thus
$G_\mathbb{C}^\vee \backslash \mathcal{N}^\vee$ can be identified with
$G_\mathbb{C} \backslash \mathcal{N}$ and $d$ can be viewed as an order
reversing map from the latter set to itself.
With this in mind, consider figure \ref{fig:topandbottomorbit}.
The map $d$ interchanges the top three orbits
with the bottom three orbits (in an order
reversing way, of course).   In particular $d$ applied to the sub-subregular
orbit is the next to minimal orbit.  The
complete calculation of $d$ is given in \cite{carter}.
\end{example}

Fix an element $\mathcal{O}^\vee$ of $G_\mathbb{C}^\vee \backslash
\mathcal{N}^\vee$.  According to the Jacobson-Morozov Theorem, there exists a Lie algebra homomorphism
\[
\phi \; : \; \mathfrak{s}\mathfrak{l}(2,\mathbb{C}) \longrightarrow
\mathfrak{g}^\vee_\mathbb{C}
\]
such that the image under of $\phi$ of
$\left ( \begin{matrix}0 & 1\\ 0 & 0\end{matrix} \right )$
lies in $\mathcal{O}^\vee$ and
\begin{equation}
\label{e:lamprovisional}
\phi \left ( \begin{matrix}1 & 0\\ 0 & -1
\end{matrix} \right ) \in
\mathfrak{h}_\mathbb{C}^\vee \simeq \mathfrak{h}_\mathbb{C}^*,
\end{equation}
with the last isomorphism as in \eqref{e:hstar}.

The element in \eqref{e:lamprovisional} depends on the choice of $\phi$. Its
Weyl group orbit is well-defined however
(independent of how $\phi$ is chosen).  So define
\begin{equation}
\label{e:lam}
\lambda(\mathcal{O}^\vee) := (1/2) \;
\phi \left ( \begin{matrix}1 & 0\\ 0 & -1
\end{matrix} \right )
\in \mathfrak{h}_\mathbb{C}^*/\Omega.
\end{equation}
According to the Harish-Chandra isomorphism, $\lambda(\mathcal{O}^\vee)$
specifies a maximal ideal $Z(\mathcal{O}^\vee)$
in the center of the enveloping algebra $U(\mathfrak{g}_\mathbb{C})$.  Recall
that an irreducible admissible representation of $G$ is said to have
infinitesimal character $\lambda(\mathcal{O}^\vee)$ if its Harish-Chandra
module is annihilated by $Z(\mathcal{O}^\vee)$.

A result of Dixmier implies that there is a unique primitive ideal
$I(\mathcal{O}^\vee)$ in $U(\mathfrak{g}_\mathbb{C})$ which is
maximal among all primitive ideals containing $Z(\mathcal{O}^\vee)$.
(A primitive ideal in $U(\mathfrak{g}_\mathbb{C})$ is, by
definition, a two-sided ideal which arises as the annihilator of a simple
$U(\mathfrak{g}_\mathbb{C})$ module.)
Given any two-sided ideal $I$ in $U(\mathfrak{g}_\mathbb{C})$, we can
consider the associated graded ideal $gr(I)$ with respect to the canonical
grading on $U(\mathfrak{g}_\mathbb{C})$.  According to the
Poincar\'e-Birkhoff-Witt Theorem, $gr(I)$ is an ideal in
$gr(U(\mathfrak{g}_\mathbb{C})) \simeq S(\mathfrak{g}_\mathbb{C})$, the
symmetric algebra of $\mathfrak{g}_\mathbb{C}$, and hence cuts out a
subvariety (the so-called associated variety, $\AV(I)$, of $I$) of
$\mathfrak{g}_\mathbb{C}^*$.

It will be convenient to identify $\mathfrak{g}_\mathbb{C}$ with
$\mathfrak{g}^*_\mathbb{C}$ (by means of the choice of an invariant form)
and view $\AV(I)$ as a subvariety of $\mathfrak{g}_\mathbb{C}$.  (The choice
of form is well-defined up to scalar; since $\AV(I)$ is a cone, $\AV(I)$
becomes a well-defined subvariety of $\mathfrak{g}_\mathbb{C}$.) A theorem
of Joseph \cite{joseph} and Borho-Brylinski \cite{bb:i} (cf.~the short proof in \cite{vogan.av}) implies that if
$I$ is primitive, $\AV(I)$ is indeed the closure of a single nilpotent orbit
of $G_\mathbb{C}$.

\begin{thm}[\!\!{\cite[Corollary A.3]{barbasch.vogan}}]
\label{t:BVav}
In the setting of the previous paragraph,
\[
\AV(I(\mathcal{O}^\vee))
= \overline{d(\mathcal{O}^\vee)}.
\]
\end{thm}

\begin{example}
\label{e:spaltenstein2}
Suppose $G_\mathbb{C}$ is simply laced and make identifications as in
Example \ref{e:spaltenstein}.  Suppose $\mathcal O^\vee$ is respectively the
regular, subregular, or sub-subregular, orbit in figure
\ref{fig:topandbottomorbit}.  Then $\AV(I(\mathcal{O}^\vee))$ is the
closure respectively of the zero, minimal, or next-to-minimal orbit.
\end{example}

\begin{defn}[{Barbasch-Vogan \cite{barbasch.vogan}}]
\label{d:su}
Fix an orbit $\mathcal{O}^\vee$ as above.
Suppose further that $\mathcal{O}^\vee$
is even or, equivalently, that $\lambda(\mathcal{O}^\vee)$  is integral.
An irreducible admissible
representation of $G$ is said to be {\sl (integral) special unipotent attached to
$\mathcal{O}^\vee$} if the
annihilator of its Harish-Chandra module is $I(\mathcal{O}^\vee)$.
\end{defn}

Note that since $I(\mathcal{O}^\vee)$ is a maximal primitive ideal, special
unipotent representations are, in a precise sense, as small as possible.

\begin{thm}
\label{t:su}
Suppose $G$ is split and $\pi$ is an irreducible spherical representation
with infinitesimal character $\lambda(\mathcal{O}^\vee)$
(with notation as
in \eqref{e:lam}).  Suppose further that $\mathcal{O}^\vee$
is even.
Then $\pi$ is
special unipotent in the sense of Definition \ref{d:su}.
\end{thm}

\begin{proof}[Sketch]
Chapter 27 in \cite{adams.barbasch.vogan} defines special unipotent Arthur packets.
Roughly speaking, such a packet is parametrized by a rational form of an orbit
$\mathcal{O}^\vee$ in $G_\mathbb{C}^\vee \backslash \mathcal{N}^\vee$
(\!\!\cite[Theorem 27.10]{adams.barbasch.vogan}).  In the case that $\mathcal O^\vee$ is even,
these packets are known  to consist of representations appearing in Definition \ref{d:su}
(\!\!\cite[Corollary 27.13]{adams.barbasch.vogan}).
As a consequence of
\cite[Definition 22.6]{adams.barbasch.vogan} (see also the discussion after
\cite[Definition 1.33]{adams.barbasch.vogan}),
such a packet also contains a (generally nontempered) L-packet.
In the case at hand, the special unipotent Arthur packet parametrize
by $\mathcal{O}^\vee$ contains the L-packet consisting of the spherical
representation with infinitesimal character $\lambda(\mathcal{O}^\vee)$.
This completes the sketch.
\end{proof}

\begin{cor}
\label{c:su}
The spherical constituents of the principal series representations $V_{\l_{\text{dom}}}$ from section~\ref{sec:NTMdetails} are integral
special unipotent attached to $\mathcal O^\vee$ (Definition \ref{d:su})
where $\mathcal O^\vee$ is, respectively, the regular, subregular, and
sub-subregular nilpotent orbit (all of which are even).  According to
Corollary \ref{t:BVav} and Example \ref{e:spaltenstein2},
the wavefront sets of these representations are, respectively, the
zero, minimal, and next to minimal orbits.
\end{cor}

\smallskip

Finally, we remark that since the special unipotent representation of
Definition \ref{d:su} are predicted by Arthur to appear in spaces of
automorphic forms, they should be unitary.

\begin{conj}
\label{conj:su}
Suppose $\pi$ is integral special unipotent in the sense of Definition \ref{d:su}.
Then $\pi$ is unitary.
\end{conj}

The representations appearing in Theorem \ref{t:su} are known
to be unitary if $G_\mathbb{C}$ is classical or of Type $G_2$.
This was proved by purely
local methods in \cite{vogan.gln},
\cite{vogan.g2}, and \cite{barbasch.spherical}.  For a summary of
results obtained by global methods, see \cite{arthur.classical}.

For completeness, we discuss the analogs of these results in the
$p$-adic case. Let $F$ be a $p$-adic field, with ring of integers
$\mathfrak O$, and finite residue field $F_q.$ The group $G$ is now
the $F$-points of a connected algebraic group $G_{\overline F}$
defined over $\overline F$. We assume for simplicity that $G$ is split
and of adjoint type. Let $K$ be the $\mathfrak O$-points of
$G_{\overline F},$ a maximal compact open subgroup of $G.$ Let $I$ be
the inverse image in $K$ under the natural projection $K\to
G_{\overline F}(F_q)$ of a Borel subgroup over $F_q$. The compact open
subgroup $I$ is called an Iwahori subgroup.

The Iwahori-Hecke algebra $\mathcal H(G,I)$ is the convolution algebra (with respect to a fixed Haar measure on $G$) of compactly supported, locally constant, $I$-biinvariant complex functions on $G$. It is a Hilbert algebra, in the sense of Dixmier, with respect to the trace function $f\mapsto f(1)$, and the $*$-operation $f^*(g)=\overline{f(g^{-1})}$, $f\in \mathcal H(G,I)$. Thus, there is a theory of unitary remodules of $\mathcal H(G,I)$ and an abstract Plancherel formula.

If $(\pi,V)$ is a complex smooth $G$-representation, such that $V^I\neq 0,$ the algebra $\mathcal H(G,I)$ acts on $V^I$ via
$$\pi(f)v=\int_G f(x) \pi(x)v~dx,\quad v\in V^I,\ f\in \mathcal H(G,I).$$

\begin{thm}[\!\!\cite{borel}]
The functor $V\to V^I$ is an equivalence of categories between the category of smooth admissible $G$-representations and finite dimensional $\mathcal H(G,I)$-modules
\end{thm}
Borel  conjectured that this functor induces a bijective correspondence of unitary representations. This conjecture was proved by Barbasch-Moy \cite{barbasch.moy} (subject to a certain technical assumption which was later removed).

\begin{thm}[\!\!\cite{barbasch.moy}]\label{t:bm}
An irreducible smooth $G$-representation $(\pi,V)$ is unitary if and only if $V^I$ is a unitary $\mathcal H(G,I)$-module.
\end{thm}

The algebra $\mathcal H(G,I)$ contains the finite Hecke algebra $\mathcal H(K,I)$ of functions whose support is in $K$. Under the functor $\eta$, $K$-spherical representations of $G$ correspond to spherical $\mathcal H(G,I)$-modules, i.e., modules whose restriction to $\mathcal H(K,I)$ contains the trivial representation of $\mathcal H(K,I)$.

The classification of simple $\mathcal H(G,I)$-modules is given by Kazhdan-Lusztig \cite{KL}.

\begin{thm}[\!\!\cite{KL}] The simple $\mathcal H(G,I)$-modules are parameterized by $G_{\mathbb C}^\vee$-conjugacy classes of triples $(s^\vee,e^\vee,\psi^\vee)$, where:
\begin{enumerate}
\item[(i)] $s^\vee\in G_{\mathbb C}^\vee$ is semisimple;
\item[(ii)] $e^\vee\in \mathcal N^\vee$ such that $Ad(s) e=qe$;
\item[(iii)] $\psi^\vee$ is an irreducible representation of Springer type of the group of components of the mutual centralizer $Z_{G_{\mathbb C}^\vee}(s^\vee,e^\vee)$ of $s^\vee$ and $e^\vee$ in $G_{\mathbb C}^\vee.$
\end{enumerate}
\end{thm}
Let $\pi(s^\vee,e^\vee,\psi^\vee)$ denote the simple $\mathcal H(G,I)$-module parametrized by $[(s^\vee,e^\vee,\psi^\vee)]$.

\begin{example}\label{ex:kl}
In the Kazhdan-Lusztig parametrization, the simple spherical $\mathcal
H(G,I)$-modules correspond to the classes of triples
$[(s^\vee,0,1)]$. Here $s^\vee$ is the Satake parameter of the
corresponding irreducible spherical $G$-representation. On the other
hand, let $\mathcal O^\vee$ be a fixed $G_{\mathbb C}^\vee$-orbit in
$\mathcal N^\vee$, and set $s^\vee_{\mathcal
  O^\vee}=q^{\lambda_0(\mathcal O^\vee)}$ where $\lambda_0(\mathcal O^\vee)$
is any choice of representative of the element in \eqref{e:lam}.
If $e^\vee_0$ belongs to the
unique open dense orbit of $Z_{G_{\mathbb C}^\vee}(s^\vee)$ on
$\mathfrak g_q^\vee=\{x\in \mathfrak g_q^\vee: Ad(s^\vee)x=qx\}$ (in
particular $e^\vee_0\in \mathcal O^\vee)$, then the simple $\mathcal
H(G,I)$-module (and the corresponding irreducible $G$-representation)
parametrized by $[(s^\vee_{\mathcal O^\vee},e^\vee_0,\psi^\vee)]$ is
tempered.
\end{example}

The Iwahori-Hecke algebra has an algebra involution $\tau$, called the Iwahori-Matsumoto involution, defined on the generators as in \cite{IM}. It induces an involution on the set of simple $\mathcal H(G,I)$-modules, which is easily seen to map unitary modules to unitary modules. The effect of $\tau$ on the set of Kazhdan-Lusztig parameters is given by a Fourier transform of perverse sheaves \cite{EM}, and therefore it is hard to compute effectively in general, except in type $A$ \cite{MW}. (For a general algorithm, see \cite{lusztig}.) However, it is easy to see that if $\pi(s^\vee_{\mathcal O^\vee},0,1)$ is a simple spherical $\mathcal H(G,I)$-module, then
\begin{equation}
\tau(\pi(s^\vee_{\mathcal O^\vee},0,1))=\pi(s^\vee_{\mathcal O^\vee},e^\vee_0,1),
\end{equation}
where the notation is as in Example \ref{ex:kl}. In particular, $\pi(s^\vee_{\mathcal O^\vee},0,1)$ is unitary. Together with Theorem \ref{t:bm}, this gives the following corollary (cf.~Conjecture \ref{conj:su}).

\begin{cor}
If $\pi$ is an irreducible spherical $G$-representation with Satake parameter $s^\vee_{\mathcal O^\vee}\in G_{\mathbb C}^\vee$, then $\pi$ is unitary.
\end{cor}

\section{Supersymmetry and instantons}
\label{susyinst}

The constraints of maximal supersymmetry are efficiently described by starting with the superalgebra generated by the
32-component Majorana spinor supercharge, $Q_\alpha = \int J^0_\alpha d^{10}x$, where $J^I_\alpha$ is the supercurrent (with spinor index $\alpha,\beta =1,\dots ,32$ and vector index $I=0,1, \dots, 10$).
This satisfies the anti-commutation relations,
 \be
\{Q_\alpha\,,  Q_\beta\} = P_{I_1} \left(\Gamma^0\Gamma^{I_1}\right)_{\alpha\beta} +
 Z_{\alpha\beta}
 \label{node2supalgebra}
 \ee
 where the central charge is
 \be
Z_{\alpha\beta}= Z_{I_1I_2}\, \left(\Gamma^0 \Gamma^{I_1
I_2}\right)_{\alpha\beta} + Z_{I_1\cdots I_5}\, \left(\Gamma^0\Gamma^{I_1\cdots I_5} \right)_{\alpha\beta}  \,,
  \label{centralcharge}
 \ee
 where $\Gamma^{I}_{\alpha\beta}$  are $SO(1,10)$  Dirac
 matrices\footnote{$\Gamma^{I_1\cdots I_r}_{\alpha\beta}$ is the
   antisymmetrized product  of $r$ Gamma matrices normalised so that  $\Gamma^{1\cdots r} = \Gamma^1 \cdots \Gamma^r$.}
and $P_{I}$  is the eleven-dimensional translation operator.

  \subsection{BPS particle states}
Positivity of the anticommutator in  \eqref{node2supalgebra}  leads to
the Bogomol'nyi bound that restricts the masses of states to be larger
than or equal to the central charge.  States saturating the bound are
BPS states that form supermultiplets, the  lengths of which depend on
the fraction of supersymmetry broken by their presence.    The
shortest multiplets are $\smallf 12$-BPS, with longer multiplets for
smaller fractions.   We refer, for instance, to~\cite{Duff:1994an,Polchinski:1996na,Polchinski:1998rr} for extensive discussions of the properties of supersymmetric branes in string theory.

 The presence of the 2-form component of the central charge indicates
 that the theory contains a membrane-like state (the $M2$-brane)
 carrying a conserved charge $Q^{(2)}$, while the 5-form component
 indicates the presence of a 5-brane state (the $M5$-brane) carrying a
 charge $Q^{(5)}$.  The 2-form and 5--form in \eqref{node2supalgebra}
 are given  by integration of  the spatial directions of the $M2$ and
 $M5$ branes over 2-cycles $A_{I_1I_2}$ or 5-cycles $A_{I_1\cdots I_5}$,
\be
Z_{I_1I_2} = Q^{(2)} \int_{A_{I_1I_2}}  \, d^2X\,, \quad
Z_{I_1\cdots I_5} = Q^{(5)}  \int_{A_{I_1\cdots I_5}} \,d^5X\,.
\label{twoformcharge}
\ee
The $M2$ and $M5$-branes are $\smallf 12$-BPS states that preserve 16 of the 32  components of supersymmetry.
The 2-form charge couples to a 3-form potential ($C^{(3)}_{I_1I_2I_3}$), with field strength $H^{(4)}=d C^{(3)}$.
This is analogous to the manner in which the Maxwell 1-form potential couples to a point-like electric charge (a $0$-brane),  and $H^{(4)}$  is the analogue of the Maxwell field.
The analogue of the dual Maxwell-field is a 7-form field-strength, which is required by consistency with supersymmetry to take the form
that $H^{(7)}=d C^{(6)} + C^{(3)} \wedge d C^{(3)}$,  where $C^{(6)}$ is the 6-form potential that couples to the $M5$-brane.
 In other words, the $M5$-brane couples to the magnetic charge that is dual to the electric charge carried by the $M2$-brane.
 The BPS condition implies that the charge on the brane is equal to its tension, $T^{(r)}$,
 \be
 Q^{(r)} = T^{(r)}\,.
 \label{tencharge}
 \ee

The integrals in \eqref{twoformcharge} are well-defined when all the
spatial directions of the branes are wound around the compact cycles
of the M-theory  torus, $\calT^{d+1}$, in which case the state is
point-like from the point of view of the $D=10-d$ non-compact
dimensions (so there are finite-mass point-like states due to wrapped
$M2$-branes when $d\ge 1$ as well as wrapped $M5$-branes when $d\ge
4$).\footnote{There is a huge literature of far more elaborate windings
  of such branes around supersymmetric cycles in curved manifolds, in
  which case a fraction of the supersymmetry may or may not be
  preserved.}   Other kinds of  $\smallf 12$-BPS states also arise in the toroidal background, such as point-like Kaluza--Klein ($KK$) charges, which are modes of the metric that contribute for any $d\ge 0$.  The magnetic dual of a $KK$ state is a $KKM$, which is described by a Taub-NUT geometry in four spatial dimensions, leaving six more spatial dimensions that are interpreted as the directions on a six-brane.  This has a finite mass when wrapped around $\calT^6$, so it can arise when $d\ge 5$.

The complete spectrum of BPS states in an arbitrary toroidal
compactification of type IIA or IIB  string theory can be deduced by
considering the toroidal compactification of the M-theory algebra
\eqref{node2supalgebra}  with appropriate rescalings of the moduli
\cite{Townsend:1997wg}.   Combining completely wrapped branes in
various combinations leads to point-like $\smallf 12$-, $\smallf 14$- and $\smallf 18$-BPS states that  are of importance in discussing the spectrum of black holes in string theory
\cite{Strominger:1996sh,Callan:1996dv}.  This spectrum is of significance in classifying the orbits of instantons that decompactify to black hole states in one higher dimension associated with the parabolic subgroup $P_{\alpha_{d+1}}$.
This will be sketched in the next subsection where we will make contact with the discussion of black hole orbits in~\cite{Ferrara:1997ci,Ferrara:1997uz,Lu:1997bg}.

\section{Orbits of BPS instantons in the decompactification limit}
\label{orbit1appendix}

 A finite action instanton in $D=10-d$ dimensions
 corresponds to an embedded euclidean world-volume that can be one of three types:
 \begin{itemize}
 \item[(a)] It has an action that does not depend on $r_d$ as $r_d
 \to \infty$ and so is also an instanton of  the $(D+1)$-dimensional
 theory -- this contributes only to the constant term in this
 parabolic and does not appear in non-zero Fourier modes;
 \item[(b)] It is  a
 euclidean world-line of  a $(D+1)$-dimensional point-like BPS black
 hole with mass $M_{BH}$, which gives a term suppressed by a factor of
 $e^{-2\pi\, r_d\, M_{BH}}$  in the amplitude in the limit
 $r_d/\ell_{D+1}\to \infty$;
 \item[(c)] It has an action that grows faster
 than $r_d/\ell_{D+1}$ so it does not decompactify to give either a
 particle state or an instanton in $D+1$ dimensions.
\end{itemize}

Thus,  the instantons of type (b) or (c) are the ones that contribute to the character variety orbits in limit (i), which is associated with the parabolic subgroup that has Levi factor  $GL(1) \times E_d(\IR)$ in $D=10-d$ dimensions,  where  the duality group is $E_{d+1}(\ZZ)$.

 In order to illustrate this pattern the following subsections summarise
 the spectrum of  $r_d$-dependent instantons (i.e.,  type (b) or (c)) in each dimension in
 the range $3\le D \le 10$ (i.e.,  $0\le d\le 7$).  Their orbits and the conditions on the charges corresponding to fractional BPS conditions are summarised in   table~\reftab{tab:bpsorbits1}.  Where appropriate we will also comment on the distinction between BPS states in dimension $D+1$ and BPS instantons in dimension $D$.


\subsection{BPS orbits in $D=10$}\label{sec:D10}\hfill\break
This degenerate case includes both $10A$ and $10B$.  Although the $10A$  theory does have a decompactification limit to 11-dimensional M-theory, it has no instantons and there is no duality symmetry group.  There are $\smallf 12$-BPS particle states in $10A$  consisting  of threshold bound states of  $D0$-branes that are manifested as instantons in the $D=9$ theory (as we will sketch in the next subsection).
There is no decompactification limit for the $10B$ theory.  In this case there are no BPS particle states but there is a  $\smallf 12$-BPS $D$-instanton, multiples of which only contribute to
amplitudes in the string perturbation limit.
There are no $\smallf 14$-BPS particle states in either $10A$ or $10B$.

\subsection{BPS instanton orbits in $D=9$}\label{sec:D9}\hfill\break
 This case may be obtained by considering M-theory on a
  2-torus, $\calT^2$, where the discrete  duality group  $SL(2,\ZZ)$ is  identified with the group of
  large diffeomorphisms of $\calT^2$.

   There  is a single type of
 BPS  instanton  that can be identified with the wrapping of  the euclidean world-line of a Kaluza--Klein state formed on one cycle around the second cycle of the 2-torus; in this sense we will refer in the following to a euclidean Kaluza--Klein state wrapping a 2-cycle on $\calT^2$.  Equivalently, this instanton can be described as a wrapped euclidean world-line of a $D0$-brane of the $10A$ string theory, which is the parameterisation manifested in  \eqref{tenawrap}.
 In this case the unipotent radical consists of $2\times 2$ upper triangular matrices with $1$'s on the diagonal, and so the one-dimensional $\smallf  12$-BPS orbit is simply
 \begin{eqnarray}
\mathcal O_{{\bf1}} = GL(1)\,.
\label{glone}
\end{eqnarray}

\subsection{BPS instanton orbits in $D=8$}\label{sec:D8}\hfill\break

 This case may be obtained by considering M-theory on a
  3-torus, $\calT^3$, where the discrete  duality group is $SL(3,\ZZ)\times
  SL(2,\ZZ)$.

There  is  one type of instanton charge from wrapping  the world-volume of the $M2$-brane around the whole of $\calT^3$.  In addition there are two types of instanton charges from  Kaluza--Klein states wrapping the 2-cycles that depend on the decompactification radius $r_2$ (a third  Kaluza--Klein state wraps the two-cycle that does not depend on $r_2$).  This  gives a total of 3 types of BPS instanton charges of type (b), which are parameterised in the same manner as the BPS particle states in $D=9$ dimensions  by a scalar $v$  and a $SL(2)$ vector $v_a$. The charges of the $\smallf 12$-BPS  states  are given  by  the  condition $v  \,
v_a=0$ and the $\smallf 14$-BPS states by  $v\, v_a\neq0$.

The $\smallf 12$-BPS instantons are those for which $vv_a=0$ \cite{Ferrara:1997ci}, giving the union of the orbits
  \begin{eqnarray}
\mathcal O_{{\bf1}} &=& GL(1)
\label{vzero}
\end{eqnarray}
 for $v_a=0$ and
  \begin{eqnarray}
\mathcal O_{{\bf 2}}&=& \f{SL(2)}{\IR}
\label{vzerova}
\end{eqnarray}
for $v=0$, arising from dense open orbits in each of the two
  factors of the duality group $SL(2)\times SL(3)$.
 The bold face subscript, in this example and in the following,  gives  the dimensions of the coset,
$\dim (\f{G_1}{G_2})=\dim(G_1)-\dim(G_2)$.
The $\smallf  14$-BPS instantons have charges satisfying  $vv_a\ne0$, giving the orbit
   \begin{eqnarray}
\mathcal O_{\bf 3}&=& {GL(1)\times SL(2)\over \mathbb R}   \,.
\label{nonvzerova}
\end{eqnarray}

\subsection{BPS instanton orbits in $D=7$}\label{sec:D7}\hfill\break
Consider M-theory on a 4-torus, $\calT^4$, with duality  group  $SL(5,\ZZ)$.

 There are  4 BPS types of instanton from euclidean  $M2$-branes wrapping 3-cycles, of which 3 depend on the decompactification radius $r_3$,
  and 6 types of instanton from the Kaluza--Klein states wrapping
  2-cycles, of which three depend on $r_3$.  This gives  a total of 10 types of BPS instanton charge, of which 6 depend on the decompactification radius $r_3$ and are of type (b).
These instantons carry charges associated with the corresponding BPS states in $D=8$ dimensions that
may be  parametrized by $v_{i\,  a}$ transforming  in the  ${\bf 3}\times
{\bf 2}$ of $SL(3)\times SL(2)$. The $\smallf 12$-BPS states are given by the
condition $\epsilon^{ab}   \,  v_{i\,a}  v_{j\,
  b}=0$~\cite{Ferrara:1997ci}   and the $\smallf 14$-BPS states by $\epsilon^{ab}   \,  v_{i\,a}  v_{j\,
  b}\neq0$.
 This determines  two BPS instanton orbits given in~\cite{Lu:1997bg} by
\begin{eqnarray}
  \label{e:E3nil}
  \tfrac12-BPS\quad     &:&\quad    \mathcal{O}_{\bf  4  }=\frac{SL(3,\IR)\times
  SL(2,\IR)}{GL( 2,\IR)\ltimes \IR^3}\,,\\
\tfrac14-BPS \quad &:&\quad \mathcal{O}_{\bf 6}=\frac{SL(3,\IR)\times
  SL(2,\IR)}{SL(2,\IR)\ltimes \IR^2}\,.
\end{eqnarray}

\subsection{BPS instanton orbits in $D=6$}\label{sec:D6}\hfill\break
Consider M-theory on a 5-torus, $\calT^5$, with duality group $Spin(5,5,\ZZ)$.

There are  10 ways of wrapping the  $M2$-brane  world-volume around 3-cycles, of which 6 depend on the decompactification radius $r_4$, and 10 ways of wrapping euclidean Kaluza--Klein states on 2-cycles, of which 4 depend on $r_4$.  This gives a total of  20 BPS instanton types of charge, of which 10 depend on $r_4$ (and so are of  type (b)).
These charges correspond to the charges of BPS states in $D=8$ dimensions and may be
parametrized by the rank-2 antisymmetric tensor $v_{[ij]}$ ($i,j = 1,\dots, 5$) that transforms in  the {\bf 10} of  $SL(5)$. The $\smallf 12$-BPS
states    are    given  in~\cite{Ferrara:1997ci}  by    the    condition
$\epsilon^{ijklm}\, v_{ij}\, v_{kl}=0$ and  the $\smallf 14$-BPS
by $\epsilon^{ijklm}\, v_{ij}\, v_{kl}\neq0$.
This determines two BPS instanton orbits given in~\cite{Lu:1997bg} by
\begin{eqnarray}
  \label{e:E4nila}
  \tfrac12-BPS\quad         &:&\quad         \mathcal{O}_{\bf        7
  }=\frac{SL(5,\IR)}{(SL(3,\IR)\times SL(2,\IR))\ltimes \IR^6}\,,\\
  \label{e:E4nilb}
\tfrac14-BPS \quad &:&\quad \mathcal{O}_{\bf 10}=\frac{SL(5,\IR)}{Spin(2,3)\ltimes \IR^4}\,.
\end{eqnarray}

\subsection{BPS instanton orbits in $D=5$}\label{sec:D5}\hfill\break
Consider  M-theory on a 6-torus, $\calT^6$,  with duality group $E_6(\ZZ)$.

There are  20 types of instanton  from the $M2$-brane world-volume wrapping 3-cycles, of which 10 depend on the decompactification radius, $r_5$;
15  types from Kaluza--Klein states wrapping
  2-cycles, of which 5 depend on $r_5$; 1 type of instanton from the world-volume of the  $M5$-brane world-volume wrapping
  the whole of $\calT^6$.  This gives a total of 36 BPS instanton charges, of which 16 depend on $r_b$ and are of type (b).

These  16 BPS  charges are parameterised by a chiral spinor $S^\alpha$ ($\alpha=1,\dots, 16)$ of
$Spin(5,5)$.   Such a spinor satisfies the identity   $\sum_{m=1}^{10}
(S\Gamma^mS) \times\break
(S\Gamma^mS) = 0$, where $\Gamma^m$ ($m=1,\dots,10)$ are Dirac matrices with suppressed spinor indices. The configurations  are $\smallf 12$-BPS  if $S$ satisfies the pure spinor condition,  $S\Gamma^mS=0$~\cite{Ferrara:1997ci}.  A standard way to analyse this condition is to  decompose $S$ into $U(5)$ representations,   ${\bf 16}={\bf 1}_{5}\oplus {\bf \bar 5}_{-3}\oplus{-\bf
 10}_1$ (where the subscripts denote the $U(1)$ charges), so it has components
\begin{equation}
  S= (s, v_a,v^{ab}), \qquad a,b=1,\dots, 5\,.
\end{equation}
The pure spinor
($\smallf 12$-BPS) condition, $S\Gamma^mS=0$ is
$v_a={s^{-1}\over  5!}\epsilon_{abcde} \, v^{bc}v^{de}$, which implies that the $\bf 5$ is not independent of the other $U(5)$ representations, so  the space of such spinors has  dimension   11.
     The    $\smallf 14$-BPS    solution  is the unconstrained spinor space (excluding $S\Gamma^mS=0$) and   has    dimension
 16.
 There are two BPS orbits given in~\cite{Lu:1997bg} by
\begin{eqnarray}
  \label{e:E5nila}
  \tfrac12-BPS\quad         &:&\quad         \mathcal{O}_{\bf      11
  }=\frac{Spin(5,5,\IR)}{SL(5,\IR)\ltimes \IR^{10}}\,,\\
 \label{e:E5nilb}
\tfrac14-BPS \quad &:&\quad \mathcal{O}_{\bf 16}=\frac{Spin(5,5,\IR)}{Spin(3,4)\ltimes \IR^8}\,.
\end{eqnarray}

\subsection{BPS instanton orbits in $D=4$}\label{sec:D4}\hfill\break
Consider M-theory on a 7-torus, $\calT^7$, with duality group $E_{7}(\ZZ)$.

There are  35 types of instanton charge from the  $M2$-brane world-volume wrapping 3-cycles, of which 15 depend on the decompactification radius $r_6$; 21 types of instanton charge from Kaluza--Klein states wrapping
  2-cycles, of which 6 depend on $r_6$; 7 types of instanton charge  from the $M5$-brane world-volume wrapping 6-cycles, of which 6 depend on $r_6$.  This gives a  total of 63 types of  BPS instanton charge, of which 27 depend on $r_6$.

The  distinct instanton charges are parameterised by
the fundamental representation, $q^i$ ($i=1,\dots, 27$),  of $E_6$ and lead
to $\smallf 12$-,  $\smallf 14$- or $\smallf 18$-BPS configurations depending on the following
conditions on the $E_6$ cubic invariant $I_3=\sum_{1\leq i,j,k\leq 27}(I_3)_{ijk} q^i q^j q^k$~\cite{Ferrara:1997ci}
\begin{eqnarray}
\label{e:I312} \tfrac12-BPS : & I_3&=0,\qquad {\partial I_3\over \partial q^i}=0, \ {\partial^2 I_3\over \partial q^i\partial q^j}\neq0\,,\\
\label{e:I314} \tfrac14-BPS : & I_3&=0, \qquad {\partial I_3\over \partial q^i}\neq0\,,\\
\label{e:I318}  \tfrac18-BPS : & I_3&\neq0\,.
\end{eqnarray}
Clearly the first of these conditions (the $\smallf 18$-BPS condition) is of dimension  27.
The other conditions may be analysed by decomposing the {\bf  27} of $E_6$
into $SO(5,5)\times U(1)$   irreducible representations,
$
  {\bf 27}= {\bf 1}_4\oplus{\bf 10}_{-2}\oplus {\bf 16}_1
$.
This  means that   $q^i$   decomposes as
\begin{equation}
 q^i = (s, v_m, S^\alpha)\,,
\end{equation}
 where $s$  is a scalar, $v_m$ is a $SO(5,5)$   vector of dimension  10   and $S^\alpha$
is a spinor of dimension   16  (the $U(1)$ charges have been suppressed).
The  cubic   invariant  $I_3$  decomposes
as $I_3= {\bf  10}_{-2}\otimes{\bf  10}_{-2}\otimes{\bf 1}_4\oplus  {\bf
  16}_1\otimes{\bf 16}_1\otimes{\bf 10}_{-2}$  ~\cite{Ferrara:1997ci}, which implies that
\begin{equation}
  I_3= s\, v\cdot v+ (S\Gamma S)\cdot v\,,
\end{equation}
where $v\cdot v$ is the $SO(5,5)$   (norm$)^2$ of the
vector $v$, and $(S\Gamma S)\cdot  v$ is the $SO(5,5)$ scalar product
between the vector $S\Gamma^mS$ and $v^m$.

The $\smallf 14$-BPS solution reduces to the condition
\begin{equation}
     s\, v\cdot v+ (S\Gamma S)\cdot v=0\,,
\end{equation}
with  non-vanishing derivative  with  respect  to  $s$,  $v_m$  and  $S_a$. Therefore the solution is given by the   26  dimensional space
\begin{equation}
  (q^i)_{\frac14-BPS}= (-(v\cdot v)^{-1}\, (S\Gamma S)\cdot v, v_m,S^\alpha)\,.
\end{equation}
The $\smallf 12$-BPS condition implies the following conditions
\begin{eqnarray}
  v\cdot v&=&0\,,\\
  (S\Gamma^m S)+s\, v^m&=&0\,,\\
  (S\Gamma^m)_a \, v^m&=&0\,,
\end{eqnarray}
which are  solved by $v^m=S\Gamma^mS$  (using the relation $(S\Gamma^mS)(S\Gamma^mS)=0$). The $\smallf 12$-BPS solution
is therefore given by the 17-dimensional solution
\begin{eqnarray}
(q^i)_{\frac12-BPS}= (s,S\Gamma^mS,S_a)\, .
\end{eqnarray}

To summarise,  in limit (i) the BPS instanton orbits in $D=4$ are given in~\cite{Lu:1997bg} by
\begin{eqnarray}
  \label{e:E6nila}
  \tfrac12-BPS\,\quad &:& \, \mathcal{O}_{\bf 17}= \frac{E_{6}}{ Spin(5,5)\ltimes
  \IR^{16}}\,,\\
  \label{e:E6nilb}
\tfrac14-BPS \quad  &:& \, \mathcal{O}_{\bf 26}=\frac{E_{6}}{Spin(4,5)\ltimes
\IR^{16}}\,, \ \ \ \ \ \text{and}
\\
  \label{e:E6nilc}
\tfrac18-BPS      \quad        &:&       \,       \mathcal{O}_{\bf
  27}= \frac{GL(1)\times E_{6}} {F_{4(4)}}\,.
\end{eqnarray}
The charges in the  $\smallf 14$-BPS  orbit can be generated by applying $E_6(\Z)$ transformations to a 2-charge instanton corresponding to  a null
vector in the 27 dimensional BPS state space.  The charges in the $\smallf 18$-BPS orbit  can be generated from a 3-charge instanton corresponding to space-like or time-like
vectors  with  $I_3\neq0$  in  the  27  dimensional  BPS  state
space (note that, unlike~\cite{Ferrara:1997uz} we have included  the
scale factor $GL(1)$ in the definition of the orbit, which is of dimension  27).
The last orbit of dimension  27  is the $\smallf 18$-BPS
  orbit of black hole states with $I_3\neq0$, and entropy
  proportional to $\sqrt{|I_3|}$.

\subsection{BPS instanton orbits in $D=3$}\label{sec:D3}\hfill\break
Consider M-theory on an 8-torus $\calT^8$ with duality group $E_8(\ZZ)$.

There are  56 types of instanton charge from  $M2$-brane world-volumes  wrapping 3-cycles, of which 21 depend on the decompactification radius, $r_7$;
28 types of instanton charge  from  Kaluza--Klein states wrapping 2-cycles, of which 7 depend on $r_7$;   28 types of instanton charge from $M5$-branes wrapping  6-cycles, of which 21 depend on $r_7$.
In addition there are 8 types of instantons that depend on $r_7$ due to $KKM$ world-volumes wrapping 8-cycles, which are distinguished by labelling which cycle corresponds to $x^\#$  (the fibre coordinate in \eqref{kkmonopole}). This gives  a total of 120  types of instanton charges, of
 which 57 depend on $r_7$.

The connection with the black hole states in $D=4$ dimensions is slightly subtle.  For one of the 8 $KKM$ instantons $x^\#$ is identified with the euclidean time dimension and gives  a vanishing contribution upon decompactification to $D=4$ dimensions (the large-$r_7$ limit), as discussed following \eqref{e:HalfOrbit}.  It is therefore of type (c) and  does not correspond to a black hole state in $D=4$ dimensions.
This accounts for the nonabelian, Heisenberg, entry in the unipotent radical for the parabolic subgroup, $GL(1)\times E_7$, of $E_8$.   The nonzero Fourier modes in limit (i) correspond to the 56 abelian components of the unipotent radical which  match the charges of BPS states in $D=4$.  These are in the fundamental representation, $q^i$ ($i=1,\dots, 56)$, of $E_7$.  The   $\smallf 12$-,   $\smallf 14$-   and   $\smallf 18$-BPS configurations  are  classified by  the  following  conditions on  the
quartic symmetric polynomial invariant $I_4$~\cite{Ferrara:1997ci,Kallosh:1996uy}
\begin{eqnarray}
\label{e:12}\tfrac12-BPS : &   I_4\ =&{\partial I_4\over \partial q^i}=\left.{\partial^2   I_4\over   \partial
    q^i\partial q^j}\right|_{Adj_{E_7}}=0,\quad {\partial^3   I_4\over   \partial
    q^i\partial q^j\partial q^k}\neq0\,,\\
\label{e:14}\tfrac14-BPS : & I_4\ =&0,\quad {\partial I_4\over \partial q^i}=0, \quad
\left.{\partial^2 I_4\over \partial q^i\partial q^j}\right|_{Adj_{E_7}}\neq0\,,\\
\label{e:18b}\tfrac18-BPS : & I_4\ =&0, \quad {\partial I_4\over \partial q^i}\neq0\,,\\
\label{e:18a}  \tfrac18-BPS : & I_4\ >& 0\,.
\end{eqnarray}
The following is a summary of the BPS orbits~\cite{Lu:1997bg,Ferrara:1997uz,Ferrara:1997ci}
\begin{eqnarray}
  \tfrac12-BPS\quad &:& \, \mathcal{O}_{\bf 28}= \frac{E_{7}}{E_{6(6)}\ltimes
  \IR^{27}}\,,    \label{e:E7nila}\\
\tfrac14-BPS\quad  &:& \,  \mathcal{O}_{\bf 45}= \frac{E_{7}}{Spin(5,6)\ltimes
(\IR^{32}\ltimes\IR)}   \label{e:E7nilb}\,,\\
\tfrac18-BPS \quad &:& \, \mathcal{O}_{\bf 55}= \frac{E_{7}}{F_{4(4)}\ltimes
\IR^{26}} \label{e:E7nilc}\,,\\
\tfrac18-BPS   \quad   &:&  \,   \mathcal{O}_{\bf
  56}= \frac{\IR^+\times E_{7}}{E_{6(2)}} \,.\label{e:E7nild}
\end{eqnarray}
The $\smallf 12$-BPS orbit  can be obtained by acting on a single
charge, the $\smallf 14$-BPS orbit can be obtained by acting on a
2-charge system, and the first $\smallf 18$-BPS (with dimension 55) has
zero entropy and can be obtained by acting on a 3-charge system. The
last orbit of dimension   56  is the $\smallf 18$-BPS orbit of black hole states with  $I_4>0$, which have  entropy proportional to $\sqrt{I_4}$; it can be obtained by acting on a  4-charge system  in  the  {\bf  56}
representation of $E_7$ as detailed in~\cite{Lu:1997bg}. We have  included the overall scale factor in
the definition of the orbit.
Another orbit of  dimension
$56$ is  $(\IR^-\times E_{7})/E_{6(2)}$ that has $I_4<0$ and does not correspond to a
 BPS solution at all~\cite{Ferrara:1997uz,Ferrara:1997ci}. All these charge orbits
 can  be   understood  in terms of the  superpositions of  branes   at  angles  and   constructed  from
 combinations of $(D0,D2,D4,D6)$~\cite{Balasubramanian:2006gi}.

Note the presence of the 33-dimensional nonabelian group in the stabilizer of ${\mathcal O}_{\bf 45}$.  It is a Heisenberg group isomorphic to the unipotent radical of the maximal parabolic subgroup $P_{\a_1}=L_{\a_1}U_{\a_1}$ of $E_7$.  This can be seen directly using the basepoint of this orbit given in \cite[\S 5.9.8]{Miller-Sahi}.  Different stabilizer groups of the same dimension have appeared in the physics literature listed.

%

\section{Euclidean $Dp$-brane instantons.}
\label{euclidDbrane}

We  here sketch the background to the analysis of the euclidean $Dp$-brane instanton configurations that contribute in the perturbative limit of string theory discussed in section~\ref{sec:BPSTorbits}, based on an analysis of supersymmetry conditions on the embeddings of world-sheets on the string theory torus $\rT^d$.  Contributions from wrapped NS5-brane world-sheets also arise for $d=6,7$ and $KK$ monopoles for $d=7$.

Wrapping a euclidean $Dp$-brane world-volume of either ten-dimensional type~II string theory on a $(p+1)$-cycle leads to an instanton in the transverse $\IR^{1,8-p}$ space-time.   This $\smallf 12$-BPS condition preserves a linear  combination of the supersymmetries  that act on the left-moving and right-moving modes of a closed superstring.  This leads to the following constraint on the  supersymmetry parameters,
\begin{equation}
  \label{e:QQ1}
  \tilde\varepsilon= \prod_{i=1}^{p+1} \Gamma^i\varepsilon
\end{equation}
where $\varepsilon$  and
$\tilde\varepsilon$ are  chiral sixteen-component $SO(1,9)$  spinors parameterizing the left- and right-moving  super symmetries and
$\Gamma^i$ are the usual $SO(1,9)$   Gamma matrices that
satisfy the Clifford algebra $\{\Gamma^i,\Gamma^j\}=-2\eta^{ij}$, where $\eta$ is the Minkowski metric
with signature $(-+\cdots +)$.

When compactifying on a $d$-torus space-time becomes $\IR^{1,9-d}\times
\rT^d$ and a $SO(1,9)$    spinor decomposes into a sum of bispinors, $\varepsilon= \hat\varepsilon\otimes \eta$,  where
$\hat\varepsilon$ is  a $SO(1,9-d)$  spinor and $\eta$ is a
$SO(d)$ spinor. The condition~(\ref{e:QQ1})  becomes a condition relating
$\eta$ and $\tilde\eta$.  T-duality transforms the $\Gamma$ matrices in~(\ref{e:QQ1}) by the action of the spin group  $SO(d,d)$,  $R^{-1}\prod_{i}\Gamma^i R$.  This, in general, transforms a wrapped $Dp$-brane into a $Dq$-brane so that the supersymmetry conditions
\begin{equation}
\label{e:QQ2}
\tilde\eta  = \prod_{i=1}^{q+1} \Gamma^i  \eta=\prod_{i=1}^{p+1} \Gamma^i \eta\,,
\end{equation}
are satisfied.
As remarked in~\cite{Berkooz:1996km}, this means the two
spinors $\prod_{i=1}^{q+1} \Gamma^i  \varepsilon$ and
$\prod_{i=1}^{p+1} \Gamma^i  \varepsilon$ must be in the same $Spin(d,d)$   orbit.

A  euclidean $Dp$-brane can be wrapped over cycles of a $d$-torus of dimension
$0\leq p+1\leq d$ with $p\equiv 0 \imod 2$ for type~IIA superstring theory
and $p\equiv 1 \imod 2$ for type~IIB.  These  instanton configurations fill out a
chiral spinor representation, $S_A$, of  dimension  $\sum_{p\equiv s\imod 2} {d \choose
  p+1}= 2^{d-1}$,   with $s=0$ or 1, of the T-duality group
$SO(d,d)$ . The BPS condition on  $Dp$-branes
wrapping  a torus in~(\ref{e:QQ2}) can be interpreted  as a condition
on the spinor $S_A$. The various brane configurations  are then
classified by orbits of $S_A$
under the action of the double cover  $Spin(d,d)$ of the T-duality group
 $SO(d,d)$. In this manner the spinor
parameterizes the  commuting set of instanton charges in the perturbative regime.

For $d=6$ or $d=7$ there are also contributions from NS5-branes wrapping
six-cycles. Such NS5-brane configurations  give contributions to the instanton charges that do not commute with
those of the wrapped $Dp$-branes.
 In other words, the $Dp$-brane charges in the spinor representation parametrize  the $\mathfrak u_{-1}$ component
part of the unipotent radical $U$ (the abelian part)
 for the standard parabolic subgroup
$P_{\alpha_1}$ of $E_{d+1}$  and
the NS5-brane charge are in
 the derived
subgroup $[U,U]$ component
part of the unipotent radical for the standard parabolic subgroup
$P_{\alpha_1}$ of $E_{d+1}$ in table~\reftab{tab:dimUnipotent}.
For $d=6$ this provides one extra charge configuration since there is a unique six-cycle. For $d=7$ there are 7 distinct six-cycles so there are 7 NS5-brane charges.  In addition there are 7 stringy $KKM$ instantons.  Recall that these arise from Kaluza--Klein monopoles in ten-dimensional string theory in which the fibre direction $x^\#$ is identified with a circle in $\rT^7$ (whereas the $D6$-brane is seen in M-theory as   a $KKM$ formed by identifying $x^\#$ with the M-theory circle).

Although it is very complicated to describe how all possible compactifications of euclidean $Dp$-branes fit into different spinor orbits, the following discussion will indicate the procedure.
For this purpose it is convenient to start in ten dimensions by
defining chiral spinors of the complexified group,
$SO(10,\IC)$  (complexification does not affect the BPS classification), by means of the raising and lowering operators,
\be
b_{k+1}
= {1\over 2} (\Gamma^{2k+1}-i\Gamma^{2k})\,, \qquad b_{k+1}^\dagger
=-{1\over 2}(\Gamma^{2k+1}+i\Gamma^{2k})\,, \quad  0\leq k\leq 4\,,
\label{creationops}
\ee
so that  $b^k=(b_k^\dagger)$ and
$\{b_k,b^l \}=\delta_k^l$ , and
$\{b_k,b_l\}=\{b^k,b^l\}=0$.
A ground state  $|-----\rangle$ is defined so that
$b_k|-----\rangle=0$, for $1\leq k\leq 5$.
Acting with $b^1$ gives the state
$b^1|-----\rangle=|+----\rangle$, with analogous states created by
any linear combination of the $b^r$'s, giving a total of $2^5$
states with $+$ or $-$ labelling each of the $5$ positions.  These
states are graded according to whether there an even or odd number of
$+$ signs.  There are therefore two chiral spinor representations of
$SO(10,\IC)$  of dimension  16.
Upon compactification on $\rT^d$ the spinor $\eta$
  in~\eqref{e:QQ2} is represented as a state of the Fock space built
  by acting with $b^i$ on the ground state
  $|-^5\rangle$. It is convenient to introduce the notation $e_{i_1\cdots i_r}:=
  b^{i_1}\cdots b^{i_r}|-^{d/2}\rangle$ and $e^*_{i_1\cdots
    i_r}:= b_{i_1}\cdots b_{i_r} |+^{d/2}\rangle$, which was used in section~\ref{spinorbit}.

 Spinors that are related by a continuous $Spin(d,d)$   transformation $\exp(\sum_{i,j}
  \,x_{ij}\gamma^{ij})$ are associated with $D$-brane
  configurations that are equivalent under T-duality. Each orbit listed in section~\ref{spinorbit}  is characterized
  by a representative $S^0$.
 Therefore  an $SO(d,d)$  pure spinor is equivalent to the ground state
of the Fock space that we can denote by  $1$, corresponding to a pure spinor defining a $D$-brane
  wrapping a supersymmetric cycle.
 The notation $e_{i_1\cdots i_r}$
  corresponds to a $D$-brane configuration wrapping the directions
  $\{i_1,\cdots, i_r\}$ in $\rT^d$ and
  $e^*_{i_1\cdots i_r}$  a $D$-brane configuration wrapping the
complementary  directions to  $\{i_1,\cdots, i_r\}$ in $\rT^d$.

Upon compactifying  on a torus of dimension $d\leq 3$, all possible brane world-volumes are parallel, up to identification under  $Spin(d,d, \Z)$, and the condition~(\ref{e:QQ1}) ensures in this case that all instanton configurations are $\smallf 12$-BPS.
These  are $p=0$ and $p=2$ wrappings in type IIA, and $p=-1$
and $p=1$ in type IIB.

The theory compactified on a 4-torus $\rT^4$ in type~IIA (for instance), includes instantons due to
wrapping $D0$-brane world-lines on any of the four 1-cycles and $D2$-brane world-volumes on any of the four
3-cycles. These configurations in general fill out an eight-dimension chiral spinor
representation of $SO(4,4)$,  $S_A= \sum_{i=a}^4 v_a b^a+
\sum_{a,b,c=1}^4 v_{abc} b^{abc}/3!$.  This
parametrization makes explicit the action of $SL(4)$ on $v_a$ or
$u^a=\epsilon^{abcd} v_{abc}$
(or $SU(4)$ in the complexified case).

With a single $D0$-brane or a single $D2$-brane world-volume wrapped on $\rT^4$ the
condition~(\ref{e:QQ1}) is always satisfied, and the configuration is  $\smallf 12$-BPS.
However,  wrapping both a $D0$-brane world-line and a $D2$-brane world-volume results in further breaking of supersymmetry unless  $v_a$ and
$u^a$ satisfy condition~(\ref{e:QQ2}).
It is easily seen that this condition is satisfied for all $\eta=|\pm\pm\rangle$ if
$v\cdot u=0$. But if $u\cdot v\neq0$
only $\eta=|+\pm\rangle$   satisfy the solution which is
$\smallf 14$-BPS.
These two conditions are invariant under the action of the T-duality
group  $Spin(4,4)$  acting on a spinor $S_A$.  The $\smallf 12$-BPS condition corresponds to imposing the
pure spinor constraint $S\cdot S=0$
while the $\smallf 14$-BPS corresponds to the complementary condition,
$S\cdot S\neq0$, which defines the configuration with the $D0$-brane world-line orthogonal to the $D2$-brane world-volume.

Extensions of these arguments lead to a classification of all BPS configurations of euclidean $Dp$-brane world-volumes
that are completely wrapped on a torus.  The orbits of such configurations are obtained by imposing generalisations of the pure spinor constraint on the $SO(d,d)$   spinor  that parameterizes the orbits.  An orbit which preserves a smaller fraction of supersymmetry is larger and is associated with a spinor satisfying weaker constraints.  The resulting orbits are described in section~\ref{spinorbit}.

\section{Properties of lattice sums   }
\label{sec:latticesums}

This appendix and appendices
  \ref{sec:latticeidentity} and \ref{sec:unfolding} together
 concern properties of lattice sums related to the Fourier expansions of certain Eisenstein series
  that appear in the coefficient functions for the cases $D=7$ and $D=6$ (i.e., for $SL(5)$ and $Spin(5,5)$, respectively).
    Those properties will later be used
    in section~\ref{lowrank} and appendices~\ref{appendixcoeff7}-\ref{sec:dfive}.
 The main result of the present appendix, proposition~\ref{prop:nonepsteinintegralrep}, is  an integral representation for the $SL(d)$ series\footnote{The labelling $\beta_i$ of the simple roots of $SL(d)$ here follows the conventional labelling as illustrated in figure~\ref{fig:A4vsE4}  for the $SL(5)$ case.} $E_{\beta_2,s}^{SL(d)}$.  The series  $E_{\beta_2,s}^{SL(d)}$ will later be related to the $Spin(d,d)$ Eisenstein series $E_{\alpha_1;s}^{Spin(d,d)}$ in proposition~\ref{prop:Ddintegralrepn} and (\ref{relationn}).

\subsection{Exponential sums of  lattice norms}

Let $g\in GL(d,\IR)$ and consider the set of points
\begin{equation}\label{setofpoints}
    \{ \,mg\,\in\,\IR^d \ | \ m\,\in\,\Z^d_{\neq 0}\,\}\,,
\end{equation}
where $m$ is thought of as a row vector.
This set of points is unchanged if $g$ is replaced by $\g g$ for any $\g\in GL(d,\Z)$,
so we may assume that $g$ lies in a fixed fundamental domain for $GL(d,\Z)\backslash GL(d,\IR)$.
   A standard result in reduction theory asserts that  every fundamental domain is
    contained in a Siegel set, so  we may also assume  that $g=nak$ where $n$ is unit
     upper triangular and lies in a fixed compact set,  $k$ lies in $O(d,\IR)$, and $a$
     is a diagonal matrix with positive diagonal entries $a_1,a_2,\ldots$ such that each
      $a_i/a_{i+1}$ is bounded below by an absolute constant (to be explicit, $\sqrt{\f34}$ \cite{KorZol}).  Therefore $a^{-1} na$ and its inverse range over a fixed compact subset of $N$, which means that the operator norms of both are bounded by a constant that depends only on the dimension $d$.  As a consequence $\|mg\|=\|mnak\|=\|mna\|=\|ma\cdot a^{-1} na\|$ is bounded above and below by multiples of $\|ma\|$:
 \begin{equation}\label{latbd}
    c_{-}\,\|ma\|^2  \ \ \le  \ \ \|mg\|^2   \ \ \le  \ \ c_{+}\,\|ma\|^2\,,
\end{equation}
 where the constants $c_{-}$ and $c_{+}$ depend only on $d$.
  Among other things, this implies the norms of  vectors in an arbitrary lattice are
crudely  similar to those of a dilation of the $\Z^d$ lattice.

Define the $\theta$-function
\begin{equation}\label{Sgt}
    S(g,t) \ \ := \ \ \sum_{m\,\in\,\Z^d_{\neq 0}}e^{-t\, \|mg\|^2}\,.
\end{equation}
The first inequality in (\ref{latbd}) shows this sum is absolutely convergent and  bounded by
\begin{equation}\label{Sgtbd1}
\sum_{m\neq 0}e^{-tc_{-}  (m_1^2 a_1^2+\cdots +m_d^2 a_d^2)} \ \ = \ \ \theta(tc_{-}a_1^2)\cdots \theta(tc_{-}a_d^2)\,-\,1\,,
\end{equation}
in terms of the Jacobi $\theta$-function $\theta(x)=\sum_{n\in\Z}e^{-n^2 x}$.  The Jacobi $\theta$-function satisfies the bounds $\theta(x)=1+O(e^{-x})$ for $x>1$, and  $\theta(x)=O(x^{-1/2})$ for $x\le 1$.  Therefore
\begin{equation}\label{Sgtbd2}
    S(g,t) \ \ = \ \ O(e^{-tc_{-}a_d^2}) \ , \ \ t \ > \ (\smallf{4}{3})^{d/2}\,(c_{-}a_d^2)^{-1}
\end{equation}
and
\begin{equation}\label{Sgtbd3}
    S(g,t) \ \ = \ \ O\(\f{t^{-d/2}}{a_1\cdots a_d}\) \ , \ \ t \ \le \  (\smallf{3}{4})^{d/2}\, (c_{-}a_1^2)^{-1}\,,
\end{equation}
with implied constants that depend only on $d$.
If $g$ is fixed we can use the fact that $\theta'(x)<0$ to bound the $t$-dependence of $S(g,t)$ by
\begin{equation}\label{Sgtbd4}
    S(g,t) \ \ \ll  \ \  \left\{
                        \begin{array}{ll}
                          e^{-t c_{-}(g)}, & t\,>\,1\,, \\
                          t^{-d/2}, & t\,<\,1\,,
                        \end{array}
                      \right.
\end{equation}
where both $c_{-}(g)$  and the implied constant in the $\ll$-inequality  depend on $g$.

\subsection{A constrained lattice sum over pairs}\label{sec:cG}

Let $\tau=\tau_1+ i\tau_2\in\U$ and define
\begin{equation}\label{bigGdef}
    {\mathcal G}(\tau,X) \ \ := \ \ \sum_{ [\srel{m}{n}]\, \in \, {\mathcal M}_{2,d}^{(2)}(\Z)}
    e^{-\pi \tau_2^{-1} (m+n\tau)X(m+n\bar\tau)^t}\,,
\end{equation}
where in the usual physics notation $X=G+B$, with $G$  a positive definite symmetric $d\times d$ matrix and  $B$  an antisymmetric $d\times d$ matrix, and
${\mathcal M}_{2,d}^{(i)}$ represents $2\times d$ matrices of rank
$i$.
This contribution is the rank 2 part of the lattice sum
  $\Gamma_{(d,d)}$ for even self-dual Lorentzian lattices.  It is necessary to use this modification of $\G_{(d,d)}$ in order to resolve some convergence issues in formal calculations involving $\G_{(d,d)}$.  However, the constraint complicates applications of Poisson summation to it in the following appendices.

 We  next analyze the convergence of this sum and give estimates.
Note that because $G=G^t$ and $B=-B^t$
\begin{equation}\label{nsntcalc}
\aligned
    (m+n\tau)X(m+n\bar\tau)^t \ \
    & = \ \ (m+n\tau_1)G(m+n\tau_1)^t + \tau_2^2 nGn^t -2 i \tau_2 mBn^t\,.
\endaligned
\end{equation}
Consider the sum
\begin{equation}\label{insidesum}
    \sum_{m\,\in\,\Z^d}e^{-\pi \tau_2^{-1}(m+x)G(m+x)^t} \,,
\end{equation}
which is absolutely convergent and represents a continuous, periodic, and hence bounded function of a row vector $x\in \IR^d$. By Poisson summation it is equal to
\begin{equation}\label{psfinsidesum}
  \tau_2^{d/2} (\det G)^{-1/2}  \sum_{\hat{m}\,\in\,\Z^d}e^{2\pi i \hat{m}\cdot x}\,e^{-\pi \tau_2 \hat{m}G^{-1}\hat{m}^t}\,,
\end{equation}
where $\hat{m}$ is thought of as a row vector.
Use (\ref{nsntcalc}) to write  (\ref{bigGdef}) as
\begin{equation}\label{bigG2}
 {\mathcal G}(\tau,G+B) \ \ = \ \    \sum_{n\,\neq\,0}e^{-\pi \tau_2 nGn^t}\sum_{m}'e^{-\pi \tau_2^{-1}(m+n\tau_1)G(m+n\tau_1)^t}e^{-2\pi i mBn^t}\,,
\end{equation}
where the prime indicates $m$ is not collinear to $n$.  The interior sum is bounded by (\ref{insidesum}) with $x=n\tau_1$ and hence
\begin{equation}\label{bigG3}
  {   \mathcal G}(\tau,G+B) \ \ \le  \ \    \tau_2^{d/2} (\det G)^{-1/2}   \sum_{n\,\neq\,0}e^{-\pi \tau_2 nGn^t}
      \sum_{\hat{m}\,\in\,\Z^d}e^{-\pi \tau_2 \hat{m}G^{-1}\hat{m}^t}\,.
\end{equation}
If we write $G=ee^t$, $e\in GL(d,\IR)$, then
\begin{equation}\label{bigG3b}
      {   \mathcal G}(\tau,ee^t+B) \ \ \le \ \    \tau_2^{d/2} (\det e)^{-1} \, S(e,\pi\tau_2)\,[1+S((e^{-1})^t,\pi\tau_2)]
\end{equation}
in terms of (\ref{Sgt}).

The earlier estimates (\ref{Sgtbd2}-\ref{Sgtbd3}) give bounds on the last two factors of (\ref{bigG3b}).  This shows that  ${\mathcal G}(\tau,G+B)$ decays rapidly   as $\tau_2\rightarrow\infty$.
  Since replacing $\tau$ by $\tau+1$  or by $-1/\tau$ in (\ref{bigGdef}) is tantamount to changing variables $(m,n)\mapsto (m+n,n)$ or $(-n,m)$, respectively, ${\mathcal G}(\tau,G+B)$ is thus automorphic in $\tau$.  Consequently, the  integral
\begin{equation}\label{Isdef}
    I(s,G+B) \ \ := \ \ \int_{SL(2,\Z)\backslash \U}E^{SL(2)}_s(\tau)\,{\mathcal G}(\tau,G+B) \,\f{d^2\tau}{\tau_2^2}
\end{equation}
is well-defined as a meromorphic function of $s$, with poles contained among the poles of the Eisenstein series $E^{SL(2)}_s(\tau)$.

\begin{prop}\label{lem:IsuSgrowth}
The integral $I(s,uG)$ decays rapidly as $u\rightarrow \infty$, and slower than some polynomial in $u>0$ as $u\rightarrow 0$.
These estimates are uniform for $G$ fixed and $\Re{s}$ ranging over a finite interval.
\end{prop}
\begin{proof}
The Eisenstein series satisfies the bound $E^{SL(2)}_s(\tau_1+i\tau_2)\ll \tau_2^z$ over the standard fundamental domain for $SL(2,\Z)\backslash \U$,
where
 $z=\max\{\Re{s},\Re{1-s}\}\ge 1/2$ (this follows from (\ref{e:ESl2fourier}-\ref{e:ESl2fourierCoef})).  Since the upper bound (\ref{bigG3b}) is independent of $\tau_1$,
\begin{multline}\label{Isbd1}
    I(s,u ee^t) \ \ \ll \\ u^{-d/2}(\det e)^{-1} \int_{\f{\sqrt{3}}{2}}^\infty \tau_2^{z+d/2-2}
    S(e,\pi\tau_2 u)\,[1+S((e^{-1})^t,\pi\tau_2u^{-1})]\,d\tau_2 \,.
\end{multline}
We now use the estimates in (\ref{Sgtbd4}).
As $u\rightarrow\infty$, $S(e,\pi\tau_2 u)$ has exponential decay in $\tau_2 u$, whereas
 the bracketed term is $O(u^{d/2}\tau_2^{-d/2})$.  Since the range of the $\tau_2$ integration is bounded below, the rapid decay assertion of the proposition immediately follows.

On the other hand, as $u\rightarrow 0$ the bracketed term in (\ref{Isbd1}) is bounded.  After a change of variables we are therefore left to showing that the integral
\begin{equation}\label{Isbd2}
\aligned
    \int_{\f{\sqrt{3}}{2}\pi u}^\infty \tau_2^{z+d/2-2}
    S(e,\tau_2 )\,d\tau_2 \ \ &= \ \   \int_{\f{\sqrt{3}}{2}\pi u}^{1} \tau_2^{z+d/2-2}
    S(e,\tau_2 )\,d\tau_2 \ \\
&+ \  \int_{1}^\infty \tau_2^{z+d/2-2}
    S(e,\tau_2 )\,d\tau_2
\endaligned
\end{equation}
is bounded by a polynomial in $u^{-1}$ as $u\rightarrow 0$.  Inserting the bounds from (\ref{Sgtbd4}) this is
\begin{equation}\label{Isbd3}
    \ll \ \  \int_{\f{\sqrt{3}}{2}\pi u}^{1}  \tau_2^{z-2} \,d\tau_2 \
+
\
 \int_{1}^\infty \tau_2^{z+d/2-2}
   \exp(-c'\tau_2 )\,d\tau_2
\end{equation}
for some constant $c'$ depending on $g$, and clearly bounded by a polynomial in $u^{-1}$.

\end{proof}

\subsection{Unfolding the lattice sum at $s=0$}\label{sec:E3}

The integral $I(s,G+B)$ in  (\ref{Isdef}) is well-defined for any value of $s$ at which the Eisenstein series $E^{SL(2)}_s(\tau)$ is holomorphic -- in particular, this includes $s=0$ where $E_0(\tau)$ is identically 1.

\begin{prop}\label{lem:unfoldlatticesum}
\begin{equation}\label{I01lemformula}
    I(0,G+B) \ \ = \ \ \sum_{[\srel{m}{n}]\, \in \, SL(2,\Z)\backslash {\mathcal M}_{2,d}^{(2)}(\Z)}
  e^{-2\pi i mBn^t}\, \f{ e^{-2\pi {\mathcal D}_{m,n,G}}}{{\mathcal D}_{m,n,G}}\,,
\end{equation}
where
\begin{equation}\label{I02}
    {\mathcal D}_{m,n,G} \ \ := \ \ \det([\srel{m}{n}] G [\srel{m}{n}]^t)^{1/2}\,.
\end{equation}
\end{prop}
\begin{proof}
Unfolding the lattice sum gives that
\begin{equation}\label{I01}
\aligned
    I(0,G+B) \ \ & = \  \ \int_{SL(2,\Z)\backslash \U}{\mathcal G}(\tau,G+B)\,\f{d^2\tau}{\tau_2^2}\\
    & = \ \ \sum_{[\srel{m}{n}]\, \in \, SL(2,\Z)\backslash {\mathcal M}_{2,d}^{(2)}(\Z)} \int_{\U}
    e^{-\pi \tau_2^{-1}(m+n \tau)(G+B)(m+n\bar\tau)^t}\,\f{d^2\tau}{\tau_2^2}\,.
\endaligned
\end{equation}
The unfolding is valid because of the absolute convergence of the series ${\mathcal G}(\tau,G+B)$ to a rapidly-decaying automorphic function in $\tau$.  The  integral in the last line can be computed as
\begin{multline}\label{I01a}
    \int_0^\infty\int_{\IR} e^{-\pi\tau_2nGn^t-\pi \tau_2^{-1}(mGm^t+2\tau_1mGn^t+\tau_1^2nGn^t)}e^{-2\pi i mBn^t}\,d\tau_1 \f{d\tau_2}{\tau_2^{2}}
   \\
 = \ \ (nGn^t)^{-1/2}
     \int_0^\infty
     e^{-\pi\tau_2nGn^t-\pi\tau_2^{-1}{\mathcal
         D}_{m,n,G}^2 (nGn^t)^{-1}} e^{-2\pi i  mBn^t} \f{d\tau_2}{\tau_2^{3/2}} \\
 = \ \   \f{ e^{-2\pi {\mathcal D}_{m,n,G}}}{{\mathcal D}_{m,n,G}}e^{-2\pi i  mBn^t}
\end{multline}
using (\ref{nsntcalc}) and the formulas
\begin{equation}\label{I01b}
    \int_{\IR}e^{-\pi\tau_2^{-1}(a+2b\tau_1+c\tau_1^2)}\,d\tau_1 \ \ = \ \ \sqrt{\smallf{\tau_2}{c}}\,e^{(b^2-ac)\pi/(\tau_2 c)} \ \ , \ \ \ \ c\,>\,0
\end{equation}
and
\begin{equation}\label{I01c}
    \int_0^\infty e^{-\pi a \tau_2-\pi b\tau_2^{-1}}\,\f{d\tau_2}{\tau_2^{3/2}} \ \ = \ \ b^{-1/2}\,e^{-2\pi\sqrt{a b}} \ \ , \ \ \ \ a,b\,>\,0\,.
\end{equation}
\end{proof}

Therefore for $\Re{s}$ sufficiently large we can compute the following integral (which converges by proposition~\ref{lem:IsuSgrowth}) as
\begin{multline}\label{I03}
    \int_0^\infty I(0,uG) \, u^{s-1}\,du \ \    =  \ \
    \sum_{[\srel{m}{n}]\, \in \, SL(2,\Z)\backslash {\mathcal M}_{2,d}^{(2)}(\Z)} \int_0^\infty \f{e^{-2\pi u {\mathcal D}_{m,n,G}}}{ {\mathcal D}_{m,n,G}}\,u^{s-2}\,du
    \\   = \ \ (2\pi)^{1-s} \, \G(s-1)
    \sum_{[\srel{m}{n}]\, \in \, SL(2,\Z)\backslash {\mathcal M}_{2,d}^{(2)}(\Z)} \det([\srel{m}{n}] G [\srel{m}{n}]^t)^{-s/2} \,.
\end{multline}
Proposition~\ref{prop:nonepsteinintegralrep} is now equivalent to
   the  identification of the righthand side of (\ref{I03}) with the multiple of the $SL(d)$ Eisenstein series  $E^{SL(d)}_{\beta_2;s}(e)$ given by Audrey Terras
in \cite[Lemma~1.1]{Terraspaper}.    For completeness  (and because the mechanism will be used later)  we shall briefly sketch her argument.  Since every relative prime vector in $\Z^d$ is the bottom row of a matrix in $SL(d,\Z)$, every element $[\srel{v_1}{v_2}]\in {\mathcal M}_{2,d}^{(2)}(\Z)$ can be factored as $  [\srel{v_1}{v_2}] \ \ = \ \ [\begin{smallmatrix} x & a \\0_{d-1} & \gcd(v_2)\end{smallmatrix}]\g$ for some nonzero vector $x\in \Z^{d-1}$, $a\in \Z$, and $\g\in SL(d,\Z)$, where $0_{d-1}$ denotes the $(d-1)$-dimensional zero vector.  Thus
\begin{equation}\label{terrasfactorization1}
    [\srel{v_1}{v_2}]\g^{-1} \ \ = \ \ \ttwo{1}{0}{0}{\gcd{v_2}}\ttwo{1}{a}{0}{1}\ttwo{\gcd{x}}{0}{0}{1}[
\begin{smallmatrix}
 x/\gcd(x) & 0 \\ 0_{d-1}&1\end{smallmatrix}]\,.
\end{equation}
Since $x/\gcd(x)$ is a relatively prime vector in $\Z^{d-1}$ it is the bottom row of a matrix in $SL(d-1,\Z)$, and so  $[\srel{v_1}{v_2}]$ can be factored as the product of $2\times 2$ upper triangular integer matrix with positive diagonal entries, times a $2\times d$ matrix which forms the bottom two rows of a matrix in $SL(d,\Z)$.  By adding integer multiples of the bottom row of this matrix to the row above it, we can reduce $b\imod {\gcd(v_2)}$ and hence assume that this $2\times 2$ matrix lies in the set ${\mathcal S}_+:= \{\ttwo{d_1}{b}{0}{d_2}|d_1,d_2\neq 0,0\le b < d_2\}$.

We now claim that the coset space $SL(2,\Z)\backslash {\mathcal M}_{2,d}^{(2)}(\Z)$ in the sum (\ref{I03}) is in bijection with products of the form $\g_1\g_2$, where $\g_1 \in {\mathcal S}_+$ and $\g_2$ are the bottom two rows  of a fixed set of coset representatives for $P_{\b_2}(\Z)\backslash SL(d,\Z)$.  Recall that the latter is the quotient by $GL(2,\Z)$ of all possible bottom two rows of matrices in $SL(d,\Z)$.
It is a standard result in the theory of Hecke operators that every right $GL(2,\Z)$ translate of an element of ${\mathcal S}_+$ is left $SL(2,\Z)$ equivalent to some element of ${\mathcal S}_+$ (this is because we allow for the possibility that $d_1<0$).  Thus the previous paragraph  shows that   every coset in $SL(2,\Z)\backslash {\mathcal M}_{2,d}^{(2)}(\Z)$ has a factorization of this asserted form, and it remains to show uniqueness.  After right multiplying by matrices in $SL(d,\Z)$ it suffices to show that if
\begin{equation}\label{terrasfactorization2}
    \ttwo{d_1}{b}{0}{d_2}[\srel{w_1}{w_2}] \ \ = \ \ \ttwo pqrs\ttwo{d_1'}{b'}{0}{d_2'} [\srel{0_{d-2}\,1\,0}{0_{d-2}\,0\,1}]
\end{equation}
for some $d_1, d_1'\neq 0$, $0\le b<d_2$, $0\le b'<d_2'$, $\ttwo pqrs \in SL(2,\Z)$, and $[\srel{w_1}{w_2}]$ which are the bottom two rows of one of these coset representatives for $P_{\b_2}(\Z)\backslash SL(d,\Z)$, then  $\ttwo{d_1}{b}{0}{d_2}=\ttwo{d_1'}{b'}{0}{d_2'}$, $\ttwo pqrs =\ttwo 1001$, and  $[\srel{w_1}{w_2}]=[\srel{0_{d-2}\,1\,0}{0_{d-2}\,0\,1}]$.  Indeed, (\ref{terrasfactorization2}) implies  that all but the last two entries of $w_1$ and $w_2$ vanish, so that $[\srel{w_1}{w_2}]$ are the bottom two rows of a matrix in $P_{\b_2}(\Z)$ and hence equal to the representative $[\srel{0_{d-2}\,1\,0}{0_{d-2}\,0\,1}]$ of its equivalence class in $P_{\b_2}\backslash SL(d,\Z)$.  Then (\ref{terrasfactorization2}) reduces to the identity $\ttwo{d_1}{b}{0}{d_2}=\ttwo pqrs \ttwo{d_1'}{b'}{0}{d_2'}$.  Since $d_2,d_2'>0$ and both sides have the same determinant, $d_1$ and $d_1'$ have the same sign.  Comparing the first columns then shows that $r=0$,  $p=1$,  and $d_1=d_1'$.  Consequently $s=1$ and  $d_2=d_2'$.  Finally, since $0\le b,b'<d_2$ differ by $qd_2$ they must be equal  and $q=0$.  This proves the claim.

Therefore  the range of summation in (\ref{I03}) can be replaced by $P_{\b_2}(\Z)\backslash SL(d,\Z)$, at the cost of multiplying the overall expression by $\sum_{d_1,d_2>0}d_2(d_1d_2)^{-s}=\sum_{n>0}\sigma_1(n)n^{-s}=\zeta(s)\zeta(s-1)$.  This, along with definition (\ref{maxparabeisdef}) and   standard $\G$-function identities results in the first of the two equivalent formulas in (\ref{I04}).

Since the definition  (\ref{Isdef}) of $I(0,uG)$ is an integral of a $\theta$-function over the modular fundamental domain, proposition~\ref{prop:nonepsteinintegralrep} indicates that the series $E^{SL(d)}_{\beta_2;s}(e)$ is  the Mellin transform of a  $\theta$-lift of   the constant function.

\section{Identification of the $Spin(d,d)$ Epstein series with a lattice sum  }
\label{sec:latticeidentity}

In this appendix we prove that Langlands' definition of the maximal parabolic Eisenstein series  $E_{\alpha_1;s}^{Spin(d,d)}$ as a sum over group cosets  is equivalent to the lattice sum definition used in  \cite{Obers:1999um,Angelantonj:2011br}.   It is easier to work directly with the group $SO(d,d,\IR)$, which is quotient of $Spin(d,d,\IR)$ by an order two subgroup of its center.
We explained in (\ref{SpinddandSoddseries}) that $E_{\alpha_1;s}^{Spin(d,d)}$   is trivial on this subgroup and hence can be computed through  $E_{\alpha_1;s}^{SO(d,d)}$.

Let $w_n$  denote the anti-diagonal identity matrix obtained  by
reversing the columns of the  $n\times n$ identity matrix.  Define  groups
\begin{equation}\label{SOdddefs}
\aligned
    G \ \ & = \ \ SO(d,d,\IR) \ \ = \ \ \{\,g\,\in\,SL(2d,\IR)\ | \ gw_{2d}g^t = w_{2d}\,\}\,, \\
    \G \ \ & = \ \ SO(d,d,\Z) \ \ = \ \ SO(d,d,\IR)\,\cap\,SL(2d,\Z)
\,,
\endaligned
\end{equation}
and
\begin{equation}\label{Pa1def}
    P \ \ = \ \ P_{\a_1} \ \ = \ \ \left\{
    \tthree{a}{\star}{\star}{0}{B}{\star}{
    0}{0}{c} \ \in \ G \ | \ a,c\,\in\,\IR^* \, , \ \ B\,\in\,SO(d-1,d-1)(\IR)\,
\right\}.
\end{equation}

\begin{prop}\label{SOddfacts}With the above definitions

\begin{enumerate}
  \item[(i)] If $g_1$, $g_2 \in G=SO(d,d,\IR)$ have the same bottom row, then there exists some $p\in P$ such that $g_1=pg_2$.
  \item[(ii)] The bottom row $v$ of a matrix in $G=SO(d,d,\IR)$ is orthogonal to its reverse $w_{2d}v$.
  In particular, if $v=[m\,n]$ then $m\perp w_d n$.
  \item[(iii)] The map from a matrix to its bottom row gives a bijection from $(\G\cap P)\backslash \G$ to $\{v\in \Z^{2d}\,|\,\gcd(v)=1$ and $v\perp w_{2d}v\}/\{\pm 1\}$.
\end{enumerate}
\end{prop}
\begin{proof}
Let $e_1,\ldots,e_{2d}$ denote the standard basis vectors of $\IR^{2d}$.
In part (i), the bottom row of the matrix $g_1g_2^{-1}$ is $e_{2d}g_1g_2^{-1}=e_{2d}$, as can be seen by multiplying both sides by $g_2$.  Thus $g_1g_2^{-1}$ has bottom row $e_{2d}$; membership of such a matrix in  $G$ forces its first column to equal a multiple of $e_1$, and so $g_1g_2^{-1}$ lies in $P$ .  Statement (ii) is a consequence of the defining property of $G$ (since the bottom right entry of $w_{2d}$ is zero).

Because of parts (i) and (ii) and the fact that the bottom row of a matrix in $\G\cap P$ is $\pm e_{2d}$, part (iii) reduces to showing that every such vector $v$ is the bottom row of some matrix  $\g$ in $\G$.  The calculation
\begin{equation}\label{SOddpf1}
    \ttwo{g_1}{}{}{g_2}\ttwo{}{w_d}{w_d}{}\ttwo{g_1^t}{}{}{g_2^t} \ \ = \ \ \ttwo{}{g_1w_dg_2^t}{g_2w_dg_1^t}{}
\end{equation}
shows that the matrix
\begin{equation}\label{Soddpf2}
    \ttwo{\widetilde{g}}{}{}{g} \ \in \ G \ \, , \ \  \text{with} \   \ \widetilde{g} \,= \,w_d(g^t)^{-1}w_d \, ,
\end{equation}
for any $g\in GL(d,\IR)$.  Since $g$ can be taken to be a matrix in $GL(d,\Z)$ with an arbitrary relatively prime bottom row, the proposition reduces to the case when $v$ has the form $v=[m\,n]$, where $m,n\in \Z^d$ and $n$ has the special form $[0\,0\,\cdots\,0\,k]$ (to see this, replace $v$ by $v\ttwo{\widetilde{g}}{}{}{g}$).  The orthogonality condition in part (ii) shows that we may furthermore take $m$ to have the form $[0\,a_2\,\cdots\,a_d]$ for integers $a_2,\cdots,a_d$.    Consider $g=\ttwo{A}{b}{}{1}\in GL(d,\Z)$, so that $\widetilde{g}=\ttwo{1}{\star}{}{\widetilde{A}}$.  Multiplying on the right by $\ttwo{\widetilde{g}}{}{}{g}$ has the effect of replacing $[a_2\,a_3\,\cdots\,a_d]$ by $[a_2\,a_3\,\cdots\,a_d]\widetilde{A}$.  Since $\widetilde{A}$ can be an arbitrary matrix in $GL(d-1,\Z)$, we may arrange that $[a_2\,a_3\,\cdots\,a_d]h=[0\,0\,\cdots\,0\,r]$ for some integer $r$.
The condition that the bottom row be relatively prime  now states that $\gcd(r,k)=1$.  Such a matrix exists because of
the homomorphism of   $SL(2,\IR)$ into $G$
which  sends a matrix $\ttwo abrk$ to one with entries $a$ on the 1st and $d$-th diagonal entries, entries $k$ on the $d+1$-st and $2d$-th diagonal entries, $-b$ in the $(1,d+1)$ position, $b$ in the $(d,2d)$ position, $-c$ in the $(d+1,1)$ position, and $c$ in the $(2d,d)$ position.
\end{proof}

\begin{prop}\label{lem:epsteinD5latticesum}
Let $h\in SO(d,d,\IR)$ and  write $hh^t=\ttwo{H_1}{H_2}{H_2^t}{H_3}$, where $H_1$, $H_2$, and $H_3$ are $d\times d$ matrices.  Then the
maximal parabolic Eisenstein series  $E^{SO(d,d)}_{\a_1;s}(h)$ associated to the first node (i.e., ``vector node'') of the $D_d$ Dynkin diagram is given by
\begin{equation}\label{epsteinD5formula}
    2\,\zeta(2s)\,E^{SO(d,d)}_{\a_1;s}(h) \ \ = \ \ \sum_{\srel{m,n\,\in\Z^d}{\srel{m\,\perp\, w_d n}{(m,n)\,\neq\,(0,0)}} } (mH_1m^t+2mH_2n^t+nH_3n^t)^{-s}
\end{equation}
for $\Re{s}$ large (where the sum is absolutely convergent).   The same formula holds for $E^{Spin(d,d)}_{\a_1;s}(h')$, where $h'\in Spin(d,d,\IR)$ projects to $h\in SO(d,d)$ under the covering map.
\end{prop}

\begin{proof}
The $SO(d,d)$ Epstein series is formed by averaging the function $f(g)=e^{2s\omega_1(H(g))}$ over $g=\g h$, $\g\in (\G\cap P)
\backslash \G$.  Recall that $f(pgk)=f(g)$ for all  $p\in P$ such that
each diagonal block  $a$, $B$, and $c$ in (\ref{Pa1def}) has
determinant $\pm 1$, and for all $k$ in the maximal compact subgroup of $G$.
Calculation of $f$ thus reduces via the Iwasawa decomposition to the
case when $g$ is diagonal, in which case $f(g)$ equals the $-2s$ power of the
bottom right entry of $g$.   The bottom right entry of the Iwasawa
factor of $g$ must be the norm of $g$'s bottom row, because of these
invariance properties.  Hence $f(\g h)$  is the norm of the bottom row of   $\g h$ to the $-2s$
power. If $v=[m\,n]\in \Z^{2d}$ is the bottom row of $\g$, then the
norm is the squareroot of $vhh^t v^t= mH_1m^t+2mH_2n^t+nH_3n^t$.  Thus
$E^{SO(d,d)}_{\a_1;s}$ is given by a sum as stated, but with a gcd  and
modulo $\pm 1$
condition which, when removed, results in the factor $2\zeta(2s)$ in (\ref{epsteinD5formula}).
\end{proof}

For later reference we remark that since $w_d^2=1$ we can present these $D_d$
Epstein series as
\begin{multline}\label{epsteinD5formulaBis}
   2\,\zeta(2s)\,E^{SO(d,d)}_{\a_1;s}(h') \ \ = \ \  2\,\zeta(2s)\,E^{SO(d,d)}_{\a_1;s}(h) \\ = \ \ \sum_{\srel{m,n\,\in\Z^d}{\srel{m\,\perp\, n}{(m,n)\,\neq\,(0,0)}} } (mH_1m^t+2mH_2w_dn^t+nw_dH_3w_dn^t)^{-s}
\end{multline}
for $\Re{s}$ large.  Also we note that the condition for a matrix of the form $\ttwo{I_d}{X}{}{I_d}$ to lie in $G=SO(d,d,\IR)$ is that
\begin{equation}\label{Soddunipcond}
    \ttwo{Xw_d}{w_d}{w_d}{} \ \ = \ \ \ttwo{I_d}{X}{}{I_d}\ttwo{}{w_d}{w_d}{} \ \ = \ \ \ttwo{}{w_d}{w_d}{}\ttwo{I_d}{}{-X^t}{I_d} \ \ = \ \ \ttwo{-w_dX^t}{w_d}{w_d}{}\,,
\end{equation}
i.e., that $Xw_d$ is antisymmetric.

\section{A theta lift between $SL(d)$ and $Spin(d,d)$ Eisenstein series }
\label{sec:unfolding}

In proposition~\ref{prop:nonepsteinintegralrep}
we stated a relation between the modular integral $I(s,G)$ and the non-Epstein
 Eisenstein series $E^{SL(d)}_{\b_2;s}$.  In this section we see another relation
 (proposition~\ref{prop:Ddintegralrepn}) to the Epstein Eisenstein series $E^{D_d}_{\a_1;s+d/2-1}$, where $D_d$
 is written as a shorthand for statements that apply both to $SO(d,d)$ and $Spin(d,d)$.
   Thus both can be thought of as $\theta$-lifts from the usual nonholomorphic $SL(2,\Z)$ Eisenstein series.
We shall do this for $\Re{s}$
large, the range of absolute convergence
of the Eisenstein series,  and then meromorphically continue to $s\in
\C$.

  Unfolding the Eisenstein series in (\ref{Isdef})
 gives the formula
\begin{equation}\label{Is1}
    I(s,G+B) \ \ = \ \  \int_0^\infty\,\int_0^1\, {\mathcal G}(\tau_1+i\tau_2,G+B)\,d\tau_1  \f{d\tau_2}{\tau_2^{2-s}}\,.
\end{equation}
 This integral is absolutely convergent for $\Re{s}$ large by
(\ref{bigG3b}) and the bounds given in (\ref{Sgtbd4}).
We write $G=ee^t$ and introduce the notation $\| v\|^2=v \bar
  v^t$ if $v$ is a complex row vector.
  Using (\ref{nsntcalc}) we  write
\begin{multline}\label{Is2}
    {\mathcal G}(\tau_1+i\tau_2,ee^t+B) \ \ = \\
    \sum_{\srel{n\,\neq\,0}{m\,\in\,\Z^d}}\exp(-\pi
      \tau_2^{-1}\|(m+n\tau_1)e\|^2-\pi \tau_2 \|ne\|^2-2\pi i mBn^t)  \\
- \  \sum_{\srel{n\,\neq\,0}{m\,\in\,\Z^d\cap \Q n}}
\exp(-\pi \tau_2^{-1}\|(m+n\tau_1)e\|^2-\pi \tau_2 \|ne\|^2)\,,
\end{multline}
the second sum including all $m \in \Z^d$ which are parallel to $n$ (a condition which forces $mBn^t=0$).  Accordingly break up $I(s,ee^t+B)$ as $I_1(s,ee^t+B)-I_2(s,ee^t+B)$, where $I_1(s,ee^t+B)$ and $I_2(s,ee^t+B)$ represent
 the integral (\ref{Is1}) with ${\mathcal G}(\tau,G+B)$ replaced by the sums in the first and second
 lines of (\ref{Is2}), respectively.
 Both integrals are absolutely convergent and can be interchanged with their respective
 summations for $\Re{s}$ sufficiently large.
  We first compute
\begin{multline}\label{Is3}
    I_2(s,G+B) \ \ = \\ \sum_{n\,\neq\,0}  \int_0^\infty \int_0^1
    \sum_{m\,\in\,\Z^d\cap \Q n} \exp(-\pi \tau_2^{-1}\|(m+n\tau_1)e\|^2 -\pi \tau_2\| ne\|^2)d\tau_1\f{d\tau_2}{\tau_2^{2-s}}\,.
\end{multline}
Since $n$ is nonzero and $m$ is a multiple of $n$, we change variables by subtracting this multiple from $\tau_1$ to eliminate the occurrence of $m$ in the integrand.
Doing so unfolds the   $\tau_1$ integration from $[0,1]$ to $\IR$ by gathering together all $m$ which are
$\Z n$-translates of each other, though we must take into account the fact that there are $\gcd(n)$ orbits of $\{m\in \Z^d\cap \Q n\}$ under $m\mapsto m+n$:
\begin{equation}\label{Is4}
\aligned
  I_2(s,G+B) \ \ & = \ \ \sum_{n\,\neq\,0} \gcd(n) \int_0^\infty
\int_\IR \exp(-\pi (\tau_2+\tau_1^2\tau_2^{-1}) \|ne\|^2)d\tau_1 \f{d\tau_2}{\tau_2^{2-s}}\\
    & = \ \
    \sum_{n\,\neq\,0} \gcd(n) \int_0^\infty \, \exp(-\pi \tau_2 \|ne\|^2)  \smallf{\sqrt{\tau_2}}{\|ne\|} \f{d\tau_2}{\tau_2^{2-s}}  \\
    & = \ \
\pi^{\smallf12-s}\,    \G(s-\smallf12)\,\sum_{n\,\neq\,0} \gcd(n) \|ne\|^{-2s}  \,.
\endaligned
\end{equation}
We now evaluate this last sum, writing $e=r^{1/2}e'$, where $\det e'=1$, and $r=(\det e)^{2/d}$.
Decomposing $n\in \Z^d_{\neq 0}$ as $n=km$, with $\gcd(n)=k$ and  $\gcd(m)=1$,
it equals
\begin{equation}\label{Is5}
   r^{-s}  \sum_{\srel{m\,\in\,\Z^d}{\gcd(m)=1}}\sum_{k\,=\,1}^\infty k^{1-2s}\,\|me'\|^{-2s} \ \ = \ \ 2 \,r^{-s}   \,\zeta(2s-1)\,E^{SL(d)}_{\b_1;s}(e')\,,
\end{equation}
so
\begin{equation}\label{I2formula}
    I_2(s,ee^t+B) = 2\, (\det e)^{-2s/d} \, \xi(2s-1)\, E^{SL(d)}_{\b_1;s}(e)\,,
\end{equation}
initially for $\Re{s}$ sufficiently large and then by meromorphic continuation to $s\in \C$.

Next we compute
\begin{multline}\label{Is5b}
    I_1(s,ee^t+B) \ \ = \\  \int_0^\infty\int_0^1
    \sum_{n\,\neq\,0\atop {m\,\in\,\Z^d}}\exp(-\pi \tau_2^{-1}\|(m+n\tau_1)e\|^2 -\pi \tau_2 \|ne\|^2-2\pi i mBn^t)d\tau_1 \f{d\tau_2}{\tau_2^{2-s}}\,.
\end{multline}
Poisson summation allows us to rewrite
\begin{multline}\label{Is6a}
\sum_{m\,\in\,\Z^d}  \exp(-\pi \tau_2^{-1}\|(m+n\tau_1)e\|^2  -2\pi i mBn^t)
 \\ = \ \    \sum_{\hat{m}\,\in\,\Z^d}
     \exp(2\pi i \hat{m}\cdot n\tau_1)
     \int_{\IR^d}\exp(-2\pi i \hat{m}\cdot m -\pi \tau_2^{-1}\|me\|^2 -2\pi i mBn^t)\,dm   \\ = \ \
      \sum_{\hat{m}\,\in\,\Z^d}
     \exp(2\pi i \hat{m}\cdot n\tau_1) (\det e)^{-1}\tau_2^{d/2} \exp(-\pi\tau_2\|(\hat{m}-nB)(e^{-1})^t\|^2)\,,
\end{multline}
where $\hat{m}$ is thought of as a row vector.
Therefore
the integration over $\tau_1$ then forces $\hat{m}\perp n$:
\begin{multline}\label{Is7}
     I_1(s,ee^t+B)  \\ = \ \ (\det e)^{-1}\,\int_0^\infty
      \sum_{\srel{n\,\neq\,0}{\srel{\hat{m}\,\in\,\Z^d}{\hat{m}\,\perp\,n}}}\exp(-\pi
        \tau_2 \|ne\|^2-\pi\tau_2\|(\hat{m}-nB)(e^{-1})^t\|^2)
      \f{d\tau_2}{\tau_2^{2-s-\frac d2}}\\
       = \ \
       \f{ \G(s+\frac d2-1)}{
      (\det e) \,\pi^{s+\smallf d2-1}} \sum_{\srel{n\,\neq\,0}{\srel{\hat{m}\,\in\,\Z^d}{\hat{m}\,\perp\,n}}}(
      nGn^t+(\hat{m}-nB)G^{-1}(\hat{m}-nB)^t)^{1-s-d/2}\,.
\end{multline}
Again write $e=r^{1/2}e'$ with $\det e'=1$, and $r=(\det e)^{2/d}$ (thus $G=re'(e')^t$).  Recall (\ref{Soddpf2}) and define an element $h\in SO(d,d,\IR)$ by
\begin{equation}\label{Is7b}
    h \ \ = \ \ \ttwo{I}{Bw_d}{}{I}\ttwo{r^{1/2}e'}{}{}{r^{-1/2}\tilde{e}'} \, , \ \ hh^t \ \ = \ \ \ttwo{G+BG^{-1}B^t}{BG^{-1}w_d}{w_d G^{-1}B^t}{w_dG^{-1}w_d}.
\end{equation}
Then the inside sum in (\ref{Is7}) is computed
using (\ref{epsteinD5formulaBis}) as
\begin{multline}\label{Is8}
2\,\zeta(2s+d-2)\,E^{D_d}_{\a_1;s+d/2-1}(h)-\sum_{\hat{m}\,\in\,\Z^d_{\neq 0}}\|\hat{m}(e^{-1})^t\|^{2-2s-d}\ \ \\
 = \ \
  2\,\zeta(2s+d-2)\,E^{D_d}_{\a_1;s+d/2-1}(h)- (\det e)^{d+2s-2\over d} \sum_{\hat{m}\,\in\,\Z^d_{\neq 0}}\|\hat{m}((e')^t)^{-1}\|^{2-2s-d}
 \\ = \ \ 2\,\zeta(2s+d-2)\(E^{D_d}_{\a_1;s+d/2-1}(h)-(\det
e)^{d+2s-2\over d} E^{SL(d)}_{\beta_{d-1};s+d/2-1}(e')\).
\end{multline}
Combining (\ref{I2formula}), (\ref{Is7}), and (\ref{Is8}) we conclude the following:
\begin{prop}\label{prop:Ddintegralrepn}
With $h$  as in (\ref{Is7b})
 \begin{multline}\label{Is9}
    I(s,ee^t+B) \ \ = \ \ 2\, (\det e)^{-1}\,
    \xi(2s+d-2)\,E^{D_d}_{\a_1;s+d/2-1}(h)
\\ -2\, (\det e)^{2s-2\over d}\,
\xi(2s+d-2)\,E^{SL(d)}_{\beta_{d-1};s+d/2-1}(e')\\
-2\,(\det e)^{-{2s\over d}}\,\xi(2s-1)\,    E^{SL(d)}_{\b_1;s}(e')\,,
 \end{multline}
 initially for $\Re{s}$ large, and then to all $s\in \C$ by meromorphic continuation.
 \end{prop}
 As before, all manipulations are valid because of the absolute convergence of the sums and integral involved and the assumption that $\Re{s}$ is sufficiently large.  Note that only the first line on the righthand side depends nontrivially on $B$.
In particular, if $B=0$ and $s=0$ the above equation provides an integral representation for $E^{D_d}_{\a_1;d/2-1}$,
 \begin{multline}\label{relationn}
    2\, u^{-d/2}   \xi(d-2)\,E^{D_d}_{\a_1;d/2-1}\ttwo{u^{1/2}e'}{}{}{u^{-1/2}\tilde{e}'} \ \  = \\
       I(0,ue'(e')^t) \, +  \, 2\,u^{-1}\, \xi(d-2)\,E^{SL(d)}_{\beta_{d-1};d/2-1}(e') \, + \, 2\, \xi(2)\,,
 \end{multline}
for any $e'\in SL(d,\IR)$.
 A similar expression appeared  in \cite{Green:2010wi,Obers:1999um,Angelantonj:2011br} but without that the last two terms in the second line.

{\bf Remark:}
 According to proposition~\ref{lem:IsuSgrowth} $I(0,uG)$ decays rapidly to
 zero as $u\rightarrow\infty$.  This is not immediately obvious from
 (\ref{relationn}), in which both other terms involving $u$ have polynomial
 behavior in that limit while the remaining term is constant.  However, the
 aggregate sum indeed does decay to zero.  This can be seen explicitly
through an analysis of the constant term of the $E^{D_d}_{\a_1;d/2-1}$ Eisenstein
series in the spinor parabolic (which determines these asymptotic behaviors).

\section{Fourier modes of Eisenstein series}
\label{modesdetails}

In this appendix we will present details of the Fourier modes of Eisenstein series that enter in the expressions for the  coefficients   $\cE^{(D)}_{(0,0)}$   and
$\cE^{(D)}_{(1,0)}$ when $D=8$, $D=7$, and $D=6$ (although the discussion of the $D=6$ case with symmetry group $Spin(5,5)$ is incomplete). This summarises and extends the string theory results in
  \cite{Green:2010wi} (see~\cite{Obers:1999um,Pioline:2010kb,Gubay:2010nd,Basu:2011he,Angelantonj:2011br} for related investigations).

\subsection{The $SL(3)\times SL(2)$ case}
\label{appendixcoeff8}

The results of this subsection are used in section \ref{exampleeight}
  in the text.   The coefficients are functions of both  the $SL(2)/SO(2)$  symmetric space, which depends on
$\calU= \calU_1+i \calU_2$ (the complex structure of the 2-torus, $\rT^2$),  and  the  $SL(3)/SO(3)$ space, which depends on  $5$ parameters.  We will parametrise the $SL(2)/SO(2)$ coset by \eqref{sl2param} (with $\Omega$ replaced by $\calU$) while
the $SL(3)/SO(3)$ coset will be parameterised by the string fluxes as
\begin{equation}\label{e:e3def}
    e_3=
    \begin{pmatrix}
      1& B_{\rm NS}&C^{(2)}+\Omega_1 B_{\rm NS}\\ 0 & 1&\Omega_1 \\ 0&0&1
    \end{pmatrix}\,
\begin{pmatrix}
 \nu_2^{-\frac13} &0&0\\ 0& \nu_2^{\frac16}\sqrt{\Omega_2} & 0\\
 0&0&{\nu_2^{\frac16}\over\sqrt{\Omega_2} }
\end{pmatrix},
\end{equation}
where $\nu_2^{-\half}=r_1r_2/\ell_{10}^2=\sqrt{\Omega_2}T_2$ is the
volume of the 2-torus in 10 dimensional Planck units and
  $T_2=r_1r_2/\ell_s^2$ is the volume in string units.  The five parameters of the coset are packaged into
$(\Omega,T,C^{(2)})$, where   $\Omega=\Omega_1+i\Omega_2$   and
$T=T_1+iT_2$ (where $T_1 =B_{\rm NS}$).    We shall also make use of the combination
 $y_8^{-1}=\Omega_2^2T_2$ , which is the square of the inverse  string coupling.
 The complex parameter  $T$  is interpreted as the K\"ahler structure of $\rT^2$.

The coefficient functions $\cE^{(8)}_{(0,0)}$   and $\cE^{(8)}_{(1,0)}$ are solutions of  \eqref{laplaceeigenone} and \eqref{laplaceeigentwo}  with $D=8$ \cite{Basu:2007ru,Green:2010wi},
\begin{eqnarray}
  \Delta^{(8)} \,\cE^{(8)}_{(0,0)}&=&6\pi
   \label{delta8op1}\\
(\Delta^{(8)}-{10\over3})\,\cE^{(8)}_{(1,0)}&=&0
\label{delta8op2}\,,
\end{eqnarray}
where the $SL(3)\times SL(2)$ Laplace operator is defined in terms of the parameters introduced above by
\be
\Delta^{(8)}:=
\Delta^{SL(3)}+2\Delta_U^{SL(2)}\,,
\label{laplace8}
\ee
with
\begin{eqnarray}
  \label{e:Delta8}
\Delta^{SL(3)}&=&\Delta_\Omega+{|\partial_{B_{\rm
      NS}}-\Omega \partial_{C^{(2)}}|^2\over
  \nu_2\Omega_2}+3\partial_{\nu_2}(\nu_2^2\partial_{\nu_2})\\
  \text{and~~~~}\ \ \ \ \Delta^{SL(2)}_Z&=&Z_2^2\, (\partial^2_{Z_1}+\partial^2_{Z_2})\,,
\end{eqnarray}
where $Z=Z_1+iZ_2$ and $Z=\Omega$ or $\calU$.
The fact that the eigenvalue in \eqref{delta8op1} vanishes, together with the presence of the $6\pi$ on the right-hand side, is related to the presence of a 1-loop  ultraviolet divergence in eight-dimensional maximally supersymmetric supergravity \cite{Green:2010sp}.

The solutions to these equations are given in terms of $SL(2)$ and $SL(3)$  Eisenstein series. The $SL(2)$ series is given by   \eqref{e:EisSl2} while  the $SL(3)$ (Epstein) Eisenstein  series is given
by
\be
\label{e:EpSl3}
  2\zeta(2s)\, E^{SL(3)}_{\alpha_1;s}(e_3)=        \sum_{M_3\in\ZZ^3\bsz}
  (m^2_{SL(3)})^{-s}\,,
  \ee
where, setting $M_3=(m_1\,m_2\,m_3)\in\ZZ^3$,  the  mass squared  is given  by
\begin{eqnarray}
\label{e:EpSl3ma}
m^2_{SL(3)}&:=&  M_3G_3M_3^t\\
&=& {\nu_2^{\frac13}\over\Omega_2}\, \left(|m_3+m_2\Omega+\calB m_1|^2+(m_1\Omega_2T_2)^2\right)\nn
\end{eqnarray}
with
\begin{equation}
\label{g3def}
 G_3:=e_3e_3^t=\nu_2^{\frac13} \,
\begin{pmatrix}
 \nu_2^{-1}+(G_2)_{ab} B^aB^b&    (G_2)_{ab} B^b
\\
 (G_2)_{ab} B^a&(G_2)_{ab}
\end{pmatrix}\,,
\end{equation}
\be\label{g2def}
G_2:={1\over\Omega_2}\,
\begin{pmatrix}
  |\Omega|^2 &\Omega_1\\
\Omega_1&1
\end{pmatrix}, \ \ \
B:=
\begin{pmatrix}
  B_{\rm NS} \\ C^{(2)}
\end{pmatrix}, \ \ \text{and} \ \  \calB:=C^{(2)}+\Omega B_{\rm NS}
\,.
\ee
The Eisenstein series  $E^{SL(3)}_{\alpha_1;s}$ is related to $E^{SL(3)}_{\alpha_2;s}$ by the  functional relation
\begin{equation}
  \label{e:SL3Function}
 \xi(2s)\,         \,            E_{\alpha_1;s}^{SL(3)}(e_3)=\xi(3-2s)\,
  E_{\alpha_2;\frac32-s}^{SL(3)}(e_3)\, .
\end{equation}
The Fourier modes of the coefficient functions can now be considered in each of the three parabolic subgroups of interest, after putting the $SL(3,\Z)$ part together with the $SL(2,\Z)$ part. The unipotent radicals in these three cases are given by:

(i) {\bf   The unipotent radical  $U_{\alpha_3}$  of the
nonmaximal  parabolic  $P_{\alpha_3}=GL(1)\times  SL(2)\times\IR^+ \times
U_{\alpha_3}$. }   As noted earlier, the relevant parabolic is nonmaximal in order to   match the
$D=9$ duality group  associated with the decompactification limit.  The unipotent radical is
parametrized  by $(C^{(2)},B_{\rm NS})$ and  takes  the block
diagonal form,
\begin{equation}\label{e:UnipotentN3E3}
U_{\alpha_3}=  \begin{pmatrix} \begin{pmatrix}
  1&B_{\rm NS}&C^{(2)}\\
0&1&0\\
0&0&1
  \end{pmatrix} &0\\
0& \begin{pmatrix}
  1&\calU_1\\
0&1\\
  \end{pmatrix}
 \end{pmatrix}
\,.
\end{equation}

In the maximal parabolic subgroup of $SL(3)$ determined by $B_{\rm{NS}}$ and $C^{(2)}$ the
  Fourier coefficients of the $SL(3)$ Eisenstein series
  in~\eqref{e:EpSl3}  are defined by\footnote{The labelling of the simple roots $\beta_1$ and $\beta_2$ on these
  Fourier coefficients uses  the conventional labelling of the $SL(3)$
  Dynkin diagram according the convention in figure~\ref{fig:E3lab}.}
\begin{equation}\label{e:F2def}
  F^{SL(3)\beta_2}_{\beta_1;s}(kp_1,kp_2):=\int_{[0,1]^2} dB_{\rm NS} dC^{(2)}\,
  e^{-2i\pi k (p_1 C^{(2)}+p_2 B_{\rm NS})}\, E^{SL(3)}_{\alpha_1;s}\,,
\end{equation}
with $\gcd(p_1,p_2)=1$. Extending the constant term computation in~\cite[Appendix B.4]{Green:2010wi}, the Fourier coefficients for $k\neq 0$ are
\begin{equation}
  \label{e:F2}
  F_{\beta_1;s}^{SL(3)\, \beta_2}(kp_1,kp_2)
  = {2\over \xi(2s)}\, \Omega_2^{1-{2s\over3}}T_2^{1-{s\over3}}\,
 {\sigma_{2s-2}(k)\over |k|^{s-1}}\, {K_{s-1}(2\pi |k|\, |p_2+p_1\Omega|T_2)\over |p_2+p_1\Omega|^{1-s}}\,.
 \end{equation}
The Fourier modes of the $SL(2)$ series are defined as
\begin{equation}\label{e:F22def}
  F^{SL(2)}_{s}(k'):=\int_{[0,1]}
d\calU_1\,
  e^{-2i\pi k'\calU_1}\, E^{SL(2)}_s(\calU)\,,
\end{equation}
where
\begin{equation}\label{e:F22}
  F^{SL(2)}_{s}(k')={2\sqrt{\calU_2}\over\xi(2s)}
  {\sigma_{2s-1}(|k'|)\over |k'|^{s-\frac12}}\, K_{s-\frac12}(2\pi|k'|\calU_2)\,
\end{equation}
for $k'\neq 0$  (cf.~(\ref{e:ESl2fourierCoef})).

 Putting this together, the Fourier modes of the product of the $SL(3)$
  and the $SL(2)$ series are given by
\begin{eqnarray}
\label{e:oldF23}
  F^{SL(3) \times SL(2)\alpha_3}_{\alpha_1;s,s'}(kp_1,kp_2,k')&:=&\int_{[0,1]^2}\!\!\!\!\!\!\! dB_{\rm NS} dC^{(2)}\,
  e^{-2i\pi k (p_1 C^{(2)}+p_2 B_{\rm NS})}\, E^{SL(3)}_{\beta_1;s }\nn\\
\nn&& \times \int_{[0,1]}\,d\calU_1\,
  e^{-2i\pi k'\calU_1}\, E^{SL(2)}_{s'}(\calU) \\
  &=&\cF^{SL(3)\beta_2}_{\beta_1;s}(kp_1,kp_2)\,\cF^{SL(2)\alpha_3}_{s'}(k')\,,
\end{eqnarray}
with $\gcd(p_1,p_2)=1$.  These results are used in~\eqref{fouriereight} and~\eqref{fouriereight2}, where we provide a physical interpretation of the Fourier modes in the decompactification regime (limit  (i) in the notation of~\eqref{notation}).

(ii)  {\bf
The unipotent radical  $U_{\alpha_1}$ of the maximal
parabolic  subgroup $P_{\alpha_1}=GL(1)\times  Spin(2,2)\times U_{\alpha_1}$ } associated with the string  perturbation
regime   is   parametrized  by
$(\Omega_1,C^{(2)})$ and takes the form,
 \begin{equation}\label{e:UnipotentN1E3}
U_{\alpha_1}=  \begin{pmatrix} \begin{pmatrix}
  1&0&C^{(2)}\\
0&1&\Omega_1\\
0&0&1
  \end{pmatrix} &0\\
0&
\begin{pmatrix}
  1&0\\0&1
\end{pmatrix}
 \end{pmatrix}
\,.
\end{equation}
In this maximal parabolic  only the $SL(3)$ series have
  non-vanishing
  Fourier coefficients, which are  defined by
\begin{equation}
  F^{SL(3)\beta_1}_{\beta_1;s}(kp_1,kp_2):=\int_{[0,1]^2} d\Omega_1 dC^{(2)}\,
  e^{-2i\pi k (p_1 C^{(2)}+p_2 \Omega_1)}\, E^{SL(3)}_{\alpha_1;s}
\end{equation}
with $\gcd(p_1,p_2)=1$. Extending the constant term calculation in~\cite[Appendix B.4]{Green:2010wi}
  leads to
  \begin{equation}
   \label{e:F1}
 F_{\beta_1;s}^{SL(3)\,\beta_1}(kp_1,kp_2)  = {2\over \xi(2s)}\, T_2^{2s\over3}\, \Omega_2^{\frac12+{s\over 3}}\,
 {\sigma_{2s-1}(k)\over |k|^{s-\frac12}}\,  {K_{s-\frac12}(2\pi  |k|\, |p_1   T+p_2|\Omega_2)\over      |p_1   T+p_2|^{s-\frac12}}\,.
 \end{equation}
These results are used in~\eqref{e:R4Dstring} and~\eqref{e:D4R4Dstring}, where we provide a physical interpretation of the Fourier modes in the perturbative regime (limit  (ii) in the notation of~\eqref{notation})

(iii) {\bf  The unipotent radical $U_{\alpha_2}$ of the maximal
parabolic   $P_{\alpha_2}=GL(1)\times   SL(3)\times U_{\alpha_2}$ } associated with the semi-classical M-theory limit    is  parametrized   by  $\calU_1$ and takes the form
\begin{equation}\label{e:UnipotentN2E3}
U_{\alpha_2}=  \begin{pmatrix} \begin{pmatrix}
  1&0&0\\
0&1&0\\
0&0&1
  \end{pmatrix} &0\\
0& \begin{pmatrix}
  1&\calU_1\\
0&1\\
  \end{pmatrix}
 \end{pmatrix}
\,,
\end{equation}
In this maximal parabolic subgroup only the  $SL(2)$ series has
  non-vanishing Fourier coefficients, which are given in (\ref{e:F22def}-\ref{e:F22}).
The evaluation of the
non-zero Fourier coefficients  of $\cE^{(8)}_{(0,0)}$ and
$\cE^{(8)}_{(1,0)}$ in each of the three limits of interest is straightforwardly obtained by using the above expressions, and is
discussed in section~\ref{exampleeight}.

\subsection{The $SL(5)$ case}
\label{appendixcoeff7}

Here we consider the Fourier modes of the   Eisenstein series that enter the expressions for the  coefficients   $\cE^{(7)}_{(0,0)}$   and
$\cE^{(7)}_{(1,0)}$
  that are used in section  \ref{exampleseven}
  in the text.

In $D=7$ dimensions the coefficient functions are automorphic under the action of the duality group $SL(5,\Z)$ and
are functions on  the  14-dimensional  coset   space   $SL(5)/SO(5)$, which  is   parametrized, using the notation that arises from string theory, by
\begin{equation}\label{e:bige5}
  e_5=
  \begin{pmatrix}
  &B^1_{\rm NS}&C^{(2)\,1}+\Omega_1 B^1_{\rm NS}\\
 u_3 &B^2_{\rm NS}&C^{(2)\,2}+\Omega_1 B^2_{\rm NS}\\
    &B^3_{\rm NS}&C^{(2)\, 3}+\Omega_1 B^3_{\rm NS}\\
    0&1&\Omega_1\\
    0&0&1
  \end{pmatrix}
\,
\begin{pmatrix}
&& &0&0\\
& \nu_3^{-2/15}D_3 &&0&0\\
&&&0&0\\
  0&0&0&{\nu_3^{\frac15}\, \sqrt{\Omega_2}}&0\\
  0&0&0&0&{\nu_3^{\frac15}\over \sqrt{\Omega_2}}
\end{pmatrix}\,,
\end{equation}
where $\Omega_2$ is the inverse string coupling constant,  $\Omega_1$ is the type IIB  $RR$ pseudoscalar,
   and  $B_{\rm NS}^i$ and $C^{(2)\, i}$ ($i=1,\dots,3$)  the $NS$ and  $RR$  charges.
 The quantity $u_3$ is  a $3\times 3$  unit upper
triangular matrix and $D_3$ is a $3\times 3$ diagonal matrix  with $\det D_3=1$.  These are defined so that
$ \tilde e_3=u_3 D_3$ or equivalently
  $\tilde G_3=\tilde e_3\cdot \tilde e_3^t$  parametrizes the coset
$SL(3)/SO(3)$  describing the perturbative string
  compactified on a three-torus.
We will make use of the following combinations,
\be
  \nu_3^{-1}=\left(r_1r_2r_3\over\ell_{10}^3\right)^2=\Omega_2^{3\over2}\,\left(r_1r_2r_3\over\ell_{s}^3\right)^2\,, \qquad
  y_7^{-1}= \Omega_2^2\, {r_1r_2r_3\over \ell_s^3}\,,
\ee
where $r_1$, $r_2$ and $r_3$ are the radii of $\rT^3$ and $y_7$ is the $7$-dimensional string coupling.
 Note that $\nu_3$ is invariant under the action of  $SL(3)\times SL(2)$.

The coset space $SL(5)/SO(5)$ is parametrized by the metric  $G_5=
e_5 e_5^t$,
\begin{equation} \label{e:g5coset}
G_5=\nu_3^{\frac25}
\begin{pmatrix}
\nu_3^{-\frac23}\,  (\tilde G_3)_{ij}+
  (G_2)_{ab} B_i^a B_j^b&  (G_2)_{ab}    B_j^b   \\
 (G_2)_{ab} B^a_j&(G_2)_{ab}\\
\end{pmatrix},
\end{equation}
 where again
\begin{equation}
\label{gambdef}
G_2={1\over \Omega_2}
\begin{pmatrix}
|\Omega|^2& \Omega_1\\
 \Omega_1&1\\
\end{pmatrix} \qquad \text{and}\qquad
B=
\begin{pmatrix}
  B^1_{\rm NS} & B^2_{\rm NS} & B^3_{\rm NS}\\
 C^{(2)\,1} & C^{(2)\,2} & C^{(2)\,3}\\
\end{pmatrix}.
 \end{equation}
The  $SL(5)$ mass squared is  given by the quadratic form
\begin{eqnarray}
\label{sl5massdef}
  m^2_{SL(5)}&:=&  M_5G_5 M_5^t\\
\nn&=&\nu_3^{\frac25} {|m_1+m_2\Omega+n^t\cdot  (C^{(2)}+\Omega B_{\rm NS})|^2\over \Omega_2}+ {n  \tilde G_3 n^t \over \nu_3^{4\over15}}\,,
\end{eqnarray}
where
$M_5:=[n_1\,n_2\,n_3\,m_2\,m_1]\in \ZZ^5\bsz$,
$n:=[n_1\,n_2\,n_3]$,  and  $B_{\rm NS}$  and $C^{(2)}$ are the
first and second rows of the matrix $B$, respectively.
This expression will later be useful for describing the $SL(5)$ Eisenstein series.

The   $\smallf 12$-BPS    and   $\smallf 14$-BPS   coefficients,    $\cE^{(7)}_{(0,0)}$   and
$\cE^{(7)}_{(1,0)}$,   that solve \eqref{laplaceeigenone} and
\eqref{laplaceeigentwo} together with the appropriate boundary
conditions      are       given\footnote{In~\cite{Green:2010wi}
   these series were defined as
   $\bE^{SL(5)}_{[1000];s}=2\zeta(2s)E^{SL(5)}_{\beta_1;s}$ and $\bE^{SL(5)}_{[0010];s}=4\zeta(2s)\zeta(2s-1)E^{SL(5)}_{\beta_3;s}$ }     in~\cite{Green:2010wi}  by linear combinations of the
$E^{SL(5)}_{\beta_1;s}$  and $E^{SL(5)}_{\beta_3;s}$    Eisenstein
series  as described in (\ref{esl510a}-\ref{esl5501a}).
The definitions and Fourier expansions  of the Eisenstein series in this expression will now be reviewed.

 \subsubsection{Fourier modes of the series $E^{SL(5)}_{\beta_1;s}$}\hfil\break
 \label{sec:sl5modes}

The $E^{SL(5)}_{\beta_1;s}$  series may be written using \eqref{sl5massdef} in the form
\begin{equation}
 2\zeta(2s)\,E^{SL(5)}_{\beta_1;s}=\sum_{M_5\in \ZZ^5\bsz} (M_5G_5 M_5^t)^{-s}\,.
   \label{e:R47d}
 \end{equation}
The constant terms with respect to the maximal parabolic subgroups
  $P_{\beta_3}$, $P_{\beta_1}$, and
  $P_{\beta_4}$  (corresponding to limits (i), (ii), and (iii), respectively) were evaluated
  in~\cite{Green:2010wi}.  Note that in terms of our matrix identification used in (\ref{e:R47d}), $P_{\a_1}$ corresponds to the subgroup of $SL(5)$ whose bottom row has the form $(0\,0\,0\,0\,\star)$.


(i) {\bf The   parabolic    $P_{\beta_3}=    GL(1)\times
  SL(2)\times SL(3)\times U_{\beta_3}$.}

The unipotent radical for this parabolic subgroup is
  abelian and is given by
\begin{equation}
  U_{\beta_3}=
  \begin{pmatrix}
    I_2&Q_4 \\
    0&I_3\\
  \end{pmatrix}\ , \ \ \text{ with}\ \ \  Q_4=
  \begin{pmatrix}
  G_{13}&  B_{\rm NS}^1 & C^{(2)\, 1}+\Omega_1B_{\rm NS}^1 \\
  G_{23}&  B_{\rm NS}^2 & C^{(2)\,2}+\Omega_1B_{\rm NS}^2 \\
  \end{pmatrix}.
  \label{Q4defa}
\end{equation}
The Fourier modes are defined by
\begin{equation}  \label{e:EpPalpha4def}
 F^{SL(5)\beta_3}_{\beta_1;s}(N_4):=\int_{[0,1]^6}\, d^6Q_4 \, e^{-2i\pi\, \tr(
   N_4 Q_4)} \, E^{SL(5)}_{\beta_1;s}\,,
\end{equation}
where $N_4\in M(3,2;\ZZ)$.

For all values of $s$ the Fourier modes are only non-zero
  when $N_4$ has rank 1.   Such a matrix can be written as  $N_4=k\, \tilde N_4$, with
  $\gcd(\tilde N_4)=1$ and
\begin{equation}\label{e:DefN4a}
   \tilde N_4=n^tm =
    \begin{pmatrix}
       m_1  n_1& m_2  n_1\\
       m_1  n_2& m_2  n_2\\
        m_1  n_3& m_2  n_3\\
    \end{pmatrix} \ , \quad n=(n_i)\in\ZZ^3,  \, m=(m_a)\in\ZZ^2
\,.
\end{equation}
The decomposition $N_4=kn^tm$ of the rank one matrix $N_4$ is unique up to signs of the factors.  Moreover, $\gcd(n_1,n_2,n_3)=\gcd(m_1,m_2)=1$.

Poisson resummation on two integers,  keeping the off-diagonal terms in the parametrisation of~\cite[section~B.5.2]{Green:2010wi}, results in the
following formula for the Fourier coefficients:
\begin{multline}  \label{e:EpPalpha4FourierModes}
 F^{SL(5)\, \beta_3}_{\beta_1;s}(k,\tilde N_4) \ \ = \ \    {2\over\xi(2s)}\,r^{3-{2s\over5} }
{\sigma_{2s-3}(|k|)\over |k|^{s-\frac32}}\,
\left(\smallf{ \|n( e_3^t)^{-1}\| }{
  \|m e_2\|}\right)^{s-{3\over 2}}\ \times \\
 K_{s-\frac32}\left(2\pi|k| \,r^2 \, \|m  e_2\| \,  \|n (e_3^t)^{-1}\| \right)  \,,
\end{multline}
where $e_2$ and $e_3$ are the $SL(2)$ and $SL(3)$ components, respectively, of the semisimple part of the Levi component of $P_{\b_3}$, and $\tilde e_3$ refers to the contragredient defined in (\ref{Soddpf2}).  Note $\|me_2\|$ and $\|n\tilde e_3\|$ are independent of the choice of factorization of $\tilde N_4=n^tm$.
The matrix $\tilde N_4$ is transformed by the action of  $SL(3,\ZZ)$
on the left and by the action of  $SL(2,\ZZ)$ on the right.  Because  $\tilde N_4$    has rank 1, it therefore
satisfies the $\smallf 12$-BPS conditions
$\epsilon_{ab} (N_4)_i{}^a(N_4)_j{}^b=0$ of section~\ref{sec:D8}.
In other words, for any value of $s$, the Fourier modes fill out $\smallf 12$-BPS orbits -- one for each value of $k$.

(ii) {\bf The  parabolic   $P_{\beta_1}=GL(1)\times
  SL(4) \times U_{\beta_1}$.}

  The unipotent radical for this parabolic is abelian and is given in our parameterisation  by
\begin{equation}
  U_{\beta_1}=
  \begin{pmatrix}
    I_4&Q_1 \\
    0&1\\
  \end{pmatrix}\,,  \ \ \text{with} \ \ \  Q_{1}=
  \begin{pmatrix}
    C^{(2)\, 1}+\Omega_1B_{\rm NS}^1\\ C^{(2)\,2}+\Omega_1B_{\rm NS}^2\\ C^{(2)\,3}+\Omega_1B_{\rm NS}^3 \\ \Omega_1
  \end{pmatrix} ,
  \label{Q1defa}
\end{equation}
where $I_4$ is the $4\times 4$  identity matrix and $Q_1$ is
a four-dimensional vector   that can also be thought of as a spinor for $Spin(3,3)$.

The Fourier modes are defined by
\begin{equation}  \label{e:EpPalpha1def}
 F^{SL(5)\beta_1}_{\beta_1;s}(k,N_1):=\int_{[0,1]^4} d^4Q_1 \, e^{-2i\pi\,k
   N_1  Q_1} \, E^{SL(5)}_{\beta_1;s}\,,
\end{equation}
where the row vector $N_1\in \ZZ^4$ is such that $\gcd(N_1)=1$. These Fourier modes  are evaluated by a straightforward extension of the expansion given in~\cite[section~B.5.1]{Green:2010wi}, which computed only the
constant terms (for which it is sufficient to set $Q_1=0$).  The result  is
\begin{equation}
F^{SL(5)\, \beta_1}_{\beta_1;s}(k,N_1)={2\over\xi(2s)}\,    r^{1+{6s\over    5}}\,
{\sigma_{2s-1}(|k|)\over |k|^{s-\frac12}}\,
{K_{s-\frac12}\left(2\pi|k|\,r^2 \,\| N_1e_4\|\right)\over \| N_1
  e_4\|^{s-{1\over2}}  } \,.
    \label{e:EpPalplha1FourierModes}
\end{equation}

\smallskip

 (iii) {\bf The parabolic   $P_{\beta_4}=GL(1)\times   SL(4)\times U_{\beta_4}$}

  The unipotent radical is abelian and given by
\begin{equation}
  U_{\beta_4}=
  \begin{pmatrix}
    1&Q_2 \\
    0&I_4\\
  \end{pmatrix}\,, \qquad Q_2=
  \begin{pmatrix}
    C_{123} \ C_{124} \ C_{234} \ C_{134}
  \end{pmatrix}\, ,
\end{equation}
where $Q_2$ is again a  $SL(4)$ (row) vector.  The notation indicates that it is
parametrized by the 3-form  flux of the $M2$-brane world-volume wrapped  on the M-theory 4-torus, $\calT^4$.  This translates into  the $NS$ components of flux,
$B_{{\rm NS}\, 12}, B_{{\rm NS}\, 23}, B_{{\rm NS}\, 13}$,  and the $RR$
$D2$-brane  flux, $C^{(3)}_{123}$.   In  type~IIB language these components become the $NS$ flux
$B_{{\rm NS}\, 12}$,  the $RR$ $D$-string flux  $C^{(2)}_{12}$ and  the
Kaluza--Klein momenta  from the components of the metric $g_{i\,3}$ with $i=1,2$.

The Fourier coefficients  in this parabolic are indexed by
$k\in\ZZ$ and $N_4\in\ZZ^4$ with $\gcd(N_4)=1$ by the formula
\begin{equation}
      \label{e:EpPalpha2Def}
F^{SL(5)\beta_4}_{\beta_1;s}(k,N_4):= \int_{[0,1]^4} d^4Q_2 \,
e^{-2i\pi\,k\, N_4^t\cdot Q_2}\, E^{SL(5)}_{\beta_1;s}\,.
\end{equation}
These coefficients can  again be evaluated  by an extension of the
computation of~\cite[section~B.5.1]{Green:2010wi}, keeping the
off-diagonal terms, which gives
\begin{multline}
  \label{e:EpPalpha2FourierModes}
F^{SL(5)\, \beta_4}_{\beta_1;s}(k,N_4) \ \ = \\
{2\over\xi(2s)}\, r^{4-{6s\over5}}\,
{\sigma_{2s-4}(|k|)\over |k|^{s-2}}\,\| N_4e_4^{-1}\|^{s-2}\,K
_{s-2}\left(2\pi|k|\, r^2\,\| N_4 e_4^{-1}\|\right)\,,
\end{multline}
where $r={\mathcal V}_4^{3/8}\ell_{11}^{-3/2}$, and again $\gcd(N_4)=1$.


\subsubsection{Fourier modes of the series $E_{\beta_2;s}^{SL(5)}$ }\hfil\break
\label{nonepsteinmodes}

Our method of determining the Fourier modes of the non-Epstein $SL(5)$ series $E_{\beta_3;s}^{SL(5)}$ is based on the integral representation described in proposition~\ref{prop:nonepsteinintegralrep}.  For computational reasons it is easier to work with the series $E_{\beta_2;s}^{SL(5)}$, which is related both by the functional equation in  (\ref{I04}) and  the contragredient map $g\mapsto \widetilde g$ defined in (\ref{Soddpf2}).  Here we shall compute its nonzero Fourier modes in each of the four standard maximal parabolic subgroups $P_{\b_1}$, $P_{\b_2}$, $P_{\b_3}$, and $P_{\b_4}$ of $SL(5)$; the relevant Fourier modes for $E_{\beta_3;s}^{SL(5)}$ will be derived  from this in section~\ref{exampleseven}.

\paragraph{ {\bf The parabolic}   $P_{\b_1}=SL(4)\times GL(1)\times U_{\b_1}$}\label{subsubsub1}\hfil\break

In this case $e$ has the special form $\ttwo{I_4}{Q}{}{1}\ttwo{e_4}{}{}{e_1}$, where $Q\in M_{4,1}(\IR)$, $e_1\neq 0$, and $e_4\in GL(4,\IR)$.  Note that we do not assume that $\det e=1$, so that we can later utilize proposition~\ref{prop:nonepsteinintegralrep}. The sum (\ref{bigGdef}) can be written as
\begin{equation}\label{bigGdeftris}
    {\mathcal G}(\tau,ee^t) \ \ := \ \ \sum_{ [\srel{p\,m}{q\,n}]\, \in \, {\mathcal M}_{2,5}^{(2)}(\Z)}
    e^{-\pi \tau_2^{-1} \|[p+q\tau\,,\,m+n\tau]e\|^2}\,,
\end{equation}
where $p,q\in \Z^4$ and $m,n\in \Z$.  For emphasis we have used commas to separate the entries of the row vectors.
The exponent is
\begin{equation}\label{nonepsinP1a}
-\pi \tau_2^{-1} \|(p+q\tau)e_4\|^2  \ - \ \pi \tau_2^{-1}e_1^2|(p+q\tau)Q+m+n\tau|^2\,.
\end{equation}
This is independent of $Q$ if both $p=q=0$.  Hence the nonzero Fourier
coefficients of (\ref{bigGdeftris}) come from terms where $[\srel pq]$
has rank 1 or 2.   We thus separate these contributions and write
\begin{equation}\label{bigGdeftrisa}
    {\mathcal G}(\tau,ee^t) \ \ = \ \ {\mathcal G}_1(\tau,ee^t) \ + \ {\mathcal G}_2(\tau,ee^t)\,,
\end{equation}
where
\begin{equation}\label{bigGitris}
   {\mathcal G}_i(\tau,ee^t) \ \ := \ \ \sum_{\srel{\operatorname{rank}[\srel{p}{q}]\,=\,i}{[\srel{p\,m}{q\,n}]\, \in \, {\mathcal M}_{2,5}^{(2)}(\Z)}}
    e^{-\pi \tau_2^{-1}\|[p+q\tau\,,\,m+n\tau]e\|^2}\,.
\end{equation}
Let us first consider $ {\mathcal G}_2(\tau,ee^t) $.  Changing $\tau$ to $\tau+1$ is equivalent to changing $(p,q,m,n)$ to $(p+q,q,m+n,n)$, while changing $\tau$ to $-\tau^{-1}$ is equivalent to changing $(p,q,m,n)$ to $(-q,p,-n,m)$.  Thus the sum is modular invariant and can be written as
 a sum over $SL(2,\Z)$ cosets:
\begin{equation}\label{bigGdeftrisb}
    {\mathcal G}_2(\tau,ee^t) \ \ = \ \ \sum_{\g\,\in\,SL(2,\Z)}
    {\mathcal G}^\circ_2(\g\tau,ee^t)\,,
\end{equation}
where
\begin{equation}\label{bigGdeftrisc}
    {\mathcal G}^\circ_2(\tau,ee^t) \  =      \sum_{\srel{[\srel{p}{q}]\,\in\,SL(2,\Z)\backslash {\mathcal M}_{2,4}^{(2)}(\Z)}{m,n\,\in\,\Z}}  e^{-\pi \tau_2^{-1} \|(p+q\tau)e_4\|^2   -  \pi \tau_2^{-1}e_1^2|(p+q\tau)Q+m+n\tau|^2}
\end{equation}
(here we have used  that $\operatorname{rank}[\srel{p}{q}]=2$ implies
that $\operatorname{rank}[\srel{p\,m}{q\,n}]=2$).  Applying Poisson summation over $m$ and $n$ this is
\begin{multline}\label{bigGdeftrisd}
    {\mathcal G}^\circ_2(\tau,ee^t) \ \ = \  \   \sum_{\srel{[\srel{p}{q}]\,\in\,SL(2,\Z)\backslash {\mathcal M}_{2,4}^{(2)}(\Z)}{\hat{m},\hat{n}\,\in\,\Z}}  e^{-\pi \tau_2^{-1} \|(p+q\tau)e_4\|^2} e^{2\pi i (\hat{m} p+\hat{n}q)Q} \ \times \\ \int_{\IR^2} e^{-2\pi i (\hat{m} m+\hat{n} n)} e^{-  \pi \tau_2^{-1}e_1^2|m+n\tau|^2}\,dm\,dn\,.
\end{multline}
Thus its Fourier coefficient for $e^{2\pi i N_1Q}$,  $N_1\in\Z^4$, is equal to
\begin{multline}\label{bigGdeftrise}
    {\mathcal F}{\mathcal G}^\circ_2(\tau,e_1,e_4;N_1) \ \ = \ \
     \sum_{\srel{[\srel{p}{q}]\,\in\,SL(2,\Z)\backslash {\mathcal M}_{2,4}^{(2)}(\Z)}{\srel{\hat{m},\hat{n}\,\in\,\Z}{\hat{n}q-\hat{m}p\,=\,N_1}}}
            e^{-\pi \tau_2^{-1} \|(p+q\tau)e_4\|^2} \ \times\\
\int_{\IR^2} e^{2\pi i \hat{m}m-(\hat{n}+\hat{m}\tau_1) n} e^{-  \pi \tau_2^{-1}e_1^2 (m^2+n^2\tau_2^2)}\,dm\,dn \\
            = \ \ e_1^{-2}
     \sum_{\srel{[\srel{p}{q}]\,\in\,SL(2,\Z)\backslash {\mathcal M}_{2,4}^{(2)}(\Z)}{\srel{\hat{m},\hat{n}\,\in\,\Z}{\hat{n}q-\hat{m}p\,=\,N_1}}}
            e^{-\pi \tau_2^{-1} \|(p+q\tau)e_4\|^2}
             e^{-\pi \tau_2^{-1}e_1^{-2}|\hat{n}+\hat{m}\tau|^2}
            \,.
\end{multline}
Analogously to (\ref{bigGdeftrisb})
\begin{equation}\label{bigGdeftrisf}
    {\mathcal G}_1(\tau,ee^t) \ \ = \ \ \sum_{\g\,\in\,\{\pm \G_\infty\}\backslash SL(2,\Z)}
    {\mathcal G}^\circ_1(\g\tau,ee^t)\,,
\end{equation}
 where $\G_\infty=\{\ttwo 1n01|n\in \Z\}$ and
\begin{equation}\label{bigGdeftrisg}
    {\mathcal G}^\circ_1(\tau,ee^t) \ \ := \  \   \sum_{\srel{p\,\neq\,0}{\srel{m\,\in\,\Z}{n\,\neq\,0}}}  e^{-\pi \tau_2^{-1} \|pe_4\|^2   -  \pi \tau_2^{-1}e_1^2|pQ+m+n\tau|^2}
\end{equation}
(this parametrization is due to the fact that any $SL(2,\Z)$ orbit in ${\mathcal M}^{(1)}_{2,4}(\Z)$ has an element with bottom row equal to zero, and that the rank 2 condition is then equivalent to the bottom right entry, $n$, being nonzero).  Applying Poisson summation over $m$ gives the formula
\begin{multline}\label{bigGdeftrish}
      {\mathcal G}^\circ_1(\tau,ee^t) \ \   = \  \   \sum_{\srel{p\,\neq\,0}{\srel{\hat{m}\,\in\,\Z}{n\,\neq\,0}}}  e^{-\pi \tau_2^{-1} \|pe_4\|^2} \int_{\IR}e^{-2\pi i \hat{m}m} \, e^{-  \pi \tau_2^{-1}e_1^2|pQ+m+n\tau|^2}\,dm \\
        =    \tau_2^{\frac12}\,e_1^{-1}\,  \sum_{\srel{p\,\neq\,0}{\srel{\hat{m}\,\in\,\Z}{n\,\neq\,0}}} e^{2\pi i \hat{m}(pQ+n\tau_1)}
  e^{-\pi \tau_2^{-1} \|pe_4\|^2-\pi\tau_2e_1^{-2}\hat{m}^2 - \pi \tau_2e_1^2n^2}\,.
\end{multline}
Since $N_1\neq 0$ its Fourier mode for $e^{2\pi i N_1Q}$ is thus
\begin{multline}\label{bigGdeftrishi}
    {\mathcal F} {\mathcal G}^\circ_1(\tau,e_1,e_4;N_1)  \ \ = \\ \tau_2^{\frac12}\,e_1^{-1}\,
    \sum_{\srel{\hat{m}p\,=\,N_1}{n\,\neq\,0}}  e^{2\pi i \hat{m}n\tau_1}  e^{-\pi \tau_2^{-1} \|pe_4\|^2-\pi\tau_2e_1^{-2}\hat{m}^2 - \pi \tau_2e_1^{2}n^2}\,.
\end{multline}

It follows using proposition~\ref{prop:nonepsteinintegralrep} that the nonzero Fourier modes of $F_{\beta_2;s}^{SL(5)\, \beta_1}$ are given by
\begin{multline}\label{nonepsinPa1}
   \smallf12\xi(2s)\xi(2s-1)F_{\beta_2;s}^{SL(5)\, \beta_1}(N_1) \ \ = \\ 2\,
    \int_0^\infty \int_{\U} {\mathcal
      {FG}}_2^\circ(\tau,u^{\frac12}e_1,u^{\frac12}e_4;N_1)\,\f{d^2\tau}{\tau_2^2}\,\f{du}{u^{1-2s}}
    \  \\
+ \ \int_0^\infty \int_{\G_\infty\backslash \U} {\mathcal{FG}}_1^\circ(\tau,u^{\frac12}e_1,u^{\frac12}e_4;N_1)\,\f{d^2\tau}{\tau_2^2}\,\f{du}{u^{1-2s}}\,,
\end{multline}
the factor of 2 coming from unfolding pairs of elements $\pm\g\in SL(2,\Z)$ that have identical actions on $\U$.
By integrating the expression given in (\ref{bigGdeftrishi}) for  $ {\mathcal F} {\mathcal G}^\circ_1(\tau,e_1,e_4;N_1)$
over the strip $\G_\infty\backslash \U$,  the $\tau_1$-integration over $[0,1]$ forces $\hat{m}n$ to vanish. Since $n\neq 0$ this means $N_1=0$, and hence there are no nontrivial Fourier contributions from ${\mathcal G}_1$.

The contribution from the modes $ {\mathcal F} {\mathcal G}^\circ_2$
is given by
\begin{multline}
 2\,  e_1^{-2}\, \int_0^\infty \int_{\mathbb H}
\sum_{\srel{[\srel{p}{q}]\,\in\,SL(2,\Z)\backslash {\mathcal M}_{2,4}^{(2)}(\Z)}{\srel{\hat{m},\hat{n}\,\in\,\Z}{\hat{n}q-\hat{m}p\,=\,N_1}}}\\
    e^{-\pi \tau_2^{-1} u \|(p+q\tau)e_4\|^2-\pi\tau_2^{-1} u^{-1}e_1^{-2}|\hat{n}+\hat m \tau|^2 }\,\f{d^2\tau}{\tau_2^2}\f{du}{u^{2-2s}}\,.
\end{multline}
Changing variables to $x=u/\tau_2$ and $y=\tau_2u$, so that $u=\sqrt{xy}$, $\tau_2=\sqrt{y/x}$, and $d\tau_2 du = \f{dx dy}{2x}$, yields
\begin{multline}\label{e:Ebeta3Pbeta1}
 e_1^{-2} \int_0^\infty  \sum_{\srel{[\srel{p}{q}]\,\in\,SL(2,\Z)\backslash {\mathcal M}_{2,4}^{(2)}(\Z)}{\srel{\hat{m},\hat{n}\,\in\,\Z}{\hat{n}q-\hat{m}p\,=\,N_1}}}
 \int_0^\infty\, e^{-\pi x  \|(p+q\tau_1)e_4\|^2 - \pi x^{-1} e_1^{-2} \hat m^2}\,\f{dx}{x^{1-s}} \\ \times \
 \int_0^\infty e^{-\pi y  \|q e_4\|^2-\pi y^{-1}e_1^{-2}(\hat{n}+\hat m \tau_1)^2    }
\f{dy}{y^{2-s}}
d\tau_1\\
=  \ \  4\, e_1^{-2} \int_0^\infty \sum_{\srel{[\srel{p}{q}]\,\in\,SL(2,\Z)\backslash {\mathcal M}_{2,4}^{(2)}(\Z)}{\srel{\hat{m},\hat{n}\,\in\,\Z}{\hat{n}q-\hat{m}p\,=\,N_1}}}
 \left(  |  \hat m|\over\|e_1(p+q\tau_1)e_4\| \right)^{s} \left(  | \hat n+\hat m \tau_1|\over \|e_1qe_4\| \right)^{s-1} \ \times
\\    K_s(\,2\pi \,  | \hat m|\,\|e_1^{-1}(p+q\tau_1)e_4\|  \,  ) \,  K_{s-1}(\,2\pi  \,  | \hat n+\hat m \tau_1| \,  \|e_1^{-1}qe_4\|\, ) d\tau_1\,.
\end{multline}

\paragraph{{\bf The parabolic}   $P_{\b_2}= GL(1)\times SL(3)\times   SL(2)\times U_{\b_2}$}\label{subsubsub2}

We may rewrite (\ref{bigGdef})  in the case of  $d=5$ as
\begin{equation}\label{bigGdefbis2}
    {\mathcal G}(\tau,ee^t) \ \ := \ \ \sum_{ [\srel{p\,m}{q\,n}]\, \in \, {\mathcal M}_{2,5}^{(2)}(\Z)}
    e^{-\pi \tau_2^{-1}\| [p+q\tau\,,\,m+n\tau]e\|^2}\,,%
\end{equation}
where $p,q\in\Z^3$ and $m,n\in\Z^2$.
Let us further take $e$ to have the special form $e=\ttwo{I_3}{Q}{}{I_2}\ttwo{e_3}{}{}{e_2}$, where  $Q\in M_{3,2}(\IR)$, $e_2\in GL(2,\IR)$, and $e_3\in GL(3,\IR)$.  We will be interested in Fourier coefficients in $Q$ for the Fourier modes $Q\mapsto e^{2\pi i \tr NQ}$, where $N\in M_{2, 3}(\Z)$.
 Break up the sum as
\begin{equation}\label{bigGbreakup2}
    {\mathcal G}(\tau,ee^t) \ \ = \ \  {\mathcal G}_0(\tau,ee^t) \ + \  {\mathcal G}_1(\tau,ee^t) \ + \  {\mathcal G}_2(\tau,ee^t)\,,\\
\end{equation}
where
\begin{equation}\label{bigGi2}
   {\mathcal G}_i(\tau,ee^t) \ \ := \ \ \sum_{\srel{\operatorname{rank}[\srel{p}{q}]\,=\,i}{ [\srel{p\,m}{q\,n}]\, \in \, {\mathcal M}_{2,5}^{(2)}(\Z)}}
    e^{-\pi \tau_2^{-1} \|[p+q\tau\,,\,m+n\tau]e\|^2}\,.%
\end{equation}
If $\operatorname{rank}[\srel{p}{q}]=2$, then $[\srel{p\,m}{q\,n}]$ automatically has rank 2.  Thus
\begin{equation}\label{bigG2ba2}
     {\mathcal G}_2(\tau,ee^t) \ \ := \ \ \sum_{\srel{\operatorname{rank}[\srel{p}{q}]\,=\,2}{ m,n\,\in\,\Z^2}}
    e^{-\pi \tau_2^{-1} \|[p+q\tau\,,\,m+n\tau]e\|^2}\,.%
\end{equation}
Using the method of orbits we may write this as an average over $SL(2,\Z)$:
\begin{equation}\label{bigG2baave2}
     {\mathcal G}_2(\tau,ee^t) \ \  = \ \ \sum_{\g\,\in\,SL(2,\Z)}  {\mathcal G}^\circ_2(\g\tau,ee^t)\,,
\end{equation}
where
\begin{equation}\label{bigG2circdef2}
    {\mathcal G}^\circ_2(\tau,ee^t) \ \ = \ \ \sum_{\srel{[\srel pq ] \, \in \, SL(2,\Z)\backslash {\mathcal M}_{2,3}^{(2)}(\Z)}{ m,n\,\in\,\Z^2}}
    e^{-\pi \tau_2^{-1} \|[p+q\tau\,,\,m+n\tau]e\|^2}\,.%
\end{equation}

Poisson summation over the inner $m,n\in\Z^2$ sum gives
\begin{multline}\label{bigG2a2}
      {\mathcal G}^\circ_2(\tau,ee^t) \ \ = \\
      \sum_{\srel{[\srel pq ] \, \in \, SL(2,\Z)\backslash {\mathcal M}_{2,3}^{(2)}(\Z)}{ \hat{m},\hat{n}\,\in\,\Z^2}}
      \int_{\IR^4}e^{-2\pi i (  m\hat{m}-n\hat{n})}
      e^{-\pi \tau_2^{-1} \|[p+q\tau\,,\,m+n\tau]e\|^2}\,dm\,dn\,,
\end{multline}
where $\hat{m}, \hat{n}\in \Z^2$ are column vectors.
With the particular form $e=\ttwo{I_3}{Q}{}{I_2}\ttwo{e_3}{}{}{e_2}$
the exponent of the second factor is
\begin{equation}\label{bigG2b2}
    -\pi\tau_2^{-1}[p+q\tau \,,\, (p+q\tau)Q+m+n\tau]\ttwo{e_3e_3^t}{}{}{e_2e_2^t}[p+q\tau \,,\, (p+q\bar\tau)Q+m+n\bar\tau]^t.
\end{equation}
Thus after changing variables $m\mapsto m-pQ$, $n\mapsto n-qQ$ (\ref{bigG2a2}) becomes
\begin{multline}\label{bigG2c2}
     {\mathcal G}^\circ_2(\tau,ee^t) \ \ = \ \
     \sum_{[\srel pq ] \, \in \, SL(2,\Z)\backslash {\mathcal M}_{2,3}^{(2)}(\Z)} e^{-\pi\tau_2^{-1}\|(p+q\tau) e_3\|^2} \sum_{ \hat{m},\hat{n}\,\in\,\Z^2}e^{2\pi i  (pQ\hat{m}- qQ\hat{n})}
   \ \times \\   \int_{\IR^4}e^{-2\pi i ( m\hat{m}-n\hat{n})}
      e^{-\pi \tau_2^{-1} \|(m+n\tau)e_2\|^2}\,dm\,dn\,.
\end{multline}
To compute this integral we change variables $m\mapsto m e_2^{-1}$, $n\mapsto ne_2^{-1}$, which has the effect of dividing both $dm$ and $dn$ each by $\det e_2$:~the integral equals $(\det e_2)^{-2}$ times
\begin{multline}\label{bigG2d2}
    \int_{\IR^4}e^{-2\pi i
   (me_2^{-1}\hat{m}-ne_2^{-1}\hat{n})}
    e^{-\pi\tau_2^{-1}\|(m+n\tau)\|^2}\,dm\,dn \ \ = \\
     \int_{\IR^4}e^{-2\pi i
      (me_2^{-1}\hat{m}-ne_2^{-1}(\hat{n}+\tau_1\hat{m}))
      }
    e^{-\pi\tau_2^{-1}\|m\|^2-\pi\tau_2\|n\|^2}\,dm\,dn
\end{multline}
after changing variables $m\mapsto m-n\tau_1$ in the last step.
 This then factors as two Fourier transforms of Gaussians and
 (\ref{bigG2c2}) is equal to
\begin{multline}\label{bigG2e2}
       {\mathcal G}^\circ_2(\tau,ee^t) \ \ =  (\det e_2)^{-2}
  \sum_{[\srel pq ] \, \in \, SL(2,\Z)\backslash {\mathcal M}_{2,3}^{(2)}(\Z)}
 e^{-\pi\tau_2^{-1}\|(p+q\tau) e_3\|^2}\\ \sum_{ \hat{m},\hat{n}\,\in\,\Z^2}e^{2\pi i  (pQ\hat{m}-qQ\hat{n})}
    e^{-\pi\tau_2\|e_2^{-1}\hat{m}\|^2-\pi\tau_2^{-1}\|e_2^{-1}(\hat{n}+\hat{m}\tau_1)\|^2}\,.
\end{multline}
The dependence on $Q$ is manifest in the exponential factors in the sum, and hence taking Fourier coefficients in $Q$ amounts to restricting $p$, $q$, $\hat{m}$, and $\hat{n}$.
In particular the Fourier coefficient for $N_4\in M_{2,3}(\Z)$ is equal to
\begin{multline}\label{bigG2FC2}
        {\mathcal F}{\mathcal G}_2^\circ(\tau,e_2,e_3;N_4) \ \ = \ \ (\det e_2)^{-2}
      \sum_{[\srel pq ] \, \in \, SL(2,\Z)\backslash {\mathcal M}_{2,3}^{(2)}(\Z)}
     \\   \sum_{\srel{\hat{m},\hat{n}\,\in\,\Z^2}{\hat{m}p-\hat{n}q=N_4}}
        e^{-\pi\tau_2^{-1}\|(p+q\tau) e_3\|^2 -\pi\tau_2\|e_2^{-1}\hat{m}\|^2-\pi\tau_2^{-1}\|e_2^{-1}(\hat{n}+\hat{m}\tau_1)\|^2}\,.
\end{multline}

Let us now consider ${\mathcal G}_1(\tau,ee^t)$, which has the contributions for $p,q\in\Z^3$ such that $\operatorname{rank}[\srel pq]=1$:
\begin{equation}\label{bigG12}
   {\mathcal G}_1(\tau,ee^t) \ \ := \ \ \sum_{\srel{\operatorname{rank}[\srel{p}{q}]\,=\,1}{ [\srel{p\,m}{q\,n}]\, \in \, {\mathcal M}_{2,5}^{(2)}(\Z)}}
    e^{-\pi \tau_2^{-1} \|[p+q\tau\,,\,m+n\tau]e\|^2}\,.
\end{equation}
We may write this as an average   over $\{\pm \G_\infty\}\backslash SL(2,\Z)$:
\begin{equation}\label{bigG1ave2}
     {\mathcal G}_1(\tau,ee^t) \ \ = \ \ \sum_{\g\,\in\,\{\pm \G_\infty\}\backslash SL(2,\Z)}
      {\mathcal G}^\circ_1(\g\tau,ee^t)\,,
\end{equation}
where
\begin{multline}\label{bigG1circdef2}
     {\mathcal G}^\circ_1(\tau,ee^t) \ \ := \ \
     \sum_{\srel{p\,\neq\,0}{ [\srel{p\,m}{0\,n}]\, \in \, {\mathcal M}_{2,5}^{(2)}(\Z)}}
    e^{-\pi \tau_2^{-1} \|[p,\,m+n\tau]e\|^2} \\
    = \ \
    \sum_{\srel{p\,\neq\,0}{\srel{n\,\neq\,0}{m\,\in\,\Z^2}}}
    e^{-\pi \tau_2^{-1}\|pe_3\|^2-\pi \tau_2^{-1} \|(pQ+m+n\tau_1)e_2\|^2-\pi \tau_2 \|ne_2\|^2}     \,.
\end{multline}
Here we used that the matrix $ [\srel{p\,m}{0\,n}]$ has rank 2 if and only if $n\neq 0$ (since $p\neq 0$).  Poisson sum over $m$ then  gives the formula
\begin{equation}\label{bigG1a2}
{\mathcal G}^\circ_1(\tau,ee^t) \ \ = \ \
 {\tau_2 \over \det e_2}
  \sum_{\srel{p\,\neq\,0}{\srel{n\,\neq\,0}{\hat{m}\,\in\,\Z^2}}}
  e^{2\pi i (pQ+n\tau_1)\hat{m}}
 e^{-\pi \tau_2^{-1}\|pe_3\|^2-\pi \tau_2 \|ne_2\|^2
-\pi \tau_2 \|e_2^{-1}\hat{m}\|^2}
     \end{equation}
for  (\ref{bigG1circdef2}), where again $\hat{m}\in\Z^2$ is a column vector.

We conclude that the Fourier coefficient of ${\mathcal G}^\circ_1(\tau,ee^t)$ for $N_4$ is equal to
\begin{multline}\label{bigG1FC2}
     {\mathcal F}{\mathcal G}_1^\circ(\tau,e_2,e_3;N_4) \ \ = \\
    {\tau_2 \over \det e_2}
        \sum_{\srel{p\,\neq\,0}{\srel{n\,\neq\,0}{\hat{m}p=N_4}}}
    e^{2\pi i \tau_1 n\hat{m}}
         e^{-\pi \tau_2^{-1}\|pe_3\|^2-\pi \tau_2 \|ne_2\|^2
-\pi \tau_2 \|e_2^{-1}\hat{m}\|^2}
        \,.
\end{multline}
Note that ${\mathcal F}{\mathcal G}_1(\tau,e_2,e_3;N_4)\equiv 0$ if $\operatorname{rank}(N_4)=2$.
Finally since $[0\,0\,0\,\star\,\star]\ttwo{I_3}{Q}{}{I_2}$ is independent of $Q$,  so too is ${\mathcal G}_0(\tau,ee^t)$, the sum over terms with $p=q=[0\,0\,0]$.  It therefore has no nontrivial Fourier coefficients.

We now return to the identity of proposition~\ref{prop:nonepsteinintegralrep},
\begin{equation}\label{nonepstrepnreversed2}
  \frac12 \xi(2s)\xi(2s-1)\,E^{SL(5)}_{\beta_2;s}(e) \ \ = \ \
  \int_0^\infty \int_{SL(2,\Z)\backslash \U} {\mathcal G}(\tau, u
  ee^t)\,\f{d^2\tau}{\tau_2^2} \,{du\over u^{1-2s}}\,,
\end{equation}
with the specialization that  $e\in SL(d,\IR)$ has the form
$e=\ttwo{I_3}{Q}{}{I_2}\ttwo{e_3}{}{}{e_2}$.  Its Fourier coefficient
for $N_4$ can be written as
\begin{multline}\label{nonepsinothernoneps12}
   \frac12\xi(2s)\xi(2s-1)F_{\b_2;s}^{SL(5)\,\b_2}(e_2,e_3;N_4) \ \ = \\
   2\ \int_0^\infty \int_{\U} {\mathcal
      {FG}}_2^\circ(\tau,u^{\frac12}e_2,u^{\frac12}e_3;N_4)\,\f{d^2\tau}{\tau_2^2}\,\f{du}{u^{1-2s}}
    \  \\
+ \ \int_0^\infty \int_{\G_\infty\backslash \U} {\mathcal{FG}}_1^\circ(\tau,u^{\frac12}e_2,u^{\frac12}e_3;N_4)\,\f{d^2\tau}{\tau_2^2}\,\f{du}{u^{1-2s}}\,.
\end{multline}

Let us consider the first integral,
\begin{multline}\label{nonepsteininothernoneps22}
   2 \ \int_0^\infty \int_{ \U} (\det e_2)^{-2}
         \sum_{[\srel pq ] \, \in \, SL(2,\Z)\backslash {\mathcal M}_{2,3}^{(2)}(\Z)}
    e^{-\pi\tau_2^{-1}u\|(p+q\tau) e_3\|^2}  \\   \sum_{\srel{\hat{m},\hat{n}\,\in\,\Z^2}{\hat{m}p-\hat{n}q=N_4}}
        e^{-\pi\tau_2 u^{-1}\|e_2^{-1}\hat{m}\|^2-\pi\tau_2^{-1}u^{-1}\|e_2^{-1}(\hat{n}+\hat{m}\tau_1)\|^2}
        \,\f{d^2\tau}{\tau_2^2}\,\f{du}{u^{3-2s}}\,.
\end{multline}
Changing variables to $x=u/\tau_2$ and $y=\tau_2 u$, so that $u=\sqrt{xy}$, $\tau_2=\sqrt{y/x}$ and $d\tau_2 du = \f{dx dy}{2x}$ the integral
becomes
\begin{multline}\label{nonepsteininothernoneps42}
    (\det e_2)^{-2} \int_{\IR} \int_0^\infty \int_0^\infty
       \sum_{[\srel pq ] \, \in \, SL(2,\Z)\backslash {\mathcal M}_{2,3}^{(2)}(\Z)}
 \sum_{\srel{\hat{m},\hat{n}\,\in\,\Z^2}{\hat{m}p-\hat{n}q=N_4}}\\
      e^{-\pi x \|(p+q\tau_1) e_3\|^2-\pi
        x^{-1}\|e_2^{-1}\hat{m}\|^2}
e^{ -\pi y \|q e_3\|^2-\pi y^{-1}\|e_2^{-1}(\hat{n}+\hat{m}\tau_1)\|^2}
      \,\f{dx}{x^{3/2-s}}\,\f{dy}{y^{5/2-s}}d\tau_1\,.
\end{multline}
Integrating over $x$ and $y$  yields
\begin{multline}\label{nonepsteininothernoneps52}
 4 \, (\det  e_2)^{-2} \int_{\IR}
 \sum_{[\srel pq ] \, \in \, SL(2,\Z)\backslash {\mathcal M}_{2,3}^{(2)}(\Z)}
 \sum_{\srel{\hat{m},\hat{n}\,\in\,\Z^2}{\hat{m}p-\hat{n}q=N_4}}
    \left(\|(p+q\tau_1) e_3\|\over \|e_2^{-1}\hat{m}\|  \right)^{1/2-s} \ \times \\
 \left( \|q e_3\|\over \|e_2^{-1}(\hat{n}+\hat{m}\tau_1)\|\right)^{3/2-s}
K_{s-1/2}(2\pi \|(p+q\tau_1) e_3\|\|e_2^{-1}\hat{m}\|) \ \times \\ K_{s-3/2}(2\pi \|q e_3\|\|e_2^{-1}(\hat{n}+\hat{m}\tau_1)\|)  d\tau_1
\end{multline}
for the first line on the righthand side of (\ref{nonepsinothernoneps12}).

Next we analyze the second integral in (\ref{nonepsinothernoneps12}), in which we assume $N_4$ has rank 1 (since it vanishes if it has rank 2):
\begin{multline}\label{nonepsintherothernonepssecondparta2}
     \smallf{1}{\det e_2} \int_0^\infty \int_{\G_\infty\backslash \U}
        \sum_{\srel{p\,\neq\,0}{\srel{n\,\neq\,0}{\hat{m}p=  N_4}}}
  e^{2\pi i \hat{m}\cdot n\tau_1} \ \times \\
         e^{-\pi \tau_2^{-1}u \|pe_3\|^2-\pi \tau_2 u \|ne_2\|^2
-\pi \tau_2 u^{-1} \|e_2^{-1}\hat{m}\|^2}
    \,\f{d^2\tau}{\tau_2 }\,\f{du}{u^{2-2s}}\,.
\end{multline}
The $\tau_1$ integration over $[0,1]$ enforces the condition that
$n\perp \hat{m}$ (which implies $n\perp   N_4$):~(\ref{nonepsintherothernonepssecondparta2}) equals
\begin{equation}\label{nonepsintherothernonepssecondpartb2}
     \smallf{1}{\det e_2}  \int_0^\infty  \int_0^\infty
        \sum_{\srel{p\,\neq\,0}{\srel{n\,\neq
        \,0}{\srel{\hat{m}\,\perp\,n}{\hat{m}p=  N_4}}}}
         e^{-\pi \tau_2^{-1}u\|pe_3\|^2-\pi \tau_2 u \|ne_2\|^2-\pi \tau_2 u^{-1} \|e_2^{-1}\hat{m}\|^2}
    \,\f{d\tau_2}{\tau_2 }\,\f{du}{u^{2-2s}}\,.
\end{equation}
As before, change variables
$x=u/\tau_2$ and $y=\tau_2 u$ so that (\ref{nonepsintherothernonepssecondpartb2}) becomes
\begin{multline}\label{nonepsintherothernonepssecondpartc2}
   \f{1}{2(\det e_2)} \sum_{\srel{p\,\neq\,0}{\srel{n\,\neq
        \,0}{\srel{\hat{m}\,\perp\,n}{\hat{m}p= N_4}}}}  \int_0^\infty  \int_0^\infty
         e^{-\pi x \|pe_3\|^2-\pi x^{-1} \|e_2^{-1}\hat{m}\|^2} e^{-\pi y \|ne_2\|^2}
 \,\f{dx}{x^{\frac32-s}}\,\f{dy}{y^{\frac32-s}} \\
     = \  \
        \f{ \G\left(s-\frac12\right)}{(\det e_2)} \sum_{\srel{p\,\neq\,0}{\srel{n\,\neq
        \,0}{\srel{\hat{m}\,\perp\,n}{\hat{m}p=  N_4}}}}
 \left(\| e_2^{-1}\hat{m}  \| \over
    \pi \|ne_2 \|^2 \|  pe_3 \|\right)^{s-1/2}
     K_{s-1/2}(2\pi\|e_2^{-1}\hat{m} \| \| pe_3 \| )  \,.
\end{multline}

The matrices $e_2$ and $e_3$ in the above argument are unconstrained except for the condition that $\det(e_2)\det(e_3)=1$.  For our application in section~\ref{exampleseven} it will be helpful to restate these calculations using the $GL(1)$ parameter $r$ from (\ref{notation}).  We set
\begin{equation}\label{rforgl5beta2}
    \ttwo{e_3}{}{}{e_2} \ \ = \ \ \ttwo{r^{4/5}e_3'}{}{}{r^{-6/5}e_2'} \, ,
\end{equation}
where $e_2'\in SL(2,\IR)$ and $e_3'\in SL(3,\IR)$.  Then after inserting  (\ref{nonepsteininothernoneps52}), and (\ref{nonepsintherothernonepssecondpartc2})
we may restate (\ref{nonepsinothernoneps12}) as
\begin{multline}\label{gl5node2h22punchline}
    F_{\b_2;s}^{SL(5)\,\b_2}(r^{-6/5}e_2',r^{4/5}e_3';N_4) \ \ = \\
    \f{8\,r^{4+4s/5}}{\xi(2s)\xi(2s-1)} \int_{\IR}
 \sum_{[\srel pq ] \, \in \, SL(2,\Z)\backslash {\mathcal M}_{2,3}^{(2)}(\Z)}
 \sum_{\srel{\hat{m},\hat{n}\,\in\,\Z^2}{\hat{m}p-\hat{n}q=N_4}}
    \left(\|(p+q\tau_1) e_3'\|\over \|e_2'^{-1}\hat{m}\|  \right)^{1/2-s} \ \times \\
 \left( \|q e_3'\|\over \|e_2'^{-1}(\hat{n}+\hat{m}\tau_1)\|\right)^{3/2-s}
K_{s-1/2}(2\pi r^{2} \|(p+q\tau_1) e_3'\|\|e_2'^{-1}\hat{m}\|) \ \times \\ K_{s-3/2}(2\pi r^{2} \|q e_3'\|\|e_2'^{-1}(\hat{n}+\hat{m}\tau_1)\|) \, d\tau_1 \\
 + \ \
  \f{2\, \G\left(s-\frac12\right)}{\xi(2s)\xi(2s-1)} r^{1+14s/5} \sum_{\srel{p\,\neq\,0}{\srel{n\,\neq
        \,0}{\srel{\hat{m}\,\perp\,n}{\hat{m}p=  N_4}}}}
 \left(\| e_2'^{-1}\hat{m}  \| \over
    \pi \|ne_2' \|^2 \|  pe_3' \|\right)^{s-1/2} \ \times \\
     K_{s-1/2}(2\pi r^{2}\|e_2'^{-1}\hat{m} \| \| pe_3' \| ) \,.
\end{multline}

\paragraph{ {\bf The parabolic}   $P_{\b_3}= GL(1)\times SL(2)\times   SL(3)\times U_{\b_3}$}\label{subsubsub3}

We may rewrite (\ref{bigGdef})  in the case of  $d=5$ as
\begin{equation}\label{bigGdefbis3}
    {\mathcal G}(\tau,ee^t) \ \ := \ \ \sum_{ [\srel{p\,m}{q\,n}]\, \in \, {\mathcal M}_{2,5}^{(2)}(\Z)}
    e^{-\pi \tau_2^{-1}\| [p+q\tau\,,\,m+n\tau]e\|^2}\,,%
\end{equation}
where $p,q\in\Z^2$ and $m,n\in\Z^3$.  We take $e$ to have the special form $e=\ttwo{I_2}{Q}{}{I_3}\ttwo{e_2}{}{}{e_3}$, where  $Q\in M_{2\times 3}(\IR)$, $e_2\in GL(2,\IR)$, and $e_3\in GL(3,\IR)$.  We will be interested in Fourier coefficients in $Q$ for the  modes $Q\mapsto e^{2\pi i \tr NQ}$, where $N\in M_{3,2}(\Z)$.
 Break up the sum as
\begin{equation}\label{bigGbreakup3}
    {\mathcal G}(\tau,ee^t) \ \ = \ \  {\mathcal G}_0(\tau,ee^t) \ + \  {\mathcal G}_1(\tau,ee^t) \ + \  {\mathcal G}_2(\tau,ee^t)\,,\\
\end{equation}
where
\begin{equation}\label{bigGi3}
   {\mathcal G}_i(\tau,ee^t) \ \ := \ \ \sum_{\srel{\operatorname{rank}[\srel{p}{q}]\,=\,i}{ [\srel{p\,m}{q\,n}]\, \in \, {\mathcal M}_{2,5}^{(2)}(\Z)}}
    e^{-\pi \tau_2^{-1} \|[p+q\tau\,,\,m+n\tau]e\|^2}\,.%
\end{equation}
If $\operatorname{rank}[\srel{p}{q}]=2$, then $[\srel{p\,m}{q\,n}]$ automatically has rank 2.  Thus
\begin{equation}\label{bigG2ba3}
     {\mathcal G}_2(\tau,ee^t) \ \ := \ \ \sum_{\srel{\operatorname{rank}[\srel{p}{q}]\,=\,2}{ m,n\,\in\,\Z^3}}
    e^{-\pi \tau_2^{-1} \|[p+q\tau\,,\,m+n\tau]e\|^2}\,.%
\end{equation}
Again we use modular invariance to write
\begin{equation}\label{bigG2baave3}
     {\mathcal G}_2(\tau,ee^t) \ \  = \ \ \sum_{\g\,\in\,\left\{\pm\ttwo 1001\right\}\backslash SL(2,\Z)}  {\mathcal G}^\circ_2(\g\tau,ee^t)\,,
\end{equation}
where
\begin{equation}\label{bigG2circdef3}
    {\mathcal G}^\circ_2(\tau,ee^t) \ \ = \ \ \sum_{\srel{\srel{p\,=\,[p_1\,p_2]}{q\,=\,[0\,q_2]}}{ \srel{p_1>0,\, 0\le p_2<|q_2|}{ m,n\,\in\,\Z^3}}}
    e^{-\pi \tau_2^{-1} \|[p+q\tau\,,\,m+n\tau]e\|^2}\,.%
\end{equation}

Poisson summation over the inner $m,n\in\Z^3$ sum gives
\begin{equation}\label{bigG2a3}
      {\mathcal G}^\circ_2(\tau,ee^t) \ \ = \ \
      \sum_{\srel{\srel{p\,=\,[p_1\,p_2]}{q\,=\,[0\,q_2]}}{ \srel{p_1>0,\, 0\le p_2<|q_2|}{ \hat{m},\hat{n}\,\in\,\Z^3}}}
      \int_{\IR^6}e^{-2\pi i (  m\hat{m}-n\hat{n})}
      e^{-\pi \tau_2^{-1} \|[p+q\tau\,,\,m+n\tau]e\|^2}\,dm\,dn\,,
\end{equation}
where $\hat{m}, \hat{n}\in \Z^3$ are column vectors.
With the particular form $e=\ttwo{I_2}{Q}{}{I_3}\ttwo{e_2}{}{}{e_3}$
the exponent of the second factor is
\begin{equation}\label{bigG2b3}
    -\pi\tau_2^{-1}[p+q\tau \,,\, (p+q\tau)Q+m+n\tau]\ttwo{e_2e_2^t}{}{}{e_3e_3^t}[p+q\tau \,,\, (p+q\bar\tau)Q+m+n\bar\tau]^t\,.
\end{equation}
Thus after changing variables $m\mapsto m-pQ$, $n\mapsto n-qQ$ (\ref{bigG2a3}) becomes
\begin{multline}\label{bigG2c3}
     {\mathcal G}^\circ_2(\tau,ee^t) \ \ = \ \
     \sum_{\srel{\srel{p\,=\,[p_1\,p_2]}{q\,=\,[0\,q_2]}}{p_1>0,\, 0\le p_2<|q_2|}} e^{-\pi\tau_2^{-1}\|(p+q\tau) e_2\|^2} \sum_{ \hat{m},\hat{n}\,\in\,\Z^3}e^{2\pi i  (pQ\hat{m}- qQ\hat{n})}
   \ \times \\   \int_{\IR^6}e^{-2\pi i ( m\hat{m}-n\hat{n})}
      e^{-\pi \tau_2^{-1} \|(m+n\tau)e_3\|^2}\,dm\,dn\,.
\end{multline}
To compute this integral we change variables $m\mapsto m e_3^{-1}$, $n\mapsto ne_3^{-1}$, which has the effect of dividing both $dm$ and $dn$ each by $\det e_3$:~the integral equals $(\det e_3)^{-2}$ times
\begin{multline}\label{bigG2d3}
    \int_{\IR^6}e^{-2\pi i
   (me_3^{-1}\hat{m}-ne_3^{-1}\hat{n})}
    e^{-\pi\tau_2^{-1}\|(m+n\tau)\|^2}\,dm\,dn \ \ = \\
     \int_{\IR^6}e^{-2\pi i
      (me_3^{-1}\hat{m}-ne_3^{-1}(\hat{n}+\tau_1\hat{m}))
      }
    e^{-\pi\tau_2^{-1}\|m\|^2-\pi\tau_2\|n\|^2}\,dm\,dn
\end{multline}
and
 (\ref{bigG2c3}) is equal to
\begin{multline}\label{bigG2e3}
       {\mathcal G}^\circ_2(\tau,ee^t) \ \ =  (\det e_3)^{-2}
   \sum_{\srel{\srel{p\,=\,[p_1\,p_2]}{q\,=\,[0\,q_2]}}{p_1>0,\, 0\le
       p_2<|q_2|}}
 e^{-\pi\tau_2^{-1}\|(p+q\tau) e_2\|^2}\\ \sum_{ \hat{m},\hat{n}\,\in\,\Z^3}e^{2\pi i  (pQ\hat{m}-qQ\hat{n})}
    e^{-\pi\tau_2\|e_3^{-1}\hat{m}\|^2-\pi\tau_2^{-1}\|e_3^{-1}(\hat{n}+\hat{m}\tau_1)\|^2}\,.
\end{multline}
The Fourier coefficient for $N_4\in M_{3,2}(\Z)$ is equal to
\begin{multline}\label{bigG2FC3}
        {\mathcal F}{\mathcal G}_2^\circ(\tau,e_2,e_3;N_4) \ \ = \ \ (\det e_3)^{-2}
      \sum_{\srel{\srel{p\,=\,[p_1\,p_2]}{q\,=\,[0\,q_2]}}{p_1>0,\, 0\le p_2<|q_2|}}
     \\   \sum_{\srel{\hat{m},\hat{n}\,\in\,\Z^3}{\hat{m}p-\hat{n}q\,=\,N_4}}
        e^{-\pi\tau_2^{-1}\|(p+q\tau) e_2\|^2 -\pi\tau_2\|e_3^{-1}\hat{m}\|^2-\pi\tau_2^{-1}\|e_3^{-1}(\hat{n}+\hat{m}\tau_1)\|^2}\,.
\end{multline}

Let us now consider ${\mathcal G}_1(\tau,ee^t)$, which has the contributions for $p,q\in\Z^2$ such that $\operatorname{rank}[\srel pq]=1$:
\begin{equation}\label{bigG13}
   {\mathcal G}_1(\tau,ee^t) \ \ := \ \ \sum_{\srel{\operatorname{rank}[\srel{p}{q}]\,=\,1}{ [\srel{p\,m}{q\,n}]\, \in \, {\mathcal M}_{2,5}^{(2)}(\Z)}}
    e^{-\pi \tau_2^{-1} \|[p+q\tau\,,\,m+n\tau]e\|^2}\,.
\end{equation}
We may write this as an average   over $\{\pm \G_\infty\}\backslash SL(2,\Z)$:
\begin{equation}\label{bigG1ave3}
     {\mathcal G}_1(\tau,ee^t) \ \ = \ \ \sum_{\g\,\in\,\{\pm \G_\infty\}\backslash SL(2,\Z)}
      {\mathcal G}^\circ_1(\g\tau,ee^t)\,,
\end{equation}
where
\begin{multline}\label{bigG1circdef3}
     {\mathcal G}^\circ_1(\tau,ee^t) \ \ := \ \
     \sum_{\srel{p\,\neq\,0}{ [\srel{p\,m}{0\,n}]\, \in \, {\mathcal M}_{2,5}^{(2)}(\Z)}}
    e^{-\pi \tau_2^{-1} \|[p,\,m+n\tau]e\|^2} \\
    = \ \
    \sum_{\srel{p\,\neq\,0}{\srel{n\,\neq\,0}{m\,\in\,\Z^3}}}
    e^{-\pi \tau_2^{-1}\|pe_2\|^2-\pi \tau_2^{-1} \|(pQ+m+n\tau_1)e_3\|^2-\pi \tau_2 \|ne_3\|^2}     \,.
\end{multline}
Here we used that the matrix $ [\srel{p\,m}{0\,n}]$ has rank 2 if and only if $n\neq 0$ (since $p\neq 0$).  Poisson sum over $m$   gives the formula
\begin{equation}\label{bigG1a3}
{\mathcal G}^\circ_1(\tau,ee^t) \ \ = \ \
 {\tau_2^{\frac32}\over \det e_3}
  \sum_{\srel{p\,\neq\,0}{\srel{n\,\neq\,0}{\hat{m}\,\in\,\Z^3}}}
  e^{2\pi i (pQ+n\tau_1)\hat{m}}
 e^{-\pi \tau_2^{-1}\|pe_2\|^2-\pi \tau_2 \|ne_3\|^2
-\pi \tau_2 \|e_3^{-1}\hat{m}\|^2}
     \end{equation}
for  (\ref{bigG1circdef3}), where $\hat m$ is a column vector.

We conclude that the Fourier coefficient of ${\mathcal G}^\circ_1(\tau,ee^t)$ for $N_4$ is equal to
\begin{multline}\label{bigG1FC3}
     {\mathcal F}{\mathcal G}_1^\circ(\tau,e_2,e_3;N_4) \ \ = \\
    {\tau_2^{\frac32}\over \det e_3}
        \sum_{\srel{p\,\neq\,0}{\srel{n\,\neq\,0}{\hat{m}p=N_4}}}
    e^{2\pi i \tau_1 n\hat{m}}
         e^{-\pi \tau_2^{-1}\|pe_2\|^2-\pi \tau_2 \|ne_3\|^2
-\pi \tau_2 \|e_3^{-1}\hat{m}\|^2}
        \,.
\end{multline}
Observe that ${\mathcal F}{\mathcal G}_1(\tau,e_2,e_3;N_4)\equiv 0$ if $\operatorname{rank}(N_4)=2$,
and again that  ${\mathcal G}_0(\tau,ee^t)$ has no nonzero Fourier coefficients because
$[0\,0\,\star\,\star\,\star]\ttwo{I_2}{Q}{}{I_3}$ is independent of $Q$.

Proposition~\ref{prop:nonepsteinintegralrep} states that
\begin{equation}\label{nonepstrepnreversed3}
  \frac12 \xi(2s)\xi(2s-1)\,E^{SL(5)}_{\beta_2;s}(e) \ \ = \ \
  \int_0^\infty \int_{SL(2,\Z)\backslash \U} {\mathcal G}(\tau, u
  ee^t)\,\f{d^2\tau}{\tau_2^2} \,{du\over u^{1-2s}}\,.
\end{equation}
Since we  have specialized  $e\in SL(d,\IR)$ to have the form
$e=\ttwo{I_2}{Q}{}{I_3}\ttwo{e_2}{}{}{e_3}$ the Fourier coefficient
for $N_4$ can be written as
\begin{multline}\label{nonepsinothernoneps13}
   \frac12\xi(2s)\xi(2s-1)F_{\b_2;s}^{SL(5)\,\b_3}(N_4) \ \ = \\
    \int_0^\infty \int_{\U} {\mathcal
      {FG}}_2^\circ(\tau,u^{\frac12}e_2,u^{\frac12}e_3;N_4)\,\f{d^2\tau}{\tau_2^2}\,\f{du}{u^{1-2s}}
    \  \\
+ \ \int_0^\infty \int_{\G_\infty\backslash \U} {\mathcal{FG}}_1^\circ(\tau,u^{\frac12}e_2,u^{\frac12}e_3;N_4)\,\f{d^2\tau}{\tau_2^2}\,\f{du}{u^{1-2s}}\,.
\end{multline}

Let us consider the first integral,
\begin{multline}\label{nonepsteininothernoneps23}
   \int_0^\infty \int_{ \U} (\det e_3)^{-2}
         \sum_{\srel{\srel{p\,=\,[p_1\,p_2]}{q\,=\,[0\,q_2]}}{p_1>0,\, 0\le p_2<|q_2|}}
    e^{-\pi\tau_2^{-1}u\|(p+q\tau) e_2\|^2}  \\   \sum_{\srel{\hat{m},\hat{n}\,\in\,\Z^3}{\hat{m}p-\hat{n}q\,=\,N_4}}
        e^{-\pi\tau_2 u^{-1}\|e_3^{-1}\hat{m}\|^2-\pi\tau_2^{-1}u^{-1}\|e_3^{-1}(\hat{n}+\hat{m}\tau_1)\|^2}
        \,\f{d^2\tau}{\tau_2^2}\,\f{du}{u^{4-2s}}\,.
\end{multline}
Changing variables to $x=u/\tau_2$ and $y=\tau_2 u$, so that $u=\sqrt{xy}$, $\tau_2=\sqrt{y/x}$ and $d\tau_2 du = \f{dx dy}{2x}$ the integral
becomes
\begin{multline}\label{nonepsteininothernoneps43}
  \frac12  (\det e_3)^{-2} \int_{\IR} \int_0^\infty \int_0^\infty
      \sum_{\srel{p=[p_1\,,\,p_2]}{\srel{q=[0\,,\,q_2]}{\srel{p_1>0 }{\srel{0\le p_2<|q_2|}{}}}}}
 \sum_{\srel{\hat{m},\hat{n}\,\in\,\Z^3}{\hat{m}p-\hat{n}q=N_4}}\\
      e^{-\pi x \|(p+q\tau_1) e_2\|^2-\pi
        x^{-1}\|e_3^{-1}\hat{m}\|^2}
e^{ -\pi y \|q e_2\|^2-\pi y^{-1}\|e_3^{-1}(\hat{n}+\hat{m}\tau_1)\|^2}
      \,\f{dx}{x^{2-s}}\,\f{dy}{y^{3-s}}d\tau_1\,.
\end{multline}
Integrating over $x$ and $y$  yields
\begin{multline}\label{nonepsteininothernoneps53}
 2 (\det  e_3)^{-2} \int_{\IR}
 \sum_{\srel{p=[p_1\,,\,p_2]}{\srel{q=[0\,,\,q_2]}{\srel{p_1>0 }{\srel{0\le p_2<|q_2|}{}}}}}
 \sum_{\srel{\hat{m},\hat{n}\,\in\,\Z^3}{\hat{m}p-\hat{n}q=N_4}}
    \left(\|(p+q\tau_1) e_2\|\over \|e_3^{-1}\hat{m}\|  \right)^{1-s} \ \times \\
 \left( \|q e_2\|\over \|e_3^{-1}(\hat{n}+\hat{m}\tau_1)\|\right)^{2-s}
K_{s-1}(2\pi \|(p+q\tau_1) e_2\|\|e_3^{-1}\hat{m}\|) \ \times \\ K_{s-2}(2\pi \|q e_2\|\|e_3^{-1}(\hat{n}+\hat{m}\tau_1)\|)  d\tau_1\,.
\end{multline}

Next we analyze the second integral in (\ref{nonepsinothernoneps13}), in which we assume $N_4$ has rank 1 (since it vanishes if it has rank 2):
\begin{multline}\label{nonepsintherothernonepssecondparta3}
     \smallf{1}{\det e_3} \int_0^\infty \int_{\G_\infty\backslash \U}
        \sum_{\srel{p\,\neq\,0}{\srel{n\,\neq\,0}{\hat{m}p=  N_4}}}
  e^{2\pi i \tau_1 n\hat{m} } \ \times \\
         e^{-\pi \tau_2^{-1}u \|pe_2\|^2-\pi \tau_2 u \|ne_3\|^2
-\pi \tau_2 u^{-1} \|e_3^{-1}\hat{m}\|^2}
    \,\f{d^2\tau}{\tau_2^{\frac12}}\,\f{du}{u^{\frac52-2s}}\,.
\end{multline}
The $\tau_1$ integration over $[0,1]$ enforces the condition that
$n\perp \hat{m}$ (which implies $n\perp   N_4$):
\begin{equation}\label{nonepsintherothernonepssecondpartb3}
     \smallf{1}{\det e_3}  \int_0^\infty  \int_0^\infty
        \sum_{\srel{p\,\neq\,0}{\srel{n\,\neq
        \,0}{\srel{\hat{m}\,\perp\,n}{\hat{m}p=  N_4}}}}
         e^{-\pi \tau_2^{-1}u\|pe_2\|^2-\pi \tau_2 u \|ne_3\|^2-\pi \tau_2 u^{-1} \|e_3^{-1}\hat{m}\|^2}
    \,\f{d\tau_2}{\tau_2^{\frac12}}\,\f{du}{u^{\frac52-2s}}\,.
\end{equation}
As before, change variables
$x=u/\tau_2$ and $y=\tau_2 u$ so that (\ref{nonepsintherothernonepssecondpartb3}) becomes
\begin{multline}\label{nonepsintherothernonepssecondpartc3}
   \f{1}{2(\det e_3)} \sum_{\srel{p\,\neq\,0}{\srel{n\,\neq
        \,0}{\srel{\hat{m}\,\perp\,n}{\hat{m}p= N_4}}}}  \int_0^\infty  \int_0^\infty
         e^{-\pi x \|pe_2\|^2-\pi x^{-1} \|e_3^{-1}\hat{m}\|^2} e^{-\pi y \|ne_3\|^2}
 \,\f{dx}{x^{2-s}}\,\f{dy}{y^{\frac32-s}} \\
     = \  \
        \f{ \G\left(s-\frac12\right)}{(\det e_3)} \sum_{\srel{p\,\neq\,0}{\srel{n\,\neq
        \,0}{\srel{\hat{m}\,\perp\,n}{\hat{m}p=  N_4}}}}
       (\pi \|ne_3 \|^2)^{\frac12-s}\,\left(\| e_3^{-1}\hat{m}  \| \over
     \|  pe_2 \|\right)^{s-1}
     K_{s-1}(2\pi\|e_3^{-1}\hat{m} \| \| pe_2 \| )  \,.
\end{multline}

\paragraph{ {\bf The parabolic}   $P_{\b_4}=GL(1)\times SL(4)\times U_{\b_4}$}\label{subsubsub4}\hfil\break

In this case $e\in GL(5,\IR)$ has the special form $\ttwo{I_1}{Q}{}{I_4}\ttwo{e_1}{}{}{e_4}$, where $Q$ is a 4-dimensional row vector, $e_1$ is a nonzero real number,  and  $e_4\in GL(4,\IR)$.  We work with a sum of the form (\ref{bigGdefbis3}) but now instead $p,q\in\Z$ and $m,n\in \Z^4$.  Then the exponent  (\ref{bigG2b3}) becomes
\begin{equation}\label{nonepsinP4a}
    -\pi \tau_2^{-1}e_1^2|p+q\tau|^2  \ - \ \pi \tau_2^{-1} \|((p+q\tau)Q+m+n\tau)e_4\|^2\,.
\end{equation}
If $p=q=0$ then the exponent and hence ${\mathcal G}(\tau,ee^t)$ is independent of $Q$.  To get nontrivial Fourier modes in $Q$, we must thus assume to the contrary that $\operatorname{rank}[\srel pq]=1$.  We write the contributions of these rank one terms as
\begin{equation}\label{nonepsinP4a2}
    {\mathcal G}_1(\tau,ee^t) \ \ = \ \ \sum_{\g\,\in\,\G_\infty\backslash \G}
    {\mathcal G}^\circ_1(\g\tau,ee^t)\,,
\end{equation}
 where
 \begin{equation}\label{nonepsinP4a3}
    {\mathcal G}^\circ_1(\tau,ee^t) \ \ := \ \ \sum_{\srel{p\,>\,0}{\srel{m,n\,\in\,\Z^4}{n\,\neq\,0}}}
    e^{  -\pi \tau_2^{-1} e_1^2 p^2   -  \pi \tau_2^{-1}  \|(pQ+m+n\tau)e_4\|^2}
 \end{equation}
 (this uses the fact that the $SL(2,\Z)$ orbits of rank one integer matrices $[\srel pq]$ each have representatives with $p>0$ and $q=0$, and that the rank 2 condition for $[\srel{p\,m}{0\,n}]$ is that $n\neq 0$).  Applying Poisson summation over $m\in\Z^4$ results in the expression
 \begin{multline}\label{nonepsinP4a4}
    {\mathcal G}^\circ_1(\tau,ee^t) \ \ = \ \ \sum_{\srel{p\,>\,0}{\srel{\hat{m},n\,\in\,\Z^4}{n\,\neq\,0}}}
    e^{  -\pi \tau_2^{-1}e_1^2p^2-\pi \tau_2\|ne_4\|^2}
    e^{2\pi i (  pQ\hat{m}+\tau_1 n\hat{m} )}\\  \times \
    \int_{\IR^4}e^{-2\pi i m \hat{m} }
    e^{-  \pi \tau_2^{-1} \|me_4\|^2}\,dm\,.
 \end{multline}
 Here we think of $\hat{m}$ as  a column vector.
 Thus  the Fourier coefficient for $e^{2\pi i QN_4}$, when the column vector $N_4\in \Z^4$ is not zero, is equal to
 \begin{multline}\label{nonepsinP4a5}
   {\mathcal F} {\mathcal G}^\circ_1(\tau,ee^t)  \ \ = \\  \tau_2^2 (\det e_4)^{-1} \sum_{\srel{p\,>\,0}{\srel{n\,\neq\,0}{p\hat{m}\,=\,N_4}}}
    e^{2\pi i \tau_1 n \hat{m} }  e^{  -\pi \tau_2^{-1} e_1^2 p^2-\pi \tau_2\|ne_4\|^2-  \pi \tau_2  \| e_4^{-1} \hat{m}\|^2} \,.
 \end{multline}

Using proposition~\ref{prop:nonepsteinintegralrep} the $N_4$-th Fourier coefficient of $\smallf 12\xi(2s)\xi(2s-1)E^{SL(5)}_{\beta_2;s}(e)$ is
\begin{multline}\label{nonepsinP4c}
  {1\over \det e_4} \int_0^\infty\int_0^\infty\int_0^1  e^{2\pi i \tau_1 n \hat{m} } \ \times \\
 \sum_{\srel{p\,>\,0}{\srel{n\,\neq\,0}{p\hat{m}\,=\,N_4}}}
    e^{  -\pi \tau_2^{-1}  u  e_1^2 p^2-\pi \tau_2 u  \|ne_4\|^2-  \pi \tau_2 u^{-1}  \|e_4^{-1}\hat{m}\|^2}
d\tau_1\,d\tau_2 \f{du}{u^{3-2s}} \\
=  {1\over \det e_4}\ \
 \int_0^\infty\int_0^\infty
  \sum_{\srel{p\,>\,0}{\srel{n\,\neq\,0}{\srel{p\hat{m}\,=\,N_4}{n\,\perp\,N_4}}}}
    e^{  -\pi \tau_2^{-1}  u e_1^2 p^2-\pi \tau_2 u \|ne_4\|^2} \ \times  \\
    e^{-  \pi \tau_2 u^{-1}  \|e_4^{-1}\hat{m}\|^2}
 \,d\tau_2\,\f{du}{u^{3-2s}}
 \,.
\end{multline}
Changing variables to $x=u/\tau_2$ and $y=\tau_2 u$, so that
$u=\sqrt{xy}$, $\tau_2=\sqrt{y/x}$ and $d\tau_2 du = \f{dx dy}{2x}$ (\ref{nonepsinP4c}) equals
\begin{multline}\label{nonepsinP4d}
  \f{1}{2 \det e_4}  \sum_{\srel{p\,>\,0}{\srel{n\,\neq\,0}{\srel{p\hat{m}\,=\,N_4}{n\,\perp\,N_4}}}} \int_0^\infty\int_0^\infty
    e^{  -\pi x e_1^2 p^2 - \pi x^{-1} \|e_4^{-1}\hat{m}\|^2-\pi  y  \|ne_4\|^2}
\,\f{dx}{x^{\frac52-s}}\,\f{dy}{y^{\frac32-s}} \ \
= \\ {\G\left(s-\frac12\right)\over\pi^{s-\frac12}(\det e_4)}
    \sum_{\srel{p\,>\,0}{\srel{n\,\neq\,0}{\srel{p\hat{m}\,=\,N_4}{n\,\perp\,N_4}}}}
( \smallf{\|e^{-1}_4\hat{m}\|}{ p|e_1|})^{s-\frac32} \|ne_4\|^{1-2s}
 K_{s-\frac32}(2\pi  e_1 p \|e_4^{-1}\hat{m}\| )
 \,.
\end{multline}

\subsection{The $Spin(5,5)$ case}
\label{sec:dfive}

 Here we analyze the Fourier modes
 of the series $E^{Spin(5,5)}_{\alpha_1;s}$, which is one of the two Eisenstein series
 appropriate to the $D=6$ case.  The results are summarized in section \ref{so55case}.
Here we shall use the expressions (\ref{Isdef}) and (\ref{Is9}), which for $d=5$ imply
\begin{multline}\label{D5EpsinEps1}
    \,E^{SO(5,5)}_{\a_1;s+3/2}\(
\ttwo{I}{Bw_5}{}{I}\ttwo{v^{1/2}e}{}{}{v^{-1/2}\tilde{e}}\) \ \ = \\
   \f{v^{5/2}}{2\,\xi(2s+3)}  \int_{SL(2,\Z)\backslash \U}E^{SL(2)}_s(\tau)\,{\mathcal G}(\tau,v ee^t+B) \,\f{d^2\tau}{\tau_2^2}
    \\
   + \   v^{s+3/2}\,   E^{SL(5)}_{\beta_{4};s+3/2}(e)\\
    +\ v^{5/2-s}\,\f{\xi(2s-1)}{\xi(2s+3)}\,
    E^{SL(5)}_{\b_1;s}(e)\,,
\end{multline}
where $v>0$ and $e\in SL(5,\IR)$. Formula (\ref{SpinddandSoddseries}) shows that the same formula is valid for $E^{Spin(d,d)}_{\a_1;s+3/2}(h')$, where $h'\in Spin(d,d,\IR)$ is any element which projects onto $\ttwo{I}{Bw_5}{}{I}\ttwo{v^{1/2}e}{}{}{v^{-1/2}\tilde{e}}$ via the covering map $Spin(d,d,\IR)\rightarrow SO(d,d,\IR)$.

(i) {\bf The parabolic   $P_{\alpha_5}=GL(1)\times SL(5)\times U_{\alpha_5}$}

The analysis in  this section also covers limit (iii), since both parabolics come from spinor nodes.
Formula (\ref{D5EpsinEps1})  shows that the nontrivial spinor parabolic Fourier coefficients (in $B$) all come from the integral on the righthand side.  In limit (i) the parameter $v$ plays the role of the parameter $r^2$ from (\ref{notation}), and so we set $v=r^2$.
Substituting the formula (\ref{bigG2}) for ${\mathcal G}(\tau,r^2ee^t+B)$ we see that the contribution to the nonzero Fourier modes of (\ref{D5EpsinEps1}) is given by
\begin{equation}\label{D5Epsinspina}
     \f{r^{5}}{2\,\xi(2s+3)} \int_{SL(2,\Z)\backslash \U}E^{SL(2)}_s(\tau)\,
        \sum_{ [\srel{m}{n}]\, \in \, {\mathcal M}_{2,5}^{(2)}(\Z)}
    e^{-\pi \tau_2^{-1}\, r^2\, \|(m+n\tau)e\|^2} e^{-2\pi i mBn^t}
       \,\f{d^2\tau}{\tau_2^2}  \,.
\end{equation}  Note that all nonzero Fourier modes have the form $B\mapsto e^{2\pi i mBn^t}$, which is precisely the $\smallf 12$-BPS condition from (\ref{472}).

  We conclude that for $N_2\in {\mathcal M}_{5,5}(\Z)$  the Fourier coefficient of $E_{\a_1;s}^{SO(5,5)}\(\ttwo{I}{Bw_d}{}{I}\ttwo{r e}{}{}{r^{-1}\tilde{e}}\)$ for  the character $B\mapsto e^{i \pi  (\tr N_2 B)}$ is equal to
\begin{equation}\label{D5Epsinspinb}
    \f{r^{5}}{2\,\xi(2s+3)}    \int_{\U}E^{SL(2)}_{s}(\tau)\,
        \sum_{\srel{ [\srel{m}{n}]\, \in \, SL(2,\Z)\backslash {\mathcal M}_{2,5}^{(2)}(\Z)}{N_2\,=\,n^t m-m^t n}}
    e^{-\pi \tau_2^{-1} r^2\,\|(m+n\tau)e\|^2}
       \,\f{d^2\tau}{\tau_2^2}
\end{equation}
 with $r^2= r_4/\ell_7$ according the identification of the
  parameters in~\cite{Green:2010kv} recalled in~\eqref{notation}.

 In the case of interest in section~\ref{so55case} the parameter $s$ is equal to zero,  and the integral was computed in (\ref{I01a})  as
  \begin{equation}\label{D5Epsinspinc}
   \f{r^{3}}{2\,\xi(3)}  \,  \sum_{\srel{ [\srel{m}{n}]\, \in \, SL(2,\Z)\backslash {\mathcal M}_{2,5}^{(2)}(\Z)}{N_2\,=\,n^t m-m^t n}} \,\f{
   e^{-2\pi r^2 \det([\srel mn]ee^t [\srel mn]^t)^{1/2}}}{ \det([\srel mn]ee^t [\srel mn]^t)^{1/2}}\,.
  \end{equation}
  In the claim following (\ref{I03}) we saw that the $[\srel mn]$ in this sum can be parametrized as
  $[\srel mn]=\ttwo{d_1}{b}{0}{d_2}[\srel{m'}{n'}]$,  where $d_1\neq 0$,  $0\le b < d_2$, and $[\srel{m'}{n'}]$ ranges over left $GL(2,\Z)$-cosets of ${\mathcal M}_{2,5}^{(2)}(\Z)':=\{$all possible bottom two rows of matrices in  $SL(5,\Z)\}$.  (This coset space is in bijective correspondence with $P_{\b_2}(\Z)\backslash SL(5,\Z)$.) The constraint $N_2=n^t m-m^t n$ then reads $N_2=d_1d_2((n')^tm'-(m')^tn')$.  As a consequence we can rewrite (\ref{D5Epsinspinc}) as
  \begin{equation}\label{D5Epsinspind}
   \f{r^{3}}{2\,\xi(3)}  \,  \sum_{\srel{ [\srel{m'}{n'}]\, \in \, GL(2,\Z)\backslash {\mathcal M}_{2,5}^{(2)}(\Z)'}{N_2\,=\,d_1d_2((n')^t m'-(m')^t n')}} \,\f{d_2}{d_1d_2}\,\f{
   e^{-2\pi r^2 d_1d_2\det([\srel{m'}{n'}]ee^t [\srel{m'}{n'}]^t)^{1/2}}}{ \det([\srel{m'}{n'}] ee^t [\srel{m'}{n'}]^t)^{1/2}}\,.
  \end{equation}
The product $d_1d_2$ obviously divides each entry of $N_2$, but the entries of $N_2=n^tm-m^tn$ can have a nontrivial common factor even if $\gcd(m)=\gcd(n)=1$.  On the other hand, the ${5\choose 2}=10$ minors of the two bottom rows $[\srel{m'}{n'}]$ must be relatively prime, since the determinant of the $SL(5,\Z)$ matrix (i.e., 1) is an integral linear combination of them.  These minors are the entries of $N_2$, up to sign.  We conclude that $d_1d_2=\gcd(N_2)$ and that
  (\ref{D5Epsinspinc}) is equal to
 \begin{equation}\label{D5Epsinspine}
   \f{r^{3}}{2\,\xi(3)}  \,  \sum_{\srel{ [\srel{m'}{n'}]\, \in \, GL(2,\Z)\backslash {\mathcal M}_{2,5}^{(2)}(\Z)'}{N_2\,=\,\gcd(N_2)((n')^t m'-(m')^t n')}} \f{\sigma_1(\gcd(N_2))}{\gcd(N_2)} \,\f{
   e^{-2\pi r^2 \gcd(N_2)\det([\srel{m'}{n'}]ee^t [\srel{m'}{n'}]^t)^{1/2}}}{ \det([\srel{m'}{n'}]ee^t [\srel{m'}{n'}]^t)^{1/2}}\,.
  \end{equation}
Again, (\ref{SpinddandSoddseries}) shows that this formula is also valid for $F^{Spin(d,d)}_{\a_1;s}(h')$ and any $h'\in Spin(d,d,\IR)$ which projects onto $\ttwo{I}{BW_d}{}{I}\ttwo{re}{}{}{r^{-1}\tilde e}$ via the covering map $Spin(d,d,\IR)\rightarrow SO(d,d,\IR)$.

(ii) {\bf The parabolic   $P_{\alpha_1}=GL(1)\times Spin(4,4)\times U_{\alpha_1}$}\hfil\break
\label{sec:pierres-analysis}

 We shall use (\ref{D5EpsinEps1}) to compute the nonzero Fourier modes of $E_{\a_1;s}^{SO(5,5)}$ and hence $E_{\a_1;s}^{Spin(5,5)}$.
Before beginning the calculation, it is helpful to explicitly write out the groups and characters involved.
The unipotent radical $U=U_{\a_1}$ of $P_{\a_1}$ is an abelian group isomorphic to $\IR^8$
under the map
\begin{equation}\label{Ua1param}
  u_1,u_2,\ldots,u_8 \  \ \mapsto \ \
    \(\begin{smallmatrix}
        1 &  u_1 & u_2 & u_3 & u_4 & u_5 & u_6 & u_7 & u_8 & -u_1u_8-u_2u_7-u_3u_6-u_4u_5 \\
 0 & 1 & 0 & 0 & 0 & 0 & 0 & 0 & 0 & -u_8 \\
 0 & 0 & 1 & 0 & 0 & 0 & 0 & 0 & 0 & -u_7 \\
 0 & 0 & 0 & 1 & 0 & 0 & 0 & 0 & 0 & -u_6\\
 0 & 0 & 0 & 0 & 1 & 0 & 0 & 0 & 0 & -u_5\\
 0 & 0 & 0 & 0 & 0 & 1 & 0 & 0 & 0 & -u_4\\
 0 & 0 & 0 & 0 & 0 & 0 & 1 & 0 & 0 & -u_3\\
 0 & 0 & 0 & 0 & 0 & 0 & 0 & 1 & 0 & -u_2\\
 0 & 0 & 0 & 0 & 0 & 0 & 0 & 0 & 1 & -u_1\\
 0 & 0 & 0 & 0 & 0 & 0 & 0 & 0 & 0 & 1
      \end{smallmatrix}
    \),
\end{equation}
and $\G\cap U$ is isomorphic to $\Z^8$ under this identification.
 The general Fourier mode is indexed $N_1 = [\srel MN]=[\srel{m^1}{n_1}\srel{m^2}{n_2}\srel{m^3}{n_3}\srel{m^4}{n_4}] \in M_{2,4}(\Z)$ from (\ref{e:N11}), and is given by the character
\begin{equation}\label{mostgeneralFouriermodeU}
\chi_{N_1}(u) \ \ :=   \ \ e^{2\pi i (m^1u_1+m^2u_2+m^3u_3+m^4u_4+n_1u_8+n_2u_7+n_3u_6+n_4u_5)}\,.
\end{equation}
The Fourier coefficient (\ref{e:477}) is given by
\begin{equation}\label{Fa1sD5a1ni}
F_{\a_1;s}^{SO(5,5)\,\a_1}(N_1) \ \ = \ \ \int_{U(\Z)\backslash U(\IR)}
E_{\a_1;s}^{SO(5,5)}(uh)\,\chi_{N_1}(u)^{-1}\, du\,.
\end{equation}
The general element $h$ of the Levi component has the form
\begin{equation}\label{La1element}
    h \ \ =  \ \  h(a,h_4) \ \ = \ \
      \left(
                                     \begin{array}{ccc}
                                      a & 0 & 0  \\
                                       0 & h_4 & 0  \\
                                       0 & 0 & 1/a  \\
                                     \end{array}
                                   \right),
\end{equation}
where $a\neq 0$ and $h_4\in SO(4,4)(\IR)$.

  Given the structure of the last two terms in (\ref{D5EpsinEps1}) (which are insensitive to $u_5,u_6,u_7,u_8$) it makes sense to treat the cases $N\neq [0\,0\,0\,0]$ and $N=[0\,0\,0\,0]$ separately.  Since $E_{\a_1;s}^{SO(5,5)}$ is invariant under the Weyl group element $h(1,w_8)$ ($w_8$ denoting the reversed-$8\times 8$ identity matrix)  and conjugating the matrix (\ref{Ua1param}) by $h(1,w_8)$ reverses the order of the $u_i$, the Fourier coefficient $F_{\a_1;s}^{SO(5,5)\,\a_1}([\srel{M}{N}])
  $  evaluated at $h(a,h_4)$ equals  $F_{\a_1;s}^{SO(5,5)\,\a_1}([\srel{N}{M}])$ evaluated at $h(a,w_8h_4)$.  Since we are studying nontrivial Fourier coefficients at least one entry of the matrix $N_1$ is nonzero. Thus the determination of these coefficients for $N_1$ of the form $[\srel{M}{0}]$ reduces to those of the form $[\srel{0}{N}]$.
 Therefore in performing these computations we can  assume that  $N\neq [0\,0\,0\,0]$, and then convert afterwards to $N=[0\,0\,0\,0]$ using this $w_8$-mechanism.  For reasons of space we will not carry out this conversion here, and instead limit our discussion in section~\ref{so55case} to the case when $N\neq [0\,0\,0\,0]$.  Thus for the remainder of the paper we assume $N\neq [0\,0\,0\,0]$.

Suppose $e\in SL(5,\IR)$ has the form $e= \ttwo{1}{Q}{}{I_4}\ttwo{v^{-1/2}r^2}{}{}{ e_4}=\ttwo{v^{-1/2}r^2}{ Qe_4}{}{ e_4}$, with $Q=[q_1\,q_2\,q_3\,q_4]$ and  $e_4\in GL(4,\IR)$ a matrix with determinant $v^{1/2}r^{-2}$.  The reason for writing $e$ this way is ensure  that $r$ plays the same role it does in (\ref{notation}).
Furthermore suppose
\begin{equation}\label{Bw5}
    Bw_5 \ \ = \ \ \(
    \begin{smallmatrix}
    b_1 & b_2 & b_3 & b_4 & 0 \\
    b_5 & b_6 & b_7 & 0 & -b_4 \\
    b_8 & b_9 & 0 & -b_7 & -b_3 \\
    b_{10} & 0 & -b_9 & -b_6 & -b_2 \\
    0 & -b_{10} & -b_8 & -b_5 & -b_1 \\
    \end{smallmatrix}
    \).
\end{equation}
Then the argument $\ttwo{I}{Bw_5}{}{I}\ttwo{v^{1/2} e}{}{}{v^{-1/2}\tilde{e}}$ of the first line of (\ref{D5EpsinEps1}) lies in $P_{\a_1}$ (recall that the parameter $v$ determines the determinant, $v^{5/2}$, of the upper left $5\times 5$ block of this matrix).
This product also has the factorization $uh$, where $u$ is the matrix (\ref{Ua1param}) with $(u_1,u_2,u_3,u_4,u_5,u_6,u_7,u_8)=(q_1,q_2,q_3,q_4,b_1-b_5q_1-b_8q_2-b_{10}q_3,b_2-b_6q_1-b_9q_2+b_{10}q_4,
b_3-b_7q_1+b_9q_3+b_8q_4,b_4+b_7q_2+b_6q_3+b_5q_4)$ and  $h=h(r^2,h_4)$, where $h_4=\ttwo{I_4}{B'w_4}{}{I_4}\ttwo{v^{1/2} e_4}{}{}{ v^{-1/2}\tilde{e}_4}$.
  Thus the character
   $ \chi_{N_1}(u)=\exp(2\pi i (m^1-n_4b_5-n_3b_6-n_2b_7)q_1+(m^2-n_4b_8-n_3b_9+n_1b_7)q_2+(m^3-n_4b_{10}+n_2b_9+n_1b_6)q_3
+(m^4+n_3b_{10}+n_2b_8+n_1b_5)q_4+n_4b_1+n_3b_2+n_2b_3+n_1b_4)$.

Recall  (\ref{bigG2}), which states
\begin{multline}\label{bigG2bis}
 {\mathcal G}(\tau,vee^t+B) \ \ = \\    \sum_{ [\srel{p\,m_1}{q\,m_2}]\, \in \, {\mathcal M}_{2,5}^{(2)}(\Z)}e^{-\pi \tau_2 v \|[q\,m_2]e\|^2 -\pi \tau_2^{-1}v \|[p+q\tau_1\ m_1+m_2\tau_1]e\|^2}e^{-2\pi i [p\,m_1]B[q\,m_2]^t}
\end{multline}
after the elements of $\Z^5$ are grouped together as an integer $p$ or $q$ and a vector  $m_1=[m_{12}\,m_{13}\,m_{14}\,m_{15}]$ or  $m_2=[m_{22}\,m_{23}\,m_{24}\,m_{25}] \in \Z^4$.
At this point identify the variables $b_1=u_5$, $b_2=u_6$, $b_3=u_7$, and $b_4=u_8$. Then $-[p\,m_1]B[q\,m_2]^t=-p(u_8m_{22}+u_7m_{23}+u_6m_{24}+u_5m_{25})+q(u_8m_{12}+ u_7m_{13}+u_6m_{14}+u_5m_{15})-m_1B'm_2^t$.
Hence the
 $[\srel{p\,m_1}{q\,m_2}]$ which contribute to the Fourier mode $(u_5,u_6,u_7,u_8)\mapsto e^{2\pi i (n_4u_5+n_3u_6+n_2u_7+n_1u_8)}$
  are those having  $pm_{22}-qm_{12}=-n_1$,
$pm_{23}-qm_{13}=-n_2$,
$pm_{24}-qm_{14}=-n_3$,
and
$pm_{25}-qm_{15}=-n_4$.
 This condition on the minors of the $2\times 5$ matrix $[\srel{p\,m_1}{q\,m_2}]$ is $SL(2,\Z)$-invariant.  Each $SL(2,\Z)$ orbit has an element with $q=0$ and $p>0$, at which the conditions simplify to
\begin{equation}\label{D5epsinepsc}
   pm_2 \ \ = \ \ p[m_{22}\,m_{23}\,m_{24}\,m_{25}] \ \ = \ \ -\, [n_1\,n_2\,n_3\,n_4] \ \ = \ \ -N \,,
\end{equation}
which cannot be zero because $[\srel{p\,m_1}{q\,m_2}]$ has rank 2.

For the rest of the paper we shall assume that  $N\neq 0$.
Under this  assumption   all contributions to the Fourier coefficient come from the second line of (\ref{D5EpsinEps1}).  Thus the terms in (\ref{bigG2bis}) which contribute to the Fourier coefficient can be written as
\begin{equation}\label{D5epsinepsda}
    \sum_{\g\,\in\,\G_\infty\backslash \G} {\mathcal G}^0_a(\g \tau,vee^t+B) \,,
\end{equation}
where
\begin{multline}\label{D5epsinepsdb}
    {\mathcal G}^0_a( \tau,vee^t+B) \ \ := \\
    \sum_{ \srel{pm_2\,=\,-N}{m_1 \,\in\,\Z^4}}e^{-\pi \tau_2 v \|[0\,m_2]e\|^2 -\pi \tau_2^{-1}v \|[p \ m_1+m_2\tau_1]e\|^2}e^{-2\pi i [p\,m_1]B[0\,m_2]^t}\,.
\end{multline}
Using the facts that $e^{-2\pi i [p\,m_1]B[0\,m_2]^t} =e^{2\pi i (n_4u_5+n_3u_6+n_2u_7+n_1u_8)}e^{-2\pi i m_1B'm_2^t}$ and $[0\,m_2]e= m_2e_4$
we now execute Poisson summation over $m_1\in \Z^4$ in (\ref{D5epsinepsdb}):
\begin{multline}\label{D5epsineps3}
    e^{-2\pi i (n_4u_5+n_3u_6+n_2u_7+n_1u_8)}{\mathcal G}^0_a( \tau,vee^t+B) \ \ = \\
    \sum_{ \srel{pm_2\,=\,-N}{\hat{m}_1 \,\in\,\Z^4}}e^{-\pi \tau_2 v \|m_2e_4\|^2}
    \int_{\IR^4}e^{2\pi i (m_2B'-\hat{m}_1)\cdot m_1}e^{-\pi
      \tau_2^{-1}v \|[p \ m_1+m_2\tau_1]e\|^2} \,dm_1 \\
    = \ \  \sum_{ \srel{pm_2\,=\,-N}{\hat{m}_1 \,\in\,\Z^4}}e^{2\pi i (\hat{m}_1-m_2B')\cdot m_2 \tau_1} e^{-\pi \tau_2  v \|m_2e_4\|^2} \ \times \\
    \int_{\IR^4}e^{2\pi i (m_2B'-\hat{m_1})\cdot m_1}e^{-\pi
      \tau_2^{-1}v \|[p \, m_1 ]e\|^2} \,dm_1
    \,.
\end{multline}
Again using the special form $e= \ttwo{1}{Q}{}{I_4}\ttwo{v^{-1/2}r^2}{}{}{ e_4}=\ttwo{v^{-1/2}r^2}{ Qe_4}{}{ e_4}$ (so
that $[p\,m_1]e=[  v^{-1/2}r^2 p \ \   (pQ+m_1)e_4]$), this equals

\begin{multline}\label{D5epsineps4}
    = \ \
    \sum_{ \srel{pm_2\,=\,-N}{\hat{m}_1 \,\in\,\Z^4}} e^{2\pi i (\hat{m}_1-m_2B')\cdot m_2 \tau_1} e^{-\pi \tau_2 v  \| m_2e_4\|^2}  \\ \times \
    \int_{\IR^4}e^{2\pi i (m_2B'-\hat{m}_1)\cdot m_1} e^{-\pi \tau_2^{-1}  p^2 r^{4}
    -\pi \tau_2^{-1}v  \|(pQ+m_1) e_4\|^2 } \,dm_1  \\
    = \ \
     \sum_{ \srel{pm_2\,=\,-N}{\hat{m}_1
         \,\in\,\Z^4}}e^{2\pi i (\hat{m}_1-m_2B')\cdot (m_2
       \tau_1+pQ)}e^{-\pi \tau_2 v  \| m_2e_4\|^2-\pi \tau_2^{-1}   p^2 r^{4} } \\
\times    \int_{\IR^4}e^{2\pi i (m_2B'-\hat{m_1})\cdot m_1}
   e^{ -\pi \tau_2^{-1} v  \|m_1 e_4\|^2 } \,dm_1 \\
    = \ \   \smallf{\tau_2^2 \,r^2 }{v^{5/2}}
     \sum_{ \srel{pm_2\,=\,-N}{\hat{m}_1 \,\in\,\Z^4}}e^{2\pi i (\hat{m}_1-m_2B') \cdot (m_2 \tau_1+pQ)}  \\
     \times e^{-\pi \tau_2 v   \| m_2e_4\|^2 -\pi \tau_2^{-1} p^2r^{4} -\pi \tau_2 v^{-1} \|(m_2B'-\hat{m}_1)(e_4^t)^{-1}\|^2}
\,,
\end{multline}
where we have used that $\det e_4=v^{1/2}r^{-2}$.

We now  use (\ref{D5epsineps4}) to determine the remaining Fourier dependence on $(u_1,u_2,u_3,u_4)$.
The dependence on $Q=[q_1\,q_2\,q_3\,q_4]$ here is from $e^{2\pi i (\hat{m}_1p-m_2B'p)Q^t}=e^{2\pi i (p\hat{m}_1+NB')Q^t}$.  Writing $\hat{m}_1=[\hat{m}_{12}\,
\hat{m}_{13}\,\hat{m}_{14}\,\hat{m}_{15}]$
the argument here is
$(p\hat{m}_{12}  -n_2b_7-n_3b_6-n_4b_5  )q_1 +
(p\hat{m}_{13}    +n_1b_7-n_3b_9-n_4b_8    )q_2 +
(p\hat{m}_{14}    +n_1b_6+n_2b_9-n_4b_{10}   )q_3 +
(p\hat{m}_{15}   +n_1b_5+n_2b_8+n_3b_{10}    )q_4$.
The character describing the Fourier mode  above was $ \chi_{N_1}(u)=\exp(2\pi i (m^1-n_4b_5-n_3b_6-n_2b_7)q_1+(m^2-n_4b_8-n_3b_9+n_1b_7)q_2+(m^3-n_4b_{10}+n_2b_9+n_1b_6)q_3
+(m^4+n_3b_{10}+n_2b_8+n_1b_5)q_4+n_4b_1+n_3b_2+n_2b_3+n_1b_4)$.  The condition that these match is thus that
$p\hat{m}_1=p[\hat{m}_{12}\,\hat{m}_{13}\,\hat{m}_{14}\,\hat{m}_{15}]=[m^1\,m^2\,m^3\,m^4]=M$. Then the
relevant Fourier coefficient
${\mathcal F}{\mathcal G}_a^0(\tau,v ee^t+B;[\srel MN])$ of ${\mathcal G}_a^0(\tau,v ee^t+B)$ is
\begin{equation}\label{D5epsineps5}
\aligned
 = \ \    \smallf{\tau_2^2 \,r^2 }{v^{5/2}}
     \sum_{ \srel{p\hat{m}_1\,=\,M}{pm_2\,=\,-N}} &e^{2\pi i
       \hat{m}_1\cdot m_2 \tau_1}  e^{-\pi \tau_2 v  \| m_2e_4\|^2}  \
     \times \\ & e^{ -\pi \tau_2^{-1} p^2 r^{4} -\pi \tau_2v^{-1}  \|(m_2B'-\hat{m}_1)(e_4^t)^{-1}\|^2} \\
     = \ \
     \smallf{\tau_2^2 \,r^2 }{v^{5/2}}  \sum_{p\,|\,\gcd(m^1,\ldots,n_4)}  &
e^{ -\pi \tau_2^{-1} p^2 r^{4}}
    e^{-2\pi i \, p^{-2}\, \tau_1 M \cdot N} \ \times \\ &
e^{-\pi\, p^{-2}\, \tau_2 v   \| N e_4\|^2-\pi\, p^{-2}\, \tau_2 v^{-1} \|(NB'+M)(e_4^t)^{-1}\|^2}
     \,,
    \endaligned
\end{equation}
the sum being over the positive common divisors $p$ of $m^1,\ldots,n_4$.

Finally, we insert (\ref{D5epsineps5}) into the second line of
(\ref{D5EpsinEps1}), and unfold to the strip.  In terms of (\ref{Fa1sD5a1ni}), this gives the Fourier coefficient  $F^{SO(5,5)\,\a_1}_{\a_1;s+3/2}$ at $h=h(r^2,h_4)$, where $h_4=\ttwo{I_4}{B'w_4}{}{I_4}\ttwo{v^{1/2} e_4}{}{}{v^{-1/2} \tilde{e}_4}$:
\begin{multline}\label{D5epsineps6}
    F^{SO(5,5)\,\a_1}_{\a_1;s+3/2}(h(r^2,h_4);[\srel{M}{N}]) \ \ = \\ \f{ \,r^{2}}{\xi(2s+3)} \int_{\G_\infty\backslash \U}
    E^{SL(2)}_s(\tau) \,d^2\tau
\sum_{p\,|\,\gcd(m^1,\ldots,n_4)}\!\!
    e^{-2\pi i \tau_1 \frac{M\cdot N}{p^2}}  \\ \times \  e^{ -\pi \tau_2^{-1} p^2 r^{4} -\pi\, p^{-2}\, \tau_2 v  \| N e_4\|^2-\pi\, p^{-2}\, \tau_2 v^{-1} \|(NB'+M)(e_4^t)^{-1}\|^2}.
\end{multline}
The matrix $e_4$ here is normalized differently than in  (\ref{PLRdefgeneral}), where it corresponds to the $SO(4,4)$ semisimple part of the Levi component.  In our setting that is instead $v^{1/2}e_4$, so that $G_4= ve_4e_4^t$.  Here $B'$ plays the role of the antisymmetric matrix $B$ and so  (\ref{PLRdefgeneral})  reads
\begin{equation}\label{PLRinappH}
    \aligned
 \sqrt{2}\, p_L \ \ & = \ \  v^{-1/2}(M \ + \ N B')  (e_4^t)^{-1} \ - \ v^{1/2} N   e_4                \\
 \sqrt{2}\,p_R \ \ & = \ \   v^{-1/2}(M \ + \ N B')
 (e_4^t)^{-1}  \ + \ v^{1/2}  N   e_4   \,.
    \endaligned
\end{equation}
It follows that
\begin{equation}\label{pl2mpr2}
   p_L^2\,+\,p_R^2\ \
= \ \ v^{-1} \| (M\ + \ N
  B') (e_4^t)^{-1}\|^2 \  + \ v \|  N   e_4  \|^2
\end{equation}
while
\begin{equation}\label{pl2mpr2b}
   p_L^2\,-\,p_R^2 \ \ = \ \ -2 \, (M \ + \ N
  B')  (e_4^t)^{-1} ( N  e_4)^t \ \
= \ \ -2 \, M\cdot N\,.
\end{equation}
With these substitutions and replacing $s$ by $s-3/2$, (\ref{SpinddandSoddseries})  and (\ref{D5epsineps6})  lead to~\eqref{so55modes}.
%



\begin{thebibliography}{MatOsh88}

\bibitem[ABV]{adams.barbasch.vogan} J.~Adams, D.~Barbasch,
  D. A. Vogan, Jr., {\it The Langlands classification and irreducible
  characters for real reductive groups}, Prog. in Math.,
 {\bf 104} (1992), Birkh\"auser (Boston).


\bibitem[Ar1]{arthur.old} J.~Arthur, {\it On some problems suggested by the
  trace formula}, in Lecture Notes in Math., {\bf 1024}(1983),
  Springer-Verlag (Berlin-Heidelberg-New York).

  \bibitem[Ar2]{arthur.conj} J.~Arthur, {\it Unipotent automorphic representations: conjectures}
  Ast\'erisque, {\bf 171} (1989), 13--71.

\bibitem[Ar3]{arthur.classical} J.~Arthur,
{\it The Endoscopic Classification of
Representations: Orthogonal and
Symplectic Groups}, preprint.


\bibitem[B]{barbasch.spherical} D.~Barbasch,
{\it The unitary spherical spectrum for split classical groups}
{J.~Inst.~Math.~Jussieu}, {\bf  9} (2010), no.~2, 265--356.

\bibitem[BV1]{barbasch.vogan}
D.~Barbasch and D. A. Vogan, Jr.,
{\it Unipotent representations of complex semisimple groups},
{Ann.~of Math.~(2)}, {\bf  121} (1985), no.~1, 41--110.

\bibitem[BM1]{barbasch.moy}
D.~Barbasch and A.~Moy,
{\it A unitarity criterion for p-adic groups,} Invent. Math. 98 (1989), no. 1, 19--37.


\bibitem[Bo]{borel}
A.~Borel,
{\it Admissible representations of a semi-simple group over a local field with vectors fixed under an Iwahori subgroup},
Invent. Math. 35 (1976), 233-- 259.

\bibitem[BoBr1]{bb:i}
W.~Borho, J.-L.~Brylinski,
{\it Differential operators on homogeneous spaces I.
Irreducibility of
the associated variety for annihilators of induced modules},
{Invent.~Math,} {\bf 69} (1982), no.~3, 437--476.

\bibitem[Ca]{carter} R.~Carter, {\it Finite groups of Lie type},
John Wiley and Sons (Chichester), 1993.

\bibitem[EM]{EM} S.~Evens, I.~Mirkovi\'c
{\it Fourier transform and the Iwahori-Matsumoto involution}, Duke Math. J. 86 (1997), no. 3, 435-- 464.

\bibitem[KL]{KL}
D.~Kazhdan, G.~Lusztig,
{\it Proof of the Deligne-Langlands conjecture for Hecke algebras},
Invent. Math. 87 (1987), no. 1, 153-- 215.

\bibitem[L]{lusztig}
G.~Lusztig,
{\it Graded Lie algebras and intersection cohomology}, Representation theory of algebraic groups and quantum groups, 191-- 224, Progr. Math., 284, Birkh\" auser/Springer, New York, 2010.


\bibitem[IM]{IM} N.~Iwahori, H.~Matsumoto
{\it On some Bruhat decomposition and the structure of the Hecke rings of p-adic Chevalley groups}, Inst. Hautes \'Etudes Sci. Publ. Math. No. 25 (1965) 5-- 48.

\bibitem[MW]{MW} C. Moeglin, J.-L. Waldspurger,
{\it Sur l'involution de Zelevinski}, J. Reine Angew. Math. 372 (1986), 136-- 177.

\bibitem[Sp]{springer} T.~Springer,
{\it Linear algebraic groups}, second edition, Birkh\"auser (Boston), 2009.

\bibitem[V1]{vogan.gln} D. A. Vogan, Jr.,
{\it The unitary dual of GL(n) over an Archimedean field},
Invent. Math. 83 (1986), no. 3, 449-- 505.

\bibitem[V2]{vogan.av} D. A. Vogan, Jr.,
{\it Associated varieties and unipotent representations}, in
{Progr. Math.}, {\bf 101} (1991), Birkh\"auser (Boston).

\bibitem[V2]{vogan.g2} D. A. Vogan, Jr.,
{\it The unitary dual of $G_2$}, Invent. Math. 116 (1994), no. 1-3, 677--791.





\end{thebibliography}

\begin{thebibliography}{99}



\bibitem{Green:2010kv}
  M.~B.~Green, S.~D.~Miller, J.~G.~Russo and P.~Vanhove,
  ``Eisenstein series for higher-rank groups and string theory amplitudes,''
 Commun.Num.Theor.Phys. {\bf 4} (2010) 551-596 [arXiv:1004.0163 [hep-th]].


\bibitem{Green:2010wi}
M.~B.~Green, J.~G.~Russo and P.~Vanhove,
``Automorphic Properties of Low Energy String Amplitudes in Various   Dimensions,''
Phys.\ Rev.\  D {\bf 81} (2010) 086008
[arXiv:1001.2535 [hep-th]].


\bibitem{Hull:1994ys}
C.~M.~Hull and P.~K.~Townsend,
``Unity of Superstring Dualities,''
Nucl.\ Phys.\ B {\bf 438} (1995) 109
[arXiv:hep-th/9410167].

\bibitem{Hull:2007}
  C.~M.~Hull,
  ``Generalised Geometry for M-Theory,''
  JHEP {\bf 0707}, 079 (2007)
  [hep-th/0701203].

\bibitem{Green:2010sp}
M.~B.~Green, J.~G.~Russo and P.~Vanhove,
``String Theory Dualities and Supergravity Divergences,''
JHEP {\bf 1006} (2010) 075
[arXiv:1002.3805 [hep-th]].

\bibitem{Green:2005ba}
  M.~B.~Green and P.~Vanhove,
  ``Duality and higher derivative terms in M theory,''
  JHEP {\bf 0601} (2006) 093
  [arXiv:hep-th/0510027].


\bibitem{Pioline:2010kb}
B.~Pioline,
``$ R^4$ Couplings and Automorphic Unipotent Representations,''
JHEP {\bf 1003} (2010) 116
[arXiv:1001.3647 [hep-th]].

\bibitem{GrossWI} B. H. Gross and N. R. Wallach, "On quaternionic discrete series representations, and their continuations," J. Reine Angew. Math. {\bf 481} (1996) 73–123.
\bibitem{GrossWII} B. H. Gross and N. R. Wallach, "A distinguished family of unitary representations for the exceptional groups of real rank = 4," in Lie theory and geometry, vol. {\bf 123} of Progr. Math., pp. 289–304. Birkh\"auser Boston, Boston, MA, 1994.

\bibitem{grs} D. Ginzburg, S. Rallis, and D. Soudry, ``On the automorphic theta representation for simply laced groups,'' Israel J. Math. 100 (1997) 61-116.




\bibitem{Green:1987mn}
M.~B.~Green, J.~H.~Schwarz and E.~Witten,
``Superstring Theory. Vol. 2: Loop Amplitudes, Anomalies and Phenomenology,''
{\it Cambridge, Uk: Univ. Pr. ( 1987) 596 P. ( Cambridge Monographs On Mathematical Physics)}



\bibitem{Green:2008bf}
M.~B.~Green, J.~G.~Russo and P.~Vanhove,
``Modular Properties of Two-Loop Maximal Supergravity and Connections with String Theory,''
JHEP {\bf 0807} (2008) 126
[arXiv:0807.0389 [hep-th]].

\bibitem{Collingwood} D. Collingwood and W. McGovern, ``Nilpotent orbits in semisimple Lie algebras'' Van Nostrand Reinhold Mathematics Series, New York, 1993.




\bibitem{joseph} A.~Joseph, ``On the associated variety of a primitive ideal'', Jour. of Algebra {\bf 93} (1985), pp. 509-523.


\bibitem{borhobryl} W.~Borho and J.-L.~Brylinski, ``Differential operators on homogeneous spaces, I'', Invent. Math. {\bf 69} (1982), pp. 437-476.


\bibitem{KazhdanSavin} D. Kazhdan and G. Savin, ``The smallest representation of simply laced groups'', Israel Mathematical Conference Proceedings, Vol. 3, I. Piatetski-Shapiro Festschrift, Weizmann Science Press of Israel, Jerusalem, 1990, pp. 209-233.


\bibitem{Kazhdan:2001nx}
  D.~Kazhdan, B.~Pioline, A.~Waldron,
  ``Minimal representations, spherical vectors, and exceptional theta series,''
  Commun.\ Math.\ Phys.\  {\bf 226 } (2002)  1-40.
  [hep-th/0107222].


  \bibitem{KazhdanPolishchuk:2002}
  D.~Kazhdan and A.~Polishchuk
  ``Minimal representations: spherical vectors and automorphic functionals,''
Algebraic groups and arithmetic,  127-198, Tata Inst. Fund. Res., Mumbai, 2004.
  [arXiv:math/0209315].


\bibitem{Neitzke}
A. Neitzke, private communication.

\bibitem{Townsend:1995kk}
  P.~K.~Townsend,
  ``The eleven-dimensional supermembrane revisited,''
  Phys.\ Lett.\  {\bf B350 } (1995)  184-187.
  [hep-th/9501068].

\bibitem{Becker:1995kb}
  K.~Becker, M.~Becker, A.~Strominger,
  ``Five-branes, membranes and nonperturbative string theory,''
  Nucl.\ Phys.\  {\bf B456 } (1995)  130-152.
  [hep-th/9507158].


\bibitem{Harvey:1999as}
  J.~A.~Harvey and G.~W.~Moore,
 ``Superpotentials and membrane instantons,''
  arXiv:hep-th/9907026.



\bibitem{Freed:2006yc}
  D.~S.~Freed, G.~W.~Moore, G.~Segal,
  ``Heisenberg Groups and Noncommutative Fluxes,''
  Annals Phys.\  {\bf 322 } (2007)  236-285.
  [hep-th/0605200].

\bibitem{Ferrara:1997ci}
S.~Ferrara and J.~M.~Maldacena,
``Branes, Central Charges and U-Duality Invariant BPS Conditions,''
Class.\ Quant.\ Grav.\  {\bf 15} (1998) 749
[arXiv:hep-th/9706097].


\bibitem{Ferrara:1997uz}
S.~Ferrara and M.~G{\"u}naydin,
``Orbits of Exceptional Groups, Duality and BPS States in String Theory,''
Int.\ J.\ Mod.\ Phys.\  A {\bf 13} (1998) 2075
[arXiv:hep-th/9708025].

\bibitem{Strominger:1996sh}
  A.~Strominger and C.~Vafa,
  ``Microscopic origin of the Bekenstein-Hawking entropy,''
  Phys.\ Lett.\  B {\bf 379} (1996) 99
  [arXiv:hep-th/9601029].

\bibitem{Callan:1996dv}
  C.~G.~Callan and J.~M.~Maldacena,
  ``D-brane approach to black hole quantum mechanics,''
  Nucl.\ Phys.\  B {\bf 472} (1996) 591
  [arXiv:hep-th/9602043].


\bibitem{Lu:1997bg}
H.~L{\"u}, C.~N.~Pope and K.~S.~Stelle,
``Multiplet Structures of BPS Solitons,''
Class.\ Quant.\ Grav.\  {\bf 15} (1998) 537
[arXiv:hep-th/9708109].

\bibitem{trombone}
E.~Cremmer, H.~Lu, C.~N.~Pope and K.~S.~Stelle,
``Spectrum generating symmetries for BPS solitons,''
Nucl.\ Phys.\ B {\bf 520}, 132 (1998)
[hep-th/9707207].



\bibitem{Popov}
 V.~L.~Popov, ``Classification of the spinors of dimension fourteen''
 Uspekhi Mat. Nauk {\bf 32} (1977) 1(193) 199--200

\bibitem{Trautman}  A.~Trautman  and  K.~Trautman, ``Generalized  pure
  spinors'' Journal of Geometry and Physics {\bf 15} (1994), 1-21.


\bibitem{Igusa}  J.~Igusa,   ``A  classification  of   spinors  up  to
  dimension twelve'' American Journal of Mathematics {\bf 92} (1970), 997-1028.


\bibitem{Miller-Sahi}  S.D~Miller and S.~Sahi, ``Fourier coefficients of automorphic forms, character variety orbits, and small representations'', Journal of Number Theory {\bf 132} (2012), 3070-3108 [arXiv:1202.0210].

\bibitem{bhargava}
M.~Bhargava, ``Higher composition laws I: A new view on Gauss composition, and quadratic generalizations'', Annals of Mathematics {\bf 159} (2004), pp. 217-250.


\bibitem{krutelevich}
S.~Krutelevich, ``Jordan algebras, exceptional groups, and Bhargava composition.''
   J. Algebra 314 (2007), no. 2, 924-977 [arXiv:math/0411104].


\bibitem{Green:1997tn}
M.~B.~Green and M.~Gutperle,
``D-Particle Bound States and the D-Instanton Measure,''
JHEP {\bf 9801} (1998) 005
[arXiv:hep-th/9711107].

\bibitem{bumpgl3}
D.~Bump, {\it Automorphic forms on ${\rm GL}(3,{\bf R})$},
   Lecture Notes in Mathematics {\bf 1083},
Springer-Verlag,
   Berlin, 1984.


\bibitem{psmult}
   I.~I.~Piatetski-Shapiro, ``Multiplicity one theorems'',
    in  {\it Automorphic forms, representations and $L$-functions} (Proc.
      Sympos. Pure Math., Oregon State Univ., Corvallis, Ore., 1977), Part
      1, Amer. Math. Soc.,
 Providence, R.I., 1979, pp.~209-212.


\bibitem{shalika}
   J.~A.~Shalika,  ``The multiplicity one theorem for ${\rm GL}\sb{n}$'',
    Ann. of Math. {\bf 100} (1974), pp. 171-193.



\bibitem{Pioline:2009qt}
  B.~Pioline, D.~Persson,
  ``The Automorphic NS5-brane,''
  Commun.\ Num.\ Theor.\ Phys.\  {\bf 3 } (2009)  697-754.
  [arXiv:0902.3274 [hep-th]].

\bibitem{Bao:2009fg}
  L.~Bao, A.~Kleinschmidt, B.~E.~W.~Nilsson, D.~Persson, B.~Pioline,
  ``Instanton Corrections to the Universal Hypermultiplet and Automorphic Forms on SU(2,1),''
  Commun.\ Num.\ Theor.\ Phys.\  {\bf 4 } (2010)  187-266.
  [arXiv:0909.4299 [hep-th]].

\bibitem{Persson:2011xi}
D.~Persson, ``Automorphic Instanton Partition Functions on Calabi-Yau Threefolds,''
[arXiv:1103.1014 [hep-th]].




\bibitem{Green:1997tv}
M.~B.~Green and M.~Gutperle,
``Effects of D-Instantons,''
Nucl.\ Phys.\  B {\bf 498} (1997) 195
[arXiv:hep-th/9701093].


\bibitem{Green:1998by}
M.~B.~Green and S.~Sethi,
``Supersymmetry Constraints on Type IIB Supergravity,''
Phys.\ Rev.\ D {\bf 59} (1999) 046006
[arXiv:hep-th/9808061].




\bibitem{Green:1998yf}
M.~B.~Green and M.~Gutperle,
``D-Instanton Partition Functions,''
Phys.\ Rev.\  D {\bf 58} (1998) 046007
[arXiv:hep-th/9804123].




\bibitem{Moore:1998et}
  G.~W.~Moore, N.~Nekrasov and S.~Shatashvili,
  ``D particle bound states and generalized instantons,''
  Commun.\ Math.\ Phys.\  {\bf 209} (2000) 77
  [arXiv:hep-th/9803265].



\bibitem{Green:1997di}
  M.~B.~Green and P.~Vanhove,
  ``D instantons, strings and M theory,''
  Phys.\ Lett.\  B {\bf 408} (1997) 122
  [arXiv:hep-th/9704145].

\bibitem{Basu:2007ru}
  A.~Basu,
  ``The D**4 R**4 term in type IIB string theory on T**2 and U-duality,''
  Phys.\ Rev.\  D {\bf 77}, 106003 (2008)
  [arXiv:0708.2950 [hep-th]].

\bibitem{Kiritsis:1997em}
  E.~Kiritsis and B.~Pioline,
  ``On R**4 threshold corrections in IIb string theory and (p, q) string
  instantons,''
  Nucl.\ Phys.\  B {\bf 508} (1997) 509
  [arXiv:hep-th/9707018].



\bibitem{Sugino:2001iq}
F.~Sugino and P.~Vanhove,
``U-Duality from Matrix Membrane Partition Function,''
Phys.\ Lett.\  B {\bf 522} (2001) 145
[arXiv:hep-th/0107145].


\bibitem{Obers:1999um}
N.~A.~Obers and B.~Pioline,
``Eisenstein Series and String Thresholds,''
Commun.\ Math.\ Phys.\ {\bf 209} (2000) 275
[arXiv:hep-th/9903113].


\bibitem{Angelantonj:2011br}
C.~Angelantonj, I.~Florakis and B.~Pioline,
``A New Look at One-Loop Integrals in String Theory,''
 Commun.Num.Theor.Phys. {\bf 6} (2012) 159-201
[arXiv:1110.5318 [hep-th]].



\bibitem{Matumoto} H.~Matumoto, ``Whittaker vectors and associate
  varieties,'' Invent. math. {\bf 89} 219-224 (1987)

\bibitem{richardson-rohrle-steinberg}
R.~Richardson, G.~R\"ohrle and R.~Steinberg,
``Parabolic subgroups with abelian unipotent radical''
Invent. Math. 110 (1992), no. 3, 649-671.



\bibitem{Rohrle} G. R\"ohrle, ``On extraspecial parabolic subgroups'', in  Linear algebraic groups and their representations (Los Angeles, CA, 1992), 143�155, Contemp. Math {\bf 153}, Amer. Math. Soc., Providence, RI, 1993.

\bibitem{moewal} C.~M{\oe}glin and J.-L.~Waldspurger, ``Mod\`eles de Whittaker d\'eg\'en\'er\'es pour des groupes
   $p$-adiques'', Math. Z. {\bf 196} 427-452 (1987).


\bibitem{Mill:2012}
S.~D.~Miller, ``Residual automorphic forms and spherical unitary
representations of exceptional groups,'' Annals of Mathematics (to
appear), arXiv:1205.0426 [math.NT].


\bibitem{GinzSayag} D.~Ginzburg and E.~Sayag, ``Construction of
  Certain Small Representations for SO(2m)'', preprint



\bibitem{langlandsSLN} R.~P.~Langlands, ``On the Functional Equations Satisfied by Eisenstein Series'',
Springer Lecture Notes in Mathematics {\bf 544}, 1965.

\bibitem{Fleig:2012xa}
P.~Fleig and A.~Kleinschmidt,
``Eisenstein Series for Infinite-Dimensional U-Duality Groups,''
JHEP {\bf 1206} (2012) 054
[arXiv:1204.3043 [hep-th]].


\bibitem{Duff:1994an}
M.~J.~Duff, R.~R.~Khuri and J.~X.~Lu,
``String Solitons,''
Phys.\ Rept.\ {\bf 259} (1995) 213
[hep-th/9412184].


\bibitem{Polchinski:1996na}
J.~Polchinski,
``Tasi Lectures on D-Branes,''
hep-th/9611050.

\bibitem{Polchinski:1998rr}
  J.~Polchinski,
  ``String theory. Vol. 2: Superstring theory and beyond,''
  Cambridge, UK: Univ. Pr. (1998) 531 p

\bibitem{Townsend:1997wg}
  P.~K.~Townsend,
  ``M theory from its superalgebra,''
  arXiv:hep-th/9712004.

\bibitem{Kallosh:1996uy}
R.~Kallosh and B.~Kol,
``$E_{7}$ Symmetric Area of the Black Hole Horizon,''
Phys.\ Rev.\  D {\bf 53} (1996) 5344
[arXiv:hep-th/9602014].



\bibitem{Balasubramanian:2006gi}
V.~Balasubramanian, E.~G.~Gimon and T.~S.~Levi,
``Four Dimensional Black Hole Microstates: from D-Branes to Spacetime Foam,''
JHEP {\bf 0801} (2008) 056
[arXiv:hep-th/0606118].


\bibitem{Berkooz:1996km}
  M.~Berkooz, M.~R.~Douglas, R.~G.~Leigh,
  ``Branes intersecting at angles,''
  Nucl.\ Phys.\  {\bf B480 } (1996)  265-278.
  [arXiv:hep-th/9606139 [hep-th]].

\bibitem{KorZol}
A.~Korkine and G.~Zolotareff, ``Sur les formes quadratiques'', Math. Ann., {\bf 6} (1873), 366-389.

\bibitem{Terraspaper}
A.~Terras, ``A generalization of Epstein's Zeta function,'' Nagoya Math. J. {\bf 42} (1971), 173-188.



\bibitem{Gubay:2010nd}
F.~Gubay, N.~Lambert and P.~West,
``Constraints on Automorphic Forms of Higher Derivative Terms from Compactification,''
JHEP {\bf 1008} (2010) 028
[arXiv:1002.1068 [hep-th]].

\bibitem{Basu:2011he}
A.~Basu,
``Supersymmetry Constraints on the R$^4$ Multiplet in Type IIB on T$^2$,''
arXiv:1107.3353 [hep-th].


\end{thebibliography}
\end{document}